\documentclass[11pt]{article}
\usepackage{titling}
\usepackage[OT1]{fontenc}
\usepackage{amsmath}
\usepackage{mathrsfs}
\usepackage{amsfonts}
\usepackage{amssymb}
\usepackage{diagbox}
\usepackage{footnote}
\usepackage{graphicx}
\usepackage{float}
\usepackage{epstopdf}
\usepackage[margin=1in]{geometry}
\usepackage{setspace}
\usepackage{tikz}
\usepackage{pgfplots}
\usepackage[round]{natbib}
\usepackage{multirow}
\usepackage[toc,page]{appendix}
\usepackage{hyperref}
\usepackage{caption}
\usepackage{subcaption}
\usepackage{enumitem}
\usepackage{booktabs}
\usepackage{bibunits}
\defaultbibliography{reference1}
\defaultbibliographystyle{apalike}
\usepackage{dsfont}
\usepackage{bbm}
\usepackage{titletoc}
\usepackage{pifont}
\newcommand{\cmark}{\ding{51}}
\newcommand{\xmark}{\ding{55}}
\allowdisplaybreaks

\usepackage{scalerel,stackengine}
\stackMath
\newcommand\reallywidehat[1]{%
	\savestack{\tmpbox}{\stretchto{%
			\scaleto{%
				\scalerel*[\widthof{\ensuremath{#1}}]{\kern-.6pt\bigwedge\kern-.6pt}%
				{\rule[-\textheight/2]{1ex}{\textheight}}
			}{\textheight}%
		}{0.5ex}}%
	\stackon[1pt]{#1}{\tmpbox}%
}
\definecolor{gblue}{rgb}{0.1, 0.6, 3.0}

\providecommand{\U}[1]{\protect\rule{.1in}{.1in}}
\providecommand{\keywords}[1]{\textsc{Keywords}: #1}

\hypersetup{
colorlinks=true,
linkcolor=gblue,
citecolor=gblue,
filecolor=gblue,
urlcolor=gblue,
breaklinks=true,
}
\newtheorem{theorem}{Theorem}[section]
\newtheorem{assumption}{Assumption}[section]

\newtheorem{corollary}{Corollary}[section]

\newtheorem{definition}{Definition}[section]
\newtheorem{example}{Example}[section]

\newtheorem{lemma}{Lemma}[section]

\newtheorem{proposition}{Proposition}[section]
\newtheorem{remark}{Remark}[section]

\newenvironment{proof}[1][Proof]{\noindent\textbf{#1.} }{\ \rule{0.5em}{0.5em}}
\newcommand{\nci}{\not\!\perp\!\!\!\perp}
\newcommand{\ci}{\perp\!\!\!\perp}

\numberwithin{equation}{section}
\numberwithin{table}{section}
\numberwithin{figure}{section}

\pgfplotsset{compat=1.18}
\begin{document}

\begin{bibunit}

\title{Pairwise Valid Instruments}
\author{Zhenting Sun\thanks{Correspondence to: Department of Economics, University of Melbourne, 
Grattan Street, Parkville Victoria 3010,
Australia. 
E-mail addresses: zhentingsun@gmail.com (Z. Sun), kasparwu@umich.edu (K. W{\"u}thrich).}\\Department of Economics\\ University of Melbourne\and Kaspar W{\"u}thrich\\Department of Economics\\ University of Michigan\\} 
\maketitle

\onehalfspacing
\begin{abstract}
Finding valid instruments is difficult. We propose Validity Set Instrumental Variable (VSIV) estimation, a method for estimating local average treatment effects (LATEs) in heterogeneous causal effect models when the instruments are partially invalid. We consider settings with pairwise valid instruments, that is, instruments that are valid for a subset of instrument value pairs. VSIV estimation exploits testable implications of instrument validity to remove invalid pairs and provides estimates of the LATEs for all remaining pairs, which can be aggregated into a single parameter of interest using researcher-specified weights. We show that the proposed VSIV estimators are asymptotically normal under weak conditions and remove or reduce the asymptotic bias relative to standard LATE estimators (that is, LATE estimators that do not use testable implications to remove invalid variation). 
We evaluate the finite sample properties of VSIV estimation in application-based simulations and apply our method to estimate the returns to college education using parental education as an instrument.

\bigskip

\end{abstract}

\bigskip

\newpage

\section{Introduction}

Instrumental variable (IV) methods based on the local average treatment effect (LATE) framework \citep{imbens1994identification,angrist1995two,angrist1996identification} rely on three assumptions:\footnote{See, for example, \citet{imbens2014instrumental}, \citet{melly2017local}, and \citet{huber2018local} for recent reviews, and \citet{angrist2008mostly}, \citet{angrist2014mastering}, and \citet{imbens2015causal} for textbook treatments.} (i) \emph{exclusion} (the instrument does not have a direct effect on the outcome), (ii) \emph{random assignment} (the instrument is independent of potential outcomes and treatments), and (iii) \emph{monotonicity} (the instrument has a monotonic impact on treatment take-up).\footnote{Some papers also include the instrument first-stage assumption as part of the LATE assumptions. We will maintain suitable first-stage assumptions.} In many applications, some of these assumptions are likely to be violated or at least questionable. This has motivated the derivation of testable restrictions and tests for IV validity in various settings \citep[e.g.,][]{balke1997bounds,imbens1997estimating,heckman2005structural,huber2015testing,kitagawa2015test,mourifie2016testing,kedagni2020generalized,carr2021testing,farbmacher2022instrument,frandsen2023judging,jiang2023testing,sun2021ivvalidity}.\footnote{There is a related literature on inference with invalid instruments in linear IV models \citep[e.g.,][]{conley2012plausibly,nevo2012identification,armstrong2021sensitivity,goh2022causal}.} The main contribution of this paper is to propose a method for exploiting the information available in the testable restrictions of IV validity to remove or reduce the asymptotic bias when estimating LATE parameters.\footnote{We define the asymptotic bias as the probability limit of the $\ell^2$ difference between an estimator and the true value.}

We consider settings where the available instruments are partially invalid. A leading example of such a setting is when there is a multivalued instrument for which only some pairs of instrument values satisfy the IV assumptions. In Section \ref{sec.application}, we revisit the analysis of the causal effect of college education on earnings using parental education as an instrument. This instrument is likely partially invalid due to parental education having a positive effect on future earnings, at least up to a certain education level \citep{kedagni2020discordant}. 
Another example is the quarter of birth (QOB) instrument of \citet{angrist1991does}. A potential concern with this instrument is that the seasonality in birth patterns renders the QOB instrument partially invalid \citep[e.g.,][]{bound1995problems,buckles2013season}, motivating some studies to only consider a subset of QOBs as instruments \citep[e.g.,][]{dahl2017s}. Another empirically relevant setting where partially invalid instruments may arise is when there are multiple instruments.\footnote{Settings with multiple instruments are common in empirical research \citep[][Section I]{mogstad2021causal}.} In applications with multiple instruments, the validity of a subset of the instruments may be questionable, or the instruments may be partially invalid because the heterogeneity in individual choice behavior renders standard monotonicity assumptions invalid \citep{mogstad2021causal}.  As an example of the latter, consider the study by \citet{thornton2008demand}, who estimates the causal effect of knowing HIV status on the likelihood of buying condoms using two randomly assigned instruments: Monetary incentives and distance to results centers. In this application, monotonicity is likely to fail due to differences in individual preferences over monetary incentives and distance \citep[][Online Appendix B.1]{mogstad2021causal}.\footnote{\citet{mogstad2021causal} propose a weaker version of monotonicity, referred to as partial monotonicity, that they argue is more plausible in this application.  We discuss the connection between our assumptions and partial monotonicity in Section \ref{sec.partial monotnicity}. \citet{jiang2023testing} develop formal tests for partial monotonicity and apply these tests in the context of the \citet{thornton2008demand} application.}

The proposed method, which we refer to as \emph{Validity Set IV (VSIV) estimation}, uses testable implications of IV validity to remove invalid variation in the instruments and provides LATE estimates based on the remaining variation in the instruments. We establish the asymptotic normality of the proposed VSIV estimators and show that they always remove or reduce the asymptotic bias relative to standard LATE estimators, that is, LATE estimators that do not exploit testable implications of IV validity to remove invalid pairs. Thus, VSIV estimation constitutes a data-driven approach for removing or reducing the asymptotic bias of LATE estimators, given all the information about IV validity available in the data. 

The use of the testable implications of IV validity in VSIV estimation is more constructive than the standard practice where researchers first test for IV validity, discard the instruments if they reject IV validity, and proceed with standard IV analyses if they do not reject IV validity. VSIV estimation uses the testable implications to remove invalid information in the instruments. Consequently, it can be used to estimate causal effects in settings where the instruments are only partially invalid so that existing tests reject the null of full IV validity.\footnote{See Appendix \ref{app:pretest} for a comparison between VSIV estimation and pairwise pretests based on existing tests for IV validity.} VSIV estimation salvages falsified instruments by exploiting the variation in the instruments not refuted by the data and thereby contributes to the literature on salvaging falsified models \citep[e.g.,][]{masten2021salvaging,kedagni2020discordant}.

Our goal is to estimate the causal effect of an endogenous treatment $D$ on an outcome of interest $Y$, using a potentially vector-valued discrete instrument $Z$. We consider binary treatments in the main text and multivalued ordered or unordered treatments in the Appendix. In the ideal case, $Z$ is fully valid, that is, the LATE assumptions hold for all instrument values (the instrument is valid for the whole population). However, full IV validity is questionable in many applications, especially when there are many instruments or instrument values. To this end, we introduce the notion of \emph{pairwise valid instruments}. Pairwise valid instruments are only valid for a subset of all pairs of instrument values, which we refer to as \emph{validity pair set}. Intuitively, the instruments are valid for some subpopulations but invalid for the others. 

Pairwise validity separates the instrument value pairs into two groups: Valid pairs for which all LATE assumptions hold and invalid pairs for which at least one of the LATE assumptions is violated. Pairwise validity does not require researchers to specify which LATE assumptions are violated for the invalid pairs and to which extent; it allows for failures of exclusion, random assignment, monotonicity, or combinations thereof. As a result, there is no information about the LATE for the invalid instrument value pairs absent additional restrictions (see Appendix \ref{sec.no information}). Pairwise validity is motivated by the fact that it is often difficult to determine why exactly specific instrument pairs are invalid based on contextual knowledge (that is, which combinations of assumptions are violated and how), especially when there are many potentially invalid pairs. If additional information is available on which LATE assumptions fail and how, we can exploit such information for partial identification and sensitivity analysis \citep[e.g.,][]{huber2014sensitivity,noack2021sensitivity,kedagni2023identifying,cui2024robust} or focus on target parameters that are identified under relaxations of the LATE assumptions \citep[e.g.,][]{de2017tolerating,frandsen2023judging}. Even in settings where such information is available, pairwise validity provides a useful benchmark and starting point.

VSIV estimation provides estimates of the LATEs for all pairs of instrument values that satisfy the testable restrictions for IV validity. Specifically, we obtain an estimator, $\widehat{\mathscr{Z}_0}$, of the set of pairs of instrument values that satisfy the testable restrictions in \citet{kitagawa2015test}, \citet{mourifie2016testing}, \citet{kedagni2020generalized}, and \citet{sun2021ivvalidity} and estimate LATEs for all pairs of instrument values in  $\widehat{\mathscr{Z}_0}$. These LATEs can then be aggregated into a single parameter of interest based on user-specified weights.

We study the theoretical properties of VSIV estimation under two scenarios. First, we assume that the estimated validity pair set, $\widehat{\mathscr{Z}_0}$, is consistent for the largest validity pair set $\mathscr{Z}_{\bar{M}}$ (that is, the union of all validity pair sets) in the sense that $\mathbb{P}(\widehat{\mathscr{Z}_0}=\mathscr{Z}_{\bar{M}})\to 1$. In this case, VSIV estimation is asymptotically unbiased and normal under standard conditions. Second, since the estimator of the validity pair set,  $\widehat{\mathscr{Z}_0}$, is typically constructed based on necessary (but not necessarily sufficient) conditions for IV validity, it could converge to a \emph{pseudo-validity pair set} $\mathscr{Z}_{0}$ that is larger than $\mathscr{Z}_{\bar{M}}$, that is, $\mathbb{P}(\widehat{\mathscr{Z}_0}=\mathscr{Z}_{0})\to 1$.\footnote{\citet[][Proposition 1.1]{kitagawa2015test} shows that there exist no sufficient conditions for IV validity when $D$ and $Z$ are both binary.} Let $\mathscr{Z}_P$ be a presumed set of valid pairs of instrument values, incorporating prior information about instrument validity ($\mathscr{Z}_P$ is equal to the set of all pairs if no such prior information is available). We prove that VSIV estimation based on $\widehat{\mathscr{Z}_0}\cap \mathscr{Z}_P$ leads to a smaller asymptotic bias than standard LATE estimators based on $\mathscr{Z}_P$. Taken together, our theoretical results show that, irrespective of whether the largest validity pair set can be estimated consistently or not, VSIV estimation leads to asymptotically normal LATE estimators with reduced asymptotic bias.

Finally, we use VSIV estimation to revisit the estimation of the causal effect of college education on earnings using parental education as an instrument. We evaluate the finite sample performance of VSIV estimation in a simulation study calibrated to this application and use these simulations to determine the choice of the tuning parameter required for VSIV estimation. Based on this choice of tuning parameter, VSIV estimation screens out the pairs of instrument values corresponding to low levels of parental education. This is consistent with the above discussion and the findings in \citet{kedagni2020discordant} in that for low levels of parental education the exclusion restriction may fail. The LATEs for the pairs of instrument values that are not screened out are positive and significant.

\paragraph{Notation:} We introduce some standard notation, following \citet{sun2021ivvalidity}. All random elements are defined on a probability space $(\Omega,\mathcal{A},\mathbb{P})$. For all $m\in\mathbb{N}$,  $\mathcal{B}_{\mathbb{R}^m}$ is the Borel $\sigma$-algebra on $\mathbb{R}^m$. 
We denote by $\mathcal{P}$ the set of probability measures such that if the data $\{(Y_i,D_i,Z_i)\}_{i=1}^n$ are i.i.d.\ and distributed according to some probability measure $Q\in\mathcal{P}$, then $Q(G)=\mathbb{P}((Y_i,D_i,Z_i)\in G)$ for all measurable sets $G$.
For every $Q\in\mathcal{P}$ and every measurable function $v$, with some abuse of notation, we define $Q(v)=\int v\, \mathrm{d}Q$.
The symbol $\leadsto$ denotes weak convergence in a metric space in the Hoffmann--J\o rgensen sense. 
	For every set $B$, let $1_B$ denote the indicator function for $B$.
Finally, to simplify the exposition of the theoretical results, 
we adopt the convention \citep[e.g.,][p.~45]{folland2013real},
that 
\begin{equation}\label{eq.0timesinfinity}
	0\cdot \infty =0.
\end{equation}

\section{Identification with Pairwise Valid Instruments}\label{sec.pairwise valid instrument binary D}

\subsection{Weakening Instrument Validity to Pairwise Validity}

Consider a setting with an outcome variable $Y\in \mathbb{R}$, a treatment $D\in \mathcal{D}$, and an instrument (vector) $Z\in \mathcal{Z}$. In the main text, we focus on the leading case where the treatment is binary with $\mathcal{D}=\left\{0,1\right\}$. The extensions to multivalued ordered and unordered treatments can be found in the Appendix. The instrument is discrete with $\mathcal{Z}=\left\{z_{1},\ldots,z_K\right\} $, and can be ordered or unordered. Let $Y_{dz}\in\mathbb{R}$ for $(d,z)\in\mathcal{D}\times \mathcal{Z}$ denote the potential outcomes and
let $D_{z}$ for $z\in \mathcal{Z}$ denote the  potential treatments. The following assumption generalizes the standard LATE assumptions with binary instruments to multivalued instruments. 

\begin{assumption}\label{ass.IV validity binary D}
IV validity for	LATEs with binary treatments and multivalued instruments: 
	\begin{enumerate}[label=(\roman*)]
\item Exclusion: For each $d\in\{0,1\}$, $Y_{dz_{1}}=\cdots=Y_{dz_{K}}$ almost surely (a.s.). \label{ass.IV validity binary D exclusion}		
  
  \item Random Assignment: $Z$ is jointly independent of $\left( Y_{0z_{1}},\ldots,Y_{0z_K}, Y_{1z_{1}
			},\ldots,Y_{1z_K}\right)$ and \\$\left(D_{z_{1}},\ldots,D_{z_{K}}\right)$. \label{ass.IV random assignment}

		\item Monotonicity: For all  $
		k\in\{1,\ldots, K-1\}$, $D_{z_{k+1}}\geq D_{z_k}$ a.s. \label{ass.IV validity binary D monotonicity}
	\end{enumerate}
	
\end{assumption}

Assumption \ref{ass.IV validity binary D} is similar to the LATE assumptions in, for example, \citet{imbens1994identification}, \citet{angrist1995two}, \citet{frolich2007nonparametric}, \citet{kitagawa2015test}, and \citet{sun2021ivvalidity}. It imposes the IV validity assumptions with respect to all possible values of the instrument $z\in \mathcal{Z}$. This assumption has a lot of identifying power: It identifies LATEs with respect to every pair of IV values $(z_{k},z_{k+1})$ with $\mathbb{P}(D_{z_{k+1}}>D_{z_{k}})>0$. However, Assumption \ref{ass.IV validity binary D} can be restrictive in applications. We therefore introduce the notion of \emph{pairwise instrument validity}, which weakens the conditions in Assumption \ref{ass.IV validity binary D}. Define the set of all possible pairs of values of $Z$ as
\begin{align*}
    \mathscr{Z}=\left\{  \left(  z_{1},z_{2}\right)  ,\ldots,\left(z_{1},z_{K}\right),\left(  z_{2},z_{3}\right)  ,\ldots,\left(z_{2},z_{K}\right)  ,\ldots,(z_{K-1},z_K),(z_{2},z_1),\ldots,\left(
z_{K},z_{K-1}\right)  \right\}.
\end{align*} 
The number of the elements in $\mathscr{Z}$
is $K\cdot\left(  K-1\right)  $. We use $\mathcal{Z}_{\left(  k,k^{\prime	}\right)  }$ to denote a pair $\left(  z_{k},z_{k^{\prime}}\right)  \in
\mathscr{Z}$. Note that we include both $(z_k,z_{k'})$ and $(z_{k'},z_k)$ in $\mathscr{Z}$ so that we do not restrict the direction of the monotonicity (Assumption \ref{ass.IV validity binary D}\ref{ass.IV validity binary D monotonicity}) a priori. 

\begin{definition}
	\label{def.partial validity pairwise binary D} The instrument $Z$ is \textbf{pairwise valid} for
	the treatment $D\in\{0,1\}$ if there is a set $\mathscr{Z}_M=\{(z_{k_1},z_{k_1^{\prime}}),\ldots,(z_{k_M},z_{k_M^{\prime}})\}\subseteq\mathscr{Z}$ such that the following conditions hold for every $(z,z')\in\mathscr{Z}_M$:
	\begin{enumerate}[label=(\roman*)]
		
		\item Exclusion: For each $d\in\{0,1\}$, $Y_{dz}=Y_{dz^{\prime}}$ a.s.\label{def.pairwise exclusion}
		
		\item Random Assignment: $Z$ is jointly independent of $(Y_{0z},Y_{0z^{\prime}},Y_{1z},Y_{1z^{\prime}},D_{z},D_{z^{\prime}}) $.\footnote{This condition can be further weakened: The conditional distribution of  $(Y_{0z},Y_{0z^{\prime}},Y_{1z},Y_{1z^{\prime}},D_{z},D_{z^{\prime}}) $ given $Z=z$ or $Z=z'$ is the same as the unconditional distribution.} 
		\label{def.pairwise random assignment}
		\item Monotonicity: $D_{z^{\prime}}\geq D_{z}$ a.s. \label{def.pairwise monotonicity}
	\end{enumerate}
	The set $\mathscr{Z}_M$ is called a \textbf{validity pair set} of $Z$.\footnote{We use $\mathscr{Z}_M$ to denote an arbitrary validity pair set throughout the paper. To simplify the notation, we therefore only index $\mathscr{Z}$ by $M$ and not by the full index set $\{(k_1,k_1'),\dots,(k_M,k_M')\}$.} The union of all validity pair sets is the largest validity pair set, denoted by $\mathscr{Z}_{\bar{M}}$. A pair of instrument values $(z,z')$ is called a \textbf{valid pair} if $(z,z')\in\mathscr{Z}_{\bar{M}}$. A pair $(z,z')$ is called an \textbf{invalid pair} if $(z,z')\notin\mathscr{Z}_{\bar{M}}$.
\end{definition}

Definition \ref{def.partial validity pairwise binary D} separates the instrument value pairs into two groups: Valid pairs for which all the LATE assumptions hold and invalid pairs for which the LATE assumptions fail due to failures of exclusion, independence, monotonicity, or combinations thereof. We show in Appendix \ref{sec.no information} that absent additional restrictions on how Definition \ref{def.partial validity pairwise binary D} can be violated, the sharp identified set for the LATEs for the invalid pairs is the entire real line; that is, there is no information about the LATEs for the invalid pairs in the data. 

Pairwise validity does not require researchers to specify which LATE assumptions are violated and to which extent for the invalid pairs. The motivation for this is that it can be difficult to determine why exactly instrument pairs are invalid in applications. While pairwise validity does not require researchers to impose additional assumptions for the invalid pairs, it allows for the possibility that valid pairs restrict which LATE assumptions are violated for invalid pairs.\footnote{For example, suppose that $\mathcal{Z}= \{z_1,z_2,z_3\}$, where the pairs $(z_{1},z_3)$ and $(z_{2},z_3)$ are valid. This configuration implies that exclusion (Definition \ref{def.partial validity pairwise binary D}\ref{def.pairwise exclusion}) cannot be violated for the pair $(z_{1},z_2)$.}

Definition \ref{def.partial validity pairwise binary D} complements the existing approaches for relaxing the LATE assumptions. These approaches typically impose more structure on which LATE assumptions are violated and how exactly the LATE assumptions are violated. They provide partial identification results and methods for performing sensitivity analyses \citep[e.g.,][]{huber2014sensitivity,noack2021sensitivity,kedagni2023identifying,cui2024robust} or focus on target parameters that are identified under weaker assumptions \citep[e.g.,][]{de2017tolerating,frandsen2023judging}. 

To illustrate Definition \ref{def.partial validity pairwise binary D}, consider a simple example where $Z\in\mathcal{Z}=\{z_1,z_2,z_3\}$. If $Z$ is fully valid as in Assumption \ref{ass.IV validity binary D} such that $D_{z_3}\ge D_{z_2}\ge D_{z_1}$ a.s., then $\mathscr{Z}_{\bar{M}}=\{(z_1,z_2),(z_1,z_3),(z_2,z_3)\}$. The orange solid lines in Figure \ref{fig:pairwise IV}(a) indicate that two instrument values, $\{z_k,z_{k'}\}$, form a valid pair: Either $(z_k,z_{k'})$ or $(z_{k'},z_k)$ satisfies the conditions in Definition \ref{def.partial validity pairwise binary D}. The full validity Assumption \ref{ass.IV validity binary D} requires that every pair of instrument values forms a valid pair. Definition \ref{def.partial validity pairwise binary D} relaxes Assumption \ref{ass.IV validity binary D} as it does not require every pair to form a valid pair. For example, it could be that only $(z_1,z_3)$ satisfies the conditions in Definition \ref{def.partial validity pairwise binary D}. The teal dashed lines in Figure \ref{fig:pairwise IV}(b) indicate that $\{z_1,z_2\}$ and $\{z_2,z_3\}$ do not form valid pairs. In this case, the instrument $Z$ is pairwise but not fully valid.  

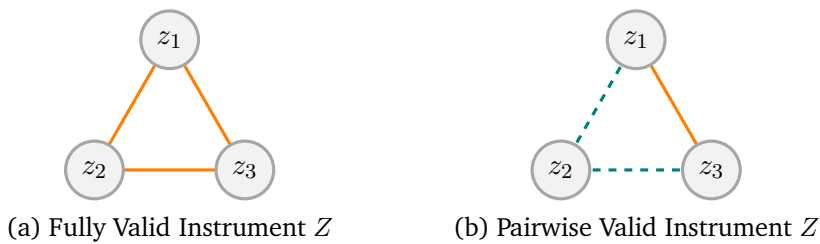
\begin{figure} [H]
	\caption{Full IV Validity vs.\ Pairwise IV Validity}
	\label{fig:pairwise IV}
	\centering
	\begin{subfigure}[b]{0.4\textwidth}
		\centering
		\begin{tikzpicture}
           \draw [orange, very thick, -] (0,0) -- (1,1.732); 
	       \draw [orange, very thick, -] (2,0) -- (1,1.732);
           \draw [orange, very thick, -] (0,0) -- (2,0);
           \node[circle, draw=gray!70, fill=gray!10, very thick, minimum size=7mm] at (0,0) {$z_2$};
           \node[circle, draw=gray!70, fill=gray!10, very thick, minimum size=7mm] at (1,1.732) {$z_1$};
	       \node[circle, draw=gray!70, fill=gray!10, very thick, minimum size=7mm] at (2,0) {$z_3$};
       \end{tikzpicture}
		\subcaption{Fully Valid Instrument $Z$}
	\end{subfigure}
	\begin{subfigure}[b]{0.4\textwidth}
		\centering
		\begin{tikzpicture}
           \draw [teal, very thick, dashed,-] (0,0) -- (1,1.732); 
           \draw [orange, very thick, -] (2,0) -- (1,1.732);
           \draw [teal, very thick, dashed, -] (0,0) -- (2,0);
           \node[circle, draw=gray!70, fill=gray!10, very thick, minimum size=7mm] at (0,0) {$z_2$};
           \node[circle, draw=gray!70, fill=gray!10, very thick, minimum size=7mm] at (1,1.732) {$z_1$};
	       \node[circle, draw=gray!70, fill=gray!10, very thick, minimum size=7mm] at (2,0) {$z_3$};
    \end{tikzpicture}
		\subcaption{Pairwise Valid Instrument $Z$}
	\end{subfigure}
\end{figure}

In applications where the instrument $Z$ is randomly assigned (e.g., in experiments with imperfect compliance), joint independence (Assumption \ref{ass.IV validity binary D}\ref{ass.IV random assignment}) holds by design. In such applications, Definition \ref{def.partial validity pairwise binary D} captures violations of exclusion and monotonicity. Such violations are easy to interpret. The pairwise exclusion assumption in Definition \ref{def.partial validity pairwise binary D}\ref{def.pairwise exclusion} requires that $Y_{dz}$, viewed as a function of $z$, is constant over some regions of $\mathcal{Z}$ and varies over others. This nests, for example, Condition E.3 in \citet{kedagni2020discordant}, which requires that $Y_{dt}\le Y_{dt'}$ for all $t\le t'$ and $Y_{dt}=Y_{dz}$ for all $t\ge z$. The pairwise monotonicity assumption (Definition \ref{def.partial validity pairwise binary D}\ref{def.pairwise monotonicity}) requires that $D_z$, viewed as a function of $z$, is monotonic over some regions of $\mathcal{Z}$ and non-monotonic over others. We discuss the relationship to existing relaxations of LATE monotonicity in more detail in Section \ref{sec.partial monotnicity}.

In many quasi-experimental applications, the instrument $Z$ is not randomly assigned, and joint independence (Assumption \ref{ass.IV validity binary D}\ref{ass.IV random assignment}) may fail. A leading and practically relevant case where joint independence fails but pairwise independence (Definition \ref{def.partial validity pairwise binary D}\ref{def.pairwise random assignment}) holds is when there are multiple instruments, and some of them are not independent of all potential variables.\footnote{It is possible that joint independence fails but pairwise independence holds even if there is only one original instrument. To illustrate, let $D$ indicate college enrollment, and let $Z\in \{1,2,3\}$ measure distance to the closest college \citep[e.g.,][]{kane1993labor}, where $Z=1$  indicates \texttt{close}, $Z=2$ indicates \texttt{far}, and $Z=3$ indicates \texttt{very far}. Consider the selection mechanism $D_z=1\left\{B_0+B_11\{z=1\}+f(z)1\{z>1\}\le 0\right\}$, where $B_0$ and $B_1$ are random coefficients, $f(2)<f(3)$, and $B_0\ci Z$. The coefficient $B_1$ captures the taste for living \texttt{close} to college relative to living farther away. If $B_1$ is correlated with actual distance $Z$, then $Z\nci D_1$ and $Z\ci (D_2,D_3)$.} To illustrate, consider the following example based on \citet[][Section II.C]{mogstad2021causal} and the empirical application in \citet{carneiro2011estimating}. Let $D$ be an indicator for college attendance. There are two binary instruments, $Z=(Z_1,Z_2)$, where $Z_1$ is an indicator for college proximity \citep[e.g.,][]{card1993geographic,kane1993labor} and $Z_2$ is an indicator for tuition subsidy.\footnote{We swap the order of the instruments relative to \citet[][Section II.C]{mogstad2021causal} for the purpose of illustration.} Individuals decide whether to attend college based on the following selection mechanism,
\begin{equation}
D_z=1\left\{B_0+B_1z_1+z_2\ge 0\right\},
\end{equation}
where $B_0$ and $B_1$ are random coefficients. For simplicity, we assume that $Z\ci B_0$. The coefficient $B_1$ measures the ``taste'' for proximity relative to tuition subsidy. If the taste for proximity, $B_1$, is correlated with actual proximity, $Z_1$, for example, due to spatial sorting, then the pair $(D_{(0,0)},D_{(0,1)})$ is independent of $Z$ but the pair $(D_{(1,0)},D_{(1,1)})$ is not.\footnote{See, for example, \citet{card1993geographic}, \citet{carneiro2011estimating}, \citet{slichter2014testing}, and \citet{kitagawa2015test} for discussions of the validity of the college proximity instrument.}

\begin{remark}[Weakening Definition \ref{def.partial validity pairwise binary D} with Multiple Instruments]
In Appendix \ref{sec.selectively valid instruments}, we introduce a weaker notion of pairwise validity (Definition \ref{def.partial validity pairwise binary D}) for settings where $Z$ contains multiple instruments:  $Z=(Z_1,\ldots,Z_L)^T$, where $Z_l$ is a scalar instrument for $l\in\{1,\ldots,L\}$.
\end{remark}

\subsection{Relationship to Other Variants and Relaxations of Monotonicity} 
\label{sec.partial monotnicity}
Here we discuss the connection between our pairwise monotonicity assumption (Definition \ref{def.partial validity pairwise binary D}\ref{def.pairwise monotonicity}) and three recently proposed variants and relaxations of the LATE monotonicity assumption.

First, \citet{mogstad2021causal} propose a partial monotonicity (PM) condition for settings with multiple instruments, which is a special case of Condition \ref{def.pairwise monotonicity} in Definition \ref{def.partial validity pairwise binary D}; see also \citet{goff2020vector} for vector monotonicity assumption.\footnote{\citet{mogstad2021causal} motivate the PM condition by showing that full monotonicity imposes strong restrictions on the heterogeneity in individual choice behavior and is therefore likely violated in many applications.}
For example, suppose that $Z=(Z_1,Z_2)\in\mathbb{R}^2$ and each element of $Z$ is binary so that $\mathcal{Z}=\{(0,0),(0,1),(1,0),(1,1)\}$. Suppose that Assumption PM of \citet{mogstad2021causal} holds with $D_{(0,0)}\ge D_{(0,1)}$, $D_{(0,0)}\ge D_{(1,0)}$, $D_{(1,1)}\ge D_{(0,1)}$, and $D_{(1,1)}\ge D_{(1,0)}$ a.s. (the sex composition instrument in \citet{angrist1998children} discussed in \citet{mogstad2021causal}), and that Conditions \ref{def.pairwise exclusion} and \ref{def.pairwise random assignment} of Definition \ref{def.partial validity pairwise binary D} hold. Then a validity pair set is $$\{((0,1),(0,0)),((1,0),(0,0)),((0,1),(1,1)),((1,0),(1,1))\}.$$

Second, \citet[][Section IV]{frandsen2023judging} study the interpretation of 2SLS under relaxations of monotonicity and exclusion. The relaxation of monotonicity, referred to as average monotonicity, requires $D_z$ to be positively correlated with the instrument propensity. Average monotonicity is fundamentally different from  pairwise monotonicity (Definition \ref{def.partial validity pairwise binary D}\ref{def.pairwise monotonicity}). Pairwise montonicity operates at the level of pairs of instrument values, whereas average monotonicity implies restrictions across all instrument values. Also, \citet{frandsen2023judging} show that average monotonicity can be used to identify averages of treatment effects. By contrast, pairwise validity identifies the LATE for $(z,z')\in\mathscr{Z}_{\bar{M}}$, but does not identify the LATE for $(z,z')\notin \mathscr{Z}_{\bar{M}}$ (Corollary \ref{cor:no_information} in Appendix \ref{sec.no information}).

Finally, \citet{noack2021sensitivity} considers a continuous relaxation of monotonicity when $Z$ is binary, parameterized by the fraction of defiers. We do not consider continuous relaxations and make no assumptions on the degree of violation. Combining VSIV estimation with continuous relaxations as in \citet{noack2021sensitivity} is an interesting direction for future research, as we discuss in Section \ref{sec:conclusion}.

\subsection{Identification under Pairwise Validity}
The following lemma establishes identification under pairwise validity.
\begin{lemma}\label{lemma.pairwise beta binary D}
	Suppose that the instrument $Z$ is pairwise valid according to Definition \ref{def.partial validity pairwise binary D} with a known validity pair set $\mathscr{Z}_M=\{(z_{k_1},z_{k_1^{\prime}}),\ldots,(z_{k_M},z_{k_M^{\prime}})\}$.\footnote{Note that mathematically we do not need to impose a first-stage assumption here due to the convention \eqref{eq.0timesinfinity}.} Then we can define a random variable $Y_d(z_{k_m},z_{k_m^{\prime}})=Y_{dz_{k_m}}=Y_{dz_{k'_m}}$ a.s. for each $d\in\{0,1\}$ and every $(z_{k_m},z_{k_m^{\prime}})\in\mathscr{Z}_M$, and the following quantity can be identified for every $(z_{k_m},z_{k_m^{\prime}})\in\mathscr{Z}_M$:
	\begin{align}\label{eq.beta binary D}
		\beta _{k_{m}^{\prime},k_{m}}&\equiv E\left[Y_{1}(z_{k_m},z_{k_m^{\prime}})-Y_{0}(z_{k_m},z_{k_m^{\prime}})\big|D_{z_{k_m'}}>D_{z_{k_m}}\right]\notag \\
		&=\frac{E\left[ Y|Z=z_{k_{m}^{\prime}}\right] -E\left[ Y|Z=z_{k_{m}}\right] }{E\left[
			D|Z=z_{k_{m}^{\prime}}\right] -E\left[ D|Z=z_{k_{m}}\right] }.
	\end{align}

\end{lemma}

Lemma \ref{lemma.pairwise beta binary D} is a direct extension of Theorem 1 of \citet{imbens1994identification} for the case where $Z$ is pairwise valid. We follow \citet{imbens1994identification} and refer to $\beta_{k_{m}^{\prime},k_m}$ as a LATE. Lemma \ref{lemma.pairwise beta binary D} shows that if a validity pair set $\mathscr{Z}_{{M}}$ is known, we can identify every $\beta_{k_{m}^{\prime},k_m}$ with $(z_{k_m},z_{k_m^{\prime}})\in\mathscr{Z}_M$.\footnote{Note that if $(z_{k_m},z_{k_m^{\prime}})\in\mathscr{Z}_M$ with $D_{z_{k_m}}=D_{z_{k'_m}}$ a.s., then $\beta_{k_{m}^{\prime},k_m}=0$ by \eqref{eq.0timesinfinity}.
Moreover, if $(z_{k_m},z_{k_m^{\prime}})\in\mathscr{Z}_{{M}}$ and $(z_{k'_m},z_{k_m})\in\mathscr{Z}_{{M}}$, then by Definition \ref{def.partial validity pairwise binary D}, $D_{z_{k_m}}=D_{z_{k'_m}}$ a.s.} In practice, however, $\mathscr{Z}_M$ is usually unknown. In this paper, we use testable implications of IV validity to estimate a pseudo-validity pair set $\mathscr{Z}_{0}$ containing $\mathscr{Z}_M$, and show how to use this estimated set to reduce the asymptotic bias in LATE estimation.

We focus on the vector of LATEs $\{\beta_{k_{m}^{\prime},k_m}\}$ as our object of interest. Traditional IV estimators estimate weighted averages of LATEs \citep[e.g.,][]{imbens1994identification} and, thus, are strictly less informative (we can always compute linear IV estimands based on the LATEs). Moreover, VSIV estimation estimates LATEs that do not enter such weighted averages (Theorem 2 of \citet{imbens1994identification}).  To illustrate, suppose $\mathcal{Z}=\{z_1,z_2,z_3\}$ and $\mathscr{Z}_{\bar{M}}=\{(z_1,z_2),(z_1,z_3),(z_2,z_3)\}$. The traditional IV estimator estimates a weighted average of $\beta_{2,1}$ and $\beta_{3,2}$, whereas our method estimates $(\beta_{2,1},\beta_{3,1},\beta_{3,2})^T$.

Importantly, VSIV estimation allows for assigning researcher-specified weights to $\{\beta_{k_{m}^{\prime},k_m}\}$. In the traditional IV estimation, the weights are determined by the estimation procedure. \citet{mogstad2021causal} show that the weights assigned to the LATEs by 2SLS could be negative under partial monotonicity. Negative weights are not an issue for VSIV estimation because the weights can be chosen by researchers, instead of being determined by the estimation procedure. For example, we may define the weighted average as
\begin{align}\label{eq.weighted average}
    \beta_w=\frac{p_{12}}{p_{123}}\beta_{2,1}+\frac{p_{13}}{p_{123}}\beta_{3,1}+\frac{p_{23}}{p_{123}}\beta_{3,2},
\end{align}
where $p_{ij}=\mathbb{P}(Z\in\{z_i,z_j\})$ and $p_{123}=\mathbb{P}(Z\in\{z_1,z_2\})+\mathbb{P}(Z\in\{z_2,z_3\})+\mathbb{P}(Z\in\{z_1,z_3\})$. The asymptotic properties of the estimated weighted averages of LATEs follow straightforwardly from the asymptotic theory in Section \ref{sec: validity set iv estimation}. See Corollary \ref{corollary.weighted average weak convergence}.

\begin{remark}[Extrapolation] The focus of VSIV estimation is on estimating the LATE parameters $\beta _{k^{\prime},k}$ for all pairs of instrument values $(z_{k},z_{k'})$ satisfying the testable restrictions of IV validity. This is because absent additional restrictions, there is no information in the data about the LATE for the invalid pairs (see Appendix \ref{sec.no information}). 

The local and DGP-dependent nature of LATE parameters has motivated the development of a variety of methods for assessing and restoring external validity \citep[e.g.,][]{heckman2003simple,angrist2013extrapolate,brinch2017beyond,mogstad2018using,wuthrich2020comparison,kowalski2023how}. The use of invalid instrument pairs will result in these approaches being biased and inconsistent. VSIV estimation constitutes a natural complement to the existing approaches to external validity. For settings where researchers are interested in externally valid effects, we recommend a two-step procedure: (i) Use VSIV estimation to eliminate invalid pairs. (ii) Apply a suitable approach to external validity based on the estimated validity pair set. In step (i), we recommend also reporting the VSIV LATE estimates because they summarize the available information about pairwise LATEs and are important inputs for approaches to external validity.
\end{remark}

\section{Validity Set IV Estimation}
\label{sec: validity set iv estimation}
\subsection{Overview}

The goal of VSIV estimation is to exclude invalid instrument pairs. Specifically, we seek to exclude $(z_k,z_{k'})\notin{\mathscr{Z}}_{\bar{M}}$ from $\mathscr{Z}$, since if $(z_k,z_{k'})\notin{\mathscr{Z}}_{\bar{M}}$, then $\beta_{k',k}$ in \eqref{eq.beta binary D} is not identified absent additional restrictions (Appendix \ref{sec.no information}).
Suppose that there is a set $\mathscr{Z}_0\subseteq\mathscr{Z}$ that satisfies the testable implications in \citet{kitagawa2015test}, \citet{mourifie2016testing}, \citet{kedagni2020generalized}, and \citet{sun2021ivvalidity}. Then we construct an estimator  $\widehat{\mathscr{Z}_0}$ for  $\mathscr{Z}_0$ and construct IV estimators based on $(z_{k},z_{k^{\prime}})\in \widehat{\mathscr{Z}_0}$. We refer to these estimators as \emph{VSIV estimators}. In the following, we assume that we have access to an estimator  $\widehat{\mathscr{Z}_0}$, which is consistent for $\mathscr{Z}_{0}$ in the sense that $\mathbb{P}(\widehat{\mathscr{Z}_0}=\mathscr{Z}_{0})\to 1$. We describe the testable implications and the construction of the proposed estimator satisfying $\mathbb{P}(\widehat{\mathscr{Z}_0}=\mathscr{Z}_{0})\to 1$ in detail in Section \ref{sec.estimation validity set}.

In Section \ref{sec.VSIV under consistency}, we study VSIV estimation under the assumption that $\mathscr{Z}_0=\mathscr{Z}_{\bar{M}}$ so that $\mathbb{P}(\widehat{\mathscr{Z}_0}=\mathscr{Z}_{\bar{M}})\to 1$. In this case, the proposed VSIV estimators are asymptotically unbiased and normal under standard weak regularity conditions. Since ${\mathscr{Z}_0}$ is constructed based on the necessary (but not necessarily sufficient) conditions for the pairwise IV validity, ${\mathscr{Z}_0}$ could be larger than ${\mathscr{Z}}_{\bar{M}}$. (There exist no sufficient testable conditions for IV validity in general \citep{kitagawa2015test}.) In Section \ref{sec.bias reduction binary D}, we show that even if ${\mathscr{Z}_0}$ is larger than ${\mathscr{Z}}_{\bar{M}}$, VSIV estimators yield asymptotic bias reductions relative to standard LATE estimators that do not exploit testable implications to remove invalid instrument value pairs. 

Note that if  $\mathscr{Z}_0=\varnothing$, VSIV estimation is trivial asymptotically since $\mathbb{P}(\widehat{\mathscr{Z}_0}=\varnothing)\to 1$. All the VSIV estimators converge to $0$ by the convention in \eqref{eq.0timesinfinity}. In this case, we do not report any IV estimates in practice.

\subsection{VSIV Estimation under Consistent Estimation of Validity Pair Set}
\label{sec.VSIV under consistency}

Suppose that $\mathscr{Z}_0=\mathscr{Z}_{\bar{M}}$ so that $\mathbb{P}(\widehat{\mathscr{Z}_0}=\mathscr{Z}_{\bar{M}})\to 1$, and we use $\widehat{\mathscr{Z}_{0}}$ to construct VSIV estimators for the LATEs. 
We impose the following standard regularity conditions. Let $g$ be a prespecified function that maps the value of $Z$ to $\mathbb{R}$. For example, we can simply set $g(z)=z$ for all $z$ if $Z$ is a scalar instrument.\footnote{The choice of $g$ may affect the efficiency of the VSIV estimators. We leave the formal analysis of the optimal choice of $g$ for future study.}

\begin{assumption}\label{ass.iid data binary D}
	$\{(Y_i,D_i,Z_i)\}_{i=1}^{n}$ is an i.i.d.\ sample from a population such that all relevant moments exist.
\end{assumption}

\begin{assumption}\label{ass.first stage binary D}
	For every $\mathcal{Z}_{(k,k')}\in\mathscr{Z}_{\bar{M}}$, 
	\begin{align}\label{eq.first stage}
	    E[g(Z_i)D_i|Z_i\in\mathcal{Z}_{(k,k^{\prime})}]-E[D_i|Z_i\in\mathcal{Z}_{(k,k^{\prime})}]\cdot E[g(Z_i)|Z_i\in\mathcal{Z}_{(k,k^{\prime})}]\neq0.
	\end{align}
\end{assumption}

Assumption \ref{ass.iid data binary D} assumes an i.i.d.\ data set and requires the existence of the relevant moments. Assumption \ref{ass.first stage binary D} imposes a first-stage condition for every $\mathcal{Z}_{(k,k')}\in\mathscr{Z}_{\bar{M}}$. Note that \eqref{eq.first stage} may not hold for $\mathcal{Z}_{(k,k')}\notin\mathscr{Z}_{\bar{M}}$. This creates additional technical difficulties when establishing the asymptotic normality of the VSIV estimators, which we discuss below. Assumption \ref{ass.first stage binary D} also implies that if $\mathcal{Z}_{(k,k')}\in\mathscr{Z}_{\bar{M}}$, then $\mathcal{Z}_{(k',k)}\notin\mathscr{Z}_{\bar{M}}$. Otherwise, by Definition \ref{def.partial validity pairwise binary D}, $D_{z_k}=D_{z_{k'}}$ a.s., and \eqref{eq.first stage} does not hold.
For every scalar random sample $\{\xi_i\}_{i=1}^n$ and every $\mathcal{A}\in\mathscr{Z}$, we define
\begin{align*}
	\mathcal{E}_{n}\left(  \xi_{i},\mathcal{A}\right)    =\frac{\frac{1}{n}\sum
		_{i=1}^{n}\xi_{i}1\left\{  Z_{i}\in\mathcal{A}\right\}}{\frac{1}{n}\sum_{i=1}^{n} 1\left\{  Z_{i}\in\mathcal{A}\right\}  }
	\text{ and }\mathcal{E}\left(  \xi_{i},\mathcal{A}\right)    =\frac{E\left[
		\xi_{i}1\left\{  Z_{i}\in\mathcal{A}\right\}  \right]}{E\left[
		1\left\{  Z_{i}\in\mathcal{A}\right\}  \right]}.
\end{align*}
We define the VSIV estimators using regression-based IV estimators, following \citet{imbens1994identification}. For every
$\mathcal{Z}_{(  k,k^{\prime})  }\in{\mathscr{Z}}$, we
run the IV regression
\begin{align}\label{eq.VSIV estimation pairwise binary D}
	Y_i 1\left\{  Z_{i}\in\mathcal{Z}_{(k,k^{\prime})}\right\} =&\,\gamma_{(k,k^{\prime})}^01\left\{Z_{i}\in\mathcal{Z}_{(k,k^{\prime})}\right\}+\gamma_{(k,k^{\prime})}^1D_i1\left\{Z_{i}\in\mathcal{Z}_{(k,k^{\prime})}\right\}+\epsilon_{i}1\left\{Z_{i}\in\mathcal{Z}_{(k,k^{\prime})}\right\},
\end{align}
using $g(Z_i)1\{Z_{i}\in\mathcal{Z}_{(k,k^{\prime})}\}$ as the instrument for  $D_i1\{Z_{i}\in\mathcal{Z}_{(k,k^{\prime})}\}$. Given the estimated validity set $\widehat{\mathscr{Z}_0}$, we set the VSIV estimator for each $\mathcal{Z}_{(k,k^{\prime})}$ as
\begin{align}\label{eq.VSIV estimator pairwise binary D}
	\widehat{\beta}_{(  k,k^{\prime})  }^1=1\left\{\mathcal{Z}_{(k,k^{\prime})}\in\widehat{\mathscr{Z}_0}\right\}\cdot\frac{\mathcal{E}_{n}\left(
		g\left(  Z_{i}\right)  Y_{i},\mathcal{Z}_{(k,k^{\prime})}\right)
		-\mathcal{E}_{n}\left(  g\left(  Z_{i}\right)  ,\mathcal{Z}_{(k,k^{\prime}
			)}\right)  \mathcal{E}_{n}\left(  Y_{i},\mathcal{Z}_{(k,k^{\prime})}\right)
	}{\mathcal{E}_{n}\left(  g\left(  Z_{i}\right)  D_{i},\mathcal{Z}
		_{(k,k^{\prime})}\right)  -\mathcal{E}_{n}\left(  g\left(  Z_{i}\right)
		,\mathcal{Z}_{(k,k^{\prime})}\right)  \mathcal{E}_{n}\left(  D_{i}
		,\mathcal{Z}_{(k,k^{\prime})}\right)  },
\end{align}
which is the IV estimator of $\gamma_{(k,k^{\prime})}^1$ in \eqref{eq.VSIV estimation pairwise binary D} multiplied by $1\{\mathcal{Z}_{(k,k^{\prime})}\in\widehat{\mathscr{Z}_0}\}$. For every $\mathcal{Z}_{(k,k')}\in\widehat{\mathscr{Z}}_0$, this IV estimation is equivalent to a conventional IV regression in the subsample of $\{(Y_i,D_i,Z_i)\}_{i=1}^n$ with $Z_i\in\mathcal{Z}_{(k,k')}$. Note that $\widehat{\beta}_{(  k,k^{\prime})  }^1=0$ if $\mathcal{Z}_{(k,k^{\prime})}\notin\widehat{\mathscr{Z}_0}$. We discuss this convention further below.

Define the vector of VSIV estimators as
\[
\widehat{\beta}_{1}=\left(  \widehat{\beta}_{\left(  1,2\right)  }^{1},\ldots
,\widehat{\beta}_{\left(  1,K\right)  }^{1},\ldots,\widehat{\beta}_{\left(
	K,1\right)  }^{1},\ldots,\widehat{\beta}_{\left(  K,K-1\right)  }^{1}\right)^T.
\]
We also define
\begin{align}\label{eq.VSIV true beta binary D}
	\beta_{(  k,k^{\prime})  }^{1}=1\left\{\mathcal{Z}_{(k,k^{\prime})}\in\mathscr{Z}_{\bar{M}}\right\}\cdot\frac{\mathcal{E}\left(  g\left(
		Z_{i}\right)  Y_{i},\mathcal{Z}_{(k,k^{\prime})}\right)  -\mathcal{E}\left(
		g\left(  Z_{i}\right)  ,\mathcal{Z}_{(k,k^{\prime})}\right)  \mathcal{E}
		\left(  Y_{i},\mathcal{Z}_{(k,k^{\prime})}\right)  }{\mathcal{E}\left(
		g\left(  Z_{i}\right)  D_{i},\mathcal{Z}_{(k,k^{\prime})}\right)
		-\mathcal{E}\left(  g\left(  Z_{i}\right)  ,\mathcal{Z}_{(k,k^{\prime}
			)}\right)  \mathcal{E}\left(  D_{i},\mathcal{Z}_{(k,k^{\prime})}\right)  }
\end{align}
and
\begin{align}\label{eq.VSIV true beta1 binary D}
\beta_{1}=\left(  \beta_{\left(  1,2\right)  }^{1},\ldots,\beta_{\left(
	1,K\right)  }^{1},\ldots,\beta_{\left(  K,1\right)  }^{1},\ldots
,\beta_{\left(  K,K-1\right)  }^{1}\right)^T  .
\end{align}
As we show formally in Theorem \ref{thm.IV estimator asymptotics pairwise binary D} below, $\beta_{(  k,k^{\prime})  }^{1}=\beta_{k^{\prime},k}$ as defined in \eqref{eq.beta binary D} for every $(z_k,z_{k^{\prime}})\in\mathscr{Z}_{\bar{M}}$. If $\mathcal{Z}_{(k,k^{\prime})}\notin {\mathscr{Z}_{\bar{M}}}$, we set ${\beta}_{(k,k^{\prime})}^1=0$ by \eqref{eq.VSIV true beta binary D} and \eqref{eq.0timesinfinity}. Similarly, if $\mathcal{Z}_{(k,k^{\prime})}\notin \widehat{\mathscr{Z}_0}$, $\widehat{\beta}_{(k,k^{\prime})}^1=0$ by \eqref{eq.VSIV estimator pairwise binary D} and \eqref{eq.0timesinfinity}. Letting them be equal to $0$ facilitates the description of the theoretical results in Theorem \ref{thm.IV estimator asymptotics pairwise binary D}, and this will not affect the estimation of the weighted average of LATEs.\footnote{It is equivalent to not including them into the averages.}
In practice, we recommend leaving $\widehat{\beta}_{(k,k^{\prime})}^1$ to be blank if $1\{\mathcal{Z}_{(k,k^{\prime})}\in\widehat{\mathscr{Z}_0}\}=0$, as in the application in Section \ref{sec.application}. We interpret the LATE corresponding to $\mathcal{Z}_{(k,k^{\prime})}\in {\mathscr{Z}_{\bar{M}}}$ in the usual way as the average treatment effects for compliers in the corresponding subgroup. We do not report the estimates for LATEs corresponding to invalid pairs since they are not identified and there is no information about them in the data absent additional restrictions (Appendix \ref{sec.no information}).

The next theorem establishes the asymptotic distribution of the VSIV estimator $\widehat{\beta}_{1}$, obtained based on the estimator of the instrument validity pair set $\widehat{\mathscr{Z}_0}$.

\begin{theorem}\label{thm.IV estimator asymptotics pairwise binary D}
	Suppose that the instrument $Z$ is pairwise valid for the treatment $D$ according to Definition \ref{def.partial validity pairwise binary D} with the largest validity pair set $\mathscr{Z}_{\bar{M}}=\{(z_{k_1},z_{k_1^{\prime}}),\ldots,(z_{k_{\bar{M}}},z_{k_{\bar{M}}^{\prime}})\}$, that the estimator $\widehat{\mathscr{Z}_0}$ satisfies $\mathbb{P}(\widehat{\mathscr{Z}_0}=\mathscr{Z}_{\bar{M}})\to 1$, and that Assumptions \ref{ass.iid data binary D} and \ref{ass.first stage binary D} hold. Then 
	\begin{align}\label{eq.beta asymptotic distribution}
	    \sqrt{n}( \widehat{\beta}_{1}-\beta_1 ) \overset{d}\to N\left( 0,\Sigma \right), 
	\end{align} 
	where $\Sigma$ is defined in \eqref{eq.weak convergence pairwise2} in the Appendix. In addition, $\beta_{(  k,k^{\prime})  }^{1}=\beta_{k^{\prime},k}$ as defined in \eqref{eq.beta binary D} for every $(z_k,z_{k^{\prime}})\in\mathscr{Z}_{\bar{M}}$. 
\end{theorem}

	Theorem \ref{thm.IV estimator asymptotics pairwise binary D} establishes the joint asymptotic normality of the VSIV estimator of the LATEs. Establishing the asymptotic distribution in \eqref{eq.beta asymptotic distribution} requires a careful treatment of the case where the first-stage Assumption \ref{ass.first stage binary D} does not hold for some pairs of instrument  values $\mathcal{Z}_{(k,k')}$ that are not in the largest validity pair set $\mathscr{Z}_{\bar{M}}$, that is, $\mathcal{Z}_{(k,k')}\notin\mathscr{Z}_{\bar{M}}$. Specifically, we show that in this case, $\mathbb{P}(\widehat{\mathscr{Z}_0}=\mathscr{Z}_{\bar{M}})\to 1$ implies that, if $\mathcal{Z}_{(k,k')}\notin\mathscr{Z}_{\bar{M}}$, then for every $\rho>0$, $n^{\rho}1\{\mathcal{Z}_{(k,k')}\in\widehat{\mathscr{Z}_0}\}=o_p(1)$. This guarantees the convergence in \eqref{eq.beta asymptotic distribution} even when \eqref{eq.first stage} does not hold for $\mathcal{Z}_{(k,k')}\notin\mathscr{Z}_{\bar{M}}$. The asymptotic covariance matrix $\Sigma$ defined in the Appendix can be consistently estimated under standard conditions. Importantly, the estimation of the instrument validity pair set does not affect the asymptotic covariance matrix such that standard inference methods can be applied. 
	
Under Theorem \ref{thm.IV estimator asymptotics pairwise binary D}, it is straightforward to obtain a consistent estimator of a weighted average of all ${\beta}_{(k,k')}^1$. Let $\widehat{W}$ be some estimated weights with
 \begin{align*}
     \widehat{W}=(\widehat{w}_{(1,2)},\ldots,\widehat{w}_{(1,K)},\ldots,\widehat{w}_{(K,1)},\ldots,\widehat{w}_{(K,K-1)})^T
 \end{align*}
such that $\widehat{W}\overset{p}\rightarrow W$ for some $W$ with 
 \begin{align*}
     {W}=({w}_{(1,2)},\ldots,{w}_{(1,K)},\ldots,{w}_{(K,1)},\ldots,{w}_{(K,K-1)})^T.
 \end{align*}
Theorem \ref{thm.IV estimator asymptotics pairwise binary D} implies that $\widehat{W}^T\widehat{\beta}_1\overset{p}\rightarrow W^T{\beta}_1$.
 To establish the asymptotic distribution of the estimated weighted average of all $\beta^1_{(k,k')}$, we assume that $(\widehat{W}^T,\widehat{\beta}_1^T)^T$ is asymptotically normal. This is a weak condition in practice. It holds, for example, if the weights are defined as in the weighted average in \eqref{eq.weighted average}. 
The following corollary summarizes the result.
\begin{corollary}\label{corollary.weighted average weak convergence}
Suppose $\sqrt{n}\{(\widehat{W}^T,\widehat{\beta}_1^T)^T-({W}^T,{\beta}_1^T)^T\}\overset{d}\to N(0,\Sigma_W)$ for some matrix $\Sigma_W$. Then it follows that 
\begin{align}\label{eq.beta asymptotic distribution weighted average}
	    \sqrt{n}(\widehat{W}^T \widehat{\beta}_{1}-{W}^T\beta_1 ) \overset{d}\to N\left( 0,(\beta_1^T,W^T)\Sigma_{W} (\beta_1^T,W^T)^T \right). 
	\end{align} 
\end{corollary}
 
	The LATE $\beta_{k',k}$ is not identified if $\mathcal{Z}_{(k,k')}\notin\mathscr{Z}_{\bar{M}}$. Let $\beta_{1S}=(\beta_{(\kappa_{1},\kappa_{1}^{\prime})}^{1},\ldots,\beta
_{(\kappa_{S},\kappa_{S}^{\prime})}^{1})^{T}$ for some $S>0$. In our context, it is interesting to test hypotheses about $\beta^1_{(k,k')}$ with $\mathcal{Z}_{(k,k')}\in\mathscr{Z}_{\bar{M}}$ ($\beta^1_{(k,k')}=\beta_{k',k}$ by Theorem \ref{thm.IV estimator asymptotics pairwise binary D}):
	\begin{align}
	    \mathrm{H}_{0}:\mathcal{Z}_{(  \kappa_{1},\kappa_{1}^{\prime})  }%
\in\mathscr{Z}_{\bar{M}},\ldots,\mathcal{Z}_{(  \kappa_{S},\kappa_{S}^{\prime
})  }\in\mathscr{Z}_{\bar{M}},~R\left(  \beta_{1S}\right)  =0,
	\end{align}
where $R$ is a (possibly nonlinear) smooth $r$-dimensional function. Let $R^{\prime
}(\beta_{S})$ be the $r\times S$ matrix of the continuous first derivative functions of $R$ at an arbitrary
value $\beta_{S}$, that is, $R^{\prime}(\beta_{S})=\partial R\left(  \beta
_{S}\right)  /\partial\beta_{S}^{T}$. Let $\mathcal{I}_{S}$ be a $S\times
\left(  K-1\right)  K$ matrix such that
\[
\mathcal{I}_{S}\beta=(\beta_{(\kappa_{1},\kappa_{1}^{\prime})},\ldots
,\beta_{(\kappa_{S},\kappa_{S}^{\prime})})^{T}
\]
for every  $\beta=(\beta_{\left(
1,2\right)  },\ldots,\beta_{\left(  1,K\right)  },\ldots,\beta_{\left(
K,1\right)  },\ldots,\beta_{\left(  K,K-1\right)  })^{T}$.
Theorem \ref{thm.IV estimator asymptotics pairwise binary D} implies that
\[
\sqrt{n}(  \widehat{\beta}_{1S}-\beta_{1S})  =\sqrt
{n}\mathcal{I}_{S}(  \widehat{\beta}_{1}-\beta_{1})  \overset
{d}{\rightarrow}N\left(  0,\Sigma_{S}\right)  ,
\]
where  $\Sigma_{S}=\mathcal{I}_{S}\Sigma\mathcal{I}_{S}^{T}$, so that by the delta method, we obtain
\[
\sqrt{n}\left\{  R(  \widehat{\beta}_{1S})  -R\left(  \beta
_{1S}\right)  \right\}  \overset{d}{\rightarrow}N\left(  0,R^{\prime
}\left(  \beta_{1S}\right)  \Sigma_{S}R^{\prime}\left(  \beta_{1S}%
\right)  ^{T}\right)  .
\]
We construct the test statistics as%
\[
TS_{1n}=\prod_{s=1}^{S}1\left\{  \mathcal{Z}_{\left(  \kappa_{s},\kappa
_{s}^{\prime}\right)  }\in\widehat{\mathscr{Z}_{0}}\right\}
\]
and
\begin{align}\label{eq.TS2_1}
TS_{2n}=&\,\sqrt{n}R( \widehat{\beta}_{1S})  ^{T}\left\{
R^{\prime}(  \widehat{\beta}_{1S})  \mathcal{I}_{S}
\widehat{\Sigma}\mathcal{I}_{S}^{T}R^{\prime}(  \widehat{\beta}_{1S}
)  ^{T}\right\}  ^{-1}\sqrt{n}R( \widehat{\beta}_{1S})  ,
\end{align}
where $\widehat{\Sigma}$ is a consistent estimator of $\Sigma$, which can be constructed based on the formula in \eqref{eq.weak convergence pairwise2}.
Suppose that Assumptions \ref{ass.iid data binary D} and \ref{ass.first stage binary D} hold and $\mathbb{P}(\widehat{\mathscr{Z}_0}=\mathscr{Z}_{\bar{M}})\to 1$. If $\mathrm{H}_{0}$ is true and $R'(\beta_{1S})$ is of full row rank, then it follows from standard arguments that
$TS_{2n}\overset{d}{\rightarrow}\chi^2_{r}$, where  $\chi^2_{r}$ denotes the  chi-square distribution with  $r$ degrees of freedom. The decision rule of the test is to reject
$\mathrm{H}_{0}$ if $TS_{1n}=0$ or $TS_{2n}>c_{r}(\alpha)$, where $c_{r}(\alpha)$ satisfies
$\mathbb{P}(\chi^2_{r}>c_{r}(\alpha))=\alpha$ for some predetermined $\alpha\in(0,1)$. The following proposition establishes the formal properties of the proposed test.

\begin{proposition}\label{prop.test}
Suppose that Assumptions \ref{ass.iid data binary D} and \ref{ass.first stage binary D} hold and $\mathbb{P}(\widehat{\mathscr{Z}_0}=\mathscr{Z}_{\bar{M}})\to 1$.
\begin{enumerate}[label=(\roman*)]
    \item If $\mathrm{H}_{0}$ is true,
$\mathbb{P}\left(  \left\{  TS_{1n}=0\right\}  \cup\left\{  TS_{2n}%
>c_{r}(\alpha)\right\}  \right)  \rightarrow\alpha$.

    \item If $\mathrm{H}_{0}$ is false,
$\mathbb{P}\left(  \left\{  TS_{1n}=0\right\}  \cup\left\{  TS_{2n}>c_{r}(\alpha)\right\}  \right)  \rightarrow1$.
\end{enumerate}
   
\end{proposition}

\subsection{Asymptotic Bias Reduction under VSIV Estimation}\label{sec.bias reduction binary D}
In Section \ref{sec.VSIV under consistency}, we show that if the estimator of the largest validity pair set is consistent, $ \mathbb{P}(\widehat{\mathscr{Z}_0}=\mathscr{Z}_{\bar{M}})\to 1$, the VSIV estimators are consistent and asymptotically normal under weak conditions. However, since ${\mathscr{Z}_0}$ is constructed based on necessary (but not necessarily sufficient) conditions for IV validity, we have $ \mathbb{P}(\widehat{\mathscr{Z}_0}=\mathscr{Z}_0)\to 1$ in general, where the \emph{pseudo-validity pair set} $\mathscr{Z}_0$ could be larger than $\mathscr{Z}_{\bar{M}}$. In this case, VSIV estimators may not be asymptotically unbiased. Consider an arbitrary presumed validity pair set $\mathscr{Z}_P$, which could incorporate prior information. If no prior information is available, we set $\mathscr{Z}_P=\mathscr{Z}$. Here we show even if ${\mathscr{Z}_0}$ is larger than $\mathscr{Z}_{\bar{M}}$, VSIV estimators based on $\widehat{\mathscr{Z}'_0}=\widehat{\mathscr{Z}_0}\cap \mathscr{Z}_P$ have weakly lower asymptotic biases than standard LATE estimators based on $\mathscr{Z}_P$. Intuitively, VSIV estimators use the information in the data about IV validity to reduce the asymptotic bias.

Since our target parameter is the vector $\beta_1$, a natural definition of the asymptotic bias is as follows.
\begin{definition}\label{def.bias}
	The asymptotic bias of an arbitrary estimator $\tilde{\beta}_1$ for the true value $\beta_1$ defined in \eqref{eq.VSIV true beta1 binary D} is defined as $\mathrm{plim}_{n\rightarrow\infty}\Vert \tilde{\beta}_1 -\beta_1\Vert_2$, where $\Vert\cdot\Vert_2$ is the $\ell^2$-norm on Euclidean spaces. 
\end{definition}

The next assumption extends Assumption \ref{ass.first stage binary D} to $\mathscr{Z}_0$.
\begin{assumption}\label{ass.first stage binary D bias}
	For every $\mathcal{Z}_{(k,k')}\in\mathscr{Z}_{0}$, 
	\begin{align}\label{eq.first stage bias}
	    E[g(Z_i)D_i|Z_i\in\mathcal{Z}_{(k,k^{\prime})}]-E[D_i|Z_i\in\mathcal{Z}_{(k,k^{\prime})}]\cdot E[g(Z_i)|Z_i\in\mathcal{Z}_{(k,k^{\prime})}]\neq0.
	\end{align}
\end{assumption}

 The following theorem shows that the VSIV estimators based on $\widehat{\mathscr{Z}'_0}$ exhibit a smaller asymptotic bias than standard LATE estimators based on $\mathscr{Z}_P$.  

\begin{theorem}\label{thm.bias reduction binary D}
	Suppose that Assumptions \ref{ass.iid data binary D} and \ref{ass.first stage binary D bias} hold and that $\mathbb{P}(\widehat{\mathscr{Z}_0}=\mathscr{Z}_0)\to 1$ with $\mathscr{Z}_0\supseteq\mathscr{Z}_{\bar{M}}$. For every presumed validity pair set $\mathscr{Z}_P$, the asymptotic bias of $\widehat{\beta}_1$ is reduced by using $\widehat{\mathscr{Z}'_0}$  in the estimation \eqref{eq.VSIV estimator pairwise binary D} compared to the asymptotic bias from using $\mathscr{Z}_P$.
\end{theorem}

As shown in Proposition \ref{prop.consistent G hat pairwise Z1 binary D} below, the pseudo-validity pair set $\mathscr{Z}_0$ can be estimated consistently by $\widehat{\mathscr{Z}_0}$ under mild conditions. Compared to constructing standard IV estimators based on $\mathscr{Z}_P$, Theorem \ref{thm.bias reduction binary D} shows that the asymptotic bias, $\mathrm{plim}_{n\rightarrow\infty}\Vert \widehat{\beta}_1-\beta_1\Vert_2$, can be reduced by using VSIV estimators based on $\widehat{\mathscr{Z}'_0}=\widehat{\mathscr{Z}_0}\cap\mathscr{Z}_P$.

The arguments used for establishing the asymptotic normality of the VSIV estimators in Section \ref{sec.VSIV under consistency} do not rely on the consistent estimation of $\mathscr{Z}_{\bar{M}}$ ($\mathbb{P}(\widehat{\mathscr{Z}_0}=\mathscr{Z}_{\bar{M}})\to 1$). If $\mathbb{P}(\widehat{\mathscr{Z}_0}=\mathscr{Z}_{0})\to 1$ with $\mathscr{Z}_{\bar{M}}\subsetneq\mathscr{Z}_{0}$, the VSIV estimators are asymptotically normal, centered at $\beta_1$ defined with $\mathscr{Z}_0$ instead of $\mathscr{Z}_{\bar{M}}$. However, note that $\beta_1$ can only be interpreted as a vector of LATEs under consistent estimation ($\mathbb{P}(\widehat{\mathscr{Z}_0}=\mathscr{Z}_{\bar{M}})\to 1$).

\begin{example}[Asymptotic Bias Reduction under VSIV Estimation] Consider a simple example\\ where $\mathcal{Z}=\{1,2,3,4\}$ as in our application and suppose that $\mathscr{Z}_{\bar{M}}=\{(1,2)\}$. In this case, by \eqref{eq.VSIV true beta binary D} and \eqref{eq.0timesinfinity}, 
\begin{align*}
	\beta_{1}=\left(  \beta_{\left(  1,2\right)  }^{1},\ldots,\beta_{\left(
		1,4\right)  }^{1},\ldots,\beta_{\left(  4,1\right)  }^{1},\ldots
	,\beta_{\left(  4,3\right)  }^{1}\right)^T=  \left(  \beta_{\left(  1,2\right)  }^{1},0,\ldots,0\right)^T.
\end{align*}
Suppose that, by mistake, we assume $Z$ is valid according to Assumption \ref{ass.IV validity binary D}  and use 
\begin{align*}
	\mathscr{Z}_P=\{(1,2),(1,3),(1,4),(2,3),(2,4),(3,4)\}	
\end{align*}
as an estimator for $\mathscr{Z}_{\bar{M}}$. Then by \eqref{eq.VSIV estimator pairwise binary D} and \eqref{eq.0timesinfinity},
\begin{align}\label{eq.beta1_hat example1 binary D}
	\widehat{\beta}_{1}=\left(  \widehat{\beta}_{\left(  1,2\right)  }^{1},\widehat{\beta}_{\left(  1,3\right)  }^{1}, \widehat{\beta}_{\left(  1,4\right)  }^{1},\widehat{\beta}_{\left(  2,3\right)  }^{1},\widehat{\beta}_{\left(  2,4\right)  }^{1},\widehat{\beta}_{\left(  3,4\right)  }^{1},0,0,0,0,0,0\right)^T,
\end{align}
where $\widehat{\beta}_{(1,3)}^1$, $\widehat{\beta}_{(1,4)}^1$, $\widehat{\beta}_{(2,3)}^1$, $\widehat{\beta}_{(2,4)}^1$, and $\widehat{\beta}_{(3,4)}^1$ may not converge to $0$ in probability. However, by definition ${\beta}_{(1,3)}^1=0$, ${\beta}_{(1,4)}^1=0$, ${\beta}_{(2,3)}^1=0$, ${\beta}_{(2,4)}^1=0$, and ${\beta}_{(3,4)}^1=0$. Thus, $\Vert \widehat{\beta}_1-\beta_1\Vert_2$ may not converge to $0$ in probability. The approach proposed in this paper helps reduce this asymptotic bias. We exploit the information in the data about IV validity to obtain the estimator $\widehat{\mathscr{Z}_0}$. Even if $\widehat{\mathscr{Z}_0}$ converges to a set larger than $\mathscr{Z}_{\bar{M}}$ (because we use the necessary but not sufficient conditions for IV validity), VSIV always reduces the asymptotic bias. Suppose that $\mathscr{Z}_0=\{(1,2),(3,4)\}$, which is larger than $\mathscr{Z}_{\bar{M}}$ but smaller than $\mathscr{Z}_P$. In this case, the VSIV estimator $\widehat{\beta}_1$ constructed by using $\widehat{\mathscr{Z}_0}\cap\mathscr{Z}_P$ converges in probability to
\begin{align}\label{eq.beta1_hat example2 binary D}
	{\beta}'_{1}=\left(  {\beta}_{\left(  1,2\right)  }^{1},0, 0,0,0,{\beta}^{1\prime}_{\left(  3,4\right)  },0,0,0,0,0,0\right)^T,
\end{align}
where ${\beta}^{1\prime}_{\left(  3,4\right)  }$ is the probability limit of $\widehat{\beta}^1_{\left(  3,4\right)  }$.
Then, clearly, VSIV reduces the probability limit of $\Vert \widehat{\beta}_1-\beta_1\Vert_2 $. 
\end{example}

\subsection{Partially Valid Instruments and Connection to Existing Results}
\label{sec.partially valid instruments binary D}

Suppose we estimate the following canonical IV regression model,
\begin{align}\label{eq.model equation binary D}
	Y_i=\alpha_{0}+\alpha_{1}D_i+\epsilon_{i},
\end{align}
using $g(Z_i)$ as the instrument for $D_i$. 
When the instrument $Z$ is fully valid, the traditional IV estimator of $\alpha_1$ is 
\begin{align}\label{eq.alpha_hat}
\widehat{\alpha}_{1}=\frac{n\sum_{i=1}^{n}g\left(  Z_{i}\right)  Y_{i}-\sum
	_{i=1}^{n}g\left(  Z_{i}\right)  \sum_{i=1}^{n}Y_{i}}{n\sum_{i=1}^{n}g\left(
	Z_{i}\right)  D_{i}-\sum_{i=1}^{n}g\left(  Z_{i}\right)  \sum_{i=1}^{n}D_{i}}.
\end{align}
The asymptotic properties of $\widehat{\alpha}_1$ can be found in \citet[p.~471]{imbens1994identification} and \citet[p.~436]{angrist1995two}.

To connect VSIV estimation to canonical IV regression with fully valid instruments, consider the following special case of pairwise IV validity.

\begin{definition}\label{def.partially validy instrument binary D}
	Suppose that the instrument $Z$ is pairwise valid for the treatment $D$ with the largest validity pair set $\mathscr{Z}_{\bar{M}}$. If there is a validity pair set $$\mathscr{Z}_{{M}}=\{(z_{k_1},z_{k_2}),(z_{k_2},z_{k_3}),\ldots,(z_{k_{M-1}},z_{k_M})\}$$ for some $M>0$, then the instrument $Z$ is called a \textbf{partially valid instrument} for the treatment $D$. The set $\mathcal{Z}_M=\{z_{k_1},\ldots,z_{k_M}\}$ is called a \textbf{validity value set} of $Z$.
\end{definition}

\begin{assumption}\label{ass.first stage binary D partial}
	The validity value set $\mathcal{Z}_M$ satisfies that 
	\begin{align}
	    E[g(Z_i)D_i|Z_i\in\mathcal{Z}_{M}]-E[D_i|Z_i\in\mathcal{Z}_{M}]\cdot E[g(Z_i)|Z_i\in\mathcal{Z}_{M}]\neq0.
	\end{align}
\end{assumption}

Suppose that $Z$ is partially valid for the treatment $D$ with a validity value set $\mathcal{Z}_M$ and that there is a consistent estimator $\widehat{\mathcal{Z}_{0}}$ of $\mathcal{Z}_M$. We then construct a VSIV estimator for $\alpha_1$ in \eqref{eq.model equation binary D} by estimating the model
\begin{align}\label{eq.VSIV estimation binary D}
	Y_i1\left\{
	Z_{i}\in\widehat{\mathcal{Z}_{0}}\right\} =\gamma_{0}1\left\{
	Z_{i}\in\widehat{\mathcal{Z}_{0}}\right\}+\gamma_{1}D_i1\left\{
	Z_{i}\in\widehat{\mathcal{Z}_{0}}\right\} +\epsilon_{i}1\left\{
	Z_{i}\in\widehat{\mathcal{Z}_{0}}\right\},
\end{align}
using $g(Z_i)1\{
Z_{i}\in\widehat{\mathcal{Z}_{0}}\}$ as the instrument for $D_i1\{
Z_{i}\in\widehat{\mathcal{Z}_{0}}\}$. We obtain the VSIV estimator for $\alpha_1$ in \eqref{eq.model equation binary D} by
\begin{align}\label{eq.VSIV estimator binary D}
	\widehat{\theta}_{1}=\frac{n_z\sum_{i=1}^{n}g\left(  Z_{i}\right)  Y_{i}1\left\{
		Z_{i}\in\widehat{\mathcal{Z}_{0}}\right\}  -\sum_{i=1}^{n}g\left(
		Z_{i}\right)  1\left\{  Z_{i}\in\widehat{\mathcal{Z}_{0}}\right\}  \sum
		_{i=1}^{n}Y_{i}1\left\{  Z_{i}\in\widehat{\mathcal{Z}_{0}}\right\}  }
	{n_z\sum_{i=1}^{n}g\left(  Z_{i}\right)  D_{i}1\left\{  Z_{i}\in\widehat
		{\mathcal{Z}_{0}}\right\}  -\sum_{i=1}^{n}g\left(  Z_{i}\right)  1\left\{
		Z_{i}\in\widehat{\mathcal{Z}_{0}}\right\}  \sum_{i=1}^{n}D_{i}1\left\{
		Z_{i}\in\widehat{\mathcal{Z}_{0}}\right\}  },
\end{align}
where $n_z=\sum_{i=1}^n1\{Z_{i}\in\widehat{\mathcal{Z}_{0}}\}$. We can see that $\widehat{\theta}_1$ is a generalized version of $\widehat{\alpha}_1$ in \eqref{eq.alpha_hat}, because when the instrument is fully valid, we can just let $\widehat{\mathcal{Z}_{0}}=\mathcal{Z}$ and then $\widehat{\theta}_1=\widehat{\alpha}_1$. 

\begin{theorem}\label{thm.IV estimator asymptotics binary D}
	Suppose that the instrument $Z$ is partially valid for the treatment $D$ according to Definition \ref{def.partially validy instrument binary D} with a validity value set $\mathcal{Z}_M=\{z_{k_1},\dots,z_{k_M} \}$, and that the estimator $\widehat{\mathcal{Z}_0}$ for $\mathcal{Z}_M$ satisfies  $\mathbb{P}(\widehat{\mathcal{Z}_{0}}=\mathcal{Z}_{M})\rightarrow 1$. Under Assumptions \ref{ass.iid data binary D} and \ref{ass.first stage binary D partial}, it follows that $\widehat{\theta}_{1}\overset{p}\rightarrow\theta_1 $, where 
	\begin{align*}
		\theta_{1}=\frac{E\left[  g\left(  Z_{i}\right)
			Y_{i}|  Z_{i}\in\mathcal{Z}_{M}  \right]  -E\left[
			Y_{i}|  Z_{i}\in\mathcal{Z}_{M}  \right]  E\left[  g\left(
			Z_{i}\right)|  Z_{i}\in\mathcal{Z}_{M}  \right]  }{E\left[
			g\left(  Z_{i}\right)  D_{i}|  Z_{i}\in\mathcal{Z}_{M}
			\right]  -E\left[  D_{i}|  Z_{i}\in\mathcal{Z}_{M}  \right]
			E\left[  g\left(  Z_{i}\right)|  Z_{i}\in\mathcal{Z}_{M} \right]  }.
	\end{align*}
	Also, $\sqrt{n}( \widehat{\theta}_{1}-\theta_1 ) \overset{d}\to N\left( 0,\Sigma_1 \right) $, where 
	$\Sigma_1$ is provided in \eqref{eq.asymptotic beta1} in the Appendix. In addition, the quantity $\theta_{1}$ can be interpreted as
	the weighted average of $\{ \beta _{k_{2},k_{1}},\ldots ,\beta_{k_{M},k_{M-1}}\} $ defined as in \eqref{eq.beta binary D}. Specifically, $\theta _{1}=\sum_{m=1}^{M-1}\mu
	_{m}\beta _{k_{m+1},k_{m}}$ with 
	\begin{align*}
		\mu _{m}=
		\frac{\left[ p\left( z_{k_{m+1}}\right) -p\left( z_{k_{m}}\right) 
			\right] \sum_{l=m}^{M-1}\mathbb{P}\left( Z_{i}=z_{k_{l+1}}|Z_i\in\mathcal{Z}_M\right) \left\{
			g\left( z_{k_{l+1}}\right) -E\left[ g\left( Z_{i}\right) |Z_i\in\mathcal{Z}_M \right] \right\} }{\sum_{l=1}^{M}\mathbb{P}\left(
			Z_{i}=z_{k_{l}}|Z_i\in\mathcal{Z}_M\right) p\left( z_{k_{l}}\right) \left\{ g\left(
			z_{k_{l}}\right) -E\left[ g\left( Z_{i}\right)|Z_i\in\mathcal{Z}_M \right] \right\} },
	\end{align*}
	$p\left( z_{k}\right) =E\left[ D_{i}|Z_{i}=z_{k}\right] $, and $\sum_{m=1}^{M-1}\mu _{m}=1$.
\end{theorem}

Theorem \ref{thm.IV estimator asymptotics binary D} is an extension of Theorem 2 of \citet{imbens1994identification} to the case where the instrument is partially but not fully valid. 
To establish a connection to existing results, Theorem \ref{thm.IV estimator asymptotics binary D} assumes consistent estimation of the validity value set such that $\mathbb{P}(\widehat{\mathcal{Z}_{0}}=\mathcal{Z}_{M})\rightarrow 1$. If $\widehat{\mathcal{Z}_{0}}$ converges to a larger set than $\mathcal{Z}_{M}$, the properties of VSIV estimation follow from the results in Section \ref{sec.bias reduction binary D} because partially valid instruments are a special case of pairwise valid instruments.

\section{Definition and Estimation of $\mathscr{Z}_{0}$}
\label{sec.estimation validity set}

Here we discuss the definition and the estimation of $\mathscr{Z}_0$ based on the testable implications in \citet{kitagawa2015test}, \citet{mourifie2016testing}, \citet{kedagni2020generalized}, and \citet{sun2021ivvalidity} for pairwise IV validity. We show that under weak assumptions, the proposed estimator $\widehat{\mathscr{Z}_{0}}$ is consistent for the pseudo-validity pair set $\mathscr{Z}_{0}$ in the sense that $\mathbb{P}(\widehat{\mathscr{Z}_{0}}=\mathscr{Z}_{0})\rightarrow 1$. As a consequence, when $\mathscr{Z}_{0}=\mathscr{Z}_{\bar{M}}$, the largest validity pair set can be estimated consistently.

Specifically, we first construct two sets, $\mathscr{Z}_1$ and $\mathscr{Z}_2$, of pairs of instrument values such that $\mathscr{Z}_1$ satisfies the testable implications in \citet{kitagawa2015test}, \citet{mourifie2016testing}, and \citet{sun2021ivvalidity}, and $\mathscr{Z}_2$ satisfies the testable implications in \citet{kedagni2020generalized}. We then construct $\mathscr{Z}_0$ as the intersection of these two sets, $\mathscr{Z}_0=\mathscr{Z}_1\cap \mathscr{Z}_2$ (see Appendices \ref{sec.estimation Z_0 ordered} and \ref{sec.estimation Z_0 unordered}). 
Lemma \ref{lemma.testable implications for Z2 weaker} shows that $\mathscr{Z}_1\subseteq\mathscr{Z}_2$ when $D$ is binary. For multivalued $D$, we are not aware of such a result.

\begin{lemma}\label{lemma.testable implications for Z2 weaker}
    If $D\in\{0,1\}$, then $\mathscr{Z}_1\subseteq\mathscr{Z}_2$, and the testable restrictions defining $\mathscr{Z}_1$ are sharp.\footnote{The statement in Lemma \ref{lemma.testable implications for Z2 weaker} holds for ordered or unordered $D\in\mathcal{D}=\{d_1,d_2\}$.}
\end{lemma}

Lemma \ref{lemma.testable implications for Z2 weaker} shows that when $D\in\{0,1\}$, the testable implications of \citet{kedagni2020generalized} are implied by those of \citet{kitagawa2015test}, \citet{mourifie2016testing}, and \citet{sun2021ivvalidity}. This result may be of independent interest.
Moreover, Lemma \ref{lemma.testable implications for Z2 weaker} establishes the sharpness of the testable implications used to define $\mathscr{Z}_1$. This follows from the arguments in Proposition 1.1(i) of \citet{kitagawa2015test}. 

Lemma \ref{lemma.testable implications for Z2 weaker} implies that $\mathscr{Z}_0=\mathscr{Z}_1\cap \mathscr{Z}_2=\mathscr{Z}_1$ when $D$ is binary. Therefore, we define $\mathscr{Z}_0$ based on the testable implications proposed in \citet{kitagawa2015test}, \citet{mourifie2016testing}, and \citet{sun2021ivvalidity}. These testable implications were originally proposed for full IV validity. In the following, we extend them to Definition \ref{def.partial validity pairwise binary D}. To describe the testable restrictions, we use the notation of \citet{sun2021ivvalidity}. Define conditional probabilities
\begin{equation*}
	P_z\left( B,C\right) =\mathbb{P}\left( Y\in B,D\in C|Z=z\right)
\end{equation*}
for all Borel sets $B,C\in\mathcal{B}_{\mathbb{R}}$ and all $z\in\mathcal{Z}$. With  $\mathscr{Z}_{\bar{M}}=\{(z_{k_1},z_{k_1^{\prime}}),\ldots,(z_{k_{\bar{M}}},z_{k_{\bar{M}}^{\prime}})\}$, for every $m\in\{1,\ldots,\bar{M}\}$, it follows that
\begin{align}\label{eq.testable implication binary D}
	P_{z_{k_m}}\left( B,\{1\}\right) \leq P_{z_{k_m^{\prime}}}\left(
	B,\{1\}\right)   \text{ and } P_{z_{k_m}}\left( B,\{0\}\right)  \geq P_{z_{k_m^{\prime}}}\left(
	B,\{0\}\right) 
\end{align}
for all $B\in\mathcal{B}_{\mathbb{R}}$.
By definition, for all $B,C\in \mathcal{B}_{\mathbb{R}}$, 
\begin{equation*}
	\mathbb{P}\left( Y\in B,D\in C|Z=z\right) =\frac{\mathbb{P}\left( Y\in
		B,D\in C,Z=z\right) }{\mathbb{P}\left( Z=z\right) }.
\end{equation*}
Define the function spaces 
\begin{align}\label{def.function spaces binary D}
	& \mathcal{G}_P=\left\{ \left( 1_{\mathbb{R}\times \mathbb{R}\times \left\{
		z_{k}\right\} },1_{\mathbb{R}\times \mathbb{R}\times \left\{ z_{k^{\prime
		}}\right\} }\right) : k,k^{\prime }\in\{1,\ldots, K\}, k\neq k'\right\} , \notag \\
	& \mathcal{H}=\left\{ \left( -1\right) ^{d}\cdot 1_{B\times \left\{
		d\right\} \times \mathbb{R}}:B\text{ is a closed interval in }\mathbb{R}%
	,d\in \{0,1\}\right\} , \text{ and }\notag \\
	& \bar{\mathcal{H}}=\left\{ \left( -1\right) ^{d}\cdot 1_{B\times
		\left\{ d\right\} \times \mathbb{R}}:B\text{ is a closed, open, or
		half-closed interval in }\mathbb{R},d\in \left\{ 0,1\right\} \right\}.
\end{align}
Similarly to \citet{sun2021ivvalidity}, by Lemma B.7 in \citet{kitagawa2015test}, we use all closed intervals $%
B\subseteq \mathbb{R}$ to construct ${\mathcal{H}}$ instead of all Borel sets.

Suppose we have access to an i.i.d.\ sample $\{\left(
Y_{i},D_{i},Z_{i}\right) \}_{i=1}^{n}$ distributed according to some probability distribution $P$ in $\mathcal{P}$, that is,  $P(G)=\mathbb{P}((Y_{i},D_{i},Z_{i})\in G)$ for all measurable $G$. 
The closure of $\mathcal{H}$ in $L^{2}(P)$ is equal to $\bar{\mathcal{H}}$
by Lemma C.1 of \citet{sun2021ivvalidity}. For every $\left( h,g\right) \in {\bar{%
		\mathcal{H}}\times \mathcal{G}_P}$ with $g=(g_{1},g_{2})$, we define 
\begin{equation*}
	\phi \left( h,g\right) =\frac{P\left( h\cdot g_{2}\right) }{P\left(
		g_{2}\right) }-\frac{P\left( h\cdot g_{1}\right) }{P\left( g_{1}\right) }
\end{equation*}
and
\begin{align}\label{eq.stat variance multi}
	\sigma^2(h,g)=\Lambda(P) \cdot \left\{\frac{ P\left(  h^2\cdot g_{2}\right)   }{P^{2}\left(g_{2}\right)  }  -\frac{ P^2\left(  h\cdot g_{2}\right)   }{P^{3}\left(
		g_{2}\right)  }  
	+\frac{ P\left(	h^2\cdot g_{1}\right)   }{P^{2}\left(  g_{1}\right)  }
	-\frac{ P^2\left(	h\cdot g_{1}\right)   }{P^{3}\left(  g_{1}\right)  }\right\},
\end{align}
where $\Lambda(P)=\prod_{k=1}^{K}P(1_{\mathbb{R}\times\mathbb{R}\times\{z_k\}})$ and $P^m(g_j)=[P(g_j)]^m$ for $m\in\mathbb{N}$ and $j\in\{1,2\}$. 
We denote the sample analog of $\phi $ as 
\begin{equation*}
	\widehat{\phi}\left( h,g\right) =\frac{\widehat{P}(h\cdot g_{2})}{\widehat{P}(g_{2})}-%
	\frac{\widehat{P}(h\cdot g_{1})}{\widehat{P}(g_{1})},
\end{equation*}%
where $\widehat{P}$ is the empirical probability measure corresponding to $P$ so that for every measurable function $v$ (by abuse of notation), 
\begin{equation}\label{eq.empirical P}
	\widehat{P}\left( v\right) =\frac{1}{n}\sum_{i=1}^{n}v\left(
	Y_{i},D_{i},Z_{i}\right).
\end{equation}
 For every $\left( h,g\right) \in \bar{\mathcal{H}}\times\mathcal{
	 G}_P$ with $g=\left( g_{1},g_{2}\right) $, define the sample analog of $\sigma^2(h,g)$ as
\begin{equation*}
	\widehat{\sigma}^{2}\left( h,g\right) =\frac{T_{n}}{n}\cdot \left\{ \frac{\widehat{P}%
		\left( h^{2}\cdot g_{2}\right) }{\widehat{P}^{2}\left( g_{2}\right) }-\frac{\widehat{%
			P}^{2}\left( h\cdot g_{2}\right) }{\widehat{P}^{3}\left( g_{2}\right) }+\frac{%
		\widehat{P}\left( h^{2}\cdot g_{1}\right) }{\widehat{P}^{2}\left( g_{1}\right) }-%
	\frac{\widehat{P}^{2}\left( h\cdot g_{1}\right) }{\widehat{P}^{3}\left( g_{1}\right) 
	}\right\} ,
\end{equation*}%
where $T_{n}=n\cdot \prod_{k=1}^{K}\widehat{P}(1_{\mathbb{R}\times \mathbb{R}%
	\times \{z_{k}\}})$. By \eqref{eq.0timesinfinity}, $\widehat{\sigma}^{2}$ is well defined. By similar arguments as in the proof of Lemma 3.1 in \citet{sun2021ivvalidity}, ${\sigma}^{2}$ and $\widehat{\sigma}^{2}$ are uniformly bounded in $(h,g)$. 
	The following lemma reformulates the testable restrictions in \eqref{eq.testable implication binary D} in terms of $\phi$. Below, we use this reformulation to define $\mathscr{Z}_0$ and the corresponding estimator $\widehat{\mathscr{Z}_0}$.

\begin{lemma}\label{lemma.superset of Z pairwise binary D}
	Suppose that the instrument $Z$ is pairwise valid for the treatment $D$ with the largest validity pair set $\mathscr{Z}_{\bar{M}}=\{(z_{k_1},z_{k_1^{\prime}}),\ldots,(z_{k_{\bar{M}}},z_{k_{\bar{M}}^{\prime}})\}$. For every $m\in\{1,\ldots,\bar{M}\}$, $\sup_{h\in {\mathcal{H}}}\phi \left( h,g\right) =0$ with $g=( 1_{\mathbb{R}\times \mathbb{R}\times \{	z_{k_m}\} },1_{\mathbb{R}\times \mathbb{R}\times \{ z_{k_{m}^{\prime}}\} })$. 
\end{lemma}

Lemma \ref{lemma.superset of Z pairwise binary D} reformulates the necessary conditions based on \citet{kitagawa2015test}, \citet{mourifie2016testing}, and \citet{sun2021ivvalidity} for the  validity pair set $\mathscr{Z}_{\bar{M}}$. Define 
\begin{equation}\label{eq.G0 pair binary D}
	\mathcal{G}_{0}=\left\{ g\in \mathcal{G}_P:\sup_{h\in {\mathcal{H}}}\phi\left(
	h,g\right) =0\right\} \text{ and } \widehat{\mathcal{G}_{0}}=\left\{ g\in 
	\mathcal{G}_P:\sqrt{T_n}\left\vert \sup_{h\in {\mathcal{H}} }\frac{\widehat{\phi}
		\left( h,g\right)}{\xi_{0}\vee \widehat{\sigma}(h,g)} \right\vert \leq \tau
	_{n}^{g}\right\},
\end{equation}
where $1/\min_{g\in\mathcal{G}_P}\tau_{n}^g\to0$ in probability and $\max_{g\in\mathcal{G}_P}\tau_{n}^g/\sqrt{n}\to0$ in probability as $n\to\infty$, and $\xi_{0}$ is a small positive number.\footnote{ In practice, we use $\xi_0=10^{-100}$.} Here, we allow $\tau_n^g$ to be different for every $g$. This added flexibility helps improve the finite sample performance of VSIV estimation. We discuss our explicit choice of $\tau_n^g$ in more detail in Section \ref{sec.simulation}.

The set $\mathcal{G}_0$ is different from the contact sets defined in \citet{Beare2015improved}, \citet{Beare2017improved}, and \citet{sun2021ivvalidity} in different contexts because of the presence of the map $\sup$. We refer to \citet{linton2010improved} and \citet{lee2018testing} for further discussions of contact set estimation.
Define ${\mathscr{Z}_0}$ as the collection of all $(z,z')$ associated with some $g\in{\mathcal{G}_0}$:
\begin{align}\label{eq.Z1 pair binary D}
	\mathscr{Z}_0=\left\{ (z_{k},z_{k^{\prime}})\in\mathscr{Z}: g=( 1_{\mathbb{R}\times \mathbb{R}\times \{	z_{k}\} },1_{\mathbb{R}\times \mathbb{R}\times \{ z_{k^{\prime}}\} })\in\mathcal{G}_0\right\}.
\end{align}
Note that $\mathscr{Z}_0$ is the set of all pairs that satisfy the testable implications in \eqref{eq.testable implication binary D}.
For example, if $K=4$ and ${\mathcal{G}_0}=\{( 1_{\mathbb{R}\times \mathbb{R}\times \left\{
	z_{1}\right\} },1_{\mathbb{R}\times \mathbb{R}\times \left\{ z_{2}\right\} }), ( 1_{\mathbb{R}\times \mathbb{R}\times \left\{
	z_{3}\right\} },1_{\mathbb{R}\times \mathbb{R}\times \left\{ z_{4}\right\} })\}$, then ${\mathscr{Z}_0}=\{(z_1,z_2),(z_3,z_4)\}$.

We use $\widehat{\mathcal{G}_0}$ to construct the estimator of $\mathscr{Z}_0$, denoted by $\widehat{\mathscr{Z}_0}$, which is defined as the set of all $(z,z^{\prime})$ associated with some $g\in\widehat{\mathcal{G}_0}$:
\begin{align}\label{eq.Z1_hat pair binary D}
	\widehat{\mathscr{Z}_0}=\left\{ (z_{k},z_{k^{\prime}})\in\mathscr{Z}: g=( 1_{\mathbb{R}\times \mathbb{R}\times \{	z_{k}\} },1_{\mathbb{R}\times \mathbb{R}\times \{ z_{k^{\prime}}\} })\in\widehat{\mathcal{G}_0}\right\}.
\end{align}
Note that \eqref{eq.Z1_hat pair binary D} is the sample analog of \eqref{eq.Z1 pair binary D}.
The following proposition establishes consistency of $\widehat{\mathscr{Z}_0}$.
\begin{proposition}
	\label{prop.consistent G hat pairwise Z1 binary D} Under Assumption \ref{ass.iid data binary D},  $\mathbb{P}(\widehat{\mathcal{G}_0}=\mathcal{G}_0)\rightarrow 1$, and thus $\mathbb{P}(\widehat{\mathscr{Z}_{0}}=\mathscr{Z}_{0})\rightarrow 1$.
\end{proposition}

Proposition \ref{prop.consistent G hat pairwise Z1 binary D} is related to the contact set estimation in \citet{sun2021ivvalidity}. Since, by definition, $\mathcal{G}_0\subseteq\mathcal{G}_P$ and $\mathcal{G}_P$ is a finite set, we can use techniques similar to those in \citet{sun2021ivvalidity} to obtain the stronger result in Proposition \ref{prop.consistent G hat pairwise Z1 binary D}, that is, $\mathbb{P}(\widehat{\mathcal{G}}_0=\mathcal{G}_0)\rightarrow 1$.

\begin{remark}[Uniqueness of $\mathscr{Z}_{\bar{M}}$ and $\mathscr{Z}_0$]
    By definition, $\mathscr{Z}_{\bar{M}}$ is the collection of all valid pairs of values of $Z$, while $\mathscr{Z}_0$ is the collection of all pairs that satisfy the testable restrictions in Lemma \ref{lemma.superset of Z pairwise binary D}. Thus, both $\mathscr{Z}_{\bar{M}}$ and $\mathscr{Z}_0$ are unique, and clearly $\mathscr{Z}_{\bar{M}}\subseteq\mathscr{Z}_0$.
\end{remark}

\begin{remark}[Local Violations of the Testable Restrictions]
The theory for VSIV estimation relies on selection consistency, $\mathbb{P}(\widehat{\mathscr{Z}_{0}}=\mathscr{Z}_{0})\rightarrow 1$. A potential concern with this approach is that selection consistency is theoretically only possible if the violations of the testable restrictions are well-separated from zero. 
However, VSIV estimation remains useful even in the presence of local (to-zero) violations. It will always yield asymptotic bias reductions from removing pairs of instrument values for which the violations are well-separated from zero. An important feature of VSIV estimation is that it proceeds pairwise, so that the presence of pairs corresponding to local violations does not impact the selection performance for pairs corresponding to well-separated violations.
We explore the performance of our method when we change the magnitude of the violations in Appendix \ref{sec.additional simulations}.
\end{remark}

\section{Simulations and Application}\label{sec.simulation}

In this section, we evaluate the finite sample performance of VSIV estimation in Monte Carlo simulations designed based on an empirical application. First, we discuss the empirical context. Second, we present the results from the simulation study and determine the choice of the tuning parameter $\tau_n^g$. Finally, we present the results from the empirical application.

\subsection{Empirical Context}

We revisit the analysis of the causal effect of college education on earnings. We use the dataset analyzed by \citet{heckman2001four} and \citet{kedagni2020discordant}, consisting of data on 1,230 white males from the National Longitudinal Survey of Youth of 1979 (NLSY). The outcome of interest ($Y$) is the log wage and the treatment ($D$) is a dummy variable for college enrollment. \citet{kedagni2020discordant} use the maximum parental education ($E$) as an instrument for college enrollment. Since the overall sample size is relatively small, we consider a coarsened version of this instrument: ${Z}=1\{E<12\}+2\cdot1\{E=12\}+3\cdot1\{12<E<16\}+4\cdot1\{E\ge 16\}$. 

The parental education instrument likely violates the exclusion restriction due to its potential positive effect on earnings, especially at lower levels of parental education. Therefore, \citet{kedagni2020discordant} consider a relaxation of IV validity, building on \citet{manski2000monotone,manski2009more}. As discussed in Section \ref{sec.pairwise valid instrument binary D}, their relaxation allows the exclusion restriction to be violated for IV values below a cutoff, which they estimate to be $E=11$ (Table 2, Column (3)). The results in \citet{kedagni2020discordant} suggest that the parental education instrument is partially invalid, especially for pairs  of instrument values corresponding to lower levels of parental education, thus providing an ideal setting for illustrating the usefulness of VSIV estimation.

We emphasize that there are two major differences between our analysis here and the one in \citet{kedagni2020discordant}. First, we consider a different relaxation of IV validity. Unlike \citet{kedagni2020discordant}, we do not impose any assumptions on how IV validity fails for invalid pairs. Second, \citet{kedagni2020discordant} focus on partial identification of the average treatment effect. By contrast, we consider the estimation of LATEs for all pairs of instrument values that are not screened out based on the testable implications.

\subsection{Simulation Evidence}

\subsubsection{Data Generating Processes}
Here we describe the data generating processes (DGPs) that we use in the simulations. All DGPs are calibrated to the joint empirical distribution of $(D,Z)$ in the application. Denote by $\widehat{\mathbb{P}}$ the empirical probability measure of ${\mathbb{P}}$. 
In the empirical application, we have $\widehat{\mathbb{P}}(Z=1)=0.1317$, $\widehat{\mathbb{P}}(Z=2)=0.4716$, $\widehat{\mathbb{P}}(Z=3)=0.1495$, $\widehat{\mathbb{P}}(Z=4)=0.2472$, 
$\widehat{\mathbb{P}}(D=1|Z=1)=0.1420$, $\widehat{\mathbb{P}}(D=1|Z=2)=0.3086$, $\widehat{\mathbb{P}}(D=1|Z=3)=0.5054$, and $\widehat{\mathbb{P}}(D=1|Z=4)=0.7796$. Given the prior information that $Z$ may be positively correlated with $D$, we assume that $\mathscr{Z}_P=\{(1,2),(1,3),(1,4),(2,3),(2,4),(3,4)\}$.\footnote{This is consistent with the fact that the empirical proportions $\widehat{\mathbb{P}}(D=1|Z=z)$ are increasing in $z$.} 
We consider five data generating processes (DGPs (0)--(4)), where Definition \ref{def.partial validity pairwise binary D} holds for each pair in $\mathscr{Z}_P$ under DGP (0) and is violated for every pair under DGPs (1)--(4). These DGPs are constructed following those in \citet{kitagawa2015test} and \citet{sun2021ivvalidity}.

For all DGPs, we specify $U\sim\mathrm{Unif}(0,1)$, $V\sim\mathrm{Unif}(0,1)$, $W\sim\mathrm{Unif}(0,1)$, $Z=1\{U \le 0.1317\}+2\times\{0.1317<U \le 0.6033\}+3\times\{0.6033<U\le 0.7528\}+4\times 1\{U>0.7528\}$,
$D_1=1\{V\le 0.1420\}$, $D_2=1\{V\le 0.3086\}$, $D_3=1\{V\le 0.5054\}$, $D_4=1\{V\le 0.7796\}$, and $D=\sum_{z=1}^4 1\{Z=z\}\times D_z$.
We let $N_Z\sim N(0,1)$, and let $N_{10a}(\sigma)\sim N(-1,\sigma^2)$, $N_{10b}(\sigma)\sim N(-0.5,\sigma^2)$, $N_{10c}(\sigma)\sim N(0,\sigma^2)$, $N_{10d}(\sigma)\sim N(0.5,\sigma^2)$, $N_{10e}(\sigma)\sim N(1,\sigma^2)$, and $N_{C}(\sigma)=1\{W\le 0.15\}\times N_{10a}(\sigma)+1\{0.15<W\le 0.35\}\times N_{10b}(\sigma)+1\{0.35<W\le 0.65\}\times N_{10c}(\sigma)+1\{0.65<W\le 0.85\}\times N_{10d}(\sigma)+1\{W>0.85\}\times N_{10e}(\sigma)$ for $\sigma>0$. 
The DGPs (0)--(4) are specified as follows. 

\begin{enumerate}[start=0,label=(\arabic*):]

    \item $N_{dz}=N_Z$ for each $d\in\{0,1\}$ and all $z\in\{1,2,3,4\}$, $Y=\sum_{z=1}^{4}1\{Z=z\}\times(\sum_{d=0}^{1} 1\{D=d\}\times N_{dz})$
    
	\item  For $(z_1,z_2)\in\mathscr{Z}_P$, $N_{1z_1}\sim N(\mu_{(z_1,z_2)},1)$, $N_{1z_2}\sim N(0,1)$, $N_{0z_1}\sim N(0,1)$, $N_{0z_2}\sim N(\mu_{(z_1,z_2)},1)$ with $\mu_{(1,2)}=-0.9$, $\mu_{(1,3)}=-1.1$, $\mu_{(1,4)}=-1.3$, $\mu_{(2,3)}=-0.9$, $\mu_{(2,4)}=-1.1$, $\mu_{(3,4)}=-0.9$, $N_{dz}\sim N(0,1)$ for $d\in\{0,1\}$ and $z\in\{1,2,3,4\}\setminus \{z_1,z_2\}$, 
    $Y=\sum_{z=1}^{4}1\{Z=z\}\times(\sum_{d=0}^{1} 1\{D=d\}\times N_{dz})$
	
	\item For $(z_1,z_2)\in\mathscr{Z}_P$, $N_{1z_1}\sim N(0,1)$, $N_{1z_2}\sim N(0,\sigma_{(z_1,z_2)}^2)$, $N_{0z_1}\sim N(0,\sigma_{(z_1,z_2)}^2)$, $N_{0z_2}\sim N(0,1)$ with $\sigma_{(1,2)}=3$, $\sigma_{(1,3)}=5$, $\sigma_{(1,4)}=7$, $\sigma_{(2,3)}=3$, $\sigma_{(2,4)}=5$, $\sigma_{(3,4)}=3$, $N_{dz}\sim N(0,1)$ for $d\in\{0,1\}$ and $z\in\{1,2,3,4\}\setminus \{z_1,z_2\}$, 
    $Y=\sum_{z=1}^{4}1\{Z=z\}\times(\sum_{d=0}^{1} 1\{D=d\}\times N_{dz})$
	
	\item For $(z_1,z_2)\in\mathscr{Z}_P$, $N_{1z_1}\sim N(0,1)$, $N_{1z_2}\sim N(0,\sigma_{(z_1,z_2)}^2)$, $N_{0z_1}\sim N(0,\sigma_{(z_1,z_2)}^2)$, $N_{0z_2}\sim N(0,1)$ with $\sigma_{(1,2)}=0.5$, $\sigma_{(1,3)}=0.45$, $\sigma_{(1,4)}=0.4$, $\sigma_{(2,3)}=0.5$, $\sigma_{(2,4)}=0.45$, $\sigma_{(3,4)}=0.5$, $N_{dz}\sim N(0,1)$ for $d\in\{0,1\}$ and $z\in\{1,2,3,4\}\setminus \{z_1,z_2\}$, 
    $Y=\sum_{z=1}^{4}1\{Z=z\}\times(\sum_{d=0}^{1} 1\{D=d\}\times N_{dz})$

	\item For $(z_1,z_2)\in\mathscr{Z}_P$, $N_{1z_1}\sim N_C(\sigma_{(z_1,z_2)})$, $N_{1z_2}\sim N(0,1)$, $N_{0z_1}\sim N(0,1)$, $N_{0z_2}\sim N_C(\sigma_{(z_1,z_2)})$ with $\sigma_{(1,2)}=0.08$, $\sigma_{(1,3)}=0.04$, $\sigma_{(1,4)}=0.02$, $\sigma_{(2,3)}=0.08$, $\sigma_{(2,4)}=0.04$, $\sigma_{(3,4)}=0.08$, $N_{dz}\sim N(0,1)$ for $d\in\{0,1\}$ and $z\in\{1,2,3,4\}\setminus \{z_1,z_2\}$, 
    $Y=\sum_{z=1}^{4}1\{Z=z\}\times(\sum_{d=0}^{1} 1\{D=d\}\times N_{dz})$

\end{enumerate}
The random variables $U$, $V$, $W$, and those generated from $N(\mu,\sigma^2)$ for some $\mu$ and $\sigma$ are mutually independent.

\subsubsection{Choice of $\tau_n^g$}
\label{sec.validity set estimation and choice of tau}

As shown in \eqref{eq.G0 pair binary D}, the tuning parameter $\tau_n^g$ should satisfy that $1/\min_{g\in\mathcal{G}_P}\tau_{n}^g\to0$ in probability and $\max_{g\in\mathcal{G}_P}\tau_{n}^g/\sqrt{n}\to0$ in probability as $n\to\infty$. In our simulations, for every $g=( 1_{\mathbb{R}\times \mathbb{R}\times \{	z_{k}\} },1_{\mathbb{R}\times \mathbb{R}\times \{ z_{k^{\prime}}\} })$, we set $\tau_n^g$ as
\begin{align}
    \tau_n^g=c\cdot \frac{(n\widehat{\mathbb{P}}(Z\in\{z_k,z_{k'}\}))^{1/5}}{\vert \widehat{\mathbb{P}}(D=1|Z=z_k)-\widehat{\mathbb{P}}(D=1|Z=z_{k'}) \vert^{1/5}}\label{eq:tuning_parameter}
\end{align}
for some constant $c>0$.\footnote{If ${\mathbb{P}}(D=1|Z=z_k)-{\mathbb{P}}(D=1|Z=z_{k'})=0$ for some pair $(z_k,z_{k'})$, we still have that $\tau_n^g/\sqrt{n}\to0$ in probability as $n\to\infty$ for every $g\in\mathcal{G}_P$, and $\tau_n^g$ satisfies the two conditions above in this case.} The numerator of \eqref{eq:tuning_parameter} captures the relevant sample size, and the denominator captures the idea that if the difference is large, violations are harder to detect. This (heuristic) adjustment in the denominator does not affect the asymptotic properties but works well in simulations.

\subsubsection{Simulation Results}

We consider two different sample sizes: $n=1230$ as in the empirical application and $n=2460$ to explore how the properties of our methods improve as the sample size increases. We report results for $\tau_n^g$ in \eqref{eq:tuning_parameter} with $c\in \{0.1,0.2,\dots,1\}$. For each simulation, we use 1,000 Monte Carlo iterations. To calculate the supremum in $\sqrt{T_n}\vert \sup_{h\in\mathcal{H}}{\widehat{\phi}
	\left( h,g\right)}/({\xi_{0}\vee \widehat{\sigma}(h,g)}) \vert $ for every $g$, we use
the approach employed by \citet{kitagawa2015test} and \citet{sun2021ivvalidity}. Specifically,
we compute the supremum based on the closed intervals $[a,b]$ with the realizations of $\{Y_i\}$ as endpoints, that is, intervals $[a,b]$ where $a,b\in\{Y_i\}$ and $a\le b$.

Tables \ref{tab:DGP0}--\ref{tab:DGP4} show the empirical probabilities with which each element of $\mathscr{Z}_P$ is selected to be in $\widehat{\mathscr{Z}_0}$ in the simulations. The results show that choosing $c$ is subject to a trade-off between the ability of our method to screen out invalid pairs and its ability to include valid pairs. Given the nature of the method, screening out invalid pairs is particularly important since LATE estimators based on these pairs are inconsistent. Our method with $c=0.6$ detects invalid pairs almost perfectly while selecting most of the valid pairs with high probability. With $c=0.6$, as $n$ increases from $1230$ to $2460$, the selection rates for valid pairs are increasing and those for invalid pairs are decreasing to $0$. Overall, the simulation results show that the proposed method performs well in identifying the validity pair set in finite samples. 
\begin{table}[h!]
	
	\centering
	\caption{Validity Pair Set Estimation for DGP (0)}
	\scalebox{0.9}{
		\begin{tabular}{  c  c  c  c  c  c  c  c  }
			\hline
			\hline
$n$  &	$c $ & (1, 2) & {(1, 3)} & {(1, 4)} 
			 & {{(2, 3)}} & {(2, 4)}
			 & {{(3, 4)}} \\
			\hline
\multirow{10}{*}{1230}   &  0.1 & 0.000  & 0.000  & 0.000  & 0.000  & 0.000  & 0.000  \\ 
    &    0.2 & 0.000  & 0.000  & 0.000  & 0.000  & 0.000  & 0.000  \\ 
    &    0.3 & 0.000  & 0.000  & 0.000  & 0.000  & 0.000  & 0.000  \\ 
    &    0.4 & 0.000  & 0.000  & 0.002  & 0.000  & 0.000  & 0.001  \\ 
    &    0.5 & 0.000  & 0.053  & 0.198  & 0.003  & 0.140  & 0.217  \\ 
    &    0.6 & 0.003  & 0.585  & 0.726  & 0.264  & 0.733  & 0.821  \\ 
    &    0.7 & 0.138  & 0.901  & 0.953  & 0.736  & 0.967  & 0.985  \\ 
    &    0.8 & 0.542  & 0.983  & 0.996  & 0.956  & 0.998  & 1.000  \\ 
    &    0.9 & 0.865  & 0.997  & 1.000  & 0.996  & 1.000  & 1.000  \\ 
    &    1 & 0.968  & 1.000  & 1.000  & 1.000  & 1.000  & 1.000  \\
        \hline
\multirow{10}{*}{2460}        
    &    0.1 & 0.000  & 0.000  & 0.000  & 0.000  & 0.000  & 0.000  \\ 
    &    0.2 & 0.000  & 0.000  & 0.000  & 0.000  & 0.000  & 0.000  \\ 
    &    0.3 & 0.000  & 0.000  & 0.000  & 0.000  & 0.000  & 0.000  \\ 
    &    0.4 & 0.000  & 0.000  & 0.000  & 0.000  & 0.000  & 0.006  \\ 
    &    0.5 & 0.000  & 0.182  & 0.370  & 0.018  & 0.290  & 0.568  \\ 
    &    0.6 & 0.018  & 0.802  & 0.898  & 0.482  & 0.908  & 0.956  \\ 
    &    0.7 & 0.344  & 0.977  & 0.995  & 0.901  & 0.998  & 0.997  \\ 
    &    0.8 & 0.821  & 0.998  & 1.000  & 0.989  & 1.000  & 1.000  \\ 
    &    0.9 & 0.968  & 1.000  & 1.000  & 1.000  & 1.000  & 1.000  \\ 
    &    1 & 0.997  & 1.000  & 1.000  & 1.000  & 1.000  & 1.000  \\ 
 \hline
			\hline
		\end{tabular}
	}
	
	\label{tab:DGP0}
\end{table}
\begin{table}[h!]
	
	\centering
	\caption{Validity Pair Set Estimation for DGP (1)}
	\scalebox{0.9}{
		\begin{tabular}{  c  c  c  c  c  c  c  c  }
			\hline
			\hline
$n$  &	$c $ & (1, 2) & {(1, 3)} & {(1, 4)} 
			 & {{(2, 3)}} & {(2, 4)} 
			 & {{(3, 4)}} \\
			\hline
\multirow{10}{*}{1230}  &   0.1 & 0.000  & 0.000  & 0.000  & 0.000  & 0.000  & 0.000  \\ 
    &    0.2 & 0.000  & 0.000  & 0.000  & 0.000  & 0.000  & 0.000  \\ 
    &    0.3 & 0.000  & 0.000  & 0.000  & 0.000  & 0.000  & 0.000  \\ 
    &    0.4 & 0.000  & 0.000  & 0.000  & 0.000  & 0.000  & 0.000  \\ 
    &    0.5 & 0.000  & 0.000  & 0.000  & 0.000  & 0.000  & 0.000  \\ 
    &    0.6 & 0.000  & 0.002  & 0.011  & 0.000  & 0.001  & 0.017  \\ 
    &    0.7 & 0.000  & 0.021  & 0.104  & 0.020  & 0.028  & 0.108  \\ 
    &    0.8 & 0.000  & 0.086  & 0.298  & 0.126  & 0.148  & 0.311  \\ 
    &    0.9 & 0.010  & 0.199  & 0.538  & 0.390  & 0.357  & 0.583  \\ 
    &    1 & 0.051  & 0.375  & 0.743  & 0.689  & 0.592  & 0.787  \\ 
        
    \hline
\multirow{10}{*}{2460}    &    0.1 & 0.000  & 0.000  & 0.000  & 0.000  & 0.000  & 0.000  \\ 
    &    0.2 & 0.000  & 0.000  & 0.000  & 0.000  & 0.000  & 0.000  \\ 
    &    0.3 & 0.000  & 0.000  & 0.000  & 0.000  & 0.000  & 0.000  \\ 
    &    0.4 & 0.000  & 0.000  & 0.000  & 0.000  & 0.000  & 0.000  \\ 
    &    0.5 & 0.000  & 0.000  & 0.000  & 0.000  & 0.000  & 0.000  \\ 
    &    0.6 & 0.000  & 0.000  & 0.004  & 0.000  & 0.000  & 0.001  \\ 
    &    0.7 & 0.000  & 0.003  & 0.058  & 0.003  & 0.010  & 0.019  \\ 
    &    0.8 & 0.001  & 0.010  & 0.219  & 0.037  & 0.046  & 0.093  \\ 
    &    0.9 & 0.001  & 0.037  & 0.405  & 0.182  & 0.164  & 0.276  \\ 
    &    1 & 0.005  & 0.118  & 0.648  & 0.509  & 0.352  & 0.533  \\ 
 \hline
			\hline
		\end{tabular}
	}
	
	\label{tab:DGP1}
\end{table}
\begin{table}[h!]
	
	\centering
	\caption{Validity Pair Set Estimation for DGP (2)}
	\scalebox{0.9}{
		\begin{tabular}{  c  c  c  c  c  c  c  c  }
			\hline
			\hline
$n$  &	$c $ & (1, 2) & {(1, 3)} & {(1, 4)} 
			 & {{(2, 3)}} & {(2, 4)} 
			 & {{(3, 4)}}\\
			\hline
\multirow{10}{*}{1230}  &   0.1 & 0.000  & 0.000  & 0.000  & 0.000  & 0.000  & 0.000  \\ 
    &    0.2 & 0.000  & 0.000  & 0.000  & 0.000  & 0.000  & 0.000  \\ 
    &    0.3 & 0.000  & 0.000  & 0.000  & 0.000  & 0.000  & 0.000  \\ 
    &    0.4 & 0.000  & 0.000  & 0.000  & 0.000  & 0.000  & 0.000  \\ 
    &    0.5 & 0.000  & 0.000  & 0.000  & 0.000  & 0.000  & 0.000  \\ 
    &    0.6 & 0.000  & 0.000  & 0.001  & 0.000  & 0.000  & 0.004  \\ 
    &    0.7 & 0.000  & 0.000  & 0.039  & 0.000  & 0.010  & 0.027  \\ 
    &    0.8 & 0.000  & 0.003  & 0.188  & 0.024  & 0.063  & 0.134  \\ 
    &    0.9 & 0.000  & 0.012  & 0.416  & 0.192  & 0.189  & 0.342  \\ 
    &    1 & 0.008  & 0.037  & 0.630  & 0.479  & 0.418  & 0.564  \\ 
        
    \hline
\multirow{10}{*}{2460}    &    0.1 & 0.000  & 0.000  & 0.000  & 0.000  & 0.000  & 0.000  \\ 
    &    0.2 & 0.000  & 0.000  & 0.000  & 0.000  & 0.000  & 0.000  \\ 
    &    0.3 & 0.000  & 0.000  & 0.000  & 0.000  & 0.000  & 0.000  \\ 
    &    0.4 & 0.000  & 0.000  & 0.000  & 0.000  & 0.000  & 0.000  \\ 
    &    0.5 & 0.000  & 0.000  & 0.000  & 0.000  & 0.000  & 0.000  \\ 
    &    0.6 & 0.000  & 0.000  & 0.001  & 0.000  & 0.000  & 0.000  \\ 
    &    0.7 & 0.000  & 0.000  & 0.030  & 0.000  & 0.002  & 0.004  \\ 
    &    0.8 & 0.000  & 0.000  & 0.185  & 0.007  & 0.021  & 0.050  \\ 
    &    0.9 & 0.000  & 0.000  & 0.436  & 0.041  & 0.093  & 0.135  \\ 
    &    1 & 0.000  & 0.000  & 0.652  & 0.185  & 0.241  & 0.282  \\
 \hline
			\hline
		\end{tabular}
	}
	
	\label{tab:DGP2}
\end{table}
\begin{table}[h!]
	
	\centering
	\caption{Validity Pair Set Estimation for DGP (3)}
	\scalebox{0.9}{
		\begin{tabular}{  c  c  c  c  c  c  c  c }
			\hline
			\hline
$n$  &	$c $ & (1, 2) & {(1, 3)} & {(1, 4)} 
			 & {{(2, 3)}} & {(2, 4)} 
			 & {{(3, 4)}} \\
			\hline
\multirow{10}{*}{1230}  &   0.1 & 0.000  & 0.000  & 0.000  & 0.000  & 0.000  & 0.000  \\ 
    &    0.2 & 0.000  & 0.000  & 0.000  & 0.000  & 0.000  & 0.000  \\ 
    &    0.3 & 0.000  & 0.000  & 0.000  & 0.000  & 0.000  & 0.000  \\ 
    &    0.4 & 0.000  & 0.000  & 0.000  & 0.000  & 0.000  & 0.000  \\ 
    &    0.5 & 0.000  & 0.000  & 0.000  & 0.000  & 0.000  & 0.000  \\ 
    &    0.6 & 0.000  & 0.000  & 0.000  & 0.000  & 0.000  & 0.002  \\ 
    &    0.7 & 0.000  & 0.005  & 0.022  & 0.002  & 0.000  & 0.050  \\ 
    &    0.8 & 0.000  & 0.033  & 0.141  & 0.064  & 0.004  & 0.282  \\ 
    &    0.9 & 0.000  & 0.164  & 0.337  & 0.270  & 0.035  & 0.668  \\ 
    &    1 & 0.003  & 0.436  & 0.630  & 0.556  & 0.182  & 0.922  \\ 
        
\hline
\multirow{10}{*}{2460}    &    0.1 & 0.000  & 0.000  & 0.000  & 0.000  & 0.000  & 0.000  \\ 
    &    0.2 & 0.000  & 0.000  & 0.000  & 0.000  & 0.000  & 0.000  \\ 
    &    0.3 & 0.000  & 0.000  & 0.000  & 0.000  & 0.000  & 0.000  \\ 
    &    0.4 & 0.000  & 0.000  & 0.000  & 0.000  & 0.000  & 0.000  \\ 
    &    0.5 & 0.000  & 0.000  & 0.000  & 0.000  & 0.000  & 0.000  \\ 
    &    0.6 & 0.000  & 0.000  & 0.000  & 0.000  & 0.000  & 0.000  \\ 
    &    0.7 & 0.000  & 0.000  & 0.001  & 0.000  & 0.000  & 0.001  \\ 
    &    0.8 & 0.000  & 0.000  & 0.006  & 0.002  & 0.000  & 0.013  \\ 
    &    0.9 & 0.000  & 0.003  & 0.050  & 0.055  & 0.000  & 0.193  \\ 
    &    1 & 0.000  & 0.058  & 0.214  & 0.237  & 0.006  & 0.593  \\ 
 \hline
			\hline
		\end{tabular}
	}
	
	\label{tab:DGP3}
\end{table}
\begin{table}[h!]
	
	\centering
	\caption{Validity Pair Set Estimation for DGP (4)}
	\scalebox{0.9}{
		\begin{tabular}{  c  c  c  c  c  c  c  c }
			\hline
			\hline
$n$  &	$c $ & (1, 2) & {(1, 3)} & {(1, 4)} 
			 & {{(2, 3)}} & {(2, 4)} 
			 & {{(3, 4)}} \\
			\hline
\multirow{10}{*}{1230}  &   0.1 & 0.000  & 0.000  & 0.000  & 0.000  & 0.000  & 0.000  \\ 
    &    0.2 & 0.000  & 0.000  & 0.000  & 0.000  & 0.000  & 0.000  \\ 
    &    0.3 & 0.000  & 0.000  & 0.000  & 0.000  & 0.000  & 0.000  \\ 
    &    0.4 & 0.000  & 0.000  & 0.000  & 0.000  & 0.000  & 0.000  \\ 
    &    0.5 & 0.000  & 0.000  & 0.000  & 0.000  & 0.000  & 0.000  \\ 
    &    0.6 & 0.000  & 0.000  & 0.000  & 0.000  & 0.000  & 0.030  \\ 
    &    0.7 & 0.000  & 0.001  & 0.001  & 0.036  & 0.002  & 0.302  \\ 
    &    0.8 & 0.002  & 0.023  & 0.051  & 0.269  & 0.054  & 0.700  \\ 
    &    0.9 & 0.043  & 0.173  & 0.239  & 0.606  & 0.241  & 0.914  \\ 
    &    1 & 0.180  & 0.433  & 0.530  & 0.843  & 0.508  & 0.983  \\ 
        
    \hline
\multirow{10}{*}{2460}    &    0.1 & 0.000  & 0.000  & 0.000  & 0.000  & 0.000  & 0.000  \\ 
    &    0.2 & 0.000  & 0.000  & 0.000  & 0.000  & 0.000  & 0.000  \\ 
    &    0.3 & 0.000  & 0.000  & 0.000  & 0.000  & 0.000  & 0.000  \\ 
    &    0.4 & 0.000  & 0.000  & 0.000  & 0.000  & 0.000  & 0.000  \\ 
    &    0.5 & 0.000  & 0.000  & 0.000  & 0.000  & 0.000  & 0.000  \\ 
    &    0.6 & 0.000  & 0.000  & 0.000  & 0.000  & 0.000  & 0.016  \\ 
    &    0.7 & 0.000  & 0.000  & 0.000  & 0.031  & 0.000  & 0.220  \\ 
    &    0.8 & 0.000  & 0.002  & 0.015  & 0.261  & 0.021  & 0.613  \\ 
    &    0.9 & 0.010  & 0.030  & 0.136  & 0.649  & 0.150  & 0.865  \\ 
    &    1 & 0.114  & 0.168  & 0.381  & 0.883  & 0.395  & 0.962  \\ 
 \hline
			\hline
		\end{tabular}
	}
	
	\label{tab:DGP4}
\end{table}

Table \ref{tab:DGP0Inference} shows the coverage rates of the confidence intervals constructed based on the asymptotic distribution in Theorem \ref{thm.IV estimator asymptotics pairwise binary D} for DGP (0) under which all pairs are valid. We construct the confidence interval as follows: If a pair is in the estimated validity pair set, then the confidence interval is constructed based on the asymptotic distribution; if the pair is not in the estimated validity pair set, then the confidence interval is $\mathbb{R}$, since in this case the LATE is not identified. The results in Table \ref{tab:DGP0Inference} show that the coverage rates are close to $1-\alpha$, where $\alpha=0.05$, for most pairs when $c=0.6$. Since the sample size for each pair is relatively small in our simulations for both choices of $n$, the coverage rates can be somewhat higher than $0.95$ because of the way we construct the confidence intervals if the pair is not in the estimated validity pair set. 

In Appendix \ref{sec.simulation balanced}, we provide additional simulation results for a variant of DGP (0) with a more balanced design. The results in Tables \ref{tab:DGP0Balanced} and \ref{tab:DGP0InferenceBalanced} show that in this case for $c=0.6$, the selection rates for all pairs are high and converging to one, and the coverage rates are converging to $95\%$.

In Tables \ref{tab:RMSEDGP1}--\ref{tab:RMSEDGP4}, we explore the asymptotic bias reduction property of VSIV estimation by presenting the Root Mean Square Errors (RMSEs) of $\sqrt{n}(\widehat{\beta}_{(k,k')}^1-{\beta}_{(k,k')}^1)$ for every pair in $\mathscr{Z}_P$ under DGPs (1)--(4) under which all instrument pairs are invalid. As shown in Theorem \ref{thm.IV estimator asymptotics pairwise binary D}, $\sqrt{n}(\widehat{\beta}_{(k,k')}^1-{\beta}_{(k,k')}^1)\to0$ in probability if $(z_k,z_{k'})$ is invalid. Our simulation results show that overall, the RMSEs are getting closer to $0$ as $n$ increases for $c=0.6$. For larger $c$, the RMSEs for some invalid pairs may not be small enough in finite samples. Also note that for larger $c$, the RMSEs may not be decreasing in the sample size due to the rescaling by $\sqrt{n}$.

According to these simulation results, we suggest using $c=0.6$ in applications. The results for other values of $c$ may also be presented for consideration.

\begin{table}[h!]
	
	\centering
	\caption{Coverage Rates of the Confidence Intervals for DGP (0)}
	\scalebox{0.9}{
		\begin{tabular}{  c  c  c  c  c  c  c  c  }
			\hline
			\hline
$n$  &	$c $ & (1, 2) & {(1, 3)} & {(1, 4)} 
			 & {{(2, 3)}} & {(2, 4)}
			 & {{(3, 4)}} \\
			\hline
\multirow{10}{*}{1230}   
    &    0.1 & 1.000  & 1.000  & 1.000  & 1.000  & 1.000  & 1.000  \\ 
    &    0.2 & 1.000  & 1.000  & 1.000  & 1.000  & 1.000  & 1.000  \\ 
    &    0.3 & 1.000  & 1.000  & 1.000  & 1.000  & 1.000  & 1.000  \\ 
    &    0.4 & 1.000  & 1.000  & 1.000  & 1.000  & 1.000  & 1.000  \\ 
    &    0.5 & 1.000  & 0.999  & 0.994  & 0.999  & 0.997  & 0.991  \\ 
    &    0.6 & 1.000  & 0.984  & 0.964  & 0.995  & 0.968  & 0.955  \\ 
    &    0.7 & 0.997  & 0.974  & 0.960  & 0.979  & 0.953  & 0.949  \\ 
    &    0.8 & 0.985  & 0.969  & 0.959  & 0.971  & 0.952  & 0.948  \\ 
    &    0.9 & 0.978  & 0.969  & 0.959  & 0.970  & 0.952  & 0.948  \\ 
    &    1 & 0.972  & 0.969  & 0.959  & 0.970  & 0.952  & 0.948  \\ 
        
        \hline
\multirow{10}{*}{2460}        
    &    0.1 & 1.000  & 1.000  & 1.000  & 1.000  & 1.000  & 1.000  \\ 
    &    0.2 & 1.000  & 1.000  & 1.000  & 1.000  & 1.000  & 1.000  \\ 
    &    0.3 & 1.000  & 1.000  & 1.000  & 1.000  & 1.000  & 1.000  \\ 
    &    0.4 & 1.000  & 1.000  & 1.000  & 1.000  & 1.000  & 1.000  \\ 
    &    0.5 & 1.000  & 0.995  & 0.979  & 1.000  & 0.983  & 0.979  \\ 
    &    0.6 & 1.000  & 0.966  & 0.964  & 0.988  & 0.940  & 0.966  \\ 
    &    0.7 & 0.986  & 0.953  & 0.954  & 0.967  & 0.934  & 0.965  \\ 
    &    0.8 & 0.961  & 0.952  & 0.954  & 0.964  & 0.934  & 0.965  \\ 
    &    0.9 & 0.954  & 0.952  & 0.954  & 0.963  & 0.934  & 0.965  \\ 
    &    1 & 0.953  & 0.952  & 0.954  & 0.963  & 0.934  & 0.965  \\   
 \hline
			\hline
		\end{tabular}
	}
	
	\label{tab:DGP0Inference}
\end{table}
\begin{table}[h!]
	
	\centering
	\caption{RMSEs for DGP (1)}
	\scalebox{0.9}{
		\begin{tabular}{  c  c  c  c  c  c  c  c  }
			\hline
			\hline
$n$  &	$c $ & (1, 2) & {(1, 3)} & {(1, 4)} 
			 & {{(2, 3)}} & {(2, 4)}
			 & {{(3, 4)}} \\
			\hline
\multirow{10}{*}{1230}   
    &   0.1 & 0.000  & 0.000  & 0.000  & 0.000  & 0.000  & 0.000  \\ 
    &    0.2 & 0.000  & 0.000  & 0.000  & 0.000  & 0.000  & 0.000  \\ 
    &    0.3 & 0.000  & 0.000  & 0.000  & 0.000  & 0.000  & 0.000  \\ 
    &    0.4 & 0.000  & 0.000  & 0.000  & 0.000  & 0.000  & 0.000  \\ 
    &    0.5 & 0.000  & 0.000  & 0.000  & 0.000  & 0.000  & 0.000  \\ 
    &    0.6 & 0.000  & 0.940  & 0.332  & 0.000  & 0.204  & 3.067  \\ 
    &    0.7 & 0.000  & 2.751  & 1.843  & 3.642  & 0.877  & 8.122  \\ 
    &    0.8 & 0.000  & 7.177  & 3.444  & 11.820  & 2.392  & 15.625  \\ 
    &    0.9 & 9.110  & 12.447  & 4.930  & 23.042  & 4.150  & 23.877  \\ 
    &    1 & 21.949  & 19.690  & 6.233  & 32.068  & 6.184  & 29.984  \\  
        
        \hline
\multirow{10}{*}{2460}        
    &        0.1 & 0.000  & 0.000  & 0.000  & 0.000  & 0.000  & 0.000  \\ 
    &    0.2 & 0.000  & 0.000  & 0.000  & 0.000  & 0.000  & 0.000  \\ 
    &    0.3 & 0.000  & 0.000  & 0.000  & 0.000  & 0.000  & 0.000  \\ 
    &    0.4 & 0.000  & 0.000  & 0.000  & 0.000  & 0.000  & 0.000  \\ 
    &    0.5 & 0.000  & 0.000  & 0.000  & 0.000  & 0.000  & 0.000  \\ 
    &    0.6 & 0.000  & 0.000  & 0.421  & 0.000  & 0.000  & 0.681  \\ 
    &    0.7 & 0.000  & 1.996  & 1.431  & 1.498  & 0.734  & 4.616  \\ 
    &    0.8 & 2.945  & 3.824  & 3.467  & 7.503  & 1.616  & 10.312  \\ 
    &    0.9 & 2.945  & 7.404  & 5.354  & 17.387  & 3.253  & 19.648  \\ 
    &    1 & 7.884  & 14.607  & 7.466  & 32.231  & 5.749  & 30.738  \\  
 \hline
			\hline
		\end{tabular}
	}
	
	\label{tab:RMSEDGP1}
\end{table}
\begin{table}[h!]
	
	\centering
	\caption{RMSEs for DGP (2)}
	\scalebox{0.9}{
		\begin{tabular}{  c  c  c  c  c  c  c  c  }
			\hline
			\hline
$n$  &	$c $ & (1, 2) & {(1, 3)} & {(1, 4)} 
			 & {{(2, 3)}} & {(2, 4)}
			 & {{(3, 4)}} \\
			\hline
\multirow{10}{*}{1230}   
    &   0.1 & 0.000  & 0.000  & 0.000  & 0.000  & 0.000  & 0.000  \\ 
    &    0.2 & 0.000  & 0.000  & 0.000  & 0.000  & 0.000  & 0.000  \\ 
    &    0.3 & 0.000  & 0.000  & 0.000  & 0.000  & 0.000  & 0.000  \\ 
    &    0.4 & 0.000  & 0.000  & 0.000  & 0.000  & 0.000  & 0.000  \\ 
    &    0.5 & 0.000  & 0.000  & 0.000  & 0.000  & 0.000  & 0.000  \\ 
    &    0.6 & 0.000  & 0.000  & 0.492  & 0.000  & 0.000  & 2.559  \\ 
    &    0.7 & 0.000  & 0.000  & 6.044  & 0.000  & 2.065  & 4.397  \\ 
    &    0.8 & 0.000  & 1.514  & 13.078  & 4.721  & 5.465  & 9.519  \\ 
    &    0.9 & 0.000  & 3.599  & 20.644  & 14.404  & 9.314  & 14.755  \\ 
   &     1 & 3.964  & 6.682  & 26.301  & 23.569  & 14.271  & 20.772  \\   
        
        \hline
\multirow{10}{*}{2460}        
    &             0.1 & 0.000  & 0.000  & 0.000  & 0.000  & 0.000  & 0.000  \\ 
    &    0.2 & 0.000  & 0.000  & 0.000  & 0.000  & 0.000  & 0.000  \\ 
    &    0.3 & 0.000  & 0.000  & 0.000  & 0.000  & 0.000  & 0.000  \\ 
    &    0.4 & 0.000  & 0.000  & 0.000  & 0.000  & 0.000  & 0.000  \\ 
    &    0.5 & 0.000  & 0.000  & 0.000  & 0.000  & 0.000  & 0.000  \\ 
    &    0.6 & 0.000  & 0.000  & 0.159  & 0.000  & 0.000  & 0.000  \\ 
    &    0.7 & 0.000  & 0.000  & 6.373  & 0.000  & 0.989  & 2.103  \\ 
    &    0.8 & 0.000  & 0.000  & 15.145  & 2.702  & 3.455  & 6.309  \\ 
    &    0.9 & 0.000  & 0.000  & 22.974  & 5.737  & 7.433  & 10.970  \\ 
    &    1 & 0.000  & 0.000  & 28.708  & 13.680  & 10.664  & 16.494  \\
 \hline
			\hline
		\end{tabular}
	}
	
	\label{tab:RMSEDGP2}
\end{table}
\begin{table}[h!]
	
	\centering
	\caption{RMSEs for DGP (3)}
	\scalebox{0.9}{
		\begin{tabular}{  c  c  c  c  c  c  c  c  }
			\hline
			\hline
$n$  &	$c $ & (1, 2) & {(1, 3)} & {(1, 4)} 
			 & {{(2, 3)}} & {(2, 4)}
			 & {{(3, 4)}} \\
			\hline
\multirow{10}{*}{1230}   
    &   0.1 & 0.000  & 0.000  & 0.000  & 0.000  & 0.000  & 0.000  \\ 
    &    0.2 & 0.000  & 0.000  & 0.000  & 0.000  & 0.000  & 0.000  \\ 
    &    0.3 & 0.000  & 0.000  & 0.000  & 0.000  & 0.000  & 0.000  \\ 
    &    0.4 & 0.000  & 0.000  & 0.000  & 0.000  & 0.000  & 0.000  \\ 
    &    0.5 & 0.000  & 0.000  & 0.000  & 0.000  & 0.000  & 0.000  \\ 
    &    0.6 & 0.000  & 0.000  & 0.000  & 0.000  & 0.000  & 0.333  \\ 
    &    0.7 & 0.000  & 0.307  & 0.366  & 0.384  & 0.000  & 1.571  \\ 
    &    0.8 & 0.000  & 0.928  & 0.921  & 3.482  & 0.149  & 4.226  \\ 
    &    0.9 & 0.000  & 2.612  & 1.429  & 7.000  & 0.460  & 6.845  \\ 
    &    1 & 0.340  & 4.145  & 2.119  & 10.064  & 1.239  & 8.501  \\ 
        
        \hline
\multirow{10}{*}{2460}        
    &        0.1 & 0.000  & 0.000  & 0.000  & 0.000  & 0.000  & 0.000  \\ 
    &    0.2 & 0.000  & 0.000  & 0.000  & 0.000  & 0.000  & 0.000  \\ 
    &    0.3 & 0.000  & 0.000  & 0.000  & 0.000  & 0.000  & 0.000  \\ 
    &    0.4 & 0.000  & 0.000  & 0.000  & 0.000  & 0.000  & 0.000  \\ 
    &    0.5 & 0.000  & 0.000  & 0.000  & 0.000  & 0.000  & 0.000  \\ 
    &    0.6 & 0.000  & 0.000  & 0.000  & 0.000  & 0.000  & 0.000  \\ 
    &    0.7 & 0.000  & 0.000  & 0.045  & 0.000  & 0.000  & 0.022  \\ 
    &    0.8 & 0.000  & 0.000  & 0.262  & 0.175  & 0.000  & 0.631  \\ 
    &    0.9 & 0.000  & 0.194  & 0.668  & 2.536  & 0.000  & 3.089  \\ 
    &    1 & 0.000  & 1.160  & 1.327  & 5.498  & 0.288  & 6.092  \\   
 \hline
			\hline
		\end{tabular}
	}
	
	\label{tab:RMSEDGP3}
\end{table}
\begin{table}[h!]
	
	\centering
	\caption{RMSEs for DGP (4)}
	\scalebox{0.9}{
		\begin{tabular}{  c  c  c  c  c  c  c  c  }
			\hline
			\hline
$n$  &	$c $ & (1, 2) & {(1, 3)} & {(1, 4)} 
			 & {{(2, 3)}} & {(2, 4)}
			 & {{(3, 4)}} \\
			\hline
\multirow{10}{*}{1230}   
    &   0.1 & 0.000  & 0.000  & 0.000  & 0.000  & 0.000  & 0.000  \\ 
    &    0.2 & 0.000  & 0.000  & 0.000  & 0.000  & 0.000  & 0.000  \\ 
    &    0.3 & 0.000  & 0.000  & 0.000  & 0.000  & 0.000  & 0.000  \\ 
    &    0.4 & 0.000  & 0.000  & 0.000  & 0.000  & 0.000  & 0.000  \\ 
    &    0.5 & 0.000  & 0.000  & 0.000  & 0.000  & 0.000  & 0.000  \\ 
    &    0.6 & 0.000  & 0.000  & 0.000  & 0.000  & 0.000  & 1.976  \\ 
    &    0.7 & 0.000  & 0.024  & 0.035  & 2.166  & 0.304  & 5.750  \\ 
    &    0.8 & 1.301  & 1.212  & 1.118  & 7.129  & 1.119  & 9.226  \\ 
    &    0.9 & 4.496  & 3.836  & 2.472  & 10.658  & 2.444  & 10.591  \\ 
    &    1 & 8.876  & 5.881  & 3.687  & 12.397  & 3.532  & 11.038  \\

        \hline
\multirow{10}{*}{2460}        
    &        0.1 & 0.000  & 0.000  & 0.000  & 0.000  & 0.000  & 0.000  \\ 
    &    0.2 & 0.000  & 0.000  & 0.000  & 0.000  & 0.000  & 0.000  \\ 
    &    0.3 & 0.000  & 0.000  & 0.000  & 0.000  & 0.000  & 0.000  \\ 
    &    0.4 & 0.000  & 0.000  & 0.000  & 0.000  & 0.000  & 0.000  \\ 
    &    0.5 & 0.000  & 0.000  & 0.000  & 0.000  & 0.000  & 0.000  \\ 
    &    0.6 & 0.000  & 0.000  & 0.000  & 0.000  & 0.000  & 1.407  \\ 
    &    0.7 & 0.000  & 0.000  & 0.000  & 1.994  & 0.000  & 5.094  \\ 
    &    0.8 & 0.000  & 0.398  & 0.786  & 6.534  & 0.755  & 8.282  \\ 
    &    0.9 & 1.803  & 1.366  & 1.715  & 10.912  & 1.790  & 9.792  \\ 
    &    1 & 5.893  & 3.564  & 3.012  & 12.921  & 2.941  & 10.315  \\ 
 \hline
			\hline
		\end{tabular}
	}
	
	\label{tab:RMSEDGP4}
\end{table}

\newpage

\subsection{Empirical Results}\label{sec.application}

In this section, we apply VSIV to estimate the returns of college education. We choose the tuning parameter $\tau_n^g$ using formula \eqref{eq:tuning_parameter}  in Section \ref{sec.validity set estimation and choice of tau}. The tuning parameter $\tau_n^g$ depends on the user-specified constant $c$. Table \ref{tab:Application_FullSample} presents the estimated validity pair set for a grid of values for $c$. The simulation results in the previous section suggest choosing $c=0.6$. For this choice of $c$, only the pairs $(2,4)$ and $(3,4)$ are selected, while the other pairs are screened out, so that the estimated validity pair set is $\widehat{\mathscr{Z}}_0=\{(2,4),(3,4)\}$.\footnote{In Appendix \ref{sec.validity set estimation using tests}, we present the results of a standard IV validity test.} This result is consistent with the results in \citet{kedagni2020discordant} in that for small values of $Z$, the exclusion condition may fail. 

The resulting VSIV estimates reported in the last row of Table \ref{tab:Application_FullSample} are $\widehat{\beta}^1_{(2,4)}=0.542$ and $\widehat{\beta}^1_{(3,4)}=0.652$. The corresponding $95\%$ confidence intervals based on heteroskedasticity-robust standard errors are $[0.38493,0.69954]$ and $[0.28062,1.0233]$, respectively.

\begin{table}[h!]
	
	\centering
	\caption{Validity Pair Set Estimation in Application}
	\scalebox{1}{
		\begin{tabular}{    c  c  c  c  c  c  c   }
			\hline
			\hline
			 {$c$}& (1, 2)
			  & {{(1, 3)}} & {{(1, 4)}} & (2, 3) & (2, 4) & {{(3, 4)}}  \\
			\hline

   0.1 &\xmark&\xmark&\xmark&\xmark&\xmark&\xmark\\ 
        0.2 &\xmark&\xmark&\xmark&\xmark&\xmark&\xmark\\ 
        0.3 &\xmark&\xmark&\xmark&\xmark&\xmark&\xmark\\ 
        0.4 &\xmark&\xmark&\xmark&\xmark&\xmark&\xmark\\ 
        0.5 &\xmark&\xmark&\xmark&\xmark&\xmark&\xmark\\ 
    \textbf{\textcolor{orange}{0.6}} &\xmark&\xmark&\xmark&\xmark&\cmark&\cmark\\ 
        0.7 &\xmark&\cmark&\xmark&\cmark&\cmark&\cmark\\ 
        0.8 &\xmark&\cmark&\cmark&\cmark&\cmark&\cmark\\ 
        0.9 &\xmark&\cmark&\cmark&\cmark&\cmark&\cmark\\ 
       1&\cmark&\cmark&\cmark&\cmark&\cmark&\cmark\\ 

\hline
$\widehat{\beta}^1_{(k,k')}$& -- & -- & -- & -- &0.542 & 0.652\\
			\hline
\hline
		\end{tabular}
	}
	
	\label{tab:Application_FullSample}
\end{table}

\section{Conclusion}\label{sec:conclusion}
We propose an approach for estimating LATEs when the instruments are partially invalid. We focus on settings where the instruments are pairwise valid. Under pairwise validity, there are two types of instrument value pairs: Pairs for which the LATE assumptions hold and pairs for which the LATE assumptions fail due to violations of exclusion, independence, monotonicity, or combinations thereof. Pairwise validity is a natural starting point and useful in applications in which it is difficult to determine which LATE assumptions fail and how. However, in settings where researchers have information about how exactly the LATE assumptions fail, it would be interesting to consider intermediate cases that incorporate such information. Under restrictions on which LATE assumptions fail and how, it is often possible to derive nontrivial bounds on LATEs \citep[e.g.,][]{huber2014sensitivity,noack2021sensitivity,kedagni2023identifying,cui2024robust}.

Throughout this paper, we focus on average treatment effects for the compliers. However, under pairwise validity, both marginal potential outcome distributions are identified for each $(z,z')\in \mathscr{Z}_{\bar{M}}$. Therefore, another interesting direction for future research is to extend VSIV to allow for the estimation of distributional treatment effects for the compliers, such as local quantile treatment effects \citep[e.g.,][]{abadie2002instrumental,frolich2013unconditional,melly2017local}.

\section*{Acknowledgements}
We are grateful to the Editor (Elie Tamer), Associate Editor, two anonymous referees, Zheng Fang, Martin Huber, Toru Kitagawa, Julian Martinez-Iriarte, D\'esir\'e K\'edagni, Ismael Mourifi\'e, Xiaoxia Shi, and all seminar and conference participants for their insightful suggestions and comments. We thank D\'esir\'e K\'edagni for sharing the data for the empirical application with us. Sun acknowledges funding by the National Natural Science Foundation of China [Grant Number 72103004]. W\"uthrich is also affiliated with CESifo. The usual disclaimer applies.

\putbib
\end{bibunit}
\newpage

\begin{bibunit}
\appendix
\onehalfspacing
\setcounter{page}{1} 
\setcounter{footnote}{0}
\begin{center}
    \LARGE{{Pairwise Valid Instruments}}
    
    \vspace{0.25cm}
    
    \Large{Online Supplementary Appendix}
    
    \vspace{0.5cm}
    
    \large{Zhenting Sun \qquad Kaspar W\"uthrich}
\end{center}

\startcontents[sections]
\printcontents[sections]{l}{1}{\setcounter{tocdepth}{2}}

\section{Proofs of Main Results}
\begin{proof}[Proof of Lemma \ref{lemma.pairwise beta binary D}]
    Lemma \ref{lemma.pairwise beta binary D} is a special case of Lemma \ref{lemma.partial beta} in Appendix \ref{sec.ordered treatment}. See the proof of Lemma \ref{lemma.partial beta} in Appendix \ref{sec.general proofs ordered treatment}.
\end{proof}

\begin{proof}[Proof of Theorem \ref{thm.IV estimator asymptotics pairwise binary D}]
Theorem \ref{thm.IV estimator asymptotics pairwise binary D} is a special case of Theorem \ref{thm.IV estimator asymptotics pairwise} in Appendix \ref{sec.ordered treatment}.
See the proof of Theorem \ref{thm.IV estimator asymptotics pairwise} in Appendix \ref{sec.general proofs ordered treatment}.
\end{proof}

\begin{proof}[Proof of Corollary \ref{corollary.weighted average weak convergence}]
The result follows directly from the delta method.
\end{proof}

\begin{proof}[Proof of Proposition \ref{prop.test}]
    If $\mathrm{H}_{0}$ is true, it can be shown that under the assumptions,
\[
\mathbb{P}\left(  \left\{  TS_{1n}=0\right\}  \cup\left\{  TS_{2n}%
>c_{r}(\alpha)\right\}  \right)  \geq\mathbb{P}\left(  TS_{2n}>c_{r}(\alpha)\right)
\rightarrow\alpha
\]
and
\begin{align*}
\mathbb{P}\left(  \left\{  TS_{1n}=0\right\}  \cup\left\{  TS_{2n}%
>c_{r}(\alpha)\right\}  \right) & \leq\mathbb{P}\left(  TS_{2n}>c_{r}(\alpha)\right)
+\mathbb{P}\left(  TS_{1n}=0\right) \\
&\le \mathbb{P}\left(  TS_{2n}>c_{r}(\alpha)\right)
+\mathbb{P}(  \widehat{\mathscr{Z}_0}\neq\mathscr{Z}_{\bar{M}}) \rightarrow\alpha,    
\end{align*}
which imply that $\mathbb{P}\left(  \left\{  TS_{1n}=0\right\}  \cup\left\{  TS_{2n}
>c_{r}(\alpha)\right\}  \right)  \rightarrow\alpha$.

Suppose $\mathrm{H}_{0}$ is false. If $\mathcal{Z}_{(\kappa_m,\kappa'_m)}\notin\mathscr{Z}_{\bar{M}}$ for some $m$, then
\[
\mathbb{P}\left(  \left\{  TS_{1n}=0\right\}  \cup\left\{  TS_{2n}%
>c_{r}(\alpha)\right\}  \right)  \geq\mathbb{P}\left(  TS_{1n}=0\right) \ge \mathbb{P}(  \widehat{\mathscr{Z}_0}=\mathscr{Z}_{\bar{M}}) \rightarrow1.
\]
If $\mathcal{Z}_{(\kappa_m,\kappa'_m)}\in\mathscr{Z}_{\bar{M}}$ for all $m\in\{1,\ldots,S\}$ but $R(\beta_{1S})\neq0$, then $R(\widehat{\beta}_{1S})\to_p R({\beta}_{1S})$ and 
\[
\mathbb{P}\left(  \left\{  TS_{1n}=0\right\}  \cup\left\{  TS_{2n}>c_{r}(\alpha)\right\}  \right)  \geq\mathbb{P}\left(  TS_{2n}>c_{r}(\alpha)\right)  \rightarrow1.
\]
\end{proof}

\begin{proof}[Proof of Theorem \ref{thm.bias reduction binary D}]
    Theorem \ref{thm.bias reduction binary D} is a special case of Theorem \ref{thm.bias reduction multi ordered} in Appendix \ref{sec.ordered treatment}. See the proof of Theorem \ref{thm.bias reduction multi ordered} in Appendix \ref{sec.general proofs ordered treatment}.
\end{proof}

\begin{proof}[Proof of Theorem \ref{thm.IV estimator asymptotics binary D}]
    Theorem \ref{thm.IV estimator asymptotics binary D} is a special case of Theorem \ref{thm.IV estimator asymptotics} in Appendix \ref{sec.ordered treatment partial}. See the proof of Theorem \ref{thm.IV estimator asymptotics} in Appendix \ref{sec.ordered treatment partial}.
\end{proof}

\begin{proof}[Proof of Lemma \ref{lemma.testable implications for Z2 weaker}]
First, we show that for every pair $\mathcal{Z}_{(k,k')}$, the testable implications of \citet{kitagawa2015test}, \citet{mourifie2016testing}, and \citet{sun2021ivvalidity} imply those of \citet{kedagni2020generalized}. Now suppose that the testable implications in \eqref{eq.testable implication binary D} hold for $\mathcal{Z}_{(k,k')}$, that is,
for all Borel sets $A$, we have
\begin{align*}
&\mathbb{P}\left(  Y\in A,D=1|Z=z_{k}\right)  \leq\mathbb{P}\left(  Y\in
A,D=1|Z=z_{k^{\prime}}\right)  \text{ and }\\
&\mathbb{P}\left(  Y\in A,D=0|Z=z_{k}\right)  \geq\mathbb{P}\left(  Y\in
A,D=0|Z=z_{k^{\prime}}\right).
\end{align*}
We let $p\left(
A,d,z\right)  =\mathbb{P}\left(  Y\in A,D=d|Z=z\right)  $ for all
$A\in\mathcal{B}_{\mathbb{R}}$, each $d\in\left\{  0,1\right\}  $, and all
$z\in\mathcal{Z}$.

\bigskip

\noindent \textbf{Step (i):} Suppose $A_{1},A_{2},A_{3}\in\mathcal{B}_{\mathbb{R}}$ with $A_{2}\cap
A_{3}=\varnothing$. Let $\bar{A}_{2}=A_{2}\cup A_{3}$. We want to show that
\begin{align}\label{eq.step (i) inequality 1-1}
&  \min_{z\in\mathcal{Z}_{(k,k^{\prime})}}\left\{  p\left(  A_{1},1,z\right)
+p\left(  \bar{A}_{2},0,z\right)  \right\}  \notag\\
 \leq&\,\min_{z\in\mathcal{Z}_{(k,k^{\prime})}}\left\{  p\left(  A_{1}%
,1,z\right)  +p\left(  A_{2},0,z\right)  \right\}  +\min_{z\in\mathcal{Z}%
_{(k,k^{\prime})}}\left\{  p\left(  A_{1},1,z\right)  +p\left(  A_{3}%
,0,z\right)  \right\}
\end{align}
and%
\begin{align}\label{eq.step (i) inequality 1-2}
&  \min_{z\in\mathcal{Z}_{(k,k^{\prime})}}\left\{  p\left(  A_{1},0,z\right)
+p\left(  \bar{A}_{2},1,z\right)  \right\} \notag \\
  \leq&\,\min_{z\in\mathcal{Z}_{(k,k^{\prime})}}\left\{  p\left(  A_{1}%
,0,z\right)  +p\left(  A_{2},1,z\right)  \right\}  +\min_{z\in\mathcal{Z}%
_{(k,k^{\prime})}}\left\{  p\left(  A_{1},0,z\right)  +p\left(  A_{3}%
,1,z\right)  \right\}  .
\end{align}

Suppose%
\begin{align*}
&  \min_{z\in\mathcal{Z}_{(k,k^{\prime})}}\left\{  p\left(  A_{1},1,z\right)
+p\left(  A_{2},0,z\right)  \right\}  +\min_{z\in\mathcal{Z}_{(k,k^{\prime})}%
}\left\{  p\left(  A_{1},1,z\right)  +p\left(  A_{3},0,z\right)  \right\}  \\
=&\,\left\{  p\left(  A_{1},1,z_{k}\right)  +p\left(  A_{2},0,z_{k}\right)
\right\}  +\left\{  p\left(  A_{1},1,z_{k^{\prime}}\right)  +p\left(
A_{3},0,z_{k^{\prime}}\right)  \right\}  .
\end{align*}
Because by assumption,%
\[
\mathbb{P}\left(  Y\in A_{2},D=0|Z=z_{k}\right)  \geq\mathbb{P}\left(  Y\in
A_{2},D=0|Z=z_{k^{\prime}}\right)  ,
\]
we have that
\begin{align*}
&  \min_{z\in\mathcal{Z}_{(k,k^{\prime})}}\left\{  p\left(  A_{1},1,z\right)
+p\left(  A_{2},0,z\right)  \right\}  +\min_{z\in\mathcal{Z}_{(k,k^{\prime})}%
}\left\{  p\left(  A_{1},1,z\right)  +p\left(  A_{3},0,z\right)  \right\}  \\
 \geq&\,\left\{  p\left(  A_{1},1,z_{k}\right)  +p\left(  A_{2},0,z_{k^{\prime
}}\right)  \right\}  +\left\{  p\left(  A_{1},1,z_{k^{\prime}}\right)
+p\left(  A_{3},0,z_{k^{\prime}}\right)  \right\}  \\
=&\,p\left(  A_{1},1,z_{k}\right)  +p\left(  A_{1},1,z_{k^{\prime}}\right)
+p\left(  \bar{A}_{2},0,z_{k^{\prime}}\right)  \\
\geq&\,\min_{z\in\mathcal{Z}_{(k,k^{\prime})}}\left\{  p\left(  A_{1}%
,1,z\right)  +p\left(  \bar{A}_{2},0,z\right)  \right\}  .
\end{align*}
Suppose
\begin{align*}
&  \min_{z\in\mathcal{Z}_{(k,k^{\prime})}}\left\{  p\left(  A_{1},1,z\right)
+p\left(  A_{2},0,z\right)  \right\}  +\min_{z\in\mathcal{Z}_{(k,k^{\prime})}%
}\left\{  p\left(  A_{1},1,z\right)  +p\left(  A_{3},0,z\right)  \right\}  \\
 =&\,\left\{  p\left(  A_{1},1,z_{k^{\prime}}\right)  +p\left(  A_{2}%
,0,z_{k^{\prime}}\right)  \right\}  +\left\{  p\left(  A_{1},1,z_{k}\right)
+p\left(  A_{3},0,z_{k}\right)  \right\}  .
\end{align*}
Because by assumption,
\[
\mathbb{P}\left(  Y\in A_{3},D=0|Z=z_{k}\right)  \geq\mathbb{P}\left(  Y\in
A_{3},D=0|Z=z_{k^{\prime}}\right)  ,
\]
we have that
\begin{align*}
&  \min_{z\in\mathcal{Z}_{(k,k^{\prime})}}\left\{  p\left(  A_{1},1,z\right)
+p\left(  A_{2},0,z\right)  \right\}  +\min_{z\in\mathcal{Z}_{(k,k^{\prime})}%
}\left\{  p\left(  A_{1},1,z\right)  +p\left(  A_{3},0,z\right)  \right\}  \\
\geq&\,\left\{  p\left(  A_{1},1,z_{k^{\prime}}\right)  +p\left(
A_{2},0,z_{k^{\prime}}\right)  \right\}  +\left\{  p\left(  A_{1}%
,1,z_{k}\right)  +p\left(  A_{3},0,z_{k^{\prime}}\right)  \right\}  \\
=&\,p\left(  A_{1},1,z_{k}\right)  +p\left(  A_{1},1,z_{k^{\prime}}\right)
+p\left(  \bar{A}_{2},0,z_{k^{\prime}}\right)  \\
\geq&\,\min_{z\in\mathcal{Z}_{(k,k^{\prime})}}\left\{  p\left(  A_{1}%
,1,z\right)  +p\left(  \bar{A}_{2},0,z\right)  \right\}  .
\end{align*}

The other two cases are trivial. Similarly, we can show that
\begin{align*}
&  \min_{z\in\mathcal{Z}_{(k,k^{\prime})}}\left\{  p\left(  A_{1},0,z\right)
+p\left(  \bar{A}_{2},1,z\right)  \right\}  \\
\leq&\,\min_{z\in\mathcal{Z}_{(k,k^{\prime})}}\left\{  p\left(  A_{1}%
,0,z\right)  +p\left(  A_{2},1,z\right)  \right\}  +\min_{z\in\mathcal{Z}%
_{(k,k^{\prime})}}\left\{  p\left(  A_{1},0,z\right)  +p\left(  A_{3}%
,1,z\right)  \right\}  .
\end{align*}

\bigskip

\noindent \textbf{Step (ii):} Suppose $A_{2},A_{3},B_{2},B_{3}\in
\mathcal{B}_{\mathbb{R}}$ with $A_{2}\cap A_{3}=\varnothing$ and $B_{2}\cap
B_{3}=\varnothing$. Let $\bar{A}_{2}=A_{2}\cup A_{3}$ and $\bar{B}_{2}%
=B_{2}\cup B_{3}$. We want to show that
\begin{align}\label{eq.step (ii) inequality 2}
&  \min_{z\in\mathcal{Z}_{(k,k^{\prime})}}\left\{  p\left(  \bar{A}%
_{2},1,z\right)  +p\left(  \bar{B}_{2},0,z\right)  \right\} \notag\\
\leq &  \,\min_{z\in\mathcal{Z}_{(k,k^{\prime})}}\left\{  p\left(
A_{2},1,z\right)  +p\left(  B_{2},0,z\right)  \right\}  +\min_{z\in
\mathcal{Z}_{(k,k^{\prime})}}\left\{  p\left(  A_{2},1,z\right)  +p\left(
B_{3},0,z\right)  \right\} \notag\\
&  +\min_{z\in\mathcal{Z}_{(k,k^{\prime})}}\left\{  p\left(  A_{3},1,z\right)
+p\left(  B_{2},0,z\right)  \right\}  +\min_{z\in\mathcal{Z}_{(k,k^{\prime})}%
}\left\{  p\left(  A_{3},1,z\right)  +p\left(  B_{3},0,z\right)  \right\}  .
\end{align}

If
\begin{align*}
&  \min_{z\in\mathcal{Z}_{(k,k^{\prime})}}\left\{  p\left(  A_{2},1,z\right)
+p\left(  B_{2},0,z\right)  \right\}  +\min_{z\in\mathcal{Z}_{(k,k^{\prime})}%
}\left\{  p\left(  A_{3},1,z\right)  +p\left(  B_{3},0,z\right)  \right\}  \\
= &  \,\left\{  p\left(  A_{2},1,z\right)  +p\left(  B_{2},0,z\right)
\right\}  +\left\{  p\left(  A_{3},1,z\right)  +p\left(  B_{3},0,z\right)
\right\}
\end{align*}
for some $z\in\mathcal{Z}_{(k,k^{\prime})}$, or if
\begin{align*}
&  \min_{z\in\mathcal{Z}_{(k,k^{\prime})}}\left\{  p\left(  A_{2},1,z\right)
+p\left(  B_{3},0,z\right)  \right\}  +\min_{z\in\mathcal{Z}_{(k,k^{\prime})}%
}\left\{  p\left(  A_{3},1,z\right)  +p\left(  B_{2},0,z\right)  \right\}  \\
= &  \,\left\{  p\left(  A_{2},1,z\right)  +p\left(  B_{3},0,z\right)
\right\}  +\left\{  p\left(  A_{3},1,z\right)  +p\left(  B_{2},0,z\right)
\right\}
\end{align*}
for some $z\in\mathcal{Z}_{(k,k^{\prime})}$, then the result is trivial.

Suppose
\begin{align*}
&  \min_{z\in\mathcal{Z}_{(k,k^{\prime})}}\left\{  p\left(  A_{2},1,z\right)
+p\left(  B_{2},0,z\right)  \right\}  +\min_{z\in\mathcal{Z}_{(k,k^{\prime})}%
}\left\{  p\left(  A_{3},1,z\right)  +p\left(  B_{3},0,z\right)  \right\}  \\
= &  \,\left\{  p\left(  A_{2},1,z_{k}\right)  +p\left(  B_{2},0,z_{k}\right)
\right\}  +\left\{  p\left(  A_{3},1,z_{k^{\prime}}\right)  +p\left(
B_{3},0,z_{k^{\prime}}\right)  \right\}
\end{align*}
and
\begin{align*}
&  \min_{z\in\mathcal{Z}_{(k,k^{\prime})}}\left\{  p\left(  A_{2},1,z\right)
+p\left(  B_{3},0,z\right)  \right\}  +\min_{z\in\mathcal{Z}_{(k,k^{\prime})}%
}\left\{  p\left(  A_{3},1,z\right)  +p\left(  B_{2},0,z\right)  \right\}  \\
= &  \,\left\{  p\left(  A_{2},1,z_{k}\right)  +p\left(  B_{3},0,z_{k}\right)
\right\}  +\left\{  p\left(  A_{3},1,z_{k^{\prime}}\right)  +p\left(
B_{2},0,z_{k^{\prime}}\right)  \right\}  .
\end{align*}
Then because by assumption,
\[
\mathbb{P}\left(  Y\in A_{3},D=1|Z=z_{k^{\prime}}\right)  \geq\mathbb{P}%
\left(  Y\in A_{3},D=1|Z=z_{k}\right)  ,
\]
we have that
\begin{align*}
&  p\left(  A_{2},1,z_{k}\right)  +p\left(  B_{2},0,z_{k}\right)  +p\left(
A_{3},1,z_{k^{\prime}}\right)  +p\left(  B_{3},0,z_{k^{\prime}}\right)  \\
&  +p\left(  A_{2},1,z_{k}\right)  +p\left(  B_{3},0,z_{k}\right)  +p\left(
A_{3},1,z_{k^{\prime}}\right)  +p\left(  B_{2},0,z_{k^{\prime}}\right)  \\
\geq &  \,p\left(  A_{2},1,z_{k}\right)  +p\left(  A_{3},1,z_{k}\right)
+p\left(  B_{2},0,z_{k}\right)  +p\left(  B_{3},0,z_{k}\right)  \\
= &  \,p\left(  \bar{A}_{2},1,z_{k}\right)  +p\left(  \bar{B}_{2}%
,0,z_{k}\right)  .
\end{align*}

Suppose%
\begin{align*}
&  \min_{z\in\mathcal{Z}_{(k,k^{\prime})}}\left\{  p\left(  A_{2},1,z\right)
+p\left(  B_{2},0,z\right)  \right\}  +\min_{z\in\mathcal{Z}_{(k,k^{\prime})}%
}\left\{  p\left(  A_{3},1,z\right)  +p\left(  B_{3},0,z\right)  \right\}  \\
= &  \,\left\{  p\left(  A_{2},1,z_{k}\right)  +p\left(  B_{2},0,z_{k}\right)
\right\}  +\left\{  p\left(  A_{3},1,z_{k^{\prime}}\right)  +p\left(
B_{3},0,z_{k^{\prime}}\right)  \right\}
\end{align*}
and
\begin{align*}
&  \min_{z\in\mathcal{Z}_{(k,k^{\prime})}}\left\{  p\left(  A_{2},1,z\right)
+p\left(  B_{3},0,z\right)  \right\}  +\min_{z\in\mathcal{Z}_{(k,k^{\prime})}%
}\left\{  p\left(  A_{3},1,z\right)  +p\left(  B_{2},0,z\right)  \right\}  \\
= &  \,\left\{  p\left(  A_{2},1,z_{k^{\prime}}\right)  +p\left(
B_{3},0,z_{k^{\prime}}\right)  \right\}  +\left\{  p\left(  A_{3}%
,1,z_{k}\right)  +p\left(  B_{2},0,z_{k}\right)  \right\}  .
\end{align*}
Then because by assumption,
\[
\mathbb{P}\left(  Y\in B_{2},D=0|Z=z_{k}\right)  \geq\mathbb{P}\left(  Y\in
B_{2},D=0|Z=z_{k^{\prime}}\right)  ,
\]
we have that
\begin{align*}
&  \left\{  p\left(  A_{2},1,z_{k}\right)  +p\left(  B_{2},0,z_{k}\right)
\right\}  +\left\{  p\left(  A_{3},1,z_{k^{\prime}}\right)  +p\left(
B_{3},0,z_{k^{\prime}}\right)  \right\}  \\
&  +\left\{  p\left(  A_{2},1,z_{k^{\prime}}\right)  +p\left(  B_{3}%
,0,z_{k^{\prime}}\right)  \right\}  +\left\{  p\left(  A_{3},1,z_{k}\right)
+p\left(  B_{2},0,z_{k}\right)  \right\}  \\
\geq &  \,p\left(  A_{2},1,z_{k^{\prime}}\right)  +p\left(  A_{3}%
,1,z_{k^{\prime}}\right)  +p\left(  B_{2},0,z_{k^{\prime}}\right)  +p\left(
B_{3},0,z_{k^{\prime}}\right)  \\
= &  \,p\left(  \bar{A}_{2},1,z_{k^{\prime}}\right)  +p\left(  \bar{B}%
_{2},0,z_{k^{\prime}}\right)  .
\end{align*}

Suppose%
\begin{align*}
&  \min_{z\in\mathcal{Z}_{(k,k^{\prime})}}\left\{  p\left(  A_{2},1,z\right)
+p\left(  B_{2},0,z\right)  \right\}  +\min_{z\in\mathcal{Z}_{(k,k^{\prime})}%
}\left\{  p\left(  A_{3},1,z\right)  +p\left(  B_{3},0,z\right)  \right\}  \\
= &  \,\left\{  p\left(  A_{2},1,z_{k^{\prime}}\right)  +p\left(
B_{2},0,z_{k^{\prime}}\right)  \right\}  +\left\{  p\left(  A_{3}%
,1,z_{k}\right)  +p\left(  B_{3},0,z_{k}\right)  \right\}
\end{align*}
and
\begin{align*}
&  \min_{z\in\mathcal{Z}_{(k,k^{\prime})}}\left\{  p\left(  A_{2},1,z\right)
+p\left(  B_{3},0,z\right)  \right\}  +\min_{z\in\mathcal{Z}_{(k,k^{\prime})}%
}\left\{  p\left(  A_{3},1,z\right)  +p\left(  B_{2},0,z\right)  \right\}  \\
= &  \,\left\{  p\left(  A_{2},1,z_{k}\right)  +p\left(  B_{3},0,z_{k}\right)
\right\}  +\left\{  p\left(  A_{3},1,z_{k^{\prime}}\right)  +p\left(
B_{2},0,z_{k^{\prime}}\right)  \right\}  .
\end{align*}
Then because by assumption,
\[
\mathbb{P}\left(  Y\in B_{3},D=0|Z=z_{k}\right)  \geq\mathbb{P}\left(  Y\in
B_{3},D=0|Z=z_{k^{\prime}}\right)  ,
\]
we have that%
\begin{align*}
&  \left\{  p\left(  A_{2},1,z_{k^{\prime}}\right)  +p\left(  B_{2}%
,0,z_{k^{\prime}}\right)  \right\}  +\left\{  p\left(  A_{3},1,z_{k}\right)
+p\left(  B_{3},0,z_{k}\right)  \right\}  \\
&  +\left\{  p\left(  A_{2},1,z_{k}\right)  +p\left(  B_{3},0,z_{k}\right)
\right\}  +\left\{  p\left(  A_{3},1,z_{k^{\prime}}\right)  +p\left(
B_{2},0,z_{k^{\prime}}\right)  \right\}  \\
\geq &  \,p\left(  A_{2},1,z_{k^{\prime}}\right)  +p\left(  A_{3}%
,1,z_{k^{\prime}}\right)  +p\left(  B_{2},0,z_{k^{\prime}}\right)  +p\left(
B_{3},0,z_{k^{\prime}}\right)  \\
= &  \,p\left(  \bar{A}_{2},1,z_{k^{\prime}}\right)  +p\left(  \bar{B}%
_{2},0,z_{k^{\prime}}\right)  .
\end{align*}

Suppose%
\begin{align*}
&  \min_{z\in\mathcal{Z}_{(k,k^{\prime})}}\left\{  p\left(  A_{2},1,z\right)
+p\left(  B_{2},0,z\right)  \right\}  +\min_{z\in\mathcal{Z}_{(k,k^{\prime})}%
}\left\{  p\left(  A_{3},1,z\right)  +p\left(  B_{3},0,z\right)  \right\}  \\
= &  \,\left\{  p\left(  A_{2},1,z_{k^{\prime}}\right)  +p\left(
B_{2},0,z_{k^{\prime}}\right)  \right\}  +\left\{  p\left(  A_{3}%
,1,z_{k}\right)  +p\left(  B_{3},0,z_{k}\right)  \right\}
\end{align*}
and
\begin{align*}
&  \min_{z\in\mathcal{Z}_{(k,k^{\prime})}}\left\{  p\left(  A_{2},1,z\right)
+p\left(  B_{3},0,z\right)  \right\}  +\min_{z\in\mathcal{Z}_{(k,k^{\prime})}%
}\left\{  p\left(  A_{3},1,z\right)  +p\left(  B_{2},0,z\right)  \right\}  \\
= &  \,\left\{  p\left(  A_{2},1,z_{k^{\prime}}\right)  +p\left(
B_{3},0,z_{k^{\prime}}\right)  \right\}  +\left\{  p\left(  A_{3}%
,1,z_{k}\right)  +p\left(  B_{2},0,z_{k}\right)  \right\}  .
\end{align*}
Then because by assumption,
\[
\mathbb{P}\left(  Y\in A_{2},D=1|Z=z_{k^{\prime}}\right)  \geq\mathbb{P}%
\left(  Y\in A_{2},D=1|Z=z_{k}\right)  ,
\]
we have that%
\begin{align*}
&  \left\{  p\left(  A_{2},1,z_{k^{\prime}}\right)  +p\left(  B_{2}%
,0,z_{k^{\prime}}\right)  \right\}  +\left\{  p\left(  A_{3},1,z_{k}\right)
+p\left(  B_{3},0,z_{k}\right)  \right\}  \\
&  +\left\{  p\left(  A_{2},1,z_{k^{\prime}}\right)  +p\left(  B_{3}%
,0,z_{k^{\prime}}\right)  \right\}  +\left\{  p\left(  A_{3},1,z_{k}\right)
+p\left(  B_{2},0,z_{k}\right)  \right\}  \\
\geq &  \,p\left(  A_{2},1,z_{k}\right)  +p\left(  A_{3},1,z_{k}\right)
+p\left(  B_{2},0,z_{k}\right)  +p\left(  B_{3},0,z_{k}\right)  \\
= &  \,p\left(  \bar{A}_{2},1,z_{k}\right)  +p\left(  \bar{B}_{2}%
,0,z_{k}\right)  .
\end{align*}

 In the following, we show that \eqref{eq.first inequality KM}--\eqref{eq.third inequality KM} are implied by \eqref{eq.testable implication binary D}.

We first consider \eqref{eq.first inequality KM}. 
For $d=0$, $\max_{z\in\mathcal{Z}_{(k,k^{\prime})}}f_{Y,D}\left(
y,d|z\right)  =f_{Y,D}\left(  y,d|z_{k}\right)  $, and for $d=1$, $\max
_{z\in\mathcal{Z}_{(k,k^{\prime})}}f_{Y,D}\left(  y,d|z\right)  =f_{Y,D}%
\left(  y,d|z_{k^{\prime}}\right)  $. Then \eqref{eq.first inequality KM} is trivial.

We then consider \eqref{eq.second inequality KM}. We let $d_{1}=0$ and $d_{2}=1$, and $P_{\mathbb{R}}^{1}$ and $P_{\mathbb{R}%
}^{2}$ be two arbitrary partitions with $P_{\mathbb{R}}^{1}=\{  A_{1}
^{1},A_{2}^{1},\ldots,A_{N_1}^{1}\}  $ and $P_{\mathbb{R}}^{2}=\{  A_{1}
^{2},A_{2}^{2},\ldots,A_{N_2}^{2}\}  $ for some $N_1$ and $N_2$. First, we have that
\[
\min_{z\in\mathcal{Z}_{(k,k^{\prime})}}\left\{  p\left(  \mathbb{R}%
,1,z\right)  +p\left(  \mathbb{R},0,z\right)  \right\}  =1.
\]
We write $\bar{A}_{2}^{j}=\cup_{l=2}^{N_j}A_{l}^{j}$. With $A_{1}^{j}%
\cup\bar{A}_{2}^{j}=\mathbb{R}$, by \eqref{eq.step (ii) inequality 2},
\begin{align*}
1 \leq&\,\min_{z\in\mathcal{Z}_{(k,k^{\prime})}}\left\{  p\left(  A_{1}%
^{1},1,z\right)  +p\left(  A_{1}^{2},0,z\right)  \right\}  +\min
_{z\in\mathcal{Z}_{(k,k^{\prime})}}\left\{  p\left(  A_{1}^{1},1,z\right)
+p\left(  \bar{A}_{2}^{2},0,z\right)  \right\}  \\
&  +\min_{z\in\mathcal{Z}_{(k,k^{\prime})}}\left\{  p\left(  \bar{A}_{2}%
^{1},1,z\right)  +p\left(  A_{1}^{2},0,z\right)  \right\}  +\min
_{z\in\mathcal{Z}_{(k,k^{\prime})}}\left\{  p\left(  \bar{A}_{2}%
^{1},1,z\right)  +p\left(  \bar{A}_{2}^{2},0,z\right)  \right\}  .
\end{align*}
Then the result follows by repeating the above procedure under \eqref{eq.step (i) inequality 1-1}--\eqref{eq.step (ii) inequality 2}.

We finally consider \eqref{eq.third inequality KM}. For $j=1$,
\[
\int_{A_{1}}\max_{z\in\mathcal{Z}_{(k,k^{\prime})}}f_{Y,D}\left(
y,d_{1}|z\right)  \mathrm{d}y=\int_{A_{1}}f_{Y,D}\left(  y,d_{1}|z_{k}\right)
\mathrm{d}y
\]
and by \eqref{eq.step (i) inequality 1-1} and \eqref{eq.step (i) inequality 1-2},
\begin{align*}
\varphi_{1}\left(  A_{1},\mathcal{Z}_{(k,k^{\prime})},P_{\mathbb{R}}%
^{1},P_{\mathbb{R}}^{2}\right)   &  =\sum_{A_{2}\in P_{\mathbb{R}}^{2}}%
\min_{z\in\mathcal{Z}_{(k,k^{\prime})}}\left\{  \int_{A_{1}}f_{Y,D}\left(
y,d_{1}|z\right)  \mathrm{d}y+\int_{A_{2}}f_{Y,D}\left(  y,d_{2}|z\right)
\mathrm{d}y\right\} \\
&  \geq\min_{z\in\mathcal{Z}_{(k,k^{\prime})}}\left\{  \int_{A_{1}}%
f_{Y,D}\left(  y,d_{1}|z\right)  \mathrm{d}y+\mathbb{P}\left(  D=d_{2}%
|Z=z\right)  \right\}  .
\end{align*}
If
\begin{align*}\min_{z\in\mathcal{Z}_{(k,k^{\prime})}}\left\{  \int_{A_{1}}f_{Y,D}\left(
y,d_{1}|z\right)  \mathrm{d}y+\mathbb{P}\left(  D=d_{2}|Z=z\right)  \right\} =\int_{A_{1}}f_{Y,D}\left(  y,d_{1}|z_{k}\right)  \mathrm{d}y+\mathbb{P}%
\left(  D=d_{2}|Z=z_{k}\right)  ,
\end{align*}
then
\[
\int_{A_{1}}f_{Y,D}\left(  y,d_{1}|z_{k}\right)  \mathrm{d}y+\mathbb{P}\left(
D=d_{2}|Z=z_{k}\right)  -\int_{A_{1}}f_{Y,D}\left(  y,d_{1}|z_{k}\right)
\mathrm{d}y\geq0.
\]
If
\begin{align*}\min_{z\in\mathcal{Z}_{(k,k^{\prime})}}\left\{  \int_{A_{1}}f_{Y,D}\left(
y,d_{1}|z\right)  \mathrm{d}y+\mathbb{P}\left(  D=d_{2}|Z=z\right)  \right\}
=\int_{A_{1}}f_{Y,D}\left(  y,d_{1}|z_{k^{\prime}}\right)  \mathrm{d}%
y+\mathbb{P}\left(  D=d_{2}|Z=z_{k^{\prime}}\right)  ,
\end{align*}
then
\begin{align*}
&  \int_{A_{1}}f_{Y,D}\left(  y,d_{1}|z_{k^{\prime}}\right)  \mathrm{d}%
y+\mathbb{P}\left(  D=d_{2}|Z=z_{k^{\prime}}\right)  -\int_{A_{1}}%
f_{Y,D}\left(  y,d_{1}|z_{k}\right)  \mathrm{d}y\\
=&\,\mathbb{P}\left(  Y\in A_{1},D=0|Z=z_{k^{\prime}}\right)  +\mathbb{P}%
\left(  D=1|Z=z_{k^{\prime}}\right)  -\mathbb{P}\left(  Y\in A_{1}%
,D=0|Z=z_{k}\right) \\
=&\,1-\mathbb{P}\left(D=0|Z=z_{k^{\prime}}\right)  +\mathbb{P}\left(
Y\in A_{1},D=0|Z=z_{k^{\prime}}\right)  -\mathbb{P}\left(  Y\in A_{1}%
,D=0|Z=z_{k}\right) \\
=&\,1-\mathbb{P}\left(  Y\in A_{1}^{c},D=0|Z=z_{k^{\prime}}\right)
-\mathbb{P}\left(  Y\in A_{1},D=0|Z=z_{k}\right) \\
\geq&\,1-\mathbb{P}\left(  Y\in A_{1}^{c},D=0|Z=z_{k}\right)  -\mathbb{P}%
\left(  Y\in A_{1},D=0|Z=z_{k}\right) \\
=&\,\mathbb{P}\left(  D=1|Z=z_{k}\right)  \geq0.
\end{align*}

For $j=2$, the proof is analogous.

Next, we show that the testable implications in \eqref{eq.testable implication binary D} are sharp for every pair $(z,z')$. 
We closely follow the proof of Proposition 1.1(i) in \citet{kitagawa2015test} and show that given the testable implications in \eqref{eq.testable implication binary D} hold for pair $(z,z')$, then there is always a joint distribution of the variables $(  Y_{0z},Y_{0z^{\prime}
},Y_{1z},Y_{1z^{\prime}},D_{z},D_{z^{\prime}},Z)  $ that satisfies all the
pairwise IV validity conditions in Definition \ref{def.partial validity pairwise binary D} and induces the observable distributions. Without loss of generality, here we suppose $Z\in\mathbb{R}$ for simplicity. It is straightforward to extend the results to multi-dimensional $Z$. 

Fix  $z,z^{\prime}\in\mathcal{Z}$. For each $d\in\{0,1\}$, suppose
$Y_{d}\left(  z,z^{\prime}\right)  =Y_{dz}=Y_{dz^{\prime}}$ a.s.
Recall that for all $z\in\mathcal{Z}$,
\[
P_{z}\left(  B,C\right)  =\mathbb{P}\left(  Y\in B,D\in C|Z=z\right)
.
\]
We first define for all $B_{1},B_{0}\in
\mathcal{B}_{\mathbb{R}}$,
\begin{align*}
&  w\left(  B_{1},B_{0},\left\{  1\right\}  ,\left\{  0\right\}  \right)  \\
= &  \,\frac{P_{z^{\prime}}\left(  B_{1},\left\{  1\right\}  \right)
-P_{z}\left(  B_{1},\left\{  1\right\}  \right)  }{P_{z^{\prime}}\left(
\mathbb{R},\left\{  1\right\}  \right)  -P_{z}\left(  \mathbb{R},\left\{
1\right\}  \right)  }\cdot\frac{P_{z}\left(  B_{0},\left\{  0\right\}
\right)  -P_{z^{\prime}}\left(  B_{0},\left\{  0\right\}  \right)  }%
{P_{z}\left(  \mathbb{R},\left\{  0\right\}  \right)  -P_{z^{\prime}}\left(
\mathbb{R},\left\{  0\right\}  \right)  }\cdot\left(  P_{z^{\prime}}\left(
\mathbb{R},\left\{  1\right\}  \right)  -P_{z}\left(  \mathbb{R},\left\{
1\right\}  \right)  \right)  ,
\end{align*}%
\begin{align*}
 w\left(  B_{1},B_{0},\left\{  0\right\}  ,\left\{  0\right\}  \right) 
= \frac{P_{z}\left(  B_{1},\left\{  1\right\}  \right)  }{P_{z}\left(
\mathbb{R},\left\{  0\right\}  \right)  }\cdot\frac{P_{z^{\prime}}\left(
B_{0},\left\{  0\right\}  \right)  }{P_{z^{\prime}}\left(  \mathbb{R},\left\{
0\right\}  \right)  }\cdot P_{z^{\prime}}\left(  \mathbb{R},\left\{
0\right\}  \right)  ,
\end{align*}%
\begin{align*}
 w\left(  B_{1},B_{0},\left\{  1\right\}  ,\left\{  1\right\}  \right)  
= \frac{P_{z}\left(  B_{1},\left\{  1\right\}  \right)  }{P_{z}\left(
\mathbb{R},\left\{  1\right\}  \right)  }\cdot\frac{P_{z^{\prime}}\left(
B_{0},\left\{  1\right\}  \right)  }{P_{z^{\prime}}\left(  \mathbb{R},\left\{
1\right\}  \right)  }\cdot P_{z}\left(  \mathbb{R},\left\{  1\right\}
\right)  ,
\end{align*}
and%
\[
w\left(  B_{1},B_{0},\left\{  0\right\}  ,\left\{  1\right\}  \right)  =0.
\]
Let $\mathscr{H}$ denote the collection of all h-intervals \citep[p.~33]{folland2013real}. That is, for every $B\in\mathscr{H}$, $B$ is of the form $(a,b]$, $(a,\infty)$, or $\varnothing$, where $-\infty\le a<b<\infty$.
Let $\mathcal{E}$ be the collection of all sets of the form $B_{1}\times B_{0}\times C_{2}\times C_{1}$ with $B_{1},B_{0},C_{2},C_{1}\in
\mathscr{H}$. 
Then we define a function $\mu_0:\mathcal{E}\to[0,\infty]$ such that
for all $B_{1},B_{0},C_{2},C_{1}\in
\mathscr{H}$,
\[
\mu_{0}\left(  B_{1}\times B_{0}\times C_{2}\times C_{1}\right)  =\sum
_{d_{2},d_{1}\in\left\{  0,1\right\}  }  1\left\{
d_{2}\in C_{2},d_{1}\in C_{1}\right\}    \cdot w\left(  B_{1},B_{0},\left\{
d_{2}\right\}  ,\left\{  d_{1}\right\}  \right)  .
\]
By the construction of $\mu_0$, it is straightforward to show that if $\{A_{j}\}_{j=1}^J\subseteq\mathcal{E}$ are mutually
disjoint and $\cup_{j=1}^{J}A_{j}\in\mathcal{E}$, then
\begin{align}\label{eq.mu0 finite additivity}
\mu_{0}\left(  \cup_{j=1}^{J}A_{j}\right)  =\sum_{j=1}^{J}\mu_{0}\left(
A_{j}\right)  .
\end{align}

By Proposition 1.7 of \citet{folland2013real}, the collection $\mathcal{A}$ of finite disjoint
unions of sets in $\mathcal{E}$ is an
algebra. By Propositions 1.4 and 1.6 of \citet{folland2013real}, the $\sigma$-algebra generated by $\mathcal{A}$ is $\mathcal{B}_{\mathbb{R}
^{4}}$. 
We now extend $\mu_{0}$ to $\mu_{1}:\mathcal{A}\rightarrow\left[
0,\infty\right]  $ such that for all $A\in\mathcal{A}$ with $A=\cup_{j=1}^J A_j$ and $\{A_{j}\}_{j=1}^J\subseteq\mathcal{E}$ which
are mutually disjoint,
\begin{align}\label{eq.mu1 finite additivity}
\mu_{1}\left(  A\right)  =\sum_{j=1}^{J}\mu_{0}\left(
A_{j}\right)  .
\end{align}
We next show that $\mu_{1}$ is well defined and is a premeasure on $\mathcal{A}$ following the idea of \citet[Proof of Proposition 1.15]{folland2013real}.\footnote{See the definition of premeasure in \citet[p.~30]{folland2013real}.} Suppose $\{A_{j}^1\}_{j=1}^J,\{A_{k}^2\}_{k=1}^K\subseteq\mathcal{E}$ are two collections of mutually disjoint sets and $\cup_{j=1}
^{J}A_{j}^{1}=\cup_{k=1}^{K}A_{k}^{2}$. Then,
it follows from \eqref{eq.mu0 finite additivity} that
\[
\sum_{j=1}^{J}\mu_{0}\left(  A_{j}^{1}\right)  =\sum_{j=1}^{J}\sum_{k=1}%
^{K}\mu_{0}\left(  A_{j}^{1}\cap A_{k}^{2}\right)  =\sum_{k=1}^{K}\mu
_{0}\left(  A_{k}^{2}\right)  .
\]
This verifies that $\mu_{1}$ is well defined.

Then we show that in general, if $\{A_{j}\}_{j=1}^J\subseteq\mathcal{E}$ (which may not be
mutually disjoint), then
\[
\mu_{1}\left(  \cup_{j=1}^{J}A_{j}\right)  \leq\sum_{j=1}^{J}\mu_{0}\left(
A_{j}\right)  .
\]
Suppose $A_{j}=B_{1}^{j}\times B_{2}^{j}\times B_{3}^{j}\times B_{4}^{j}$, where $B_{m}^{j}\in \mathscr{H}$ for $m\in\{1,2,3,4\}$.
Then for every $m$, we can find $\{C_{m}^{j}\}_{j=1}^{J_{m}}\subseteq \mathscr{H}$ such that $\{C_{m}^{j}\}_{j=1}^{J_{m}}$ are mutually disjoint and every $B_{m}^{j}$ can be
written as a union of elements in $\{C_{m}^{j}\}_{j=1}^{J_{m}}$.
We next define mutually disjoint sets
\[
A_{j_{1}j_{2}j_{3}j_{4}}=C_{1}^{j_{1}}\times C_{2}^{j_{2}}\times C_{3}^{j_{3}%
}\times C_{4}^{j_{4}} \text{ for } j_m\in\{1,\ldots,J_m\}.
\]
Then $\cup_{j=1}^{J}A_{j}$ can be written as
\[
\cup_{j=1}^{J}A_{j}=\cup_{E\in\mathcal{S}}E,
\]
where $\mathcal{S}$ is some collection of $A_{j_{1}j_{2}j_{3}j_{4}}$. By the
construction of $\mu_{1}$,
\[
\mu_{1}\left(  \cup_{j=1}^{J}A_{j}\right)  =\sum_{E\in\mathcal{S}}\mu
_{0}\left(  E\right)  .
\]
Also, in this way, every $A_{j}$ can be written as a union of $A_{j_{1}%
j_{2}j_{3}j_{4}}$, that is,
\[
A_{j}=\cup_{E\in\mathcal{S}_{j}}E,
\]
where $\mathcal{S}_{j}$ is some collection of $A_{j_{1}j_{2}j_{3}j_{4}}$.
Note that here $\{A_{j}\}_{j=1}^J$ may not be mutually disjoint. So $\{\mathcal{S}_{j}\}_{j=1}^J$ may not be mutually disjoint. Also, $\mathcal{S=\cup}_{j}%
\mathcal{S}_{j}$ by definition. Then it follows that
\[
\sum_{j=1}^{J}\mu_{0}\left(  A_{j}\right)  =\sum_{j=1}^{J}\sum_{E\in
\mathcal{S}_{j}}\mu_{0}\left(  E\right)  \geq\sum_{E\in\mathcal{S}}\mu
_{0}\left(  E\right)  =\mu_{1}\left(  \cup_{j=1}^{J}A_{j}\right)  .
\]
Similarly, it is then straightforward to show that if $\{A_{j}\}_{j=1}^J\subseteq\mathcal{A}$ (which may not be
mutually disjoint), then
\[
\mu_{1}\left(  \cup_{j=1}^{J}A_{j}\right)  \leq\sum_{j=1}^{J}\mu_{1}\left(
A_{j}\right)  .
\]

Suppose $\{A_{j}\}\subseteq\mathcal{A}$, $\cup_{j=1}^{\infty}A_{j}%
\in\mathcal{A}$, and $\{A_{j}\}$ are mutually disjoint. Next we show that
\[
\mu_{1}\left(  \cup_{j=1}^{\infty}A_{j}\right)  =\sum_{j=1}^{\infty}\mu
_{1}\left(  A_{j}\right)  .
\]
Since $\cup_{j=1}^{\infty}A_{j}\in\mathcal{A}$, $\cup_{j=1}^{\infty}A_{j}$ can
be written as a finite disjoint union of elements in $\mathcal{E}$. Thus,
$\{A_{j}\}$ can be partitioned into finitely many subsets such that the union
of each subset is an element in $\mathcal{E}$. For simplicity, we consider
the case where $\cup_{j=1}^{\infty}A_{j}\in\mathcal{E}$. Otherwise, we may
consider every subset separately and \eqref{eq.mu1 finite additivity} gives
the result. Now suppose $A=\cup_{j=1}^{\infty}A_{j}=B_{1}\times B_{2}\times
B_{3}\times B_{4}$, where $B_{m}\in\mathscr{H}$. Fix an arbitrary $J$. The set $A\setminus\cup_{j=1}^{J}A_{j}\in\mathcal{A}$ by the definition of an algebra. Then it follows that
\[
\mu_{1}\left(  A\right)  =\mu_{1}\left(  \cup_{j=1}^{J}A_{j}\right)  +\mu
_{1}\left(  A\setminus\cup_{j=1}^{J}A_{j}\right)  \geq\sum_{j=1}^{J}\mu
_{1}\left(  A_{j}\right)  .
\]
Then we have $\mu_{1}\left(  A\right)  \geq\sum_{j=1}^{\infty}\mu_{1}(A_{j})$
by letting $J\rightarrow\infty$.

Fix an arbitrary $\varepsilon>0$. Suppose $A=B_{1}\times B_{2}\times
B_{3}\times B_{4}$ and $A_{j}=\cup_{k=1}^{K_j}B_{1}^{jk}\times B_{2}
^{jk}\times B_{3}^{jk}\times B_{4}^{jk}$ with $B_m=(a_m,b_m]$ and $B_{m}^{jk}=(a_{m}^{jk},b_{m}^{jk}]$, where $a_m,b_m,a_{m}^{jk},b_{m}^{jk}\in\mathbb{R}$ and $\{B_{1}^{jk}\times B_{2}
^{jk}\times B_{3}^{jk}\times B_{4}^{jk}\}_{k=1}^{K_j}$ are mutually disjoint.
By the construction of $\mu_{1}$, there is some $\delta>0$ such that%
\[
\mu_{1}\left(  A\right)  -\mu_{1}\left(  A_{\delta}\right)  <\varepsilon,
\]
where $A_{\delta}=B_{1\delta}\times B_{2\delta}\times B_{3\delta}\times
B_{4\delta}$ and $B_{m\delta}=(a_{m}+\delta,b_{m}]$. Also, there is
$\delta_{j}>0$ such that
\[
\mu_{1}\left(  A_{j\delta}\right)  -\mu_{1}\left(  A_{j}\right)
<\varepsilon2^{-j},
\]
where $A_{j\delta}=\cup_{k=1}^{K_j}B_{1\delta}^{jk}\times B_{2\delta}^{jk}\times B_{3\delta}%
^{jk}\times B_{4\delta}^{jk}$ and $B_{m\delta}^{jk}=(a_{m}^{jk},b_{m}^{jk}%
+\delta_{j}]$. Since by construction, the open sets $\{\cup_{k=1}^{K_j}\prod_{m\in\left\{  1,2,3,4\right\}
}(a_{m}^{jk},b_{m}^{jk}+\delta_{j})\}_{j}$ cover the compact set $\prod
_{m\in\left\{  1,2,3,4\right\}  }\left[  a_{m}+\delta,b_{m}\right]  $, there
is a finite subcover. Then by relabeling, we can find a cover $\{\cup_{k=1}^{K_j}\prod
_{m\in\left\{  1,2,3,4\right\}  }(a_{m}^{jk},b_{m}^{jk}+\delta_{j}]\}_{j=1}^{J}$
of $\prod_{m\in\left\{  1,2,3,4\right\}  }(  a_{m}+\delta,b_{m}]
$. It now follows that
\begin{align*}
\mu_{1}\left(  A\right)   &  <\mu_{1}\left(  A_{\delta}\right)  +\varepsilon
\leq\mu_{1}\left(  \cup_{j=1}^{J}\left\{  \cup_{k=1}^{K_j}\prod_{m\in\left\{  1,2,3,4\right\}
}(a_{m}^{jk},b_{m}^{jk}+\delta_{j}]\right\}  \right)  +\varepsilon\\
&  \leq\sum_{j=1}^{J}\mu_{1}\left(  A_{j\delta}\right)  +\varepsilon\leq
\sum_{j=1}^{\infty}\mu_{1}\left(  A_{j}\right)  +2\varepsilon.
\end{align*}
Since $\varepsilon$ is arbitrary, $\mu_{1}\left(  A\right)  \leq\sum
_{j=1}^{\infty}\mu_{1}(A_{j})$. 

If $B_{m}=(-\infty,b_m]$ with $b_m\in\mathbb{R}$ for all $m$, then for every large $M>0$, we set
$A_{M}=\prod_{m=1}^{4}B_{mM}$, where $B_{mM}=(-M,b_{m}]$. By a similar proof,
we can show that $\mu_{1}(A_{M})\leq\sum_{j=1}^{\infty}\mu_{1}(A_{j})$. By the
construction of $\mu_{1}$, $\mu_{1}(A)\leq\sum_{j=1}^{\infty}\mu_{1}(A_{j})$.
If $B_{m}=(a_{m},\infty)$ with $a_{m}\in\mathbb{R}$ for all $m$, then for
every large $M>0$, we set $A_{M}=\prod_{m=1}^{4}B_{mM}$, where $B_{mM}=(a_{m},M]$. By a similar proof, we can show that $\mu_{1}(A_{M})\leq
\sum_{j=1}^{\infty}\mu_{1}(A_{j})$. By the construction of $\mu_{1}$, $\mu
_{1}(A)\leq\sum_{j=1}^{\infty}\mu_{1}(A_{j})$. The other cases are analogous.
Thus, we finally have that $\mu_{1}(\cup_{j=1}^{\infty}A_{j})=\sum
_{j=1}^{\infty}\mu_{1}(A_{j})$.
This implies that $\mu_{1}$ is a premeasure on $\mathcal{A}$.

Then by
Theorem 1.14 of \citet{folland2013real}, there is a measure $\mu$ on $\mathcal{B}
_{\mathbb{R}^{4}}$ whose restriction to $\mathcal{A}$ is $\mu_{1}$: For every $E\in\mathcal{B}_{\mathbb{R}^4}$,
\[
\mu(E)=\inf\left\{
\sum_{j=1}^{\infty}\mu_1(A_j): A_j\in\mathcal{A}, E\subseteq \cup_{j=1}^{\infty}A_j
\right\}.
\]
Clearly,
\begin{align*}
\mu(\mathbb{R}^{4}) =\mu_{0}(\mathbb{R}^{4}) =&\,w\left(  \mathbb{R},\mathbb{R},\left\{  1\right\}  ,\left\{  0\right\}
\right)  +w\left(  \mathbb{R},\mathbb{R},\left\{  0\right\}  ,\left\{
0\right\}  \right)  \\
&  +w\left(  \mathbb{R},\mathbb{R},\left\{  1\right\}  ,\left\{  1\right\}
\right)  +w\left(  \mathbb{R},\mathbb{R},\left\{  1\right\}  ,\left\{
0\right\}  \right)  
=1.
\end{align*}
Thus, $\mu$ is a probability measure. Let $\nu$ denote the marginal
distribution of $Z$. Then the product measure $\mu\times\nu$ \citep[p.~64]{folland2013real} is
also a probability measure such that for all $A_{1}\in
\mathcal{B}_{\mathbb{R}^{4}}$ and $A_{2}\in\mathcal{B}_{\mathbb{R}}$,
\[
\mu\times\nu\left(  A_{1}\times A_{2}\right)  =\mu\left(  A_{1}\right)
\nu\left(  A_{2}\right)  .
\]
Suppose the random variables $\left(  Y_{1}\left(  z,z^{\prime}\right)  ,Y_{0}\left(  z,z^{\prime
}\right)  ,D_{z'},D_{z},Z\right)  $ have the joint distribution $\mu\times\nu$ so that
for all $B_{1},B_{0},C_{2},C_{1},A\in\mathscr{H}$,
\begin{align*}
& \mathbb{P}\left(  Y_{1}\left(  z,z^{\prime}\right)  \in B_{1},Y_{0}\left(
z,z^{\prime}\right)  \in B_{0},D_{z^{\prime}}\in C_{2},D_{z}\in C_{1},Z\in
A\right)  \\
=&\,\mu\left(  B_{1}\times B_{0}\times C_{2}\times C_{1}\right)  \times
\nu\left(  A\right)  .
\end{align*}
Then it follows that under the probability measure $\mu\times\nu$,
\begin{align*}
\mathbb{P}\left(  Y\in B_{1},D=1|Z=z^{\prime}\right)  = &  \,\mathbb{P}\left(
Y_{1}\left(  z,z^{\prime}\right)  \in B_{1},Y_{0}\left(  z,z^{\prime}\right)
\in\mathbb{R},D_{z^{\prime}}=1,D_{z}=1|Z=z^{\prime}\right)  \\
&  +\mathbb{P}\left(  Y_{1}\left(  z,z^{\prime}\right)  \in B_{1},Y_{0}\left(
z,z^{\prime}\right)  \in\mathbb{R},D_{z^{\prime}}=1,D_{z}=0|Z=z^{\prime
}\right)  \\
 =&\,\mu\left(  B_{1}\times\mathbb{R}\times\left\{  1\right\}  \times\left\{
1\right\}  \right)  +\mu\left(  B_{1}\times\mathbb{R}\times\left\{  1\right\}
\times\left\{  0\right\}  \right)  \\
=&\,P_{z^{\prime}}\left(  B_{1},\{1\}\right)  ,
\end{align*}%
\begin{align*}
\mathbb{P}\left(  Y\in B_{1},D=1|Z=z\right)    =&\,\mathbb{P}\left(
Y_{1}\left(  z,z^{\prime}\right)  \in B_{1},Y_{0}\left(  z,z^{\prime}\right)
\in\mathbb{R},D_{z^{\prime}}=1,D_{z}=1|Z=z\right)  \\
&  +\mathbb{P}\left(  Y_{1}\left(  z,z^{\prime}\right)  \in B_{1},Y_{0}\left(
z,z^{\prime}\right)  \in\mathbb{R},D_{z^{\prime}}=0,D_{z}=1|Z=z\right)  \\ =&\,\mu\left(  B_{1}\times\mathbb{R}\times\left\{  1\right\}  \times\left\{
1\right\}  \right)  +\mu\left(  B_{1}\times\mathbb{R}\times\left\{  0\right\}
\times\left\{  1\right\}  \right)  \\
=&\,P_{z}\left(  B_{1},\{1\}\right)  ,
\end{align*}%
\begin{align*}
\mathbb{P}\left(  Y\in B_{0},D=0|Z=z^{\prime}\right)     =&\,\mathbb{P}\left(
Y_{1}\left(  z,z^{\prime}\right)  \in\mathbb{R},Y_{0}\left(  z,z^{\prime
}\right)  \in B_{0},D_{z^{\prime}}=0,D_{z}=0|Z=z^{\prime}\right)  \\
&  +\mathbb{P}\left(  Y_{1}\left(  z,z^{\prime}\right)  \in\mathbb{R}%
,Y_{0}\left(  z,z^{\prime}\right)  \in B_{0},D_{z^{\prime}}=0,D_{z}%
=1|Z=z^{\prime}\right)  \\
=&\,\mu\left(  \mathbb{R}\times B_{0}\times\left\{  0\right\}  \times\left\{
0\right\}  \right)  +\mu\left(  \mathbb{R}\times B_{0}\times\left\{
0\right\}  \times\left\{  1\right\}  \right)  \\
=&\,P_{z^{\prime}}\left(  B_{0},\{0\}\right)  ,
\end{align*}
\begin{align*}
\mathbb{P}\left(  Y\in B_{0},D=0|Z=z\right)   =&\,\mathbb{P}\left(
Y_{1}\left(  z,z^{\prime}\right)  \in\mathbb{R},Y_{0}\left(  z,z^{\prime
}\right)  \in B_{0},D_{z^{\prime}}=0,D_{z}=0|Z=z\right)  \\
&  +\mathbb{P}\left(  Y_{1}\left(  z,z^{\prime}\right)  \in\mathbb{R}%
,Y_{0}\left(  z,z^{\prime}\right)  \in B_{0},D_{z^{\prime}}=1,D_{z}%
=0|Z=z\right)  \\
=&\,\mu\left(  \mathbb{R}\times B_{0}\times\left\{  0\right\}  \times\left\{
0\right\}  \right)  +\mu\left(  \mathbb{R}\times B_{0}\times\left\{
1\right\}  \times\left\{  0\right\}  \right)  \\
=&\,P_{z}\left(  B_{0},\{0\}\right)  .
\end{align*}
Here, we use
\begin{align*}
  P_{z^{\prime}}\left(  \mathbb{R},\left\{  1\right\}  \right)
-P_{z}\left(  \mathbb{R},\left\{  1\right\}  \right)   
&  =\mathbb{P}\left(  D=1|Z=z'\right)  -\mathbb{P}\left(  D=1|Z=z
\right)  \\
&  =1-\mathbb{P}\left(  D=0|Z=z'\right)  -1+\mathbb{P}\left(
D=0|Z=z\right)  \\
&  =P_{z}\left(  \mathbb{R},\left\{  0\right\}  \right)  -P_{z^{\prime}%
}\left(  \mathbb{R},\left\{  0\right\}  \right)  .
\end{align*}
This implies that given $P_{z}$ and $P_{z^{\prime}}$ that satisfy \eqref{eq.testable implication binary D}, we can
always construct a joint distribution of $\left(  Y_{0z},Y_{0z^{\prime}%
},Y_{1z},Y_{1z^{\prime}},D_{z},D_{z^{\prime}},Z\right)  $ that satisfies all the
pairwise IV validity conditions in Definition \ref{def.partial validity pairwise binary D} and induces $P_{z}$ and $P_{z^{\prime}}$ \citep[Theorem A.1.5]{durrett2019probability}.
\end{proof} 

\begin{proof}[Proof of Lemma \ref{lemma.superset of Z pairwise binary D}]
    Lemma \ref{lemma.superset of Z pairwise binary D} is a special case of Lemma \ref{lemma.superset of Z pairwise} in Appendix \ref{sec.Z1 estimation ordered}. See the proof of Lemma \ref{lemma.superset of Z pairwise} in Appendix \ref{sec.Z1 estimation ordered}.
\end{proof}

\begin{proof}[Proof of Proposition \ref{prop.consistent G hat pairwise Z1 binary D}]
    Proposition \ref{prop.consistent G hat pairwise Z1 binary D} is a special case of Proposition \ref{prop.consistent G hat pairwise Z1} in Appendix \ref{sec.Z1 estimation ordered}. See the proof of Proposition \ref{prop.consistent G hat pairwise Z1} in Appendix \ref{sec.Z1 estimation ordered}.
\end{proof}

\section{No Information on LATEs for Invalid Pairs}
\label{sec.no information}

\subsection{Theoretical Results and Discussion}\label{sec.no information results}

In this section, we show that absent restrictions on how pairwise validity (Definition \ref{def.partial validity pairwise binary D}) can be violated for a pair $(z_k,z_{k'})$, there is no information about $\beta_{k',k}$  in the data. To state the formal results, we define the target parameter when the LATE assumptions are violated as $\beta_{k',k}=\beta^1_{k',k}-\beta^0_{k',k}$,  where
\begin{align*}
\beta^d_{k',k}=\sum_{z\in\mathcal{Z}}E[Y_{dz}|D_{z_{k^{\prime}}}>D_{z_{k}},Z=z]\mathbb{P}(Z=z|D_{z_{k^{\prime}}}>D_{z_{k}}) \text{ for } d\in\{0,1\}.
\end{align*}
This expression for $\beta_{k',k}$ is an adaptation of the LATE and ATE under violations of exclusion in, for example, \citet{cui2024robust} and \citet{kedagni2020discordant}. Let $\mathscr{Z}^1$ be the set of all pairs that satisfy Definition \ref{def.partial validity pairwise binary D}\ref{def.pairwise exclusion}, $\mathscr{Z}^2$ the set of all pairs that satisfy Definition \ref{def.partial validity pairwise binary D}\ref{def.pairwise random assignment}, and $\mathscr{Z}^3$ the set of all pairs that satisfy Definition \ref{def.partial validity pairwise binary D}\ref{def.pairwise monotonicity}. Clearly, $\mathscr{Z}_{\bar{M}}=\cap_j \mathscr{Z}^j$. 

The following proposition shows that even if only Definition \ref{def.partial validity pairwise binary D}\ref{def.pairwise exclusion} or only Definition \ref{def.partial validity pairwise binary D}\ref{def.pairwise random assignment} is violated for pair $(z_k,z_{k'})$, the sharp identified set for $\beta_{k',k}$ is $\mathbb{R}$. It is related to recent results on the identification of the LATE when a subset of LATE assumptions is violated \citep[e.g.,][]{huber2017sharp,noack2021sensitivity,kedagni2023identifying,cui2024robust}.
\begin{proposition}\label{prop:no_information}
Fix a distribution of $(Y,D,Z)$ so that the testable restrictions of Definition \ref{def.partial validity pairwise binary D} are satisfied for the supposed valid pairs, $\mathbb{P}\left( D=1|Z=z_{k'}\right)-\mathbb{P}\left( D=1|Z=z_{k}\right)>0$ for some pair $(z_{k},z_{k'})$, and $\mathbb{P}(Z=z)>0$ for all $z\in \mathcal{Z}$. 
\begin{enumerate}[label=(\roman*)]
\item Suppose that $(z_k,z_{k'})\in (\mathscr{Z}^1)^c\cap \mathscr{Z}^2\cap \mathscr{Z}^3$. Then, the sharp identified set for $\beta_{k',k}$ is $\mathbb{R}$.
\item Suppose that $(z_k,z_{k'})\in \mathscr{Z}^1\cap (\mathscr{Z}^2)^c\cap \mathscr{Z}^3$. Then, the sharp identified set for $\beta_{k',k}$ is $\mathbb{R}$.
\end{enumerate}
\end{proposition}

In Proposition \ref{prop:no_information}, we maintain the first-stage assumption $\mathbb{P}( D=1|Z=z_{k'})-\mathbb{P}( D=1|Z=z_{k})>0$ since otherwise $\mathbb{P}( D_{z_{k'}}>D_{z_k})$ could be equal to zero in which case the LATE may not be a meaningful object of interest.

Proposition \ref{prop:no_information} shows that violations of exclusion or independence alone can lead to $\beta_{k',k}$ taking any value in $\mathbb{R}$. Interestingly, the same result does not hold when only monotonicity (Definition \ref{def.partial validity pairwise binary D}\ref{def.pairwise monotonicity}) fails. See, for example, \citet{cui2024robust} for non-trivial sharp bounds for the LATE under exclusion and independence when the instrument is binary.

An immediate corollary of Proposition \ref{prop:no_information} is that absent restrictions on how Definition \ref{def.partial validity pairwise binary D} can be violated, there is no information in the data about $\beta_{k',k}$ for $(z_k,z_{k'})\notin \mathscr{Z}_{\bar{M}}$.

\begin{corollary}\label{cor:no_information}
    Consider the setup in Proposition \ref{prop:no_information}. Suppose that there are no further restrictions for a pair $(z_k,z_{k'})\notin \mathscr{Z}_{\bar{M}}$. Then, the sharp identified set for $\beta_{k',k}$ is $\mathbb{R}$.
\end{corollary}

We conclude this section by discussing a limitation of the theoretical results in this section. A key feature of Definition \ref{def.partial validity pairwise binary D} is that there can be valid and invalid pairs. The presence of valid pairs may restrict how Definition \ref{def.partial validity pairwise binary D} can be violated for the invalid pairs. Whenever independence or exclusion can be violated, the results in Proposition \ref{prop:no_information} apply. However, there exist configurations of valid and invalid pairs under which Proposition \ref{prop:no_information} (and Corollary \ref{cor:no_information}) does not apply. Suppose, for instance, that $\mathcal{Z}= \{z_1,z_2,z_3\}$, where $(z_{1},z_3)$ and $(z_{2},z_3)$ are valid. This configuration implies that (a) $(z_{1},z_2)\in \mathscr{Z}^1$ with $Y_{dz}=Y_d$ a.s. and (b) $(Y_{1},Y_0,D_{z_1})\ci Z$ and $(Y_{1},Y_0,D_{z_2})\ci Z$. While (b) is weaker than joint independence (Definition \ref{def.partial validity pairwise binary D}\ref{def.pairwise random assignment}), it is sufficient (together with exclusion and monotonicity) for point identification of the pairwise LATE. Therefore, under this particular configuration, the only way in which the LATE may not be point identified is if $(z_{1},z_2)$ violates monotonicity (Definition \ref{def.partial validity pairwise binary D}\ref{def.pairwise monotonicity}), so that the result in 
Proposition \ref{prop:no_information} does not apply.

\subsection{Proofs for Appendix \ref{sec.no information results}} 

\begin{proof}[Proof of Proposition \ref{prop:no_information}]
We prove both parts separately. Our proof strategies are similar to and build on those  in \citet[][Proposition 1.1]{kitagawa2015test} and \citet[][Theorem 1]{kedagni2023identifying}. Because $\mathbb{R}$ is trivially a valid identified set for $\beta_{k',k}$, we focus on proving sharpness.

\smallskip

\noindent\textbf{Proof of (i).} First, to highlight the main
arguments, we consider the simple case where
$\mathcal{Z}=\{z_{1},z_{2},z_{3}\}$. To prove the result, we find a
distribution of
\[
(Y_{0z_{1}},Y_{0z_{2}},Y_{0z_{3}},Y_{1z_{1}},Y_{1z_{2}},Y_{1z_{3}},D_{z_{1}%
},D_{z_{2}},D_{z_{3}},Z)
\]
that is consistent with the observed distribution of $(Y,D,Z)$ and violates
Definition \ref{def.partial validity pairwise binary D}\ref{def.pairwise exclusion} for the pair $(z_{2},z_{3})$ so that $\beta_{3,2}$ can take
any value in $\mathbb{R}$.\footnote{To provide a full measure theoretic description of the distribution, we may follow the strategy in the proof of Lemma \ref{lemma.testable implications for Z2 weaker}. That is, we first construct a premeasure that is consistent with the observed distributions, and then extend it to a probability measure based on Theorem 1.14 of \citet{folland2013real}.} The proof for any other pair is symmetric. 

Fix an arbitrary value $\beta\in\mathbb{R}$.
We find a distribution that satisfies
(a) $Y_{1z_{1}}=Y_{1z_{2}}=Y_{1}$ and $Y_{0z_{1}}=Y_{0z_{2}}=Y_{0}
$ a.s., but $Y_{dz_{3}}\neq Y_{d}$ for $d\in \{0,1\}$ with positive probability, (b) $Z$ is independent of $\left(
Y_{1z_{3}},Y_{0z_{3}},Y_{1},Y_{0},D_{z_{3}},D_{z_{2}},D_{z_{1}}\right)$, and (c)
$D_{z_{3}}\geq D_{z_{2}}\geq D_{z_{1}}$ a.s., so that
$\left(  z_{1},z_{2}\right)  \in\mathscr{Z}_{\bar{M}}$ and $\left(
z_{1},z_{3}\right)  ,\left(  z_{2},z_{3}\right)  \in(\mathscr{Z}^{1})^{c}\cap \mathscr{Z}^2\cap \mathscr{Z}^3$. Note that the marginal distribution of $Z$, which satisfies
$\mathbb{P}(Z=z)>0$ for $z\in\mathcal{Z}$, by assumption, does not matter for the
argument below.\footnote{This is similar, for example, to the proof of
Proposition 1.1 in \citet{kitagawa2015test}.} Therefore, we find a
distribution of $\left(  Y_{1z_{3}},Y_{0z_{3}},Y_{1},Y_{0},D_{z_{3}}%
,D_{z_{2}},D_{z_{1}}\right)  $ 
so that $\beta_{3,2}=\beta$ that is consistent with the 
conditional distribution of $(Y,D)$ given $Z$, which is fully characterized by
$P_{z}\left(  B,\{d\}\right)  $ for all $(B,d,z)\in\mathcal{B}_{\mathbb{R}%
}\times\{0,1\}\times\mathcal{Z}$.

Consider a distribution of $(  Y_{1z_{3}},Y_{0z_{3}},Y_{1},Y_{0},D_{z_{3}},D_{z_{2}%
},D_{z_{1}})  $ given $Z$ that satisfies (b), that is,
\begin{align*}
&\mathbb{P}\left(  Y_{1z_{3}}\in B_{1z_3},Y_{0z_{3}}\in B_{0z_3},Y_{1}\in B_{1},Y_{0}\in B_{0},D_{z_{3}%
}=d_{3},D_{z_{2}}=d_{2},D_{z_{1}}=d_{1}|Z=z\right)\\
=&\,\mathbb{P}\left(  Y_{1z_{3}}\in B_{1z_3},Y_{0z_{3}}\in B_{0z_3},Y_{1}\in B_{1},Y_{0}\in B_{0},D_{z_{3}%
}=d_{3},D_{z_{2}}=d_{2},D_{z_{1}}=d_{1}\right)
\end{align*}
for all $B_{1z_3},B_{0z_3},B_{1},B_{0}\in\mathcal{B}_{\mathbb{R}}$, all $d_{1},d_{2}
,d_{3}\in\{0,1\}$, and all  $z\in \mathcal{Z}$, and also satisfies 
\begin{align*}
& \mathbb{P}\left(  Y_{1z_{3}}\in B_{1z_3},Y_{0z_{3}}\in B_{0z_3},Y_{1}\in B_{1},Y_{0}\in B_{0}|D_{z_{3}%
}=d_{3},D_{z_{2}}=d_{2},D_{z_{1}}=d_{1}\right)\\
=&\,  \mathbb{P}\left(  Y_{1z_3}\in B_{1z_3}|D_{z_{3}}=d_{3},D_{z_{2}}=d_{2},D_{z_{1}%
}=d_{1}\right)  \times \mathbb{P}\left(  Y_{0z_3}\in B_{0z_3}|D_{z_{3}}=d_{3},D_{z_{2}}=d_{2},D_{z_{1}%
}=d_{1}\right)\\
& \times \mathbb{P}\left(  Y_{1}\in B_{1}|D_{z_{3}}%
=d_{3},D_{z_{2}}=d_{2},D_{z_{1}}=d_{1}\right)  \times\mathbb{P}\left(  Y_{0}\in B_{0}|D_{z_{3}}=d_{3},D_{z_{2}}%
=d_{2},D_{z_{1}}=d_{1}\right)
\end{align*}
for all $B_{1z_3},B_{0z_3},B_{1},B_{0}\in\mathcal{B}_{\mathbb{R}}$ and all $d_{1},d_{2}%
,d_{3}\in\{0,1\}$, which implies (a) by setting the sets $B_{1z_3},B_{0z_3},B_1,B_0$ to be disjoint. Choose
\begin{align*}
\mathbb{P}\left(  V\in B,D_{z_{3}}=0,D_{z_{2}}=0,D_{z_{1}}=1\right)   &  =0,\\
\mathbb{P}\left(  V\in B,D_{z_{3}}=0,D_{z_{2}}=1,D_{z_{1}}=0\right)   &  =0,\\
\mathbb{P}\left(  V\in B,D_{z_{3}}=0,D_{z_{2}}=1,D_{z_{1}}=1\right)   &  =0,\\
\mathbb{P}\left(  V\in B,D_{z_{3}}=1,D_{z_{2}}=0,D_{z_{1}}=1\right)   &  =0,
\end{align*}
for each $V\in\{Y_{1z_{3}},Y_{0z_{3}},Y_{1},Y_{0}\}$ and all $B\in\mathcal{B}_{\mathbb{R}}$,
which implies (c). Moreover, choose
\begin{align*}
\mathbb{P}\left(  Y_{1}\in B,D_{z_{3}}=1,D_{z_{2}}=1,D_{z_{1}}=1\right)   &
=P_{z_{1}}\left(  B,\left\{  1\right\}  \right)  ,\\
\mathbb{P}\left(  Y_{1}\in B,D_{z_{3}}=1,D_{z_{2}}=1,D_{z_{1}}=0\right)   &
=P_{z_{2}}\left(  B,\left\{  1\right\}  \right)  -P_{z_1}\left(  B,\left\{
1\right\}  \right)  ,
\end{align*}
and 
\begin{align*}
&\mathbb{P}\left(  Y_{0}\in B,D_{z_{3}}=1,D_{z_{2}}=1,D_{z_{1}}=0\right)   
=P_{z_{1}}\left(  B,\left\{  0\right\}  \right)  -P_{z_{2}}\left(  B,\left\{
0\right\}  \right)  ,\\
&\mathbb{P}\left(  Y_{0}\in B,D_{z_{3}}=0,D_{z_{2}}=0,D_{z_{1}}=0\right)  +\mathbb{P}\left(  Y_{0}\in B,D_{z_{3}}=1,D_{z_{2}}=0,D_{z_{1}}=0\right)
=P_{z_{2}}\left(  B,\left\{
0\right\}  \right) .
\end{align*}
Note that $P_{z_{2}}\left(  B,\left\{  1\right\}  \right)  -P_{z_1}\left(  B,\left\{
1\right\}  \right)\ge 0$ and $P_{z_{1}}\left(  B,\left\{  0\right\}  \right)  -P_{z_{2}}\left(  B,\left\{
0\right\}  \right) \ge 0$, since the distribution of $(Y,D,Z)$ satisfies the testable implications of Definition \ref{def.partial validity pairwise binary D} for the valid pair $(z_1,z_2)$ by assumption.
Finally, choose
\begin{align*}
&  \mathbb{P}\left(  Y_{1z_3}\in B,D_{z_{3}}=1,D_{z_{2}}=1,D_{z_{1}}=1\right)
+\mathbb{P}\left(  Y_{1z_3}\in B,D_{z_{3}}=1,D_{z_{2}}=1,D_{z_{1}}=0\right)  \\
&  +\mathbb{P}\left(  Y_{1z_3}\in B,D_{z_{3}}=1,D_{z_{2}}=0,D_{z_{1}}=0\right)
=P_{z_{3}}\left(  B,\left\{  1\right\}  \right)  
\end{align*}
and 
\begin{align*}
\mathbb{P}\left(  Y_{0z_3}\in B,D_{z_{3}}=0,D_{z_{2}}=0,D_{z_{1}}=0\right) =P_{z_{3}}\left(  B,\left\{  0\right\}  \right).
\end{align*}
The above construction is consistent with the distribution of
$(Y,D)$ given $Z$ for any choice of $\mathbb{P}\left(  Y_{1}\in
B,D_{z_{3}}=1,D_{z_{2}}=0,D_{z_{1}}=0\right) $.\footnote{The same is true for other joint distributions that do not enter the above construction. To prove the desired result, it suffices to focus on $\mathbb{P}\left(  Y_{1}\in
B,D_{z_{3}}=1,D_{z_{2}}=0,D_{z_{1}}=0\right) $.}

Note that
\[
\mathbb{P}\left(  Y_{1}\in B|D_{z_{3}}>D_{z_{2}}\right)  =\frac{\mathbb{P}\left(  Y_{1}\in B,D_{z_{3}}>D_{z_{2}}\right)}{\mathbb{P}\left(D_{z_{3}}>D_{z_{2}}\right)}=\frac{\mathbb{P}
\left(  Y_{1}\in B,D_{z_{3}}=1,D_{z_2}=0,D_{z_{1}}=0\right)  }{\mathbb{P}\left( D=1|Z=z_{3}\right)-\mathbb{P}\left( D=1|Z=z_{2}\right)},
\]
where $\mathbb{P}\left( D=1|Z=z_{3}\right)-\mathbb{P}\left( D=1|Z=z_{2}\right)>0$ by assumption. Therefore, $\mathbb{P}(  Y_{1}\in
B,D_{z_{3}}=1,D_{z_{2}}=0,D_{z_{1}}=0) $ being arbitrary implies that $\mathbb{P}\left(  Y_{1}\in B|D_{z_{3}}>D_{z_{2}}\right)  $ is
arbitrary. It follows that $E[Y_{1}|D_{z_{3}}>D_{z_{2}}]$ can take any value in
$\mathbb{R}$.

To complete the proof, note that $\beta_{3,2}=\beta_{3,2}^{1}-\beta_{3,2}^{0}%
$, where, under the above construction, 
\[
\beta_{3,2}^{1}=E[Y_{1}|D_{z_{3}}>D_{z_{2}}]\mathbb{P}(Z\in\{z_{1}%
,z_{2}\})+E[Y_{1z_{3}}|D_{z_{3}}>D_{z_{2}}]\mathbb{P}(Z=z_{3}).
\]
Since there are no cross restrictions between $\beta_{3,2}^{1}$ and
$\beta_{3,2}^{0}$ as well as between $E[Y_{1}|D_{z_{3}}>D_{z_{2}}]$ and $E[Y_{1z_{3}}|D_{z_{3}}>D_{z_{2}}]$, and $\mathbb{P}(Z\in\{z_{1}
,z_{2}\})>0$ by assumption, for any value of $E[Y_{1z_{3}}|D_{z_{3}}>D_{z_{2}}]$
and $\beta_{3,2}^{0}$, there is always a conditional distribution $\mathbb{P}(  Y_{1}\in B|D_{z_{3}}>D_{z_{2}})  $ such that $\beta_{3,2}=\beta_{3,2}^{1}-\beta_{3,2}^{0}=\beta$. Since $\beta\in \mathbb{R}$ is arbitrary, the result follows.

Consider now the general case where $\mathcal{Z}=\{z_{1},\ldots,z_{K}\}$.
To prove the result, we find a distribution of
\[
(Y_{0z_{1}},\ldots,Y_{0z_{K}},Y_{1z_{1}},\ldots,Y_{1z_{K}},D_{z_{1}}%
,\ldots,D_{z_{K}},Z)
\]
that is consistent with the observed distribution of $(Y,D,Z)$ and violates
Definition \ref{def.partial validity pairwise binary D}\ref{def.pairwise exclusion} for the pair $(z_{K-1},z_{K})$ so that $\beta_{K,K-1}$ can
take any value in $\mathbb{R}$.  The proof for any other pair is symmetric. 

Fix an arbitrary value $\beta\in\mathbb{R}$.
We find a distribution 
that satisfies (a) $Y_{1z_{1}}=\cdots=Y_{1z_{K-1}}=Y_{1}$ and $Y_{0z_{1}}=\cdots
=Y_{0z_{K-1}}=Y_{0}$ a.s., but $Y_{dz_{K}}\neq Y_{d}$ for $d\in\{0,1\}$ with
positive probability, (b) $Z$ is independent of $(  Y_{1z_{K}},Y_{0z_{K}},Y_{1}
,Y_{0},D_{z_{K}},\ldots,D_{z_{1}})$, and (c) $D_{z_{K}}\geq\cdots\geq D_{z_{1}}$ a.s., so that $\left(  z_{k}%
,z_{k^{\prime}}\right)  \in\mathscr{Z}_{\bar{M}}$ for all $k<k^{\prime}<K$ and
$\left(  z_{k},z_{K}\right)  \in(\mathscr{Z}^{1})^{c}\cap \mathscr{Z}^2\cap \mathscr{Z}^3$ for all $k<K$. Note
that the marginal distribution of $Z$, which satisfies $\mathbb{P}(Z=z)>0$ for all $z\in\mathcal{Z}$ by assumption, does not matter for the argument
below. 
Therefore, we find a distribution of
$\left(  Y_{1z_{K}},Y_{0z_{K}},Y_{1},Y_{0},D_{z_{K}},\ldots,D_{z_{1}}\right)
$ so that $\beta_{K,K-1}=\beta$ that is
consistent with the observable conditional distribution of $(Y,D)$ given $Z$,
which is fully characterized by $P_{z}\left(  B,\{d\}\right)  $ for all
$(B,d,z)\in\mathcal{B}_{\mathbb{R}}\times\{0,1\}\times\mathcal{Z}$.

Consider a distribution of $(  Y_{1z_{K}},Y_{0z_{K}},Y_{1},Y_{0}
,D_{z_{K}},\ldots,D_{z_{1}})  $ given $Z$ that satisfies (b), that is,
\begin{align*}
&  \mathbb{P}\left(  Y_{1z_{K}}\in B_{1z_{K}},Y_{0z_{K}}\in B_{0z_{K}}%
,Y_{1}\in B_{1},Y_{0}\in B_{0},D_{z_{K}}=d_{K},\ldots,D_{z_{1}}=d_{1}%
|Z=z\right)  \\
=&\,\mathbb{P}\left(  Y_{1z_{K}}\in B_{1z_{K}},Y_{0z_{K}}\in B_{0z_{K}}%
,Y_{1}\in B_{1},Y_{0}\in B_{0},D_{z_{K}}=d_{K},\ldots,D_{z_{1}}=d_{1}\right)
\end{align*}
for all $B_{1z_{K}},B_{0z_{K}},B_{1},B_{0}\in\mathcal{B}_{\mathbb{R}}$, all
$d_{1},\ldots,d_{K}\in\{0,1\}$, and all $z\in\mathcal{Z}$, and also satisfies
\begin{align*}
& \mathbb{P}\left(  Y_{1z_{K}}\in B_{1z_{K}},Y_{0z_{K}}\in B_{0z_{K}},Y_{1}\in
B_{1},Y_{0}\in B_{0}|D_{z_{K}}=d_{K},\ldots,D_{z_{1}}=d_{1}\right)  \\
=&\,\mathbb{P}(Y_{1z_{K}}\in B_{1z_{K}}|D_{z_{K}}=d_{K},\ldots,D_{z_{1}}
=d_{1})\times\mathbb{P}\left(  Y_{0z_{K}}\in B_{0z_{K}}|D_{z_{K}}=d_{K}
,\ldots,D_{z_{1}}=d_{1}\right)  \\
& \times\mathbb{P}\left(  Y_{1}\in B_{1}|D_{z_{K}}=d_{K},\ldots,D_{z_{1}%
}=d_{1}\right)  \times\mathbb{P}\left(  Y_{0}\in B_{0}|D_{z_{K}}=d_{K}%
,\ldots,D_{z_{1}}=d_{1}\right)
\end{align*}
for all $B_{1z_{K}},B_{0z_{K}},B_{1},B_{0}\in\mathcal{B}_{\mathbb{R}}$ and all
$d_{1},\ldots,d_{K}\in\{0,1\}$, which implies (a). Choose
\[
\mathbb{P}\left(  V\in B,D_{z_{K}}=d_{K},\ldots,D_{z_{1}}=d_{1}\right)  =0,
\]
for each $V\in\{Y_{1z_{3}},Y_{0z_{3}},Y_{1},Y_{0}\}$, all $B\in\mathcal{B}%
_{\mathbb{R}}$, and all $d_{k}>d_{k^{\prime}}$ with $k<k^{\prime}$, which
implies (c). Moreover, choose
\begin{align*}
\mathbb{P}\left(  Y_{1}\in B,D_{z_{K}}=1,D_{z_{K-1}}=1,\ldots,D_{z_{1}%
}=1\right)  =P_{z_{1}}\left(  B,\left\{  1\right\}  \right)  ,
\end{align*}
\[
\vdots
\]
\begin{align*}
&  \mathbb{P}\left(  Y_{1}\in B,D_{z_{K}}=1,D_{z_{K-1}}=1,\ldots,D_{z_{1}%
}=1\right)  +\cdots\\
&  +\mathbb{P}\left(  Y_{1}\in B,D_{z_{K}}=1,D_{z_{K-1}}=1,\ldots,D_{z_{1}%
}=0\right)  =P_{z_{K-1}}\left(  B,\left\{  1\right\}  \right)  ,
\end{align*}%
\begin{align*}
&  \mathbb{P}\left(  Y_{1z_{K}}\in B,D_{z_{K}}=1,D_{z_{K-1}}=1,\ldots
,D_{z_{1}}=1\right)  +\mathbb{P}\left(  Y_{1z_{K}}\in B,D_{z_{K}}%
=1,D_{z_{K-1}}=1,\ldots,D_{z_{1}}=0\right)  \\
&  +\cdots+\mathbb{P}\left(  Y_{1z_{K}}\in B,D_{z_{K}}=1,D_{z_{K-1}}%
=0,\ldots,D_{z_{1}}=0\right)  =P_{z_{K}}\left(  B,\left\{  1\right\}  \right)
,
\end{align*}
and
\begin{align*}
&  \mathbb{P}\left(  Y_{0}\in B,D_{z_{K}}=1,D_{z_{K-1}}=1,\ldots,D_{z_{1}%
}=0\right)  +\cdots\\
&  +\mathbb{P}\left(  Y_{0}\in B,D_{z_{K}}=0,D_{z_{K-1}}=0,\ldots,D_{z_{1}%
}=0\right)  =P_{z_{1}}\left(  B,\left\{  0\right\}  \right)  .
\end{align*}%
\[
\vdots
\]%
\begin{align*}
& \mathbb{P}\left(  Y_{0}\in B,D_{z_{K}}=1,D_{z_{K-1}}=0,\ldots,D_{z_{1}%
}=0\right)  +\mathbb{P}\left(  Y_{0}\in B,D_{z_{K}}=0,D_{z_{K-1}}%
=0,\ldots,D_{z_{1}}=0\right)  \\
& =P_{z_{K-1}}\left(  B,\left\{  0\right\}  \right)  ,
\end{align*}%
\[
\mathbb{P}\left(  Y_{0z_{K}}\in B,D_{z_{K}}=0,D_{z_{K-1}}=0,\ldots,D_{z_{1}%
}=0\right)  =P_{z_{K}}\left(  B,\left\{  0\right\}  \right)  .
\]
The above construction is consistent with the distribution of $(Y,D)$ given
$Z$ for any choice of $\mathbb{P}(  Y_{1}\in B,D_{z_{K}}=1,D_{z_{K-1}
}=0,\ldots,D_{z_{1}}=0)  $.\footnote{Note that some elements in the above construction can be solved for explicitly as in the case where $\mathcal{Z}=\{z_1,z_2,z_3\}$.}

Note that
\begin{align*}
\mathbb{P}\left(  Y_{1}\in B|D_{z_{K}}>D_{z_{K-1}}\right)  &=\frac
{\mathbb{P}\left(  Y_{1}\in B,D_{z_{K}}>D_{z_{K-1}}\right)  }{\mathbb{P}
\left(  D_{z_{K}}>D_{z_{K-1}}\right)  }\\
&=\frac{\mathbb{P}\left(  Y_{1}\in
B,D_{z_{K}}=1,D_{z_{K-1}}=0,\ldots,D_{z_{1}}=0\right)  }{\mathbb{P}\left(
D=1|Z=z_{K}\right)  -\mathbb{P}\left(  D=1|Z=z_{K-1}\right)  },
\end{align*}
where $\mathbb{P}\left(  D=1|Z=z_{K}\right)  -\mathbb{P}\left(  D=1|Z=z_{K-1}%
\right)  >0$ by assumption. Therefore, $\mathbb{P}(  Y_{1}\in B,D_{z_{K}%
}=1,D_{z_{K-1}}=0,\ldots,D_{z_{1}}=0)  $ being arbitrary implies that
$\mathbb{P}(  Y_{1}\in B|D_{z_{K}}>D_{z_{K-1}})  $ is arbitrary. It
follows that $E[Y_{1}|D_{z_{K}}>D_{z_{K-1}}]$ can take any value in
$\mathbb{R}$.

To complete the proof, note that $\beta_{K,K-1}=\beta_{K,K-1}^{1}%
-\beta_{K,K-1}^{0}$, where, under the above construction,
\[
\beta_{K,K-1}^{1}=E[Y_{1}|D_{z_{K}}>D_{z_{K-1}}]\mathbb{P}(Z\in\{z_{1}%
,\ldots,z_{K-1}\})+E[Y_{1z_{K}}|D_{z_{K}}>D_{z_{K-1}}]\mathbb{P}(Z=z_{K}).
\]
Since there are no cross restrictions between $\beta_{K,K-1}^{1}$ and
$\beta_{K,K-1}^{0}$ as well as between $E[Y_{1}|D_{z_{K}}>D_{z_{K-1}}]$ and $E[Y_{1z_{K}}|D_{z_{K}}>D_{z_{K-1}}]$,
and $\mathbb{P}(Z\in\{z_{1},\ldots,z_{K-1}\})>0$ by
assumption, there is always a conditional distribution $\mathbb{P}(  Y_{1}\in B|D_{z_{K}}>D_{z_{K-1}})  $ such that $\beta_{K,K-1}=\beta_{K,K-1}^{1}
-\beta_{K,K-1}^{0}=\beta$. Since $\beta\in \mathbb{R}$ is arbitrary, the result follows.

\smallskip
\noindent \textbf{Proof of (ii).} First, to highlight the main arguments, we consider the simple case where
$\mathcal{Z}=\{z_{1},z_{2},z_{3}\}$. To prove the result, we find a distribution of
\[
(Y_{0z_{1}},Y_{0z_{2}},Y_{0z_{3}},Y_{1z_{1}},Y_{1z_{2}},Y_{1z_{3}},D_{z_{1}%
},D_{z_{2}},D_{z_{3}},Z)
\]
that is consistent with the observed distribution of $(Y,D,Z)$ and violates
Definition \ref{def.partial validity pairwise binary D}\ref{def.pairwise random assignment} for the pair $(z_{2},z_{3})$ so that $\beta_{3,2}$ can take
any value in $\mathbb{R}$. The proof for any other pair is symmetric. 

Fix an arbitrary value $\beta\in\mathbb{R}$. We find a distribution that satisfies
(a) $Y_{1z_{1}}=Y_{1z_{2}}=Y_{1z_{3}}=Y_{1}$ and $Y_{0z_{1}}=Y_{0z_{2}%
}=Y_{0z_{3}}=Y_{0}$ a.s., (b) $\left(  Y_{1},Y_{0},D_{z_{3}},D_{z_{2}}\right)$ is not independent of $Z$, 
and (c) $D_{z_{3}}\geq D_{z_{2}}\geq D_{z_{1}}$ a.s., so that $\left(  z_{2},z_{3}\right)  \in\mathscr{Z}^1\cap (\mathscr{Z}^2)^c\cap \mathscr{Z}^3$. Thus, we find a
distribution of $\left(  Y_{1},Y_{0},D_{z_{3}},D_{z_{2}},D_{z_{1}}\right)  $
given $Z$ so
that $\beta_{3,2}=\beta$ that is consistent with the distribution of $(Y,D)$ given $Z$.

Consider a distribution of $\left( Y_{1},Y_{0},D_{z_3},D_{z_2}
,D_{z_1},Z\right)$ that satisfies 
\begin{align*}
&  \mathbb{P}\left( Y_{1}\in B_{1},Y_{0}\in B_{0}|D_{z_3}=d_3,D_{z_2}
=d_2,D_{z_1}=d_1,Z=z\right)\notag  \\
=&\,\mathbb{P}\left(  Y_{1}\in B_{1}|D_{z_3}=d_3,D_{z_2}
=d_2,D_{z_1}=d_1,Z=z\right) \\&\times\mathbb{P}\left(  Y_{0}\in B_{0}|D_{z_3}=d_3,D_{z_2}
=d_2,D_{z_1}=d_1,Z=z\right)
\end{align*}
for all $B_1,B_0\in\mathcal{B}_{\mathbb{R}}$, all $d_3,d_2,d_1\in\{0,1\}$, and all $z\in \mathcal{Z}$. Choose
\begin{align*}
\mathbb{P}\left(  V\in B,D_{z_3}=0,D_{z_2}=0,D_{z_1}=1,Z=z\right)&=0,\\
\mathbb{P}\left(  V\in B,D_{z_3}=0,D_{z_2}=1,D_{z_1}=0,Z=z\right)&=0,\\
\mathbb{P}\left(  V\in B,D_{z_3}=0,D_{z_2}=1,D_{z_1}=1,Z=z\right)  &=0,\\
\mathbb{P}\left(  V\in B,D_{z_3}=1,D_{z_2}=0,D_{z_1}=1,Z=z\right)&=0
\end{align*}
for each $V\in \{Y_{1},Y_0\}$, all $B\in\mathcal{B}_{\mathbb{R}}$, and all $z\in \mathcal{Z}$,  which implies (c). 
Moreover, choose
\begin{eqnarray*}
&&\mathbb{P}\left(  Y_{0}\in B,D_{z_3}=0,D_{z_2}=0,D_{z_1}=0|Z=z_3\right)=P_{z_3}\left(  B,\left\{  0\right\}  \right),\\
&&\mathbb{P}\left(  Y_{0}\in B,D_{z_3}=1,D_{z_2}=0,D_{z_1}=0|Z=z_2\right) +\mathbb{P}\left(  Y_{0}\in B,D_{z_3}=0,D_{z_2}=0,D_{z_1}=0|Z=z_2\right)\\
&&=P_{z_2}\left(  B,\left\{  0\right\}  \right),\\
&&\mathbb{P}\left(  Y_{0}\in B,D_{z_3}=1,D_{z_2}=1,D_{z_1}=0|Z=z_1\right)  +\mathbb{P}\left(  Y_{0}\in B,D_{z_3}=1,D_{z_2}=0,D_{z_1}=0|Z=z_1\right)  \\
&&  +\mathbb{P}\left(  Y_{0}\in B,D_{z_3}=0,D_{z_2}=0,D_{z_1}=0|Z=z_1\right) = P_{z_1}\left(  B,\left\{  0\right\}  \right) ,
\end{eqnarray*}
and 
\begin{eqnarray*}
&&\mathbb{P}\left(  Y_{1}\in B,D_{z_3}=1,D_{z_2}=1,D_{z_1}=1|Z=z_3\right)  +\mathbb{P}\left(  Y_{1}\in B,D_{z_3}=1,D_{z_2}=1,D_{z_1}=0|Z=z_3\right)  \\
&&  +\mathbb{P}\left(  Y_{1}\in B,D_{z_3}=1,D_{z_2}=0,D_{z_1}=0|Z=z_3\right) =P_{z_3}\left(  B,\left\{  1\right\}  \right) ,\\
&&\mathbb{P}\left(  Y_{1}\in B,D_{z_3}=1,D_{z_2}=1,D_{z_1}=1|Z=z_2\right)   +\mathbb{P}\left(  Y_{1}\in B,D_{z_3}=1,D_{z_2}=1,D_{z_1}=0|Z=z_2\right) \\
&&=P_{z_2}\left(  B,\left\{  1\right\}  \right), \\
&&\mathbb{P}\left(  Y_{1}\in B,D_{z_3}=1,D_{z_2}=1,D_{z_1}=1|Z=z_1\right) =P_{z_1}\left(  B,\left\{  1\right\}  \right).
\end{eqnarray*}
Finally, we choose $\mathbb{P}(D_{z_3}>D_{z_2}|Z=z_2)>0$. 

Note that this distribution is consistent with the distribution of $(Y,D)$ given $Z$ for any choice of 
$\mathbb{P}\left(  Y_{1}\in B,D_{z_3}=1,D_{z_2}=0,D_{z_1}=0|Z=z\right)$ for $z\in \{z_1,z_2\}$. Then we can choose an arbitrary distribution
\begin{align*}
\mathbb{P}\left(  Y_{1}\in B,D_{z_3}=1,D_{z_2}=0,D_{z_1}=0|Z=z_2\right)=\mathbb{P}\left(  Y_{1}\in B,D_{z_3}=1,D_{z_2}=0|Z=z_2\right),
\end{align*}
and choose
\begin{align*}
&\mathbb{P}\left(  Y_{1}\in B,D_{z_3}=1,D_{z_2}=0,D_{z_1}=0|Z=z_1\right)=\mathbb{P}\left(  Y_{1}\in B,D_{z_3}=1,D_{z_2}=0|Z=z_1\right)\\
&=a\cdot\mathbb{P}\left(  Y_{1}\in B,D_{z_3}=1,D_{z_2}=0|Z=z_2\right)
\end{align*}
for some constant $a\in(0,1)$ and all $B$, so that (b) holds.

Note that
\begin{align*}
\mathbb{P}\left(  Y_{1}\in B|D_{z_3}>D_{z_2},Z=z_2\right)&=\frac{\mathbb{P}\left(  Y_{1}\in B,D_{z_3}>D_{z_2}|Z=z_2\right)}{\mathbb{P}(D_{z_3}>D_{z_2}|Z=z_2)}
\\
&=\frac{\mathbb{P}\left(  Y_{1}\in B,D_{z_{3}}=1,D_{z_2}=0,D_{z_1}=0|Z=z_2\right)}{\mathbb{P}(D_{z_3}>D_{z_2}|Z=z_2)}.
\end{align*}
Therefore, $\mathbb{P}(  Y_{1}\in B,D_{z_{3}}=1,D_{z_2}=0,D_{z_1}=0|Z=z_2)$ being arbitrary implies that the conditional distribution $\mathbb{P}(  Y_{1}\in B|D_{z_3}>D_{z_2},Z=z_2)$ is arbitrary. Thus, $E [Y_{1}|D_{z_3}>D_{z_2},Z=z_2]$ may take any value in $\mathbb{R}$. 

To complete the proof, note that $\beta_{3,2}=\beta_{3,2}^1-\beta_{3,2}^0$, where under the above construction,
$$
\beta_{3,2}^1=\sum_{z\in\mathcal{Z}}E[Y_{1}|D_{z_3}>D_{z_2},Z=z]\mathbb{P}(Z=z|D_{z_3}>D_{z_2}).
$$
Since there are no cross-restrictions between $\beta_{3,2}^1$ and $\beta_{3,2}^0$ as well as between $E[Y_{1}|D_{z_{3}}>D_{z_{2}},Z=z_2]$ and $E[Y_{1}|D_{z_{3}}>D_{z_{2}},Z=z_3]$, and 
$$\mathbb{P}(Z=z_2|D_{z_3}>D_{z_2})=\frac{\mathbb{P}(D_{z_3}>D_{z_2}|Z=z_2)\mathbb{P}(Z=z_2)}{\mathbb{P}(D_{z_3}>D_{z_2})}>0$$ 
under our construction since $\mathbb{P}(Z=z_2)>0$, for any value of $E[Y_{1}|D_{z_{3}}>D_{z_{2}},Z=z_3]$
and $\beta_{3,2}^{0}$, there is always a conditional distribution $\mathbb{P}(  Y_{1}\in B|D_{z_{3}}>D_{z_{2}},Z=z_2)  $ such that $\beta_{3,2}=\beta_{3,2}^{1}-\beta_{3,2}^{0}=\beta$. Since $\beta\in \mathbb{R}$ is arbitrary, the result follows.

Consider now the general case where $\mathcal{Z}=\{z_{1}
,\ldots,z_{K}\}$. To prove the result, we find a distribution of
\[
(Y_{0z_{1}},\ldots,Y_{0z_{K}},Y_{1z_{1}},\ldots,Y_{1z_{K}},D_{z_{1}}%
,\ldots,D_{z_{K}},Z)
\]
that is consistent with the observed distribution of $(Y,D,Z)$ and violates
Definition \ref{def.partial validity pairwise binary D}\ref{def.pairwise random assignment} for the pair $(z_{K-1},z_{K})$ so that $\beta_{K,K-1}$ can
take any value in $\mathbb{R}$.  The proof for any other pair is symmetric.

Fix an arbitrary value $\beta\in\mathbb{R}$. We find a distribution that satisfies
(a) $Y_{1z_{1}}=\cdots=Y_{1z_{K}}=Y_{1}$ and $Y_{0z_{1}}=\cdots
=Y_{0z_{K}}=Y_{0}$ a.s., (b) $(  Y_{1},Y_{0},D_{z_{K}},D_{z_{K-1}})  $
is not independent of $Z$, and (c) $D_{z_{K}}\geq\cdots\geq D_{z_{1}}$ a.s., so that
$\left(  z_{K-1},z_{K}\right)  \in\mathscr{Z}^1\cap (\mathscr{Z}^2)^c\cap \mathscr{Z}^3$. Thus, we find a distribution of $\left(  Y_{1}
,Y_{0},D_{z_{K}},\ldots,D_{z_{1}}\right)  $ given $Z$ so that
$\beta_{K,K-1}=\beta$ that is consistent with
the observable conditional distribution of $(Y,D)$ given $Z$, which is fully characterized by $P_{z}\left(  B,\{d\}\right)  $ for all
$(B,d,z)\in\mathcal{B}_{\mathbb{R}}\times\{0,1\}\times\mathcal{Z}$.

Consider a distribution of $\left(  Y_{1},Y_{0},D_{z_{K}},\ldots,D_{z_{1}%
},Z\right)  $ that satisfies 
\begin{align*}
& \mathbb{P}\left(  Y_{1}\in B_{1},Y_{0}\in B_{0}|D_{z_{K}}=d_{K}%
,\ldots,D_{z_{1}}=d_{1},Z=z\right)  \\
=&\,\mathbb{P}\left(  Y_{1}\in B_{1}|D_{z_{K}}=d_{K},\ldots,D_{z_{1}}%
=d_{1},Z=z\right)  \times\mathbb{P}\left(  Y_{0}\in B_{0}|D_{z_{K}}%
=d_{K},\ldots,D_{z_{1}}=d_{1},Z=z\right)
\end{align*}
for all $B_{1},B_{0}\in\mathcal{B}_{\mathbb{R}}$, all
$d_{1},\ldots,d_{K}\in\{0,1\}$, and all $z\in\mathcal{Z}$. Choose
\[
\mathbb{P}\left(  V\in B,D_{z_{K}}=d_{K},\ldots,D_{z_{1}}=d_{1},Z=z\right)
=0
\]
for each $V\in\left\{  Y_{1},Y_{0}\right\}  $, all $B\in\mathcal{B}%
_{\mathbb{R}}$, all $d_{k}>d_{k^{\prime}}$ with $k<k^{\prime}$, and all
$z\in\mathcal{Z}$, which implies (c). Moreover, choose
\begin{align*}
\mathbb{P}\left(  Y_{1}\in B,D_{z_{K}}=1,D_{z_{K-1}}=1,\ldots,D_{z_{1}%
}=1|Z=z_{1}\right)  =P_{z_{1}}\left(  B,\left\{  1\right\}  \right)  ,
\end{align*}
\[
\vdots
\]%
\begin{align*}
&  \mathbb{P}\left(  Y_{1}\in B,D_{z_{K}}=1,D_{z_{K-1}}=1,\ldots,D_{z_{1}
}=1|Z=z_{K}\right) \\ &+\mathbb{P}\left(  Y_{1}\in B,D_{z_{K}}=1,D_{z_{K-1}%
}=1,\ldots,D_{z_{1}}=0|Z=z_{K}\right)  \\
&  +\cdots+\mathbb{P}\left(  Y_{1}\in B,D_{z_{K}}=1,D_{z_{K-1}}=0,\ldots
,D_{z_{1}}=0|Z=z_{K}\right) \\
&=P_{z_{K}}\left(  B,\left\{  1\right\}  \right)
,
\end{align*}
and
\begin{align*}
&  \mathbb{P}\left(  Y_{0}\in B,D_{z_{K}}=1,D_{z_{K-1}}=1,\ldots,D_{z_{1}%
}=0|Z=z_{1}\right)  +\cdots\\
&  +\mathbb{P}\left(  Y_{0}\in B,D_{z_{K}}=0,D_{z_{K-1}}=0,\ldots,D_{z_{1}%
}=0|Z=z_{1}\right)  =P_{z_{1}}\left(  B,\left\{  0\right\}  \right)  ,
\end{align*}
\[
\vdots
\]
\begin{align*}
\mathbb{P}\left(  Y_{0}\in B,D_{z_{K}}=0,D_{z_{K-1}}=0,\ldots,D_{z_{1}%
}=0|Z=z_{K}\right)  =P_{z_{K}}\left(  B,\left\{  0\right\}  \right)  .
\end{align*}
Finally, we choose $\mathbb{P}(  D_{z_{K}}>D_{z_{K-1}}|Z=z_{K-1})
>0$. 

The above construction is consistent with the distribution of $(Y,D)$ given
$Z$ for any choice of $\mathbb{P}(  Y_{1}\in B,D_{z_{K}}=1,D_{z_{K-1}%
}=0,\ldots,D_{z_{1}}=0|Z=z)  $ for $z\in\left\{  z_{1},\ldots
,z_{K-1}\right\}  $. Then we can choose an arbitrary distribution 
\begin{align*}
&\mathbb{P}\left(  Y_{1}\in B,D_{z_{K}}=1,D_{z_{K-1}}=0,\ldots,D_{z_{1}
}=0|Z=z_{K-1}\right)  \\
=&\mathbb{P}\left(  Y_{1}\in B,D_{z_{K}}=1,D_{z_{K-1}
}=0|Z=z_{K-1}\right),
\end{align*}
and choose
\begin{align*}
&\mathbb{P}\left(  Y_{1}\in B,D_{z_{K}}=1,D_{z_{K-1}}=0,\ldots,D_{z_{1}%
}=0|Z=z\right) =\mathbb{P}\left(  Y_{1}\in B,D_{z_{K}}=1,D_{z_{K-1}}=0|Z=z\right) \\
& =a\cdot\mathbb{P}\left(  Y_{1}\in B,D_{z_{K}}=1,D_{z_{K-1}}=0|Z=z_{K-1}\right)
\end{align*}
for some constant $a\in(0,1)$, all $B$, and all $z\in\{z_1,\ldots,z_{K-2}\}$, so that (b) holds. 

Note that
\begin{align*}
\mathbb{P}\left(  Y_{1}\in B|D_{z_{K}}>D_{z_{K-1}},Z=z_{K-1}\right)    &
=\frac{\mathbb{P}\left(  Y_{1}\in B,D_{z_{K}}>D_{z_{K-1}}|Z=z_{K-1}\right)
}{\mathbb{P}\left(  D_{z_{K}}>D_{z_{K-1}}|Z=z_{K-1}\right)  }\\
& =\frac{\mathbb{P}\left(  Y_{1}\in B,D_{z_{K}}=1,D_{z_{K-1}}=0,\ldots
,D_{z_{1}}=0|Z=z_{K-1}\right)  }{\mathbb{P}\left(  D_{z_{K}}>D_{z_{K-1}%
}|Z=z_{K-1}\right)  }.
\end{align*}
Therefore, $\mathbb{P}(  Y_{1}\in B,D_{z_{K}}=1,D_{z_{K-1}}=0,\ldots,D_{z_{1}}=0|Z=z_{K-1})  $ being arbitrary implies that the conditional distribution
$\mathbb{P}(  Y_{1}\in B|D_{z_{K}}>D_{z_{K-1}},Z=z_{K-1})  $ is
arbitrary. It follows that $E[Y_{1}|D_{z_{K}}>D_{z_{K-1}},Z=z_{K-1}]$ may
take any value in $\mathbb{R}$.

To complete the proof, note that $\beta_{K,K-1}=\beta_{K,K-1}^{1}
-\beta_{K,K-1}^{0}$, where under the above construction,
\[
\beta_{K,K-1}^{1}=\sum_{z\in\mathcal{Z}}E[Y_{1}|D_{z_{K}}>D_{z_{K-1}%
},Z=z]\mathbb{P}(Z=z|D_{z_{K}}>D_{z_{K-1}}).
\]
Since there are no cross restrictions between $\beta_{K,K-1}^{1}$ and
$\beta_{K,K-1}^{0}$ as well as between $E[Y_{1}|D_{z_{K}}>D_{z_{K-1}
},Z=z_{K-1}]$ and $E[Y_{1}|D_{z_{K}}>D_{z_{K-1}},Z=z_{K}]$, and $\mathbb{P}(Z=z_{K-1}|D_{z_{K}}>D_{z_{K-1}})>0$ by Bayes' Theorem because $\mathbb{P}(Z=z)>0$, there is always a conditional distribution $\mathbb{P}(  Y_{1}\in B|D_{z_{K}}>D_{z_{K-1}},Z=z_{K-1})  $ such that $\beta_{K,K-1}=\beta_{K,K-1}^{1}
-\beta_{K,K-1}^{0}=\beta$. Since $\beta\in \mathbb{R}$ is arbitrary, the result follows.
\end{proof}

\begin{proof}[Proof of Corollary \ref{cor:no_information}] The result follows directly from Proposition \ref{prop:no_information}.
\end{proof}

\section{Extensions: Multivalued Ordered and Unordered Treatments}\label{sec.multivalued treatment}

In this section, we generalize the results in the main text to multivalued ordered and unordered treatments.

\subsection{Ordered Treatments}\label{sec.ordered treatment}

Suppose, in general, that the observable treatment variable $D\in\mathcal{D}=\left\{
d_{1},\ldots,d_{J}\right\}$. Without loss of generality, suppose $d_1<\cdots<d_J$. The following assumption is a straightforward generalization of Assumption \ref{ass.IV validity binary D} to ordered treatments \citep[e.g.,][]{sun2021ivvalidity}.

\begin{assumption}\label{ass.IV validity ordered}
IV validity for LATEs with ordered treatments and multivalued instruments:  \label{ass.IV validity for
multivalued Z} 

\begin{enumerate}[label=(\roman*)]

\item Exclusion: For all $d\in\mathcal{D}$, $Y_{dz_{1}}=Y_{dz_{2}} 
=\cdots=Y_{dz_{K}}$ a.s.

\item Random Assignment: $Z$ is jointly independent of $\left( Y_{d_{1}z_{1}},\ldots,Y_{d_{1}z_K},\ldots, Y_{d_{J}z_{1}
},\ldots,Y_{d_{J}z_K}\right)$ and $\left(
D_{z_{1}},\ldots,D_{z_{K}}\right) $.

\item Monotonicity: For all  $
k=1,\ldots, K-1$, $D_{z_{k+1}}\geq D_{z_k}$ a.s. 
\end{enumerate}

\end{assumption}

We next introduce the definition of pairwise valid instruments for ordered treatments.
\begin{definition}
	\label{def.partial validity pairwise} The instrument $Z$ is \textbf{pairwise valid} for
	the ordered treatment $D\in\mathcal{D}=\left\{
	d_{1},\ldots,d_{J}\right\} $ if there is a set $\mathscr{Z}_M=\{(z_{k_1},z_{k_1^{\prime}}),\ldots,(z_{k_M},z_{k_M^{\prime}})\}$ with $z_{k_1},z_{k_1^{\prime}},\ldots,z_{k_M},z_{k_M^{\prime}}\in\mathcal{Z}$ such that the following conditions hold for every  $(z,z')\in\mathscr{Z}_M$:
	\begin{enumerate}[label=(\roman*)]
		
		\item Exclusion: For all $d\in\mathcal{D}$, $Y_{dz}=Y_{dz^{\prime}}$ a.s.
		
		\item Random Assignment: $Z$ is jointly independent of $\left( Y_{d_{1}z},Y_{d_{1}z'},\ldots, Y_{d_{J}z},Y_{d_{J}z'},D_z,D_{z'}\right)$.
		
		\item Monotonicity: $D_{z^{\prime}}\geq D_{z}$ a.s. 
	\end{enumerate}
The set $\mathscr{Z}_M$ is called a \textbf{validity pair set} of $Z$. The union of all validity pair sets is the largest validity pair set, denoted by ${\mathscr{Z}}_{\bar{M}}$. 
\end{definition}

	With the exclusion condition, for every $(z,z')\in\mathscr{Z}_{\bar{M}}$, define $Y_d(z,z')$ such that $Y_d(z,z')=Y_{dz}=Y_{dz'}$ a.s.\ for all $d\in\mathcal{D}$.
\begin{lemma}\label{lemma.partial beta}
	Suppose that the instrument $Z$ is pairwise valid as defined in Definition \ref{def.partial validity pairwise} with a known validity pair set $\mathscr{Z}_M=\{(z_{k_1},z_{k_1^{\prime}}),\ldots,(z_{k_M},z_{k_M^{\prime}})\}$. Then for every $m\in\{1,\ldots,M\}$, the following quantity can be identified:
	\begin{align}\label{eq.beta}
		\beta _{k_{m}',k_{m}}&\equiv \sum_{j=2}^{J}\omega
		_{j}\cdot E\left[ \left( Y_{d_j}(z_{k_m},z_{k'_m})-Y_{d_{j-1}}(z_{k_m},z_{k'_m})\right) |D_{z_{k_{m}'}}\geq
		d_j>D_{z_{k_{m}}}\right]\notag\\
		&=\frac{E\left[ Y|Z=z_{k_{m}'}\right] -E\left[ Y|Z=z_{k_{m}}\right] }{E\left[
			D|Z=z_{k_{m}'}\right] -E\left[ D|Z=z_{k_{m}}\right] },
	\end{align}
	where 
	\begin{equation*}
		\omega _{j}=\frac{\mathbb{P}\left( D_{z_{k_{m}'}}\geq
			d_j>D_{z_{k_{m}}}\right) }{\sum_{l=2}^{J}\left(d_l-d_{l-1}\right)\mathbb{P}\left( D_{z_{k_{m}'}}\geq
			d_l>D_{z_{k_{m}}}\right) }.
	\end{equation*}
\end{lemma}

Lemma \ref{lemma.partial beta} is an extension of Theorem 1 of \citet{imbens1994identification} and Theorem 1 of \citet{angrist1995two} to the case where $Z$ is pairwise valid. We follow \citet{angrist1995two} and refer to $\beta_{k_{m}',k_m}$ as the average causal response (ACR). Lemma \ref{lemma.partial beta} shows that if a validity pair set $\mathscr{Z}_M$ is known, we can identify every $\beta_{k_{m}',k_{m}}$.
In practice, however, $\mathscr{Z}_M$ is usually unknown. We show how to estimate the largest validity pair set $\mathscr{Z}_{\bar{M}}$ and use this estimator to estimate the ACRs. 

The estimation of $\mathscr{Z}_{\bar{M}}$ is similar to that in Section \ref{sec.pairwise valid instrument binary D}. Suppose that there are subsets $\mathscr{Z}_1\subseteq\mathscr{Z}$ and $\mathscr{Z}_2\subseteq\mathscr{Z}$ that satisfy the testable implications in \citet{kitagawa2015test}, \citet{mourifie2016testing}, and \citet{sun2021ivvalidity}, and those in \citet{kedagni2020generalized}, respectively. We let $\mathscr{Z}_0=\mathscr{Z}_1\cap\mathscr{Z}_2$ so that $\mathscr{Z}_0$ satisfies all the above necessary conditions. We first construct the estimators  $\widehat{\mathscr{Z}_1}$ and $\widehat{\mathscr{Z}_2}$ for  $\mathscr{Z}_1$ and  $\mathscr{Z}_2$, respectively, and then construct the estimator  $\widehat{\mathscr{Z}_0}$ for  $\mathscr{Z}_0$ as $\widehat{\mathscr{Z}_0}=\widehat{\mathscr{Z}_1}\cap\widehat{\mathscr{Z}_2}$. See Appendix \ref{sec.estimation Z_0 ordered} for details.

\begin{assumption}\label{ass.iid data}
	$\{(Y_i,D_i,Z_i)\}_{i=1}^{n}$ is an i.i.d.\ sample from a population such that all relevant moments exist. 
\end{assumption}

\begin{assumption}\label{ass.first stage}
	For every $\mathcal{Z}_{(k,k')}\in\mathscr{Z}_{\bar{M}}$, 
	\begin{align}
	    E[g(Z_i)D_i|Z_i\in\mathcal{Z}_{(k,k^{\prime})}]-E[D_i|Z_i\in\mathcal{Z}_{(k,k^{\prime})}]\cdot E[g(Z_i)|Z_i\in\mathcal{Z}_{(k,k^{\prime})}]\neq0.
	\end{align}
\end{assumption}

As in Section \ref{sec.pairwise valid instrument binary D}, we first suppose that $\mathscr{Z}_{\bar{M}}$ can be estimated consistently by the estimator $\widehat{\mathscr{Z}_{0}}$.
We use the same notation as in Section \ref{sec.pairwise valid instrument binary D}. 
For
$\mathcal{Z}_{(  k,k^{\prime})  }\in{\mathscr{Z}}$, we 
run the regression 
\begin{align}\label{eq.VSIV estimation pairwise}
	Y_i1\left\{  Z_{i}\in\mathcal{Z}_{(k,k^{\prime})}\right\} =&\,\gamma_{(k,k^{\prime})}^01\left\{  Z_{i}\in\mathcal{Z}_{(k,k^{\prime})}\right\}+\gamma_{(k,k^{\prime})}^1D_i1\left\{  Z_{i}\in\mathcal{Z}_{(k,k^{\prime})}\right\}+\epsilon_{i}1\left\{  Z_{i}\in\mathcal{Z}_{(k,k^{\prime})}\right\},
\end{align}
using $g(Z_i)1\{  Z_{i}\in\mathcal{Z}_{(k,k^{\prime})}\}$ as the instrument for $D_i1\{  Z_{i}\in\mathcal{Z}_{(k,k^{\prime})}\}$. 
Given the estimated validity set $\widehat{\mathscr{Z}_0}$, we define the VSIV estimator for each $\mathcal{Z}_{(k,k^{\prime})}$ as
\begin{align}\label{eq.VSIV estimator pairwise}
	\widehat{\beta}_{(  k,k^{\prime})  }^1=1\left\{\mathcal{Z}_{(k,k^{\prime})}\in\widehat{\mathscr{Z}_0}\right\}\cdot\frac{\mathcal{E}_{n}\left(
		g\left(  Z_{i}\right)  Y_{i},\mathcal{Z}_{(k,k^{\prime})}\right)
		-\mathcal{E}_{n}\left(  g\left(  Z_{i}\right)  ,\mathcal{Z}_{(k,k^{\prime}%
			)}\right)  \mathcal{E}_{n}\left(  Y_{i},\mathcal{Z}_{(k,k^{\prime})}\right)
	}{\mathcal{E}_{n}\left(  g\left(  Z_{i}\right)  D_{i},\mathcal{Z}%
		_{(k,k^{\prime})}\right)  -\mathcal{E}_{n}\left(  g\left(  Z_{i}\right)
		,\mathcal{Z}_{(k,k^{\prime})}\right)  \mathcal{E}_{n}\left(  D_{i}%
		,\mathcal{Z}_{(k,k^{\prime})}\right)  },
\end{align}
which is the IV estimator for $\gamma_{(k,k^{\prime})}^1$ in \eqref{eq.VSIV estimation pairwise} multiplied by $1\{\mathcal{Z}_{(k,k')}\in\widehat{\mathscr{Z}_0}\}$.
As in Section \ref{sec.pairwise valid instrument binary D}, we define
\[
\widehat{\beta}_{1}=\left(  \widehat{\beta}_{\left(  1,2\right)  }^{1},\ldots
,\widehat{\beta}_{\left(  1,K\right)  }^{1},\ldots,\widehat{\beta}_{\left(
	K,1\right)  }^{1},\ldots,\widehat{\beta}_{\left(  K,K-1\right)  }^{1}\right)^T  ,
\]
\begin{align}\label{eq.VSIV true beta}
\beta_{(  k,k^{\prime})  }^{1}=1\left\{\mathcal{Z}_{(k,k^{\prime})}\in{\mathscr{Z}_{\bar{M}}}\right\}\cdot\frac{\mathcal{E}\left(  g\left(
	Z_{i}\right)  Y_{i},\mathcal{Z}_{(k,k^{\prime})}\right)  -\mathcal{E}\left(
	g\left(  Z_{i}\right)  ,\mathcal{Z}_{(k,k^{\prime})}\right)  \mathcal{E}%
	\left(  Y_{i},\mathcal{Z}_{(k,k^{\prime})}\right)  }{\mathcal{E}\left(
	g\left(  Z_{i}\right)  D_{i},\mathcal{Z}_{(k,k^{\prime})}\right)
	-\mathcal{E}\left(  g\left(  Z_{i}\right)  ,\mathcal{Z}_{(k,k^{\prime}%
		)}\right)  \mathcal{E}\left(  D_{i},\mathcal{Z}_{(k,k^{\prime})}\right)  },
\end{align}
and
\[
\beta_{1}=\left(  \beta_{\left(  1,2\right)  }^{1},\ldots,\beta_{\left(
	1,K\right)  }^{1},\ldots,\beta_{\left(  K,1\right)  }^{1},\ldots
,\beta_{\left(  K,K-1\right)  }^{1}\right)^T  .
\]

\begin{theorem}\label{thm.IV estimator asymptotics pairwise}
	Suppose that the instrument $Z$ is pairwise valid for the treatment $D$ as defined in Definition \ref{def.partial validity pairwise} with the largest validity pair set $\mathscr{Z}_{\bar{M}}=\{(z_{k_1},z_{k_1^{\prime}}),\ldots,(z_{k_{\bar{M}}},z_{k_{\bar{M}}^{\prime}})\}$ and that the estimator $\widehat{\mathscr{Z}_0}$ satisfies $\mathbb{P}(\widehat{\mathscr{Z}_0}=\mathscr{Z}_{\bar{M}})\to 1$. Under Assumptions \ref{ass.iid data} and \ref{ass.first stage}, $\sqrt{n}( \widehat{\beta}_{1}-\beta_1 ) \overset{d}\to N\left( 0,\Sigma \right) $, where 
	$\Sigma$ is defined in \eqref{eq.weak convergence pairwise2}. In addition, $\beta_{(  k,k^{\prime})  }^{1}=\beta_{k^{\prime},k}$ as defined in \eqref{eq.beta} for every $(z_k,z_{k^{\prime}})\in\mathscr{Z}_{\bar{M}}$.
\end{theorem}

     Next, we generalize the results in Section \ref{sec.bias reduction binary D} and show that VSIV estimation always reduces the asymptotic bias when the treatments are ordered. Given a presumed validity pair set $\mathscr{Z}_P$, we apply VSIV estimation based on $\widehat{\mathscr{Z}'_0}=\widehat{\mathscr{Z}_0}\cap\mathscr{Z}_P$.

\begin{assumption}\label{ass.first stage bias}
	For every $\mathcal{Z}_{(k,k')}\in\mathscr{Z}_{0}$, 
	\begin{align}
	    E[g(Z_i)D_i|Z_i\in\mathcal{Z}_{(k,k^{\prime})}]-E[D_i|Z_i\in\mathcal{Z}_{(k,k^{\prime})}]\cdot E[g(Z_i)|Z_i\in\mathcal{Z}_{(k,k^{\prime})}]\neq0.
	\end{align}
\end{assumption}
    
    \begin{theorem}\label{thm.bias reduction multi ordered}
    		Suppose that Assumptions \ref{ass.iid data} and \ref{ass.first stage bias} hold and that $\mathbb{P}(\widehat{\mathscr{Z}_0}=\mathscr{Z}_0)\to 1$ with $\mathscr{Z}_0\supseteq\mathscr{Z}_{\bar{M}}$. For every presumed validity pair set $\mathscr{Z}_P$, the asymptotic bias $\mathrm{plim}_{n\rightarrow\infty}\Vert \widehat{\beta}_1-\beta_1\Vert_2$ is always reduced by using $\widehat{\mathscr{Z}'_0}$  in the estimation \eqref{eq.VSIV estimator pairwise} compared to the asymptotic bias from using $\mathscr{Z}_P$.
    \end{theorem}

As shown in Propositions \ref{prop.consistent G hat pairwise Z1} and \ref{prop.consistent G hat pairwise Z2}, the pseudo-validity pair set $\mathscr{Z}_0$ can always be estimated consistently by $\widehat{\mathscr{Z}_0}$ under mild conditions. Theorem \ref{thm.bias reduction multi ordered} shows that VSIV estimation based on $\widehat{\mathscr{Z}_0}\cap\mathscr{Z}_P$ always reduces the asymptotic bias.

    \begin{remark}
    	In Section \ref{sec.pairwise valid instrument binary D}, we provide the definition of partial IV validity for the binary treatment case. See Appendix \ref{sec.ordered treatment partial} for the extension to multivalued ordered treatments.
    \end{remark}

\subsection{Unordered Treatments}\label{sec.unordered treatment}
\subsubsection{Setup}
Here, we extend our results to unordered treatments using the framework of \citet{heckman2018unordered}. The
treatment (choice) $D$ is discrete with support
$\mathcal{D}=\left\{  d_{1},\ldots,d_{J}\right\}  $, which is unordered. 
 \citet[p.~15]{heckman2018unordered} (Assumption A-3) consider the following monotonicity assumption. 
\begin{assumption}\label{ass.monotonicity unordered}
	For all $d\in \mathcal{D}$ and all $z,z^{\prime}\in \mathcal{Z}$, $1\left\{
	D_{  z^{\prime}}  =d\right\}  \geq1\left\{  D_{  z}
	=d\right\}  $ for all $\omega\in\Omega$, or $1\left\{  D_{z'}
	=d\right\}  \leq1\left\{  D_z  =d\right\}  $ for all $\omega\in\Omega$.\footnote{More precisely, the potential treatments should be written as functions of $\omega$, $D_z(\omega)$ and $D_{z^{\prime}}(\omega)$. For simplicity of notation, we omit $\omega$ whenever there is no confusion. The inequalities can be modified to hold a.s.}
\end{assumption}

Based on Assumption \ref{ass.monotonicity unordered}, we introduce the definition of the pairwise IV validity for the unordered treatment case.\footnote{\citet{fusejima2022identification} combines a similar assumption with rank similarity \citep{chernozhukov2005iv} to identify effects with multivalued treatments.} 
\begin{definition}
	\label{def.partial validity unordered pairwise} The instrument $Z$ is \textbf{pairwise valid} for
	the unordered treatment $D$ if there is a set $\mathscr{Z}_M=\{(z_{k_1},z_{k_1^{\prime}}),\ldots,(z_{k_M},z_{k_M^{\prime}})\}$ with $z_{k_1},z_{k_1^{\prime}},\ldots,z_{k_M},z_{k_M^{\prime}}\in\mathcal{Z}$  and $k_m< k'_m$ for every $m$ such that the following conditions hold for every $(z,z')\in\mathscr{Z}_M$: 
	\begin{enumerate}[label=(\roman*)]
		
		\item Exclusion: For all $d\in\mathcal{D}$, $Y_{dz}=Y_{dz^{\prime}}$ a.s.
		
		\item Random Assignment: $Z$ is jointly independent of $\left( Y_{d_{1}z},Y_{d_{1}z'},\ldots, Y_{d_{J}z},Y_{d_{J}z'},D_z,D_{z'}\right)$.
		
		\item Monotonicity: For all $d\in \mathcal{D}$, $1\left\{
		D_{  z^{\prime}}  =d\right\}  \geq1\left\{  D_{  z}
		=d\right\}  $ for all $\omega\in\Omega$, or $1\left\{  D_{z'}
		=d\right\}  \leq1\left\{  D_{z}  =d\right\}  $ for all $\omega\in\Omega$.
	\end{enumerate}
	The set $\mathscr{Z}_M$ is called a \textbf{validity pair set} of $Z$. The union of all validity pair sets is the largest validity pair set, denoted by ${\mathscr{Z}}_{\bar{M}}$. 
\end{definition}

Suppose the instrument $Z$ is pairwise valid for the treatment $D$ with the largest validity pair set $\mathscr{Z}_{\bar
	{M}}=\{(z_{k_{1}},z_{k_{1}^{\prime}}),\ldots,(z_{k_{\bar{M}}},z_{k_{\bar{M}%
	}^{\prime}})\}$. Define $Y_{d}(z,z^{\prime})$
for every $d\in\mathcal{D}$ and every $(z,z^{\prime})\in\mathscr{Z}_{\bar{M}}$
such that $Y_{d}(z,z^{\prime})=Y_{dz}=Y_{dz^{\prime}}$ a.s.
Following \citet{heckman2018unordered}, we introduce the following notation. Define the response vector $S$ as a $K$-dimensional random vector of potential treatments with $Z$ fixed at each value of its support: 
\[
S=\left(  D_{z_{1}}  ,\ldots,D_{  z_{K}}  \right)  ^{T}.
\]
The finite support of $S$ is $\mathcal{S=}\left\{
\xi_{1},\ldots,\xi_{N_{S}}\right\}$, where
$N_{S}$ is the number of possible values of $S$. 
The response matrix $R$ is an array of response-types defined over
$\mathcal{S}$, $R=\left(  \xi_{1},\ldots,\xi_{N_{S}}\right)  $. 

For every
$\mathcal{Z}_{(  k,k^{\prime})  }\in{\mathscr{Z}}$, there is a $2\times K$ binary matrix
$\mathcal{M}_{(  k,k^{\prime})}$ such that $$\mathcal{M}_{(  k,k^{\prime})}\left(  z_{1},\ldots,z_{K}\right)
^{T}=\left(  z_{k},z_{k^{\prime}}\right)  ^{T}.$$ For example, if $K=5$ and $(  k,k^{\prime})  =\left(  3,5\right)
$, then
\[
\mathcal{M}_{\left(  3,5\right)}=\left(
\begin{array}
	[c]{ccccc}%
	0 & 0 & 1 & 0 & 0\\
	0 & 0 & 0 & 0 & 1
\end{array}
\right)  .
\]
We define a transformation $\mathcal{K}_{(  k,k^{\prime})}$ such that if $A$ is a
$K\times L$ matrix, $\mathcal{K}_{(  k,k^{\prime})}A$ is the matrix that consists of all the unique columns of $\mathcal{M}_{(  k,k^{\prime})}A$ in the
same order as in $\mathcal{M}_{(  k,k^{\prime})}A$. In the above
example, if $A=((  x_{1},\ldots,x_{5})  ^{T},(x_{1},\ldots,x_{5})  ^{T},(y_{1},\ldots,y_{5})  ^{T})$, then $\mathcal{K}_{\left(  3,5\right)}A=((  x_{3},x_{5})  ^{T},(  y_{3},y_{5})  ^{T})$.
We write $\mathcal{K}_{(  k,k^{\prime})}R=(s_1,\ldots,s_{L_{(  k,k^{\prime})}})$, where $L_{(k,k')}$ is the column number of $\mathcal{K}_{(  k,k^{\prime})}R$. Let $B_{d{(  k,k^{\prime})}}$ denote a binary matrix of the same dimension as
$\mathcal{K}_{(  k,k^{\prime})} R$, whose elements are equal to $1$ if the corresponding element in $\mathcal{K}_{(  k,k^{\prime})} R$ is
equal to $d$, and equal to $0$ otherwise. We denote the element in the
$m$th row and $l$th column of the matrix $B_{d(k,k')}$ by $B_{d{(  k,k^{\prime})}}\left(  m,l\right)  $.
Finally, we use $B_{d{(  k,k^{\prime})}}=1\{ \mathcal{K}_{(  k,k^{\prime})} R=d\}  $ to denote
$B_{d{(  k,k^{\prime})}}$. 

\begin{lemma}\label{lemma.equivalent characterizations unordered pairwise}
	Suppose that the instrument $Z$ is pairwise valid for the treatment $D$ with the largest validity pair set $\mathscr{Z}_{\bar{M}}=\{(z_{k_1},z_{k_1'}),\dots,(z_{k_{\bar{M}}},z_{k_{\bar{M}}'})\}$.
	The following statements are equivalent:
	\begin{enumerate}[label=(\roman*)]
		\item For every $(z_k,z_{k'})\in\mathscr{Z}_{\bar{M}}$, the binary matrix ${B}_{d(k,k')}=1\{  \mathcal{K}_{(k,k')}R=d\}  $ is lonesum
		for every $d\in\mathcal{D}$.\footnote{As defined in \citet[p.~20]{heckman2018unordered}, a binary matrix
			is \emph{lonesum} if it is uniquely determined by its row and column sums.}
		
		\item For every $(z_k,z_{k'})\in\mathscr{Z}_{\bar{M}}$ and all $d,d^{\prime},d^{\prime\prime}\in\mathcal{D}$, there are no $2\times2$ sub-matrices of $\mathcal{K}_{(k,k')}R$ of the type
		\[
		\left(
		\begin{array}
			[c]{cc}%
			d & d^{\prime}\\
			d^{\prime\prime} & d
		\end{array}
		\right)  \text{ or }\left(
		\begin{array}
			[c]{cc}%
			d^{\prime} & d\\
			d & d^{\prime\prime}%
		\end{array}
		\right)
		\]
		with $d^{\prime}\neq d$ and $d^{\prime\prime}\neq d$.
		
		\item For every $(z_k,z_{k'})\in\mathscr{Z}_{\bar{M}}$ and every $d\in\mathcal{D}$, the
		following inequalities hold:
		\begin{align*}
			1\left\{  D_{z_{k'}}
			=d\right\}  \ge 1\left\{  D_{z_k}  =d\right\} \text{ for all $\omega\in\Omega$},\text{ or }
			1\left\{  D_{z_{k'}}
			=d\right\}  \leq1\left\{  D_{z_k}  =d\right\} \text{ for all $\omega\in\Omega$}.
		\end{align*}

	\end{enumerate}

\end{lemma}

Lemma \ref{lemma.equivalent characterizations unordered pairwise} is an extension of Theorem T-3 of \citet{heckman2018unordered} for pairwise valid instruments. It provides equivalent conditions for the monotonicity condition (iii) in Definition \ref{def.partial validity unordered pairwise}.

To describe our results, following \citet{heckman2018unordered}, we define some additional notation.
Let $B_{d{(  k,k^{\prime})}}^{+}$ denote the Moore--Penrose pseudo-inverse of $B_{d{(  k,k^{\prime})}}$. Let $\kappa: \mathbb{R}\to \mathbb{R}$ be an arbitrary function of interest. Define for all $d\in\mathcal{D}$,
\[
\bar{P}_{Z}\left(  d\right)  =\left(  \mathbb{P}\left(  D=d|Z=z_{{1}}\right)
,\ldots,\mathbb{P}\left(  D=d|Z=z_{{K}}\right)  \right)  ^{T},
\]
\[\bar{Q}_{Z}\left(  d\right)  =\left(  E\left[  \kappa\left(  Y\right)
\cdot1\left\{  D=d\right\}  |Z=z_{{1}}\right]  ,\ldots,E\left[
\kappa\left(  Y\right)  \cdot1\left\{  D=d\right\}  |Z=z_{K}\right]
\right)  ^{T},
\] 
\[
P_{Z(  k,k^{\prime})}\left(  d\right) =\mathcal{M}_{(  k,k^{\prime})} \bar{P}_{Z}\left(  d\right)=\left(  \mathbb{P}\left(  D=d|Z=z_{k}\right)
,\mathbb{P}\left(  D=d|Z=z_{k^{\prime}}\right)  \right)  ^{T},
\]
and
\begin{align*}
	Q_{Z{(  k,k^{\prime})}}\left(  d\right) & =\mathcal{M}_{(  k,k^{\prime})}\bar{Q}_{Z}\left(  d\right)\\
	&=\left(  E\left[  \kappa\left(  Y\right)
	\cdot1\left\{  D=d\right\}  |Z=z_{k}\right],E\left[
	\kappa\left(  Y\right)  \cdot1\left\{  D=d\right\}  |Z=z_{k^{\prime}}\right]
	\right)  ^{T}.
\end{align*}
Moreover, we define
\begin{align*}
	&P_{Z{(  k,k^{\prime})}}=\left(  P_{Z{(  k,k^{\prime})}}\left(  d_{1}\right)  ,\ldots,P_{Z{(  k,k^{\prime})}}\left(  d_{J}\right)
	\right)  ^{T} \text{  and }\\
	& P_{S{(  k,k^{\prime})}}=\left(  \mathbb{P}(  \mathcal{M}_{(  k,k^{\prime})}S=s_{1})  ,\ldots,\mathbb{P}(
	\mathcal{M}_{(  k,k^{\prime})}S=s_{L_{(k,k')}})  \right)  ^{T},
\end{align*}
and for every $(z_{k},z_{k'})\in\mathscr{Z}_{\bar{M}}$, we define
\begin{align*}
&Q_{S{(  k,k^{\prime})}}\left(  d\right) =\\
 &\,\left(  E\left[  \kappa\left(  Y_d(z_k,z_{k'})
\right)  \cdot1\left\{  \mathcal{M}_{(  k,k^{\prime})}S=s_{1}\right\}  \right]  ,\ldots,E\left[
\kappa\left(  Y_d(z_k,z_{k'})  \right)  \cdot1\left\{  \mathcal{M}_{(  k,k^{\prime})}S=s_{L_{(k,k')}}\right\}  \right]  \right)^T 
\end{align*}
for all $d\in\mathcal{D}$.
Define $\Sigma_{d{(  k,k^{\prime})}}\left(  t\right)  $ to be the set of response-types in which
$d$ appears exactly $t$ times, that is, for every $d\in\mathcal{D}$ and every $t\in\left\{  0,1,2\right\}  $, define
\[
\Sigma_{d{(  k,k^{\prime})}}\left(  t\right)  =\left\{  s:s\text{ is some }l\text{th column of
}\mathcal{K}_{(  k,k^{\prime})}R\text{ with}\sum_{m=1}^{2}B_{d{(  k,k^{\prime})}}\left(  m,l\right)
=t\right\}  .
\] 
Let $b_{d{(  k,k^{\prime})}}(t)$ be a
$L_{(k,k')}$-dimensional binary row-vector that indicates if every column of
$\mathcal{K}_{(  k,k^{\prime})}R$ belongs to $\Sigma_{d{(  k,k^{\prime})}}\left(  t\right)  $, that is,
$b_{d{(  k,k^{\prime})}}\left(  t\right)  \left(  l\right)  =1$ if $s_{l}\in\Sigma_{d{(  k,k^{\prime})}}\left(
t\right)  $, and $b_{d{(  k,k^{\prime})}}\left(  t\right)  \left(  l\right)  =0$ otherwise for every $l\in\{1,\ldots,L_{(k,k')}\}$, where $s_{l}$ is the $l$th column of
$\mathcal{K}_{(  k,k^{\prime})}R$.
In this section, we let $$\mathscr{Z}=\{(z_1,z_2),\ldots,(z_1,z_K),\ldots,(z_{K-1},z_K)\}.$$
Finally, define
$
\mathds{1}(\mathscr{A})=(1\{(z_1,z_{2})\in \mathscr{A}\}, \ldots, 1\{(z_{K-1},z_{K})\in \mathscr{A}\})^T
$
for every $\mathscr{A}\subseteq\mathscr{Z}$. 

\subsubsection{VSIV Estimation under Consistent Estimation of Validity Pair Set}

Here, we study the properties of VSIV Estimation when the validity pair set can be estimated consistently, that is, there is an estimator $\widehat{\mathscr{Z}_0}$ such that $\mathbb{P}(\widehat{\mathscr{Z}_0}=\mathscr{Z}_{\bar{M}})\to 1$. 
Suppose that there are subsets $\mathscr{Z}_1\subseteq\mathscr{Z}$ and $\mathscr{Z}_2\subseteq\mathscr{Z}$ that satisfy the testable implications in \citet{sun2021ivvalidity}, and those in \citet{kedagni2020generalized}, respectively. As in Section \ref{sec.ordered treatment}, we let $\mathscr{Z}_0=\mathscr{Z}_1\cap\mathscr{Z}_2$ so that $\mathscr{Z}_0$ satisfies all the above necessary conditions. We first construct the estimators  $\widehat{\mathscr{Z}_1}$ and $\widehat{\mathscr{Z}_2}$ for  $\mathscr{Z}_1$ and  $\mathscr{Z}_2$, respectively, and then construct the estimator  $\widehat{\mathscr{Z}_0}$ for  $\mathscr{Z}_0$ as $\widehat{\mathscr{Z}_0}=\widehat{\mathscr{Z}_1}\cap\widehat{\mathscr{Z}_2}$. See Appendix \ref{sec.estimation Z_0 unordered} for details. Under mild conditions, $\mathbb{P}(\widehat{\mathscr{Z}_0}=\mathscr{Z}_{0})\to 1$. If $\mathscr{Z}_0=\mathscr{Z}_{\bar{M}}$, then it follows that $\mathbb{P}(\widehat{\mathscr{Z}_0}=\mathscr{Z}_{\bar{M}})\to 1$.

To state the results, define
\[
P_{DZ}\left(  d\right)  =\left(  \mathbb{P}\left(  D=d,Z=z_{1}\right)
,\ldots,\mathbb{P}\left(  D=d,Z=z_{K}\right)  \right)  ^{T},
\]%
\[
Q_{YDZ}\left(  d\right)  =\left(  E\left[  \kappa\left(  Y\right)  1\left\{
D=d,Z=z_{1}\right\}  \right]  ,\ldots,E\left[  \kappa\left(  Y\right)
1\left\{  D=d,Z=z_{K}\right\}  \right]  \right)  ^{T},
\]
for every $d\in\mathcal{D}$, and
\[
Z_{P}=\left(  \mathbb{P}\left(  Z=z_{1}\right)  ,\ldots,\mathbb{P}\left(
Z=z_{K}\right)  \right),
\]
\[
W=\left(  Z_{P},P_{DZ}\left(  d_{1}\right)  ^{T},\ldots,P_{DZ}\left(
d_{J}\right)  ^{T},Q_{YDZ}\left(  d_{1}\right)  ^{T},\ldots,Q_{YDZ}\left(
d_{J}\right)  ^{T}\right)  ^{T}.
\]
Suppose we have a random sample $\{(Y_i,D_i,Z_i)\}_{i=1}^{n}$. Define the following sample analogs:
\[
\widehat{\mathbb{P}}\left(  Z=z\right)     =\frac{1}{n}\sum_{i=1}
^{n}1\left\{  Z_{i}=z\right\}  \text{ for all }z,
\]
\[	\widehat{\mathbb{P}}\left(  D=d,Z=z\right)     =\frac{1}{n}\sum_{i=1}
^{n}1\left\{  D_{i}=d,Z_{i}=z\right\}  \text{ for all }d\text{ and all }z,
\]
\[	\widehat{E}\left[  \kappa\left(  Y\right)  1\left\{  D=d,Z=z\right\}  \right]
=\frac{1}{n}\sum_{i=1}^{n}\kappa\left(  Y_{i}\right)  1\left\{
D_{i}=d,Z_{i}=z\right\}  \text{ for all }d\text{ and all }z,
\]
\[	\widehat{P_{DZ}\left(  d\right) }    =\left(  \widehat{\mathbb{P}}\left(
D=d,Z=z_{1}\right)  ,\ldots,\widehat{\mathbb{P}}\left(  D=d,Z=z_{K}\right)
\right)  ^{T} \text{ for all }d,
\]
\[	\widehat{Q_{YDZ}\left(  d\right) }   =\left(  \widehat{E}\left[
\kappa\left(  Y\right)  1\left\{  D=d,Z=z_{1}\right\}  \right] 
,\ldots,\widehat{E}\left[  \kappa\left(  Y\right)  1\left\{  D=d,Z=z_{K}
\right\}  \right]  \right)  ^{T} \text{ for all }d,
\]
\[	\widehat{Z_{P}}   =\left(  \widehat{\mathbb{P}}\left(  Z=z_{1}\right)
,\ldots,\widehat{\mathbb{P}}\left(  Z=z_{K}\right)  \right), 
\]
and
\[	\widehat{W}   =\left(  \widehat{Z_{P}},\widehat{P_{DZ}\left(  d_{1}\right)  }
^{T},\ldots,\widehat{P_{DZ}\left(  d_{J}\right)  }^{T},\widehat{Q_{YDZ}\left(
	d_{1}\right)  }^{T},\ldots,\widehat{Q_{YDZ}\left(  d_{J}\right)  }^{T}\right)
^{T}.
\]

We impose the following weak regularity conditions.
\begin{assumption}\label{ass.iid data unordered}
	$\{(Y_i,D_i,Z_i)\}_{i=1}^{n}$ is an i.i.d.\ sample from a population such that all relevant moments exist. 
\end{assumption}

The next theorem presents the identification and estimation results under pairwise IV validity with unordered treatments.

\begin{theorem}\label{theorem.counterfactuals identified pairwise}
	Suppose that the instrument $Z$ is pairwise valid for the treatment $D$ as defined in Definition \ref{def.partial validity unordered pairwise} with the largest validity pair set $\mathscr{Z}_{\bar{M}}=\{(z_{k_1},z_{k_1^{\prime}}),\ldots,(z_{k_{\bar{M}}},z_{k_{\bar{M}}^{\prime}})\}$.  The following response-type
	probabilities and counterfactuals are identified for every $d\in\mathcal{D}$, each $t\in\{1,2\}$, and every $(z_k,z_{k^{\prime}})\in\mathscr{Z}_{\bar{M}}$:
	\begin{align}\label{eq.counterfactuals pairwise}
		&\mathbb{P}\left(  \mathcal{M}_{(  k,k^{\prime})}S\in\Sigma_{d{(  k,k^{\prime})}}\left(  t\right) \right)   
		=b_{d(  k,k^{\prime})}\left(  t\right)  B_{d(  k,k^{\prime})}^{+}  P_{Z(  k,k^{\prime})}\left(  d\right) \text{ and } 	\notag\\
		&E[  \kappa\left(  Y_d(z_k,z_{k'})  \right)  |\mathcal{M}_{(  k,k^{\prime})}S\in\Sigma_{d{(  k,k^{\prime})}}\left(  t\right)  ]  =\frac{b_{d(  k,k^{\prime})}\left(  t\right)  B_{d(  k,k^{\prime})}^{+}  Q_{Z(  k,k^{\prime})}\left(  d\right) }{b_{d(  k,k^{\prime})}\left(  t\right)  B_{d(  k,k^{\prime})}^{+}  P_{Z(  k,k^{\prime})}\left(  d\right) }.
	\end{align}
	In addition, under Assumption \ref{ass.iid data unordered}, if $\mathbb{P}(\widehat{\mathscr{Z}_0}=\mathscr{Z}_{\bar{M}})\to 1$, we have that
	\begin{align*}
		\sqrt{n}\left\{  \left(  \widehat{W}^T,\mathds{1}(\widehat{\mathscr{Z}_{0}})^T\right)^T  -\left(  W^T, \mathds{1}({\mathscr{Z}_{\bar{M}}})^T \right)^T  \right\} 
		\overset{d}\rightarrow\left(  N\left(  0,\Sigma_{W}\right)^T  ,0^T\right)^T, 
	\end{align*} 
	where $\Sigma_{W}$ is given in \eqref{eq.sigmaW}.
\end{theorem}

	 Theorem \ref{theorem.counterfactuals identified pairwise} is an extension of Theorem T-6 of \citet{heckman2018unordered} for pairwise valid instruments. As shown in Remark 7.1 in \citet{heckman2018unordered} and Theorem \ref{theorem.counterfactuals identified pairwise}, if $(z_k,z_{k'})\in\mathscr{Z}_{\bar{M}}$ and $\Sigma_{d(k,k')}(t)=\Sigma_{d^{\prime}(k,k')}(t^{\prime})$ for some $d,d^{\prime}\in\mathcal{D}$ and some $t,t^{\prime}\in\{1,2\}$, the mean treatment effect of $d$ relative to $d^{\prime}$ for $\Sigma_{d(k,k')}(t)$ can be identified, which is $E[Y_d(z_k,z_{k'})-Y_{d^{\prime}}(z_k,z_{k'})|\mathcal{M}_{(k,k')}S\in\Sigma_{d(k,k')}(t)]$. 

For all $d,d^{\prime}\in\mathcal{D}$, all $t,t^{\prime}\in\{1,2\}$, and all $k< k'$, following \citet{heckman2018unordered}, we define 
\begin{align*}
	\beta_{(k,k')}(d,d',t,t')\equiv& 1\{(z_k,z_{k^{\prime}})\in\mathscr{Z}_{\bar{M}},  \Sigma_{d(k,k')}(t)=\Sigma_{d'(k,k')}(t')  \} \\ &\cdot E[Y_{dz_k}-Y_{d^{\prime}z_{k'}}|\mathcal{M}_{(k,k')}S\in\Sigma_{d(k,k')}(t)].
\end{align*}
When $(z_k,z_{k^{\prime}})\in\mathscr{Z}_{\bar{M}}$ and $\Sigma_{d(k,k')}(t)=\Sigma_{d'(k,k')}(t')$, we have that
\begin{align*}
    \beta_{(k,k')}(d,d',t,t')= E[Y_d(z_k,z_{k'})-Y_{d^{\prime}}(z_k,z_{k'})|\mathcal{M}_{(k,k')}S\in\Sigma_{d(k,k')}(t)],
\end{align*}
which is the mean treatment effect of $d$ relative to $d'$ for $\Sigma_{d(k,k')}(t)$.

\begin{lemma} \label{lemma.beta unordered}
Let $\kappa (y)=y$ for all $y\in\mathbb{R}$.	The mean treatment effect $\beta_{(k,k')}(d,d',t,t')$ can be expressed as
	\begin{align}\label{eq.beta k unordered}
		\beta_{(k,k')}(d,d',t,t')=&\,1\{(z_k,z_{k^{\prime}})\in\mathscr{Z}_{\bar{M}},  \Sigma_{d(k,k')}(t)=\Sigma_{d'(k,k')}(t')  \}
		\notag\\
		&\cdot\left\{\frac{b_{d(  k,k^{\prime})}\left(  t\right)  B_{d(  k,k^{\prime})}^{+}  Q_{Z(  k,k^{\prime})}\left(  d\right) }{b_{d(  k,k^{\prime})}\left(  t\right)  B_{d(  k,k^{\prime})}^{+}  P_{Z(  k,k^{\prime})}\left(  d\right) }
		-\frac{b_{d'(  k,k^{\prime})}\left(  t'\right)  B_{d'(  k,k^{\prime})}^{+}  Q_{Z(  k,k^{\prime})}\left(  d'\right) }{b_{d'(  k,k^{\prime})}\left(  t'\right)  B_{d'(  k,k^{\prime})}^{+}  P_{Z(  k,k^{\prime})}\left(  d'\right) }\right\}.
	\end{align}
\end{lemma}

We now define 
\begin{align}
	\beta_{(k,k')}(d,d')=(\beta_{(k,k')}(d,d',1,1),\beta_{(k,k')}(d,d',1,2),\beta_{(k,k')}(d,d',2,1),\beta_{(k,k')}(d,d',2,2))
\end{align}
for all $d,d^{\prime}\in\mathcal{D}$ and all $k< k'$. For all $k< k'$, we let
\begin{align*}
\beta_{(k,k')}=(\beta_{(k,k')}(d_1,d_{2}),\ldots,\beta_{(k,k')}(d_1,d_{J}),\ldots,\beta_{(k,k')}(d_J,d_{1}),\dots, \beta_{(k,k')}(d_J,d_{J-1})).
\end{align*}
Finally, we define
\begin{align}\label{eq.beta unordered}
	\beta=(\beta_{(1,2)},\ldots, \beta_{(1,K)},\ldots, \beta_{(K-1,K)})^T.
\end{align}
Note that if $(z_k,z_{k'})\notin \mathscr{Z}_{\bar{M}}$, then $\beta_{(k,k')}=0$. For the sample analogs, we define 
	\begin{align}\label{eq.beta k estimate unordered}
	\widehat{\beta}_{(k,k')}(d,d',t,t')=&\,1\{(z_k,z_{k^{\prime}})\in\widehat{\mathscr{Z}_{0}},  \Sigma_{d(k,k')}(t)=\Sigma_{d'(k,k')}(t')  \}\notag\\
	&\cdot\left\{\frac{b_{d(  k,k^{\prime})}\left(  t\right)  B_{d(  k,k^{\prime})}^{+}  \widehat{Q_{Z(  k,k^{\prime})}\left(  d\right)} }{b_{d(  k,k^{\prime})}\left(  t\right)  B_{d(  k,k^{\prime})}^{+}  \widehat{P_{Z(  k,k^{\prime})}\left(  d\right)} }
	-\frac{b_{d'(  k,k^{\prime})}\left(  t'\right)  B_{d'(  k,k^{\prime})}^{+}  \widehat{Q_{Z(  k,k^{\prime})}\left(  d'\right)} }{b_{d'(  k,k^{\prime})}\left(  t'\right)  B_{d'(  k,k^{\prime})}^{+}  \widehat{P_{Z(  k,k^{\prime})}\left(  d'\right)} }\right\},
\end{align}
where $\widehat{P_{Z(  k,k^{\prime})}(d)}$ and $\widehat{Q_{Z(  k,k^{\prime})}(d)}$ can be obtained by transformations of $\widehat{W}$. We let
\begin{align}
	\widehat{\beta}_{(k,k')}&(d,d')=(\widehat{\beta}_{(k,k')}(d,d',1,1),\widehat{\beta}_{(k,k')}(d,d',1,2),\widehat{\beta}_{(k,k')}(d,d',2,1),\widehat{\beta}_{(k,k')}(d,d',2,2))
\end{align}
for all $d,d^{\prime}\in\mathcal{D}$ and all $k< k'$. For all $k< k'$, we define
\begin{align}
	\widehat{\beta}_{(k,k')}=(\widehat{\beta}_{(k,k')}(d_1,d_{2}),\ldots,\widehat{\beta}_{(k,k')}(d_1,d_{K}),\ldots,\widehat{\beta}_{(k,k')}(d_K,d_{1}),\dots, \widehat{\beta}_{(k,k')}(d_K,d_{K-1})).
\end{align}
Finally, define
\begin{align}\label{eq.beta estimate unordered}
	\widehat{\beta}=(\widehat{\beta}_{(1,2)},\ldots, \widehat{\beta}_{(1,K)},\ldots, \widehat{\beta}_{(K-1,K)})^T.
\end{align}
Asymptotic properties of the VSIV estimator in \eqref{eq.beta estimate unordered} can be obtained by Theorem \ref{theorem.counterfactuals identified pairwise} with $\mathbb{P}(\widehat{\mathscr{Z}_0}=\mathscr{Z}_{\bar{M}})\to 1$.

\subsubsection{Asymptotic Bias Reduction for Mean Treatment Effects}

Here, we extend the results in Section \ref{sec.bias reduction binary D} and show that VSIV estimation always reduces the asymptotic bias for estimating mean treatment effects with unordered treatments.

With $\beta$ and $\widehat{\beta}$ defined in \eqref{eq.beta unordered} and \eqref{eq.beta estimate unordered}, the following theorem shows that VSIV estimation always reduces the asymptotic bias.  
\begin{theorem}\label{thm.bias reduction unordered}
	Suppose that Assumption \ref{ass.iid data unordered} holds, $b_{d(  k,k^{\prime})}\left(  t\right)  B_{d(  k,k^{\prime})}^{+}  P_{Z(  k,k^{\prime})}\left(  d\right)\neq0$ for all $d\in\mathcal{D}$, each $t\in\{1,2\}$, and all $(z_{k},z_{k'})\in\mathscr{Z}_0$, and  $\mathbb{P}(\widehat{\mathscr{Z}_0}=\mathscr{Z}_0)\to 1$ with $\mathscr{Z}_0\supseteq\mathscr{Z}_{\bar{M}}$. For every presumed validity pair set $\mathscr{Z}_P$, the asymptotic bias $\mathrm{plim}_{n\rightarrow\infty}\Vert \widehat{\beta}-\beta\Vert_2$ is always reduced by using $\widehat{\mathscr{Z}'_0}=\widehat{\mathscr{Z}_0}\cap\mathscr{Z}_P$  in the estimation for \eqref{eq.beta unordered} compared to that from using $\mathscr{Z}_P$.
\end{theorem}

As shown in Propositions \ref{prop.consistent G hat pairwise Z2} and \ref{prop.consistent G hat unordered pairwise}, the pseudo-validity pair set $\mathscr{Z}_0$ can always be estimated consistently by $\widehat{\mathscr{Z}_0}$ under mild conditions. Theorem \ref{thm.bias reduction unordered} shows that VSIV estimation based on $\widehat{\mathscr{Z}_0}\cap\mathscr{Z}_P$ reduces the asymptotic bias relative to standard IV methods based on $\mathscr{Z}_P$.

\section{Proofs and Supplementary Results for Appendix  \ref{sec.ordered treatment}}
\label{sec.proofs ordered treatment}

The results in Section \ref{sec.pairwise valid instrument binary D} are for the special case where $D$ is binary and follow from the general results for ordered treatments in Appendix \ref{sec.ordered treatment}. The proofs of these general results are in Appendix \ref{sec.general proofs ordered treatment}.

\subsection{Proofs for Appendix \ref{sec.ordered treatment}} \label{sec.general proofs ordered treatment}

\begin{proof}[Proof of Lemma \ref{lemma.partial beta}]
The proof closely follows the strategy of that of Theorem 1 in \citet{angrist1995two}.
Let $d_{0}<d_1$ and $Y_{d_{0}}(  z_{k_{m}},z_{k_{m}^{\prime}})  =0$ for every $m$.
Let $d_{J+1}$ be some number such that $d_{J+1}>d_{J}$. We can write
\[
Y=\sum_{k=1}^{K}1\left\{  Z=z_{k}\right\}  \cdot\left\{  \sum_{j=1}%
^{J}1\left\{  D=d_{j}\right\}  Y_{d_{j}z_{k}}\right\}  .
\]
Now we have that%
\begin{align*}
	&  E\left[  Y|Z=z_{k_{m}^{\prime}}\right]  -E\left[  Y|Z=z_{k_{m}}\right]  \\
	= &  \,E\left[  \sum_{j=1}^{J}Y_{d_{j}}\left(  z_{k_{m}},z_{k_{m}^{\prime}%
	}\right)  \left(
	\begin{array}
		[c]{c}%
		1\left\{  D_{z_{k'_m}}\geq d_{j}\right\}  -1\left\{  D_{z_{k'_m}}\geq
		d_{j+1}\right\}  \\
		-1\left\{  D_{z_{k_{m}}}\geq d_{j}\right\}  +1\left\{  D_{z_{k_{m}}}\geq
		d_{j+1}\right\}
	\end{array}
	\right)  \right]  \\
	= &  \,\sum_{j=1}^{J}E\left[  \left(  Y_{d_{j}}\left(  z_{k_{m}}%
	,z_{k_{m}^{\prime}}\right)  -Y_{d_{j-1}}\left(  z_{k_{m}},z_{k_{m}^{\prime}%
	}\right)  \right)  \left(  1\left\{  D_{z_{k'_m}}\geq d_{j}\right\}
	-1\left\{  D_{z_{k_{m}}}\geq d_{j}\right\}  \right)  \right]  .
\end{align*}
By Definition \ref{def.partial validity pairwise}, $(1\{D_{z_{k'_m}}\geq
d_{j}\}-1\{  D_{z_{k_{m}}}\geq d_{j}\}  )\in\left\{  0,1\right\}
$. Then we have that
\begin{align*}
	&  \sum_{j=1}^{J}E\left[  \left(  Y_{d_{j}}\left(  z_{k_{m}},z_{k_{m}^{\prime
	}}\right)  -Y_{d_{j-1}}\left(  z_{k_{m}},z_{k_{m}^{\prime}}\right)  \right)
	\left(  1\left\{  D_{z_{k'_m}}\geq d_{j}\right\}  -1\left\{  D_{z_{k_{m}}%
	}\geq d_{j}\right\}  \right)  \right]  \\
	= &  \sum_{j=1}^{J}\bigg\{E\left[  \left(  Y_{d_{j}}\left(  z_{k_{m}}%
	,z_{k_{m}^{\prime}}\right)  -Y_{d_{j-1}}\left(  z_{k_{m}},z_{k_{m}^{\prime}%
	}\right)  \right)  \big|1\left\{  D_{z_{k'_m}}\geq d_{j}\right\}
	-1\left\{  D_{z_{k_{m}}}\geq d_{j}\right\}  =1\right]  \\
	& \qquad \cdot\mathbb{P}\left(  1\left\{  D_{z_{k'_m}}\geq d_{j}\right\}
	-1\left\{  D_{z_{k_{m}}}\geq d_{j}\right\}  =1\right)  \bigg\}\\
	= &  \sum_{j=1}^{J}E\left[  \left(  Y_{d_{j}}\left(  z_{k_{m}},z_{k_{m}%
		^{\prime}}\right)  -Y_{d_{j-1}}\left(  z_{k_{m}},z_{k_{m}^{\prime}}\right)
	\right)  |D_{z_{k'_m}}\geq d_{j}>D_{z_{k_{m}}}\right]   \cdot\mathbb{P}%
	\left(  D_{z_{k'_m}}\geq d_{j}>D_{z_{k_{m}}}\right)  .
\end{align*}
Similarly, we have
\begin{align*}
	&  E\left[  D|Z=z_{k'_m}\right]  -E\left[  D|Z=z_{k_{m}}\right]  \\
	= &  \,E\left[  \sum_{j=1}^{J}d_{j}\left(  1\left\{  D_{z_{k'_m}}\geq
	d_{j}\right\}  -1\left\{  D_{z_{k_{m}}}\geq d_{j}\right\}  \right)  \right]
	\\
	&  -E\left[  \sum_{j=1}^{J}d_{j}\left(  1\left\{  D_{z_{k'_m}}\geq
	d_{j+1}\right\}  -1\left\{  D_{z_{k_{m}}}\geq d_{j+1}\right\}  \right)
	\right]  \\
	= &  \,E\left[  \sum_{j=1}^{J}d_{j}\cdot1\left\{  D_{z_{k'_m}}\geq
	d_{j}>D_{z_{k_{m}}}\right\}  \right]  -E\left[  \sum_{j=1}^{J}d_{j-1}%
	\cdot1\left\{  D_{z_{k'_m}}\geq d_{j}>D_{z_{k_{m}}}\right\}  \right]  \\
	= &  \sum_{j=1}^{J}\left(  d_{j}-d_{j-1}\right)  \mathbb{P}\left(
	D_{z_{k'_m}}\geq d_{j}>D_{z_{k_{m}}}\right)  .
\end{align*}
Thus, finally we have that
\begin{align*}
	\beta_{k'_m,k_{m}}\equiv & \, \sum_{j=1}^{J}\omega_{j}\cdot E\left[  \left(  Y_{d_{j}}\left(  z_{k_{m}%
	},z_{k_{m}^{\prime}}\right)  -Y_{d_{j-1}}\left(  z_{k_{m}},z_{k_{m}^{\prime}%
	}\right)  \right)  |D_{z_{k'_m}}\geq d_{j}>D_{z_{k_{m}}}\right]\\
	= & \,\frac{E\left[  Y|Z=z_{k'_m}\right]
		-E\left[  Y|Z=z_{k_{m}}\right]  }{E\left[  D|Z=z_{k'_m}\right]  -E\left[
		D|Z=z_{k_{m}}\right]  }  ,
\end{align*}
where
\[
\omega_{j}=\frac{\mathbb{P}\left(  D_{z_{k'_m}}\geq d_{j}>D_{z_{k_{m}}%
	}\right)  }{\sum_{l=1}^{J}\left(  d_{l}-d_{l-1}\right)  \mathbb{P}\left(
	D_{z_{k'_m}}\geq d_{l}>D_{z_{k_{m}}}\right)  }.
\]
Note that by definition, $\mathbb{P}(D_{z_{k'_m}}\geq d_{1}>D_{z_{k_{m}}%
})=0$.
\end{proof}

\begin{proof}[Proof of Theorem \ref{thm.IV estimator asymptotics pairwise}]
	For every $\mathcal{Z}_{(  k,k^{\prime})  }\in\mathscr{Z}$, we define
	\[
	W_{i}\left(  \mathcal{Z}_{(  k,k^{\prime})  }\right)  =\left(
	\begin{array}
		[c]{c}%
		g\left(  Z_{i}\right)  Y_{i}1\left\{  Z_{i}\in\mathcal{Z}_{\left(
			k,k^{\prime}\right)  }\right\}  \\
		Y_{i}1\left\{  Z_{i}\in\mathcal{Z}_{\left(
			k,k^{\prime}\right)  }\right\}  \\
		g\left(  Z_{i}\right)  1\left\{  Z_{i}\in\mathcal{Z}_{\left(
			k,k^{\prime}\right)  }\right\}  \\
		g\left(  Z_{i}\right)  D_{i}1\left\{  Z_{i}\in\mathcal{Z}_{\left(
			k,k^{\prime}\right)  }\right\}  \\
		D_{i}1\left\{  Z_{i}\in\mathcal{Z}_{\left(
			k,k^{\prime}\right)  }\right\}  \\
		1\left\{  Z_{i}\in\mathcal{Z}_{\left(
			k,k^{\prime}\right)  }\right\}
	\end{array}
	\right)  ,
	\]%
	\[
	\widehat{W}_{n}\left(  \mathcal{Z}_{(  k,k^{\prime})  }\right)
	=\,\frac{1}{n}\sum_{i=1}^{n}W_{i}\left(  \mathcal{Z}_{(  k,k^{\prime})  }\right)
	, \text{ and } W\left(  \mathcal{Z}_{(  k,k^{\prime})
	}\right)  =E\left[  W_{i}\left(  \mathcal{Z}_{(  k,k^{\prime})
	}\right)  \right]  .
	\]
	Also, we let
	\begin{align*}
		&\widehat{W}_{n}   =\left(  \widehat{W}_{n}\left(  \mathcal{Z}_{\left(
			1,2\right)  }\right)^{T}  ,\ldots,\widehat{W}_{n}\left(  \mathcal{Z}_{\left(
			1,K\right)  }\right)^{T}  ,\ldots,\widehat{W}_{n}\left(  \mathcal{Z}_{\left(
			K,1\right)  }\right)^{T}  ,\ldots,\widehat{W}_{n}\left(  \mathcal{Z}_{\left(
			K,K-1\right)  }\right)^{T}  \right)  ^{T}\\
		&\text{and }W  =\left(  W\left(  \mathcal{Z}_{\left(  1,2\right)
		}\right)^{T}  ,\ldots,W\left(  \mathcal{Z}_{\left(  1,K\right)  }\right)^{T}
		,\ldots,W\left(  \mathcal{Z}_{\left(  K,1\right)  }\right)^{T}  ,\ldots
		,W\left(  \mathcal{Z}_{\left(  K,K-1\right)  }\right)^{T}  \right)  ^{T}.
	\end{align*}
	By the multivariate central limit theorem,
	\begin{align}
		\sqrt{n}\left(  \widehat{W}_{n}-W\right)   &  =\sqrt{n}\left(
		\begin{array}
			[c]{c}%
			\widehat{W}_{n}\left(  \mathcal{Z}_{\left(  1,2\right)  }\right)  -W\left(
			\mathcal{Z}_{\left(  1,2\right)  }\right)  \\
			\vdots\\
			\widehat{W}_{n}\left(  \mathcal{Z}_{\left(  K,K-1\right)  }\right)  -W\left(
			\mathcal{Z}_{\left(  K,K-1\right)  }\right)
		\end{array}
		\right)  \nonumber\label{eq.asymptotic beta1 pairwise}\\
		&  =\sqrt{n}\frac{1}{n}\sum_{i=1}^{n}\left(
		\begin{array}
			[c]{c}%
			W_{i}\left(  \mathcal{Z}_{\left(  1,2\right)  }\right)  -W\left(
			\mathcal{Z}_{\left(  1,2\right)  }\right)  \\
			\vdots\\
			W_{i}\left(  \mathcal{Z}_{\left(  K,K-1\right)  }\right)  -W\left(
			\mathcal{Z}_{\left(  K,K-1\right)  }\right)
		\end{array}
		\right)   \overset{d}\rightarrow N\left(  0,\Sigma_{P}\right)
		,
	\end{align}
	where $\Sigma_{P}=E\left[  V_{P}V_{P}^{T}\right]  $ and
	\[
	V_{P}=\left(
	\begin{array}
		[c]{c}%
		W_{i}\left(  \mathcal{Z}_{\left(  1,2\right)  }\right)  -W\left(
		\mathcal{Z}_{\left(  1,2\right)  }\right)  \\
		\vdots\\
		W_{i}\left(  \mathcal{Z}_{\left(  K,K-1\right)  }\right)  -W\left(
		\mathcal{Z}_{\left(  K,K-1\right)  }\right)
	\end{array}
	\right)  .
	\]
	Define a function $f:\mathbb{R}^{6}\rightarrow\bar{\mathbb{R}}$ by
	\[
	f\left(  x\right)  =\frac{x_{1}/x_{6}-x_{2}x_{3}/x_{6}^{2}}{x_{4}/x_{6}%
		-x_{5}x_{3}/x_{6}^{2}}
	\]
	for every $x\in\mathbb{R}^{6}$ with $x=\left(  x_{1},x_{2},x_{3},x_{4}
	,x_{5},x_6\right)  ^{T}$ such that $f(x)$ is well defined.
	We can obtain the gradient
	of $f$, denoted $f^{\prime}$, by $f^{\prime}\left(  x\right)  =\left(
	f_{1}^{\prime}\left(  x\right)  ,f_{2}^{\prime}\left(  x\right)
	,f_{3}^{\prime}\left(  x\right)  ,f_{4}^{\prime}\left(  x\right)
	,f_{5}^{\prime}\left(  x\right)  ,f_{6}^{\prime}\left(  x\right)  \right)
	^{T}$ with
	\begin{align*}
		f_{1}^{\prime}\left(  x\right)   &  =\frac{x_{6}}{x_{4}x_{6}-x_{5}
			x_{3}},f_{2}^{\prime}\left(  x\right)  =\frac{-x_{3}
		}{x_{4}x_{6}-x_{5}x_{3}},f_{3}^{\prime}\left(  x\right)
		=\frac{-x_{2}x_{4}x_{6}+x_{5}x_{1}x_{6}}{\left(  x_{4}x_{6}-x_{5}x_{3}\right)
			^{2}},\\
		f_{4}^{\prime}\left(  x\right)   &  =-\frac{\left(  x_{1}x_{6}-x_{2}%
			x_{3}\right)  x_{6}}{\left(  x_{4}x_{6}-x_{5}x_{3}\right)  ^{2}},f_{5}%
		^{\prime}\left(  x\right)  =\frac{x_{3}\left(  x_{1}x_{6}-x_{2}x_{3}\right)
		}{\left(  x_{4}x_{6}-x_{5}x_{3}\right)  ^{2}},\text{ and }f_{6}^{\prime
		}\left(  x\right)  =\frac{-x_{1}x_{5}x_{3}+x_{2}x_{3}x_{4}}{\left(  x_{4}%
			x_{6}-x_{5}x_{3}\right)  ^{2}}
	\end{align*}
	for every  $x=\left(  x_{1},x_{2},x_{3},x_{4},x_{5},x_6\right)  ^{T}$ such that all the above derivatives are well defined.

For every $\mathcal{Z}_{(k,k^{\prime})}$, by assumption we have that for every $\rho\geq0$,
\begin{align}\label{eq.set consistency}
\mathbb{P}\left(  n^{\rho}\left\vert 1\left\{  \mathcal{Z}_{(k,k^{\prime}%
)}\in\widehat{\mathscr{Z}_{0}}\right\}  -1\left\{  \mathcal{Z}_{(k,k^{\prime
})}\in\mathscr{Z}_{\bar{M}}\right\}  \right\vert >\varepsilon\right)    \leq\mathbb{P}\left(  \widehat{\mathscr{Z}_{0}}\neq\mathscr{Z}_{\bar{M}%
}\right)  \rightarrow0.
\end{align}
This implies that if $1\{\mathcal{Z}_{(k,k^{\prime})}\in\mathscr{Z}_{\bar{M}%
}\}=0$, then
\begin{align}\label{eq.I consistency}
n^{\rho}1\{  \mathcal{Z}_{(k,k^{\prime})}\in\widehat{\mathscr{Z}_{0}%
}\}  =o_{p}\left(  1\right).
\end{align}

Without loss of generality, we suppose $\mathscr{Z}_{\bar{M}}=\{
\mathcal{Z}_{(1,2)},\mathcal{Z}_{(1,3)},\ldots,\mathcal{Z}_{(K-1,K)}\}
$ and $\mathscr{Z}\setminus\mathscr{Z}_{\bar{M}}=\{
\mathcal{Z}_{(2,1)},\mathcal{Z}_{(3,1)},\ldots,\mathcal{Z}_{(K,K-1)}\}
$ for simplicity. For every 
$\mathcal{Z}_{\left(  k,k'\right)  }\notin\mathscr{Z}_{\bar{M}}$, by
Assumption \ref{ass.first stage}, it is possible that
\begin{align}\label{eq.first stage eq}
E\left[  g\left(  Z_{i}\right)  D_{i}|Z_{i}\in\mathcal{Z}_{\left(
k,k'\right)  }\right]  -E\left[  g\left(  Z_{i}\right)  |Z_{i}\in
\mathcal{Z}_{\left(  k,k'\right)  }\right]  E\left[  D_{i}|Z_{i}%
\in\mathcal{Z}_{\left(  k,k'\right)  }\right]  =0.
\end{align}
For every $w=(  w_{1}^T,\ldots,w_{\left(  K-1\right)K  }^T)^T  $ with
$w_{j}=(  w_{j1},\ldots,w_{j6})^T  $ for every $j $, define
\begin{align*}
\mathcal{F}_{1}\left(  w\right)  =\left(  f\left(  w_{1}\right)
,\ldots,f\left(  w_{(K-1)K/2  }\right)  \right)  ^{T}\text{ and
}
\mathcal{F}_{0}\left(  w\right)  =\left(  f\left(  w_{K(K-1)/2+1}\right)
,\ldots,f\left(  w_{(K-1)K  }\right)  \right)  ^{T}
.
\end{align*}
For every
$\mathscr{Z}_{s}\subseteq\mathscr{Z}$, define
\begin{align*}
\mathcal{I}_{1}\left(  \mathscr{Z}_{s}\right)   =  \left(
\begin{array}
[c]{cccc}%
1\left\{  \mathcal{Z}_{(1,2)}\in\mathscr{Z}_{s}\right\}   &  &  & \\
& 1\left\{  \mathcal{Z}_{(1,3)}\in\mathscr{Z}_{s}\right\}   &  & \\
&  & \ddots & \\
&  &  & 1\left\{  \mathcal{Z}_{(K-1,K)}\in\mathscr{Z}_{s}\right\}
\end{array}
\right) 
\end{align*}
and 
\begin{align*}
\mathcal{I}_{0}\left(  \mathscr{Z}_{s}\right)   =  \left(
\begin{array}
[c]{cccc}%
1\left\{  \mathcal{Z}_{(2,1)}\in\mathscr{Z}_{s}\right\}   &  &  & \\
& 1\left\{  \mathcal{Z}_{(3,1)}\in\mathscr{Z}_{s}\right\}   &  & \\
&  & \ddots & \\
&  &  & 1\left\{  \mathcal{Z}_{(K,K-1)}\in\mathscr{Z}_{s}\right\}
\end{array}
\right)  .
\end{align*}
Then we can write
\[
\sqrt{n}\left(  \widehat{\beta}_{1}-\beta_{1}\right)  =\sqrt{n}\left\{
\left(\begin{array}
[c]{c}%
\mathcal{I}_{1}\left(  \widehat{\mathscr{Z}_{0}}\right)  \mathcal{F}%
_{1}\left(  \widehat{W}_{n}\right)  \\
\mathcal{I}_{0}\left(  \widehat{\mathscr{Z}_{0}}\right)  \mathcal{F}%
_{0}\left(  \widehat{W}_{n}\right)
\end{array}\right)
-%
\left(\begin{array}
[c]{c}%
\mathcal{I}_{1}\left(  \mathscr{Z}_{\bar{M}}\right)  \mathcal{F}_{1}\left(
W\right)  \\
\mathcal{I}_{0}\left(  \mathscr{Z}_{\bar{M}}\right)  \mathcal{F}_{0}\left(
W\right)
\end{array}\right)
\right\}  .
\]

First, we have that
\begin{align*}
  \sqrt{n}\left\{  \mathcal{I}_{1}\left(  \widehat{\mathscr{Z}_{0}}\right)
\mathcal{F}_{1}\left(  \widehat{W}_{n}\right)  -\mathcal{I}_{1}\left(
\mathscr{Z}_{\bar{M}}\right)  \mathcal{F}_{1}\left(  W\right)  \right\}  
  =&\,\sqrt{n}\left\{  \mathcal{I}_{1}\left(  \widehat{\mathscr{Z}_{0}}\right)
\mathcal{F}_{1}\left(  \widehat{W}_{n}\right)  -\mathcal{I}_{1}\left(
\widehat{\mathscr{Z}_{0}}\right)  \mathcal{F}_{1}\left(  W\right)  \right\}
\\
&  +\sqrt{n}\left\{  \mathcal{I}_{1}\left(  \widehat{\mathscr{Z}_{0}}\right)
\mathcal{F}_{1}\left(  W\right)  -\mathcal{I}_{1}\left(  \mathscr{Z}_{\bar{M}%
}\right)  \mathcal{F}_{1}\left(  W\right)  \right\}  .
\end{align*}
The Jacobian matrix $\mathcal{F}_{1}^{\prime}\left(  W\right)  $ of
$\mathcal{F}_{1}$ at $W$ can be obtained with the derivatives of $f$. Then by \eqref{eq.set consistency} and delta method, it is easy to show that
\begin{align*}
  \sqrt{n}\left\{  \mathcal{I}_{1}\left(  \widehat{\mathscr{Z}_{0}}\right)
\mathcal{F}_{1}\left(  \widehat{W}_{n}\right)  -\mathcal{I}_{1}\left(
\mathscr{Z}_{\bar{M}}\right)  \mathcal{F}_{1}\left(  W\right)  \right\}   &=\mathcal{I}_{1}\left(  \widehat{\mathscr{Z}_{0}}\right)  \sqrt{n}\left\{
\mathcal{F}_{1}\left(  \widehat{W}_{n}\right)  -\mathcal{F}_{1}\left(
W\right)  \right\}  +o_{p}\left(  1\right)  \\
&  \overset{d}{\rightarrow}\mathcal{I}_{1}\left(  \mathscr{Z}_{\bar{M}%
}\right)  \mathcal{F}_{1}^{\prime}\left(  W\right)  N\left(  0,\Sigma
_{P}\right)  .
\end{align*}
Second, by assumption and \eqref{eq.0timesinfinity}, 
\begin{align*}
\sqrt{n}\left\{  \mathcal{I}_{0}\left(  \widehat{\mathscr{Z}_{0}}\right)
\mathcal{F}_{0}\left(  \widehat{W}_{n}\right)  -\mathcal{I}_{0}\left(
\mathscr{Z}_{\bar{M}}\right)  \mathcal{F}_{0}\left(  W\right)  \right\}  
=\sqrt{n}\mathcal{I}_{0}\left(  \widehat{\mathscr{Z}_{0}}\right)  \mathcal{F}_{0}\left(\widehat{W}_{n}  \right)  .
\end{align*}
For every $\mathcal{Z}_{(k,k')}\notin\mathscr{Z}_{\bar{M}}$ such that \eqref{eq.first stage eq} holds, 
\begin{align*}
\sqrt{n}1\{\mathcal{Z}_{(k,k')}\in\widehat{\mathscr{Z}_0}\}f(\widehat{W}_n(\mathcal{Z}_{(k,k')}))= {n}1\{\mathcal{Z}_{(k,k')}\in\widehat{\mathscr{Z}_0}\} \frac{A_{n}}{\sqrt{n}B_{n}},
\end{align*}
where%
\begin{align*}
A_{n} =&\,\frac{1}{n}\sum_{i=1}^{n}g\left(  Z_{i}\right)  Y_{i}1\left\{
Z_{i}\in\mathcal{Z}_{\left(  k,k'\right)  }\right\}  \frac{1}{n}\sum
_{i=1}^{n}1\left\{  Z_{i}\in\mathcal{Z}_{\left(  k,k'\right)  }\right\}  \\
&  -\frac{1}{n}\sum_{i=1}^{n}g\left(  Z_{i}\right)  1\left\{  Z_{i}%
\in\mathcal{Z}_{\left(  k,k'\right)  }\right\}  \frac{1}{n}\sum_{i=1}%
^{n}Y_{i}1\left\{  Z_{i}\in\mathcal{Z}_{\left(  k,k'\right)  }\right\} 
\end{align*}
and%
\begin{align*}
B_{n}  =&\,\frac{1}{n}\sum_{i=1}^{n}g\left(  Z_{i}\right)  D_{i}1\left\{
Z_{i}\in\mathcal{Z}_{\left(  k,k'\right)  }\right\}  \frac{1}{n}\sum
_{i=1}^{n}1\left\{  Z_{i}\in\mathcal{Z}_{\left(  k,k'\right)  }\right\}  \\
&  -\frac{1}{n}\sum_{i=1}^{n}g\left(  Z_{i}\right)  1\left\{  Z_{i}%
\in\mathcal{Z}_{\left(  k,k'\right)  }\right\}  \frac{1}{n}\sum_{i=1}%
^{n}D_{i}1\left\{  Z_{i}\in\mathcal{Z}_{\left(  k,k'\right)  }\right\}  .
\end{align*}
Define a map $h$ such that for every $x\in\mathbb{R}^6$ with $x=\left(  x_{1},\ldots,x_{6}\right)^T  $,
\[
h\left(  x\right)  =x_{4}x_{6}-x_{3}x_{5}.
\]
Let $h^{\prime}(  W(  \mathcal{Z}_{(  k,k')  }))  $ be the Jacobian matrix of $h$ at $W(  \mathcal{Z}_{\left(k,k'\right)  })  $. Then by delta method,
\begin{align*}
\sqrt{n}B_{n} &  =\sqrt{n}\left(  h\left(  \widehat{W}_{n}\left(
\mathcal{Z}_{\left(  k,k'\right)  }\right)  \right)  -h\left(  W\left(
\mathcal{Z}_{\left(  k,k'\right)  }\right)  \right)  \right) \overset{d}{\rightarrow}h^{\prime}\left(  W\left(  \mathcal{Z}_{\left(k,k'\right)  }\right)  \right)  N\left(  0,\Sigma_{\left(  k,k'\right)
}\right)  ,
\end{align*}
where $$\Sigma_{\left(  k,k'\right)  }=E\left[\left\{  W_{i}\left(  \mathcal{Z}%
_{\left(  k,k'\right)  }\right)  -W\left(  \mathcal{Z}_{\left(  k,k'\right)
}\right)  \right\}  \left\{  W_{i}\left(  \mathcal{Z}_{\left(  k,k'\right)
}\right)  -W\left(  \mathcal{Z}_{\left(  k,k'\right)  }\right)  \right\}
^{T}\right].$$ Also, it is easy to show that
\begin{align*}
A_{n}\overset{p}{\rightarrow}&\,E\left[  g\left(  Z_{i}\right)  Y_{i}1\left\{
Z_{i}\in\mathcal{Z}_{\left(  k,k'\right)  }\right\}  \right]  E\left[
1\left\{  Z_{i}\in\mathcal{Z}_{\left(  k,k'\right)  }\right\}  \right]  \\
&  -E\left[  g\left(  Z_{i}\right)  1\left\{  Z_{i}\in\mathcal{Z}_{\left(
k,k'\right)  }\right\}  \right]  E\left[  Y_{i}1\left\{  Z_{i}\in
\mathcal{Z}_{\left(  k,k'\right)  }\right\}  \right]  .
\end{align*}
Note that by \eqref{eq.I consistency}, $n\mathcal{I}_{0}(\widehat{\mathscr{Z}_{0}})=o_{p}\left(
1\right)  $. Thus, $\sqrt{n}1\{\mathcal{Z}_{(k,k')}\in\widehat{\mathscr{Z}_0}\}f(\widehat{W}_n(\mathcal{Z}_{(k,k')}))\overset{p}{\rightarrow} 0$.
Similarly, for every $\mathcal{Z}_{(k,k')}\notin\mathscr{Z}_{\bar{M}}$ such that \eqref{eq.first stage eq} does not hold, it is easy to show that
\begin{align*}
\sqrt{n}1\{\mathcal{Z}_{(k,k')}\in\widehat{\mathscr{Z}_0}\}f(\widehat{W}_n(\mathcal{Z}_{(k,k')}))= \sqrt{n}1\{\mathcal{Z}_{(k,k')}\in\widehat{\mathscr{Z}_0}\} \frac{A_{n}}{B_{n}}\overset{p}{\rightarrow} 0.
\end{align*}
This implies that
\[
\sqrt{n}\left\{  \mathcal{I}_{0}\left(  \widehat{\mathscr{Z}_{0}}\right)
\mathcal{F}_{0}\left(  \widehat{W}_{n}\right)  -\mathcal{I}_{0}\left(
\mathscr{Z}_{\bar{M}}\right)  \mathcal{F}_{0}\left(  W\right)  \right\}
\overset{p}{\rightarrow}0.
\]
By Lemma 1.10.2(iii) and Example 1.4.7 (Slutsky's lemma) of \citet{van1996weak},
\begin{align}\label{eq.weak convergence pairwise2}
\sqrt{n}\left(  \widehat{\beta}_{1}-\beta_{1}\right) & =\sqrt{n}\left\{
\left(\begin{array}
[c]{c}%
\mathcal{I}_{1}\left(  \widehat{\mathscr{Z}_{0}}\right)  \mathcal{F}%
_{1}\left(  \widehat{W}_{n}\right)  \\
\mathcal{I}_{0}\left(  \widehat{\mathscr{Z}_{0}}\right)  \mathcal{F}%
_{0}\left(  \widehat{W}_{n}\right)
\end{array}\right)
-%
\left(\begin{array}
[c]{c}%
\mathcal{I}_{1}\left(  \mathscr{Z}_{\bar{M}}\right)  \mathcal{F}_{1}\left(
W\right)  \\
\mathcal{I}_{0}\left(  \mathscr{Z}_{\bar{M}}\right)  \mathcal{F}_{0}\left(
W\right)
\end{array}\right)
\right\}  \notag\\
&  \overset{d}{\rightarrow}
\left(\begin{array}
[c]{c}%
\mathcal{I}_{1}\left(  \mathscr{Z}_{\bar{M}}\right)  \mathcal{F}_{1}^{\prime
}\left(  W\right)  N\left(  0,\Sigma_{P}\right)  \\
0
\end{array}\right).
\end{align}

Next, we show that $\beta_{(  k,k^{\prime})  }^{1}=\beta_{k^{\prime},k}$ for every $(z_k,z_{k^{\prime}})\in\mathscr{Z}_{\bar{M}}$.
We have that for every $\mathcal{Z}_{(k,k')}\in\mathscr{Z}_{\bar{M}}$,
\begin{align*}
	&  \frac{E\left[  g\left(  Z_{i}\right)  Y_{i}1\left\{  Z_{i}\in
		\mathcal{Z}_{(k,k^{\prime})}\right\}  \right]  }{\mathbb{P}\left(  Z_{i}%
		\in\mathcal{Z}_{(k,k^{\prime})}\right)  }-\frac{E\left[  Y_{i}1\left\{
		Z_{i}\in\mathcal{Z}_{(k,k^{\prime})}\right\}  \right]  }{\mathbb{P}\left(
		Z_{i}\in\mathcal{Z}_{(k,k^{\prime})}\right)  }\frac{E\left[  g\left(
		Z_{i}\right)  1\left\{  Z_{i}\in\mathcal{Z}_{(k,k^{\prime})}\right\}  \right]
	}{\mathbb{P}\left(  Z_{i}\in\mathcal{Z}_{(k,k^{\prime})}\right)  }\\
	= &  \sum_{l=1}^{K}\left\{
	\begin{array}
		[c]{c}%
		\frac{\mathbb{P}\left(  Z_{i}=z_{l}\right)  }{\mathbb{P}\left(  Z_{i}%
			\in\mathcal{Z}_{(k,k^{\prime})}\right)  }E\left[  Y_{i}1\left\{  Z_{i}%
		\in\mathcal{Z}_{(k,k^{\prime})}\right\}  |Z_{i}=z_{l}\right]  \\
		\cdot\left\{  g\left(  z_{l}\right)  1\left\{  z_{l}\in\mathcal{Z}%
		_{(k,k^{\prime})}\right\}  -\frac{E\left[  g\left(  Z_{i}\right)  1\left\{
			Z_{i}\in\mathcal{Z}_{(k,k^{\prime})}\right\}  \right]  }{\mathbb{P}\left(
			Z_{i}\in\mathcal{Z}_{(k,k^{\prime})}\right)  }\right\}
	\end{array}
	\right\}  \\
	= &  \mathbb{P}\left(  Z_{i}=z_{k}|Z_{i}\in\mathcal{Z}_{(k,k^{\prime}%
		)}\right)  E\left[  Y_{i}|Z_{i}=z_{k}\right]  \left\{  g\left(  z_{k}\right)
	-E\left[  g\left(  Z_{i}\right)  |Z_{i}\in\mathcal{Z}_{(k,k^{\prime})}\right]
	\right\}  \\
	&  +\mathbb{P}\left(  Z_{i}=z_{k^{\prime}}|Z_{i}\in\mathcal{Z}_{(k,k^{\prime
		})}\right)  E\left[  Y_{i}|Z_{i}=z_{k^{\prime}}\right]  \left\{  g\left(
	z_{k^{\prime}}\right)  -E\left[  g\left(  Z_{i}\right)  |Z_{i}\in
	\mathcal{Z}_{(k,k^{\prime})}\right]  \right\}  .
\end{align*}
By \eqref{eq.beta}, we have
\[
E\left[  Y_{i}|Z_{i}=z_{k^{\prime}}\right]  =\beta_{k^{\prime},k}\left(
E\left[  D_{i}|Z_{i}=z_{k^{\prime}}\right]  -E\left[  D_{i}|Z_{i}%
=z_{k}\right]  \right)  +E\left[  Y_{i}|Z_{i}=z_{k}\right]  ,
\]
and thus it follows that
\begin{align*}
	&  \mathbb{P}\left(  Z_{i}=z_{k}|Z_{i}\in\mathcal{Z}_{(k,k^{\prime})}\right)
	E\left[  Y_{i}|Z_{i}=z_{k}\right]  \left\{  g\left(  z_{k}\right)  -E\left[
	g\left(  Z_{i}\right)  |Z_{i}\in\mathcal{Z}_{(k,k^{\prime})}\right]  \right\}
	\\
	&  +\mathbb{P}\left(  Z_{i}=z_{k^{\prime}}|Z_{i}\in\mathcal{Z}_{(k,k^{\prime
		})}\right)  E\left[  Y_{i}|Z_{i}=z_{k^{\prime}}\right]  \left\{  g\left(
	z_{k^{\prime}}\right)  -E\left[  g\left(  Z_{i}\right)  |Z_{i}\in
	\mathcal{Z}_{(k,k^{\prime})}\right]  \right\}  \\
	 =&\,\mathbb{P}\left(  Z_{i}=z_{k^{\prime}}|Z_{i}\in\mathcal{Z}_{(k,k^{\prime
		})}\right)  \beta_{k^{\prime},k}\left(  E\left[  D_{i}|Z_{i}=z_{k^{\prime}%
	}\right]  -E\left[  D_{i}|Z_{i}=z_{k}\right]  \right)  \\
	&  \cdot\left\{  g\left(  z_{k^{\prime}}\right)  -E\left[  g\left(
	Z_{i}\right)  |Z_{i}\in\mathcal{Z}_{(k,k^{\prime})}\right]  \right\}  ,
\end{align*}
where we use the equality that
\begin{align}\label{eq.useful equality}
	& \mathbb{P}\left(  Z_{i}=z_{k}|Z_{i}\in\mathcal{Z}_{(k,k^{\prime})}\right)
	\left\{  g\left(  z_{k}\right)  -E\left[  g\left(  Z_{i}\right)  |Z_{i}%
	\in\mathcal{Z}_{(k,k^{\prime})}\right]  \right\}  \notag\\
	&+\mathbb{P}\left(  Z_{i}=z_{k^{\prime}}|Z_{i}\in\mathcal{Z}_{(k,k^{\prime}%
		)}\right)  \left\{  g\left(  z_{k^{\prime}}\right)  -E\left[  g\left(
	Z_{i}\right)  |Z_{i}\in\mathcal{Z}_{(k,k^{\prime})}\right]  \right\}     =0.
\end{align}
Similarly, we have
\begin{align*}
	&  \frac{E\left[  g\left(  Z_{i}\right)  D_{i}1\left\{  Z_{i}\in
		\mathcal{Z}_{(k,k^{\prime})}\right\}  \right]  }{\mathbb{P}\left(  Z_{i}%
		\in\mathcal{Z}_{(k,k^{\prime})}\right)  }-\frac{E\left[  D_{i}1\left\{
		Z_{i}\in\mathcal{Z}_{(k,k^{\prime})}\right\}  \right]  }{\mathbb{P}\left(
		Z_{i}\in\mathcal{Z}_{(k,k^{\prime})}\right)  }\frac{E\left[  g\left(
		Z_{i}\right)  1\left\{  Z_{i}\in\mathcal{Z}_{(k,k^{\prime})}\right\}  \right]
	}{\mathbb{P}\left(  Z_{i}\in\mathcal{Z}_{(k,k^{\prime})}\right)  }\\
	=&\,\mathbb{P}\left(  Z_{i}=z_{k^{\prime}}|Z_{i}\in\mathcal{Z}_{(k,k^{\prime
		})}\right)  \left\{  p(z_{k^{\prime}})-p(z_{k})\right\}  \left\{  g\left(
	z_{k^{\prime}}\right)  -E\left[  g\left(  Z_{i}\right)  |Z_{i}\in
	\mathcal{Z}_{(k,k^{\prime})}\right]  \right\}  ,
\end{align*}
where $p(z)=E\left[  D_{i}|Z_{i}=z\right]  $ for all $z$ and we use the
equality in \eqref{eq.useful equality} again. The result then follows from taking the ratio of the above expressions. 
\end{proof}

\begin{proof}[Proof of Theorem \ref{thm.bias reduction multi ordered}]
Recall that	for every random variable $\xi_{i}$ and every $\mathcal{A}\in\mathscr{Z}$,
	\[
	\mathcal{E}_{n}\left(  \xi_{i},\mathcal{A}\right)  =\frac{\frac{1}%
		{n}\sum_{i=1}^{n}\xi_{i}1\left\{  Z_{i}\in\mathcal{A}\right\}  }{\frac{1}{n}\sum_{i=1}^{n}1\left\{  Z_{i}%
		\in\mathcal{A}\right\}  }\text{ and }%
	\mathcal{E}\left(  \xi_{i},\mathcal{A}\right)  =\frac{E\left[
		\xi_{i}1\left\{  Z_{i}\in\mathcal{A}\right\}
		\right]  }{E\left[  1\left\{  Z_{i}\in\mathcal{A}\right\}  \right]  }.
	\]
	Then we obtain the VSIV estimator using $\mathscr{Z}_{P}$ for each ACR as
	\[
	\widehat{\beta}_{(k,k^{\prime})}^{\prime}=1\{\mathcal{Z}_{(k,k^{\prime})}\in\mathscr{Z}_{P}\}\cdot\frac{\mathcal{E}_{n}\left(  g\left(  Z_{i}\right)  Y_{i},\mathcal{Z}_{(k,k^{\prime})}\right)
		-\mathcal{E}_{n}\left(  g\left(  Z_{i}\right)  ,\mathcal{Z}%
		_{(k,k^{\prime})}\right)  \mathcal{E}_{n}\left(  Y_{i},\mathcal{Z}_{(k,k^{\prime})}\right)  }{\mathcal{E}_{n}\left(  g\left(
		Z_{i}\right)  D_{i},\mathcal{Z}_{(k,k^{\prime})}\right)  -\mathcal{E}_{n}\left(  g\left(  Z_{i}\right)  ,\mathcal{Z}_{(k,k^{\prime}%
			)}\right)  \mathcal{E}_{n}\left(  D_{i},\mathcal{Z}_{(k,k^{\prime}%
			)}\right)  },
	\]
	which converges in probability to%
	\[
	\beta_{(k,k^{\prime})}^{\prime}=1\{\mathcal{Z}_{(k,k^{\prime})}\in\mathscr{Z}_{P}\}\cdot\frac{\mathcal{E}\left(  g\left(
		Z_{i}\right)  Y_{i},\mathcal{Z}_{(k,k^{\prime})}\right)  -\mathcal{E}\left(  g\left(  Z_{i}\right)  ,\mathcal{Z}_{(k,k^{\prime})}\right)
		\mathcal{E}\left(  Y_{i},\mathcal{Z}_{(k,k^{\prime})}\right)
	}{\mathcal{E}\left(  g\left(  Z_{i}\right)  D_{i},\mathcal{Z}%
		_{(k,k^{\prime})}\right)  -\mathcal{E}\left(  g\left(  Z_{i}\right)
		,\mathcal{Z}_{(k,k^{\prime})}\right)  \mathcal{E}\left(
		D_{i},\mathcal{Z}_{(k,k^{\prime})}\right)  }.
	\]
	We obtain the VSIV estimator using $\widehat{\mathscr{Z}_{0}^{\prime}%
	}$ for each ACR as
	\[
	\widehat{\beta}_{(k,k^{\prime})}^{\prime\prime}=1\{\mathcal{Z}_{(k,k^{\prime})}\in\widehat{\mathscr{Z}_{0}
			^{\prime}}\}\cdot\frac{\mathcal{E}_{n}\left(  g\left(  Z_{i}\right)  Y_{i},\mathcal{Z}_{(k,k^{\prime
			})}\right)  -\mathcal{E}_{n}\left(  g\left(  Z_{i}\right)
		,\mathcal{Z}_{(k,k^{\prime})}\right)  \mathcal{E}_{n}\left(
		Y_{i},\mathcal{Z}_{(k,k^{\prime})}\right)  }{\mathcal{E}_{n}\left(  g\left(  Z_{i}\right)  D_{i},\mathcal{Z}_{(k,k^{\prime})}\right)
		-\mathcal{E}_{n}\left(  g\left(  Z_{i}\right)  ,\mathcal{Z}%
		_{(k,k^{\prime})}\right)  \mathcal{E}_{n}\left(  D_{i}%
		,\mathcal{Z}_{(k,k^{\prime})}\right)  },
	\]
	which converges in probability to%
	\[
	\beta_{(k,k^{\prime})}^{\prime\prime}=1\{\mathcal{Z}_{(k,k^{\prime})}\in {\mathscr{Z}_{0}
			^{\prime}}\}\cdot\frac{\mathcal{E}\left(
		g\left(  Z_{i}\right)  Y_{i},\mathcal{Z}_{(k,k^{\prime})}\right)
		-\mathcal{E}\left(  g\left(  Z_{i}\right)  ,\mathcal{Z}
		_{(k,k^{\prime})}\right)  \mathcal{E}\left(  Y_{i},\mathcal{Z}_{(k,k^{\prime})}\right)  }{\mathcal{E}\left(	g\left(  Z_{i}\right)  D_{i},\mathcal{Z}_{(k,k^{\prime})}\right)
		-\mathcal{E}\left(  g\left(  Z_{i}\right)  ,\mathcal{Z}%
		_{(k,k^{\prime})}\right)  \mathcal{E}\left(  D_{i},\mathcal{Z}_{(k,k^{\prime})}\right)  },
	\]
	where $\mathscr{Z}'_0=\mathscr{Z}_0\cap\mathscr{Z}_P$.
	
If $\mathcal{Z}_{(k,k^{\prime})}\notin\mathscr{Z}_{\bar{M}}$ and
$\mathcal{Z}_{(k,k^{\prime})}\in\mathscr{Z}_{P}$, then $\beta_{(k,k^{\prime}%
	)}^{1}=0$. In this case, it is possible that $\mathcal{Z}_{(k,k^{\prime}%
	)}\notin$ $\mathscr{Z}_{0}^{\prime}$ and $\beta_{(k,k^{\prime})}^{\prime
	\prime}=0$, because by definition $\mathscr{Z}_{0}^{\prime}\subseteq
\mathscr{Z}_{P}$. Note that if $\mathcal{Z}_{(k,k^{\prime})}\in$
$\mathscr{Z}_{0}^{\prime}$, then $\beta_{(k,k^{\prime})}^{\prime\prime}%
=\beta_{(k,k^{\prime})}^{\prime}$ by definition.

If $\mathcal{Z}_{(k,k^{\prime})}\notin\mathscr{Z}_{\bar{M}}$ and
$\mathcal{Z}_{(k,k^{\prime})}\notin\mathscr{Z}_{P}$, then $\beta
_{(k,k^{\prime})}^{1}=\beta_{(k,k^{\prime})}^{\prime}=0$. Similarly, in this
case, $\beta_{(k,k^{\prime})}^{\prime\prime}=\beta_{(k,k^{\prime})}^{1}=0$,
because $\mathscr{Z}_{0}^{\prime}\subseteq\mathscr{Z}_{P}$.

If $\mathcal{Z}_{(k,k^{\prime})}\in\mathscr{Z}_{\bar{M}}$ and $\mathcal{Z}%
_{(k,k^{\prime})}\in\mathscr{Z}_{P}$, then $\beta_{(k,k^{\prime})}^{1}%
=\beta_{(k,k^{\prime})}^{\prime}=\beta_{(k,k^{\prime})}^{\prime\prime}$,
because $\mathscr{Z}_{0}\supseteq\mathscr{Z}_{\bar{M}}$.

If $\mathcal{Z}_{(k,k^{\prime})}\in\mathscr{Z}_{\bar{M}}$ and $\mathcal{Z}%
_{(k,k^{\prime})}\notin\mathscr{Z}_{P}$, then $\beta_{(k,k^{\prime})}^{\prime
}=\beta_{(k,k^{\prime})}^{\prime\prime}=0$ because $\mathscr{Z}_{0}^{\prime
}\subseteq\mathscr{Z}_{P}$.
\end{proof}

\subsection{Selectively Pairwise Valid Multiple Instruments}\label{sec.selectively valid instruments}

Here we introduce a weaker notion of pairwise validity that is available when $Z$ contains multiple instruments. Specifically, suppose the instrument $Z$ is a vector with $Z=(Z_1,\ldots,Z_L)^T$, where $Z_l$ is a scalar instrument for every $l\in\{1,\ldots,L\}$. There are $C_L=2^{L}$ combinations of scalar instruments $\{Z_1,\ldots,Z_{L}\}$. We refer to each combination as a \textit{subinstrument} of $Z$, denoted by $V_c$ for every $c\in\{1,\ldots,C_L\}$ with $V_c\in\{v_{c1},\ldots,v_{c K_c}\}$ for some $K_c>1$. Every $V_c$ can be a scalar or vector instrument, and we define the set of all pairs of values of $V_c$ by $$\mathscr{Z}_c=\{(v_{c1},v_{c2}),\ldots,(v_{c1},v_{cK_c}),\ldots,(v_{cK_c},v_{c1}),\ldots,(v_{cK_c},v_{cK_c-1})\}.$$

The following definition weakens Definition \ref{def.partial validity pairwise binary D}.

\begin{definition}
	\label{def.partial validity pairwise binary D multiple selection} The instrument $Z$ is \textbf{selectively pairwise valid} for
	the treatment $D\in\{0,1\}$ if there is a subinstrument $V_c$ that is pairwise valid according to Definition \ref{def.partial validity pairwise binary D}.
\end{definition}

To illustrate that Definition \ref{def.partial validity pairwise binary D multiple selection} is weaker than Definition \ref{def.partial validity pairwise binary D}, consider the following example.
\begin{example} \looseness=-1  Suppose that $Z=(Z_1,Z_2,Z_3)^T$, where $Z_1$ is correlated with all potential variables and $(Z_2,Z_3)^T$ satisfies the conditions in Assumption \ref{ass.IV validity binary D}. Then $Z$ is not pairwise valid by Definition \ref{def.partial validity pairwise binary D}, but it is selectively pairwise valid.
\end{example}

For every subinstrument $V_c$, we can define the largest validity pair set $\mathscr{Z}_{c\bar{M}_c}\subseteq\mathscr{Z}_{c}$. Then the identification and estimation of $\mathscr{Z}_{c\bar{M}_c}$ and the VSIV estimation of LATEs can proceed as described in Section \ref{sec.VSIV under consistency}. Asymptotic normality and bias reduction can be established accordingly. The notion of selectively pairwise valid instruments can be straightforwardly generalized to multivalued ordered or unordered treatments.

\subsection{Testable Implications of \citet{kedagni2020generalized}}\label{sec.testable implication kedagni}
We consider the case where $D\in \mathcal{D}=\{d_1,\ldots,d_J\}$. 
Suppose $Y\in\mathbb{R}$ is continuous. Results for discrete $Y$ can be obtained similarly.
The testable implications in
\citet{kedagni2020generalized} are for exclusion ($Y_{dz_{k_{m}}
}=Y_{dz_{k_{m}^{\prime}}}$ for all $d\in\mathcal{D}$) and statistical independence ($Z$ is independent of $(Y_{d_1z_{k_{m}}},Y_{d_1z_{k_{m}^{\prime}}
},\ldots,Y_{d_Jz_{k_{m}}},Y_{d_Jz_{k_{m}^{\prime}}})$) for every $m\in
\{1,\ldots,\bar{M}\}$ with the largest validity pair set $\mathscr{Z}_{\bar
	{M}}=\{(z_{k_{1}},z_{k_{1}^{\prime}}),\ldots,(z_{k_{\bar{M}}},z_{k_{\bar{M}%
	}^{\prime}})\}$. In the following, we show that these testable implications are also for Conditions (i) and (ii) in Definition \ref{def.partial validity pairwise}.
	Under Conditions (i) and (ii) in Definition \ref{def.partial validity pairwise}, we can define $Y_{d}(z,z^{\prime})$ by $Y_{d}(z,z^{\prime})=Y_{dz}=Y_{dz'}$ a.s. for every $d\in\mathcal{D}$ and every $(z,z')\in\mathscr{Z}_{\bar{M}}$.
Define
\[
f_{Y,D}\left(  y,d|z\right)  =f_{Y|D,Z}\left(  y|d,z\right)  \mathbb{P}\left(
D=d|Z=z\right)
\]
for every $y\in\mathbb{R}$, every $d\in\mathcal{D}$, and every $z\in
\mathcal{Z}$, where $f_{Y|D,Z}\left(  y|d,z\right)  $ is the conditional
density function of $Y$ given $D=d$ and $Z=z$.
For every $\mathcal{Z}_{(k,k^{\prime})}=(z_{k},z_{k^{\prime}})\in
\mathscr{Z}_{\bar{M}}$, every $A\in\mathcal{B}_{\mathbb{R}}$, every
$d\in\mathcal{D}$, and each $z\in\mathcal{Z}_{(k,k^{\prime})}$,
\[
\mathbb{P}\left(  Y\in A,D=d|Z=z\right)  \leq\mathbb{P}\left(  Y_{dz}\in
A|Z=z\right)  =\mathbb{P}\left(  Y_{d}(z_{k},z_{k^{\prime}})\in A\right)  ,
\]
and
\begin{align*}
	\mathbb{P}\left(  Y\in A,D=d|Z=z\right)   &  =\frac{\mathbb{P}\left(  Y\in
		A,D=d,Z=z\right)  }{\mathbb{P}\left(  Z=z\right)  }\\
	&  =\mathbb{P}\left(  Y\in A|D=d,Z=z\right)  \mathbb{P}\left(  D=d|Z=z\right)
	.
\end{align*}
Then, by the discussion in Section 4.1 of \citet{kedagni2020generalized}, for (almost) all $y$,
\[
f_{Y,D}\left(  y,d|z\right)  =f_{Y|D,Z}\left(  y|d,z\right)  \mathbb{P}\left(
D=d|Z=z\right)  \leq f_{Y_{d}(z_{k},z_{k^{\prime}})}\left(  y\right)  ,
\]
where $f_{Y_{d}(z_{k},z_{k^{\prime}})}  $ is the density
function of the potential outcome $Y_{d}(z_{k},z_{k^{\prime}})$. Thus, for
every $d\in\mathcal{D}$,
\begin{align}\label{eq.density inequality}
	\max_{z\in\mathcal{Z}_{(k,k^{\prime})}}f_{Y,D}\left(  y,d|z\right)  \leq
	f_{Y_{d}(z_{k},z_{k^{\prime}})}\left(  y\right)  ,
\end{align}
and we obtain the first inequality \citep[p.~666]{kedagni2020generalized}:
\begin{align}\label{eq.first inequality KM}
\max_{d\in\mathcal{D}}\int_{\mathbb{R}}\max_{z\in\mathcal{Z}_{(k,k^{\prime})}%
}f_{Y,D}\left(  y,d|z\right)  \mathrm{d}y\leq1.
\end{align}
Also, for all $A_{1},\ldots,A_{J}\in\mathcal{B}_{\mathbb{R}}$,
\begin{align*}
	&  \mathbb{P}\left(  Y_{d_{1}}(z_{k},z_{k^{\prime}})\in A_{1},\ldots,Y_{d_{J}%
	}(z_{k},z_{k^{\prime}})\in A_{J}\right)  \\
	=&\,\min_{z\in\mathcal{Z}_{(k,k^{\prime})}}\mathbb{P}\left(  Y_{d_{1}}%
	(z_{k},z_{k^{\prime}})\in A_{1},\ldots,Y_{d_{J}}(z_{k},z_{k^{\prime}})\in
	A_{J}|Z=z\right)  \\
	=& \,  \min_{z\in\mathcal{Z}_{(k,k^{\prime})}}\sum_{j=1}^{J}\mathbb{P}\left(
	Y_{d_{1}}(z_{k},z_{k^{\prime}})\in A_{1},\ldots,Y_{d_{J}}(z_{k},z_{k^{\prime}%
	})\in A_{J},D=d_{j}|Z=z\right)  \\
	\leq&\,\min_{z\in\mathcal{Z}_{(k,k^{\prime})}}\sum_{j=1}^{J}\mathbb{P}\left(
	Y\in A_{j},D=d_{j}|Z=z\right)  .
\end{align*}
Let $P_{\mathbb{R}}^{j}$ be an arbitrary partition of $\mathbb{R}$ for
$j\in\left\{  1,\ldots,J\right\}  $, that is, $P_{\mathbb{R}}^{j}=\{C_{1}%
^{j},\ldots,C_{N_{j}}^{j}\}$ with $\cup_{l=1}^{N_{j}}C_{l}^{j}=\mathbb{R}$ and
$C_{l'}^{j}\cap C_{l}^{j}=\varnothing$ for all $l'\neq l$. Then
\begin{align*}
	1 &  =\sum_{A_{1}\in P_{\mathbb{R}}^{1}}\cdots\sum_{A_{J}\in P_{\mathbb{R}%
		}^{J}}\mathbb{P}\left(  Y_{d_{1}}(z_{k},z_{k^{\prime}})\in A_{1}%
	,\ldots,Y_{d_{J}}(z_{k},z_{k^{\prime}})\in A_{J}\right)  \\
	&  \leq\sum_{A_{1}\in P_{\mathbb{R}}^{1}}\cdots\sum_{A_{J}\in P_{\mathbb{R}%
		}^{J}}\min_{z\in\mathcal{Z}_{(k,k^{\prime})}}\sum_{j=1}^{J}\mathbb{P}\left(
	Y\in A_{j},D=d_{j}|Z=z\right)  .
\end{align*}
Then we obtain the second inequality  \citep[p.~666]{kedagni2020generalized}:
\begin{align}\label{eq.second inequality KM}
\inf_{\left\{  P_{\mathbb{R}}^{1},\ldots,P_{\mathbb{R}}^{J}\right\}  }%
\sum_{A_{1}\in P_{\mathbb{R}}^{1}}\cdots\sum_{A_{J}\in P_{\mathbb{R}}^{J}}%
\min_{z\in\mathcal{Z}_{(k,k^{\prime})}}\sum_{j=1}^{J}\mathbb{P}\left(  Y\in
A_{j},D=d_{j}|Z=z\right)  \geq1,
\end{align}
where the infimum is taken over all partitions $\{P_{\mathbb{R}}^{1},\ldots,P_{\mathbb{R}}^{J}\}  $. Next, for all $A_{1},\ldots
,A_{J}\in\mathcal{B}_{\mathbb{R}}$,
\begin{align*}
	&  \mathbb{P}\left(  Y_{d_j}(z_{k},z_{k^{\prime}})\in A_{j}\right)  \\
	=&\,\sum_{A_{1}\in P_{\mathbb{R}}^{1}}\cdots\sum_{A_{j-1}\in P_{\mathbb{R}%
		}^{j-1}}\sum_{A_{j+1}\in P_{\mathbb{R}}^{j+1}}\ldots\sum_{A_{J}\in
		P_{\mathbb{R}}^{J}}\mathbb{P}\left(  Y_{d_1}(z_{k},z_{k^{\prime}})\in
	A_{1},\ldots,Y_{d_J}(z_{k},z_{k^{\prime}})\in A_{J}\right)  \\
	\leq&\,\sum_{A_{1}\in P_{\mathbb{R}}^{1}}\cdots\sum_{A_{j-1}\in P_{\mathbb{R}%
		}^{j-1}}\sum_{A_{j+1}\in P_{\mathbb{R}}^{j+1}}\ldots\sum_{A_{J}\in
		P_{\mathbb{R}}^{J}}\min_{z\in\mathcal{Z}_{(k,k^{\prime})}}\sum_{\xi=1}%
	^{J}\mathbb{P}\left(  Y\in A_{\xi},D=d_{\xi}|Z=z\right)  ,
\end{align*}
which, together with \eqref{eq.density inequality}, implies the third inequality
 \citep[p.~666]{kedagni2020generalized}:
\begin{align}\label{eq.third inequality KM}
\sup_{\left\{  P_{\mathbb{R}}^{1},\ldots,P_{\mathbb{R}}^{J}\right\}  }%
\max_{j\in\{1,\ldots,J\}}\sup_{A_{j}\in\mathcal{B}_{\mathbb{R}}}\left\{
\int_{A_{j}}\max_{z\in\mathcal{Z}_{(k,k^{\prime})}}f_{Y,D}\left(
y,d_j|z\right)  \mathrm{d}y-\varphi_{j}\left(  A_{j},\mathcal{Z}_{(k,k^{\prime
	})},P_{\mathbb{R}}^{1},\ldots,P_{\mathbb{R}}^{J}\right)  \right\}
\leq0,    
\end{align}
where%
\begin{align*}
\varphi_{j}(  A_{j},\mathcal{W},P_{\mathbb{R}}^{1},\ldots,P_{\mathbb{R}
}^{J})=\sum_{A_{1}\in P_{\mathbb{R}}^{1}}\cdots\sum_{A_{j-1}\in
	P_{\mathbb{R}}^{j-1}}\sum_{A_{j+1}\in P_{\mathbb{R}}^{j+1}}\ldots\sum
_{A_{J}\in P_{\mathbb{R}}^{J}}\min_{z\in\mathcal{W}}\sum_{\xi=1}^{J}\int_{A_{\xi}
}f_{Y,D}\left(  y,d_{\xi}|z\right)  \mathrm{d}y
\end{align*}
for all $\mathcal{W}\subseteq\mathcal{Z}$.

\subsection{Definition and Estimation of $\mathscr{Z}_{0}$}\label{sec.estimation Z_0 ordered}

We estimate $\mathscr{Z}_{0}=\mathscr{Z}_1\cap\mathscr{Z}_2$ as $\widehat{\mathscr{Z}_0}=\widehat{\mathscr{Z}_1}\cap\widehat{\mathscr{Z}_2}$, where $\widehat{\mathscr{Z}_1}$ and $\widehat{\mathscr{Z}_2}$ are estimators of $\mathscr{Z}_1$ and  $\mathscr{Z}_2$, respectively.

\subsubsection{Definition and Estimation of $\mathscr{Z}_1$}\label{sec.Z1 estimation ordered}
The testable implications proposed by \citet{sun2021ivvalidity} are for full IV validity. Here we extend them to pairwise valid instruments (Definition \ref{def.partial validity pairwise}). We follow the notation of \citet{sun2021ivvalidity} to introduce the definition of $\mathscr{Z}_1$ and the corresponding estimator. Define conditional probabilities 
\begin{equation*}
	P_z\left( B,C\right) =\mathbb{P}\left( Y\in B,D\in C|Z=z\right)
\end{equation*}
for all Borel sets $B,C\in\mathcal{B}_{\mathbb{R}}$ and all $z\in\mathcal{Z}$. The testable implications extended from \citet{sun2021ivvalidity} for the conditions in Definition \ref{def.partial validity pairwise} are that for every $m\in\{1,\ldots,\bar{M}\}$, 
\begin{align}
	P_{z_{k_m}}\left( B,\{d_{J}\}\right) \leq P_{z_{k_m'}}\left(
	B,\{d_{J}\}\right)   \text{ and } P_{z_{k_m}}\left( B,\{d_{1}\}\right)  \geq P_{z_{k_m'}}\left(
	B,\{d_{1}\}\right) 
	\label{eq.testable implication multivalue}
\end{align}
for all $B\in\mathcal{B}_{\mathbb{R}}$, and 
\begin{align}  \label{eq.fosd multi}
	P_{z_{k_m}}\left( \mathbb{R},C\right) \ge P_{z_{k_m'}}\left( 
	\mathbb{R},C\right)
\end{align}
for all $C=(-\infty,c]$ with $c\in\mathbb{R}$.
Without loss of generality, we assume that $d_{1}=0$ and $d_{J}=1$. By definition, for all $%
B,C\in \mathcal{B}_{\mathbb{R}}$, 
\begin{equation*}
	\mathbb{P}\left( Y\in B,D\in C|Z=z\right) =\frac{\mathbb{P}\left( Y\in
		B,D\in C,Z=z\right) }{\mathbb{P}\left( Z=z\right) }.
\end{equation*}%

Next, we reformulate the testable restrictions to define $\mathscr{Z}_1$ and its estimator.
Define the following function spaces
\begin{align}\label{def.function spaces}
	& \mathcal{G}_P=\left\{ \left( 1_{\mathbb{R}\times \mathbb{R}\times \left\{
		z_{k}\right\} },1_{\mathbb{R}\times \mathbb{R}\times \left\{ z_{k^{\prime
		}}\right\} }\right) : k,k^{\prime }\in\{1,\ldots, K\}, k\neq k'\right\} , \notag \\
	& \mathcal{H}_{1}=\left\{ \left( -1\right) ^{d}\cdot 1_{B\times \left\{
		d\right\} \times \mathbb{R}}:B\text{ is a closed interval in }\mathbb{R}%
	,d\in \{0,1\}\right\} ,\notag \\
	& \bar{\mathcal{H}}_{1}=\left\{ \left( -1\right) ^{d}\cdot 1_{B\times
		\left\{ d\right\} \times \mathbb{R}}:B\text{ is a closed, open, or
		half-closed interval in }\mathbb{R},d\in \left\{ 0,1\right\} \right\} , \notag\\
	& \mathcal{H}_{2}=\left\{ 1_{\mathbb{R}\times C\times \mathbb{R}}:C=(-\infty
	,c],c\in \mathbb{R}\right\} ,\notag \\
	& \bar{\mathcal{H}}_{2}=\left\{ 1_{\mathbb{R}\times C\times \mathbb{R}%
	}:C=(-\infty ,c]\text{ or }C=(-\infty ,c),c\in \mathbb{R}\right\} ,\notag \\
	& \mathcal{H}=\mathcal{H}_{1}\cup \mathcal{H}_{2},\text{ and }\bar{\mathcal{H%
	}}=\bar{\mathcal{H}}_{1}\cup \bar{\mathcal{H}}_{2}.
\end{align}
Let $P$ and $\widehat{P}$ be defined as in Section \ref{sec.estimation validity set}. Let $\phi$, $\sigma^2$, $\widehat{\phi}$, and $\widehat{\sigma}^2$ be defined in a way similar to that in Section \ref{sec.estimation validity set} but for all $\left( h,g\right) \in \bar{	\mathcal{H}}\times \mathcal{G}_P$ in \eqref{def.function spaces}. Also, we let $\Lambda(P)=\prod_{k=1}^{K}P(1_{\mathbb{R}\times\mathbb{R}\times\{z_k\}})$  and $T_{n}=n\cdot \prod_{k=1}^K\widehat{P}(1_{\mathbb{R}\times \mathbb{R}
	\times \{z_k\}})$. By similar arguments as in the proof of Lemma 3.1 in \citet{sun2021ivvalidity}, ${\sigma}^{2}$ and $\widehat{\sigma}^{2}$ are uniformly bounded in $(h,g)\in \bar{	\mathcal{H}}\times \mathcal{G}_P$.

The following lemma reformulates the testable restrictions in terms of $\phi$.
\begin{lemma}\label{lemma.superset of Z pairwise}
	Suppose that the instrument $Z$ is pairwise valid for the treatment $D$ with the largest validity pair set $\mathscr{Z}_{\bar{M}}=\{(z_{k_1},z_{k_1^{\prime}}),\ldots,(z_{k_{\bar{M}}},z_{k'_{\bar{M}}})\}$. For every $m\in\{1,\ldots,\bar{M}\}$,  $\sup_{h\in {\mathcal{H}}}\phi \left( h,g\right) =0$ with $g=( 1_{\mathbb{R}\times \mathbb{R}\times \{	z_{k_m}\} },1_{\mathbb{R}\times \mathbb{R}\times \{ z_{k_{m}^{\prime}}\} })$.
\end{lemma}

\begin{proof}[Proof of Lemma \ref{lemma.superset of Z pairwise}]
	Note that for every $g\in\mathcal{G}_P$, we can always find some $a\in\mathbb{R}$ such that $\phi\left(  h,g\right)  =0$ with
	$h=1_{\{a\}\times\{0\}\times\mathbb{R}}$. So $\sup_{h\in\mathcal{H}}%
	\phi\left(  h,g\right)  \ge 0$ for every $g\in\mathcal{G}_P$. 
	Under assumption, for every $g=(1_{\mathbb{R}\times\mathbb{R}\times
		\{z_{k_{m}}\}},1_{\mathbb{R}\times\mathbb{R}\times\{z_{k_{m}'}\}})$, by Lemma 2.1 of \citet{sun2021ivvalidity},
	$\phi\left(  h,g\right)  \leq0$ for all $h\in\mathcal{H}$.
	Thus, $\sup_{h\in\mathcal{H}}	\phi\left(  h,g\right)  =0$.    
\end{proof}

Lemma \ref{lemma.superset of Z pairwise} reformulates the necessary conditions for $\mathscr{Z}_{\bar{M}}$.  
By Lemma \ref{lemma.superset of Z pairwise}, we define  
\begin{equation}\label{eq.G0 pair}
	\mathcal{G}_{1}=\left\{ g\in \mathcal{G}_P:\sup_{h\in {\mathcal{H}}}\phi\left(
	h,g\right) =0\right\} \text{ and } \widehat{\mathcal{G}_{1}}=\left\{ g\in 
	\mathcal{G}_P:\sqrt{T_n}\left\vert \sup_{h\in {\mathcal{H}} }\frac{\widehat{\phi}
		\left( h,g\right)}{\xi_{0}\vee \widehat{\sigma}(h,g)} \right\vert \leq \tau^g
	_{n}\right\}
\end{equation}
where $1/\min_{g\in\mathcal{G}_P}\tau_{n}^g\to0$ in probability with $\max_{g\in\mathcal{G}_P}\tau_{n}^g/\sqrt{n}\to0$ in probability as $n\to\infty$, and $\xi_{0}$ is a small positive number. 
We define ${\mathscr{Z}_1}$ as the collection of all $(z,z^{\prime})$ that are associated with some $g\in{\mathcal{G}_1}$:
\begin{align}\label{eq.Z1 pair}
	\mathscr{Z}_1=\left\{ (z_{k},z_{k^{\prime}})\in\mathscr{Z}: g=( 1_{\mathbb{R}\times \mathbb{R}\times \{	z_{k}\} },1_{\mathbb{R}\times \mathbb{R}\times \{ z_{k^{\prime}}\} })\in\mathcal{G}_1\right\}.
\end{align}
We use $\widehat{\mathcal{G}_1}$ to construct the estimator of $\mathscr{Z}_1$, denoted by $\widehat{\mathscr{Z}_1}$, which is defined as the set of all $(z,z^{\prime})$ that are associated with some $g\in\widehat{\mathcal{G}_1}$ in the same way $\mathscr{Z}_1$ is defined based on $\mathcal{G}_1$:
\begin{align}\label{eq.Z1_hat pair}
	\widehat{\mathscr{Z}_1}=\left\{ (z_{k},z_{k^{\prime}})\in\mathscr{Z}: g=( 1_{\mathbb{R}\times \mathbb{R}\times \{	z_{k}\} },1_{\mathbb{R}\times \mathbb{R}\times \{ z_{k^{\prime}}\} })\in\widehat{\mathcal{G}_1}\right\}.
\end{align}

To establish consistency of $\widehat{\mathscr{Z}_{1}}$, we state and prove an auxiliary lemma.

\begin{lemma}
	\label{lemma.uniform and weak convergence} Under Assumption \ref{ass.iid data},  $\widehat{\phi}\to \phi$, $T_n/n\to \Lambda(P)$, and $\widehat{\sigma}\to \sigma$ almost uniformly.\footnote{See the definition of almost uniform convergence in \citet[p.~52]{van1996weak}.} In addition, $\sqrt{T_n}(  \widehat{\phi}-\phi)  \leadsto\mathbb{G}$ for some random element
	$\mathbb{G}$, and for all $\left(  h,g\right)  \in\bar{\mathcal{H}}\times\mathcal{G}_P$ with $g=(g_1,g_2)$, the variance $Var\left(  \mathbb{G}\left(  h,g\right)  \right)=\sigma^2(h,g)$. 
\end{lemma}

\begin{proof}[Proof of Lemma \ref{lemma.uniform and weak convergence}]
	Note that the $\mathcal{ G}_P$ defined in \eqref{def.function spaces} is only slightly different from the $\mathcal{G}$ defined in (7) of \citet{sun2021ivvalidity}. The lemma can be proved following a strategy similar to that of the proofs of Lemmas C.11 and 3.1 of \citet{sun2021ivvalidity}.
\end{proof}

The following proposition establishes consistency of $\widehat{\mathscr{Z}_{1}}$.

\begin{proposition}
	\label{prop.consistent G hat pairwise Z1} Under Assumption \ref{ass.iid data},  $\mathbb{P}(\widehat{\mathcal{G}_1}=\mathcal{G}_1)\rightarrow 1$, and thus $\mathbb{P}(\widehat{\mathscr{Z}_{1}}=\mathscr{Z}_{1})\rightarrow 1$.
\end{proposition}

\begin{proof}[Proof of Proposition \ref{prop.consistent G hat pairwise Z1}]
	First, suppose $\mathcal{G}_1\neq\varnothing$. Under the constructions, we
	have that  
	\begin{align*}
		& \lim_{n\rightarrow\infty}\mathbb{P}\left( \mathcal{G}_1\setminus 
		\widehat{\mathcal{G}_1}\neq\varnothing\right) \\
		\leq&\lim_{n\rightarrow\infty}\mathbb{P}\left( \max_{g\in\mathcal{G}_1} \sqrt{%
			T_{n}}\left\vert \sup_{h\in\mathcal{H}}\left( \frac{\widehat{\phi}\left(
			h,g\right) }{\xi_{0}\vee\widehat{\sigma}\left( h,g\right) }\right) -\sup _{h\in%
			\mathcal{H}}\left( \frac{\phi\left( h,g\right) }{\xi_{0}\vee \widehat{\sigma}%
			\left( h,g\right) }\right) \right\vert >\min_{g\in\mathcal{G}_P}\tau_{n}^g\right) \\
		\leq&\lim_{n\rightarrow\infty}\mathbb{P}\left( \max_{g\in\mathcal{G}_1}
		\sup_{h\in\mathcal{H}}\sqrt{T_{n}}\left\vert \frac{\widehat{\phi}\left(
			h,g\right) -\phi\left( h,g\right) }{\xi_{0}\vee\widehat{\sigma}\left( h,g\right) 
		}\right\vert >\min_{g\in\mathcal{G}_P}\tau_{n}^g\right) .
	\end{align*}
	By Lemma \ref{lemma.uniform and weak convergence}, $\sqrt{T_{n}}(\widehat{\phi}-\phi)\leadsto\mathbb{G}$ and  $\widehat{\sigma}\rightarrow\sigma$ almost uniformly, which implies
	that  $\widehat{\sigma}\leadsto\sigma$ by Lemmas 1.9.3(ii) and 1.10.2(iii) of  %
	\citet{van1996weak}. Consequently, by Example 1.4.7 (Slutsky's lemma) and Theorem
	1.3.6  (continuous mapping) of \citet{van1996weak},  
	\begin{equation*}
		\max_{g\in\mathcal{G}_1}\sup_{h\in\mathcal{H}}\sqrt{T_{n}}\left\vert \frac {%
			\widehat{\phi}\left( h,g\right) -\phi\left( h,g\right) }{\xi_{0}\vee \widehat{\sigma}%
			\left( h,g\right) }\right\vert \leadsto\max_{g\in\mathcal{G}_1 }\sup_{h\in%
			\mathcal{H}}\left\vert \frac{\mathbb{G}\left( h,g\right) } {%
			\xi_{0}\vee\sigma\left( h,g\right) }\right\vert .  
	\end{equation*}
	Since $1/\min_{g\in\mathcal{G}_P}\tau_{n}^g\rightarrow0$ in probability, we have that $\lim_{n\rightarrow\infty }%
	\mathbb{P}(\mathcal{G}_1\setminus\widehat{\mathcal{G}_1}\neq 
	\varnothing)=0$.
	
	If $\mathcal{G}_1=\mathcal{G}_P$, then clearly $\lim_{n\rightarrow\infty }%
	\mathbb{P(}\widehat{\mathcal{G}_1}\setminus\mathcal{G}_1\neq 
	\varnothing)=0$. Suppose $\mathcal{G}_1\neq\mathcal{G}_P$. Since  $\mathcal{G}_P$ is a finite set and $\widehat{\sigma}$ is uniformly bounded in $(h,g)$ by construction, then there  is a $\delta>0$ such that $\min_{g\in\mathcal{G}_P\setminus\mathcal{G}_1
	}\left\vert \sup_{h\in\mathcal{H}}\phi\left( h,g\right) /(\xi_{0}\vee  \widehat{\sigma}\left( h,g\right) )\right\vert >\delta$. Thus, we have that  
	\begin{align*}
		& \lim_{n\rightarrow\infty}\mathbb{P}\left( \widehat{\mathcal{G}_1 }%
		\setminus\mathcal{G}_1\neq\varnothing\right) \\
		\leq & \lim_{n\rightarrow\infty}\mathbb{P}\left( \max_{g\in\widehat {%
				\mathcal{G}_1}\backslash\mathcal{G}_1}\left\vert \sup_{h\in\mathcal{H} }%
		\frac{\phi\left( h,g\right) }{\xi_{0}\vee\widehat{\sigma}\left( h,g\right) }%
		\right\vert >\delta,\max_{g\in\widehat{\mathcal{G}_1}\backslash \mathcal{G}_{1}}\sqrt{T_{n}}\left\vert \sup_{h\in\mathcal{H}}\frac{\widehat{\phi }\left(
			h,g\right) }{\xi_{0}\vee\widehat{\sigma}\left( h,g\right) }\right\vert
		\leq\max_{g\in\mathcal{G}_P}\tau_{n}^g\right) .
	\end{align*}
	By Lemma \ref{lemma.uniform and weak convergence}, $\widehat{\phi}\rightarrow\phi$ almost uniformly. Thus, for every $\varepsilon>0$, there is a measurable set  
	$A$ with $\mathbb{P}(A)\geq1-\varepsilon$ such that for sufficiently large  $%
	n$,  
	\begin{equation*}
	\max_{g\in\mathcal{G}_P}\left\vert\left\vert 
		\sup_{h\in\mathcal{H}}\frac{\widehat{\phi}\left( h,g\right) }{\xi_{0}\vee \widehat{%
				\sigma}\left( h,g\right) }\right\vert -\left\vert \sup_{h\in\mathcal{H} }\frac{{\phi}%
			\left( h,g\right) }{\xi_{0}\vee\widehat{\sigma}\left( h,g\right) }\right\vert \right\vert \le
		\frac{\delta}{2} 
	\end{equation*}
	uniformly on $A$. We now have that  
	\begin{align*}
		& \lim_{n\rightarrow\infty}\mathbb{P}\left( \widehat{\mathcal{G}_1 }%
		\setminus\mathcal{G}_1\neq\varnothing\right) \\
		\leq & \lim_{n\rightarrow\infty}\mathbb{P}\left( 
		\begin{array}{c}
			\left\{ \max_{g\in\widehat{\mathcal{G}_1}\backslash\mathcal{G}_1
			}\left\vert \sup_{h\in\mathcal{H}}\frac{\phi\left( h,g\right) }{\xi_{0} \vee%
				\widehat{\sigma}\left( h,g\right) }\right\vert >\delta\right\} \\ 
			\cap\left\{ \max_{g\in\widehat{\mathcal{G}_1}\backslash\mathcal{G}_1 }%
			\sqrt{T_{n}}\left\vert \sup_{h\in\mathcal{H}}\frac{\widehat{\phi}\left(
				h,g\right) }{\xi_{0}\vee\widehat{\sigma}\left( h,g\right) }\right\vert \leq
			\max_{g\in\mathcal{G}_P}\tau_{n}^g\right\} \cap A%
		\end{array}
		\right) +\mathbb{P}(A^{c}) \\
		\leq & \lim_{n\rightarrow\infty}\mathbb{P}\left( \sqrt{\frac{T_{n}}{n}} 
		\frac{\delta}{2}<\max_{g\in\widehat{\mathcal{G}_1}\backslash\mathcal{G}
			_{1} }\sqrt{\frac{{T_{n}}}{{n}}}\left\vert \sup_{h\in\mathcal{H}}\frac{\widehat{%
				\phi }\left( h,g\right) }{\xi_{0}\vee\widehat{\sigma}\left( h,g\right) }%
		\right\vert \leq\frac{\max_{g\in\mathcal{G}_P}\tau_{n}^g}{\sqrt{n}}\right) +\varepsilon=\varepsilon,
	\end{align*}
	because $\max_{g\in\mathcal{G}_P}\tau_{n}^g/\sqrt{n}\rightarrow0$ in probability as $n\rightarrow\infty$. Since  $\varepsilon$ can be arbitrarily small, we have that $\mathbb{P}( 
	\widehat{\mathcal{G}_1}=\mathcal{G}_1)\rightarrow1$, because  $\mathbb{P(%
	}\mathcal{G}_1\backslash\widehat{\mathcal{G}_1}\neq 
	\varnothing)\rightarrow0$ and $\mathbb{P(}\widehat{\mathcal{G}_1} \setminus%
	\mathcal{G}_1\neq\varnothing)\rightarrow0$.
	
	Second, suppose $\mathcal{G}_1=\varnothing $. This implies that $%
	\min_{g\in \mathcal{G}_P}\left\vert \sup_{h\in \mathcal{H}}\phi \left(
	h,g\right) /(\xi _{0}\vee \widehat{\sigma}\left( h,g\right)) \right\vert >\delta $
	for some $\delta >0$. Since by Lemma \ref{lemma.uniform and weak convergence}, $\widehat{\phi}\rightarrow \phi $ almost uniformly, then there is a
	measurable set $A$ with $\mathbb{P}(A)\geq 1-\varepsilon $ such that for
	sufficiently large $n$, 
	\begin{equation*}
	\max_{g\in\mathcal{G}_P}\left\vert\left\vert 
		\sup_{h\in\mathcal{H}}\frac{\widehat{\phi}\left( h,g\right) }{\xi_{0}\vee \widehat{%
				\sigma}\left( h,g\right) }\right\vert -\left\vert \sup_{h\in\mathcal{H} }\frac{{\phi}%
			\left( h,g\right) }{\xi_{0}\vee\widehat{\sigma}\left( h,g\right) }\right\vert \right\vert \le
		\frac{\delta}{2} 
	\end{equation*}
	uniformly on $A$. Thus we now have that 
	\begin{align*}
		&\lim_{n\rightarrow \infty }\mathbb{P}\left( \widehat{\mathcal{G}_1}\neq
		\varnothing \right)\\  
		\leq & \lim_{n\rightarrow \infty }\mathbb{P}\left( 
		\begin{array}{c}
			\left\{ \max_{g\in \widehat{\mathcal{G}_1}}\left\vert \sup_{h\in \mathcal{H%
			}}\frac{\phi \left( h,g\right) }{\xi _{0}\vee \widehat{\sigma}\left( h,g\right) }%
			\right\vert >\delta \right\}  \\ 
			\cap \left\{ \max_{g\in \widehat{\mathcal{G}_1}}\sqrt{T_{n}}\left\vert
			\sup_{h\in \mathcal{H}}\frac{\widehat{\phi}\left( h,g\right) }{\xi _{0}\vee \widehat{%
					\sigma}\left( h,g\right) }\right\vert \leq \max_{g\in\mathcal{G}_P}\tau_{n}^g\right\} \cap A%
		\end{array}%
		\right) +\mathbb{P}(A^{c}) \\
		\leq & \lim_{n\rightarrow \infty }\mathbb{P}\left( \sqrt{\frac{T_{n}}{n}}%
		\frac{\delta }{2}<\max_{g\in \widehat{\mathcal{G}_1}}\sqrt{\frac{{T_{n}}}{{%
					n}}}\left\vert \sup_{h\in \mathcal{H}}\frac{\widehat{\phi}\left( h,g\right) }{%
			\xi _{0}\vee \widehat{\sigma}\left( h,g\right) }\right\vert \leq \frac{\max_{g\in\mathcal{G}_P}\tau_{n}^g
		}{\sqrt{n}}\right) +\varepsilon =\varepsilon ,
	\end{align*}%
	because $\max_{g\in\mathcal{G}_P}\tau_{n}^g/\sqrt{n}\rightarrow 0$ in probability as $n\rightarrow \infty $. Since $\varepsilon $ can be arbitrarily small, we have that $\mathbb{P}(\widehat{\mathcal{G%
		}_{1}}=\mathcal{G}_1)=1-\mathbb{P(}\widehat{\mathcal{G}_1}\neq
	\varnothing )\rightarrow 1$.
\end{proof}

As mentioned after Proposition \ref{prop.consistent G hat pairwise Z1 binary D}, Proposition \ref{prop.consistent G hat pairwise Z1} and its proof are related to the contact set estimation in \citet{sun2021ivvalidity}. Since  $\mathcal{G}_1\subseteq\mathcal{G}_P$ and $\mathcal{G}_P$ is a finite set, we can use techniques similar to those in \citet{sun2021ivvalidity} to obtain the stronger result in Proposition \ref{prop.consistent G hat pairwise Z1}, that is, $\mathbb{P}(\widehat{\mathcal{G}_1}=\mathcal{G}_1)\rightarrow 1$.

\subsubsection{Definition and Estimation of $\mathscr{Z}_2$}\label{sec.Z2 multi ordered}

The definition of $\mathscr{Z}_{2}$ relies on the testable implications in
\citet{kedagni2020generalized}. Under Conditions (i) and (ii) in Definition \ref{def.partial validity pairwise}, we can define $Y_{d}(z,z^{\prime})$
for every $d\in\mathcal{D}$ and every $(z,z^{\prime})\in\mathscr{Z}_{\bar{M}}$
such that $Y_{d}(z,z^{\prime})=Y_{dz}=Y_{dz^{\prime}}$ a.s.
We consider the case where $Y$ is continuous. Similar results can be obtained
easily when $Y$ is discrete. To avoid theoretical and computational complications, we introduce the following testable implications
that are slightly weaker than (and implied by) the original testable restrictions in \citet{kedagni2020generalized} (see Appendix \ref{sec.testable implication kedagni}).

Let $\mathcal{R}$ denote the collection of all subsets $C\subseteq\mathbb{R}$
such that $C=(a,b]$ or $C=(a,\infty)$ with $-\infty\le a<b<\infty$. For every $\mathcal{Z}_{(k,k^{\prime})}=(z_{k},z_{k^{\prime}})\in
\mathscr{Z}_{\bar{M}}$, every $A\in\mathcal{B}_{\mathbb{R}}$, every
$d\in\mathcal{D}$, and each $z\in\mathcal{Z}_{(k,k^{\prime})}$,
\[
\mathbb{P}\left(  Y\in A,D=d|Z=z\right)  \leq\mathbb{P}\left(  Y_{dz}\in
A|Z=z\right)  =\mathbb{P}\left(  Y_{d}(z_{k},z_{k^{\prime}})\in A\right)  ,
\]
which implies that
\begin{align}\label{eq.density inequality simplified}
	\max_{z\in\mathcal{Z}_{(k,k^{\prime})}}\mathbb{P}\left(  Y\in
	A,D=d|Z=z\right)  \leq\mathbb{P}\left(  Y_{d}(z_{k},z_{k^{\prime}})\in
	A\right)  .
\end{align}
Let $\mathscr{P}$ be a prespecified finite collection of partitions $P_{\mathbb{R}}$ of $\mathbb{R}$ such that
$P_{\mathbb{R}}=\{C_{1},\ldots,C_{N}\}$ for some $N$ with $C_{k}\in\mathcal{R}$ for all
$k$, $\cup_{k=1}^{N}C_{k}=\mathbb{R}$, and $C_{k}\cap C_{l}=\varnothing$ for
all $k\neq l$. Then we obtain the first condition:
\begin{align}\label{eq.first condition ordered multi simplified}
	\max_{P_{\mathbb{R}}\in\mathscr{P}}\max_{d\in\mathcal{D}}\sum_{A\in
		P_{\mathbb{R}}}\max_{z\in\mathcal{Z}_{(k,k^{\prime})}}\mathbb{P}\left(  Y\in
	A,D=d|Z=z\right)  \leq\max_{P_{\mathbb{R}}\in\mathscr{P}}\max_{d\in
		\mathcal{D}}\sum_{A\in P_{\mathbb{R}}}\mathbb{P}\left(  Y_{d}\left(
	z_{k},z_{k^{\prime}}\right)  \in A\right)  =1.
\end{align}
Also, for all $A_{1},\ldots,A_{J}\in\mathcal{B}_{\mathbb{R}}$,
\begin{align*}
	&  \mathbb{P}\left(  Y_{d_{1}}(z_{k},z_{k^{\prime}})\in A_{1},\ldots,Y_{d_{J}%
	}(z_{k},z_{k^{\prime}})\in A_{J}\right)  \\
	=&\,\min_{z\in\mathcal{Z}_{(k,k^{\prime})}}\mathbb{P}\left(  Y_{d_{1}}%
	(z_{k},z_{k^{\prime}})\in A_{1},\ldots,Y_{d_{J}}(z_{k},z_{k^{\prime}})\in
	A_{J}|Z=z\right)  \\
	=&\,  \min_{z\in\mathcal{Z}_{(k,k^{\prime})}}\sum_{j=1}^{J}\mathbb{P}\left(
	Y_{d_{1}}(z_{k},z_{k^{\prime}})\in A_{1},\ldots,Y_{d_{J}}(z_{k},z_{k^{\prime}%
	})\in A_{J},D=d_{j}|Z=z\right)  \\
	\leq&\,\min_{z\in\mathcal{Z}_{(k,k^{\prime})}}\sum_{j=1}^{J}\mathbb{P}\left(
	Y\in A_{j},D=d_{j}|Z=z\right)  .
\end{align*}
Let $P_{\mathbb{R}}^{1},\ldots,P_{\mathbb{R}}^{J}\in\mathscr{P}$. It follows that
\begin{align*}
	1 &  =\sum_{A_{1}\in P_{\mathbb{R}}^{1}}\cdots\sum_{A_{J}\in P_{\mathbb{R}%
		}^{J}}\mathbb{P}\left(  Y_{d_{1}}(z_{k},z_{k^{\prime}})\in A_{1}%
	,\ldots,Y_{d_{J}}(z_{k},z_{k^{\prime}})\in A_{J}\right)  \\
	&  \leq\sum_{A_{1}\in P_{\mathbb{R}}^{1}}\cdots\sum_{A_{J}\in P_{\mathbb{R}%
		}^{J}}\min_{z\in\mathcal{Z}_{(k,k^{\prime})}}\sum_{j=1}^{J}\mathbb{P}\left(
	Y\in A_{j},D=d_{j}|Z=z\right)  .
\end{align*}
Then we obtain the second condition:
\begin{align}\label{eq.second condition ordered multi simplified}
	\min_{P_{\mathbb{R}}^{1},\ldots,P_{\mathbb{R}}^{J}\in\mathscr{P}}\sum
	_{A_{1}\in P_{\mathbb{R}}^{1}}\cdots\sum_{A_{J}\in P_{\mathbb{R}}^{J}}%
	\min_{z\in\mathcal{Z}_{(k,k^{\prime})}}\sum_{j=1}^{J}\mathbb{P}\left(  Y\in
	A_{j},D=d_{j}|Z=z\right)  \geq1.
\end{align}
Next, for every $j$ and every $A_{j}\in\mathcal{B}_{\mathbb{R}}$,
\begin{align*}
	&  \mathbb{P}\left(  Y_{d_{j}}(z_{k},z_{k^{\prime}})\in A_{j}\right)  \\
	 =&\,\sum_{A_{1}\in P_{\mathbb{R}}^{1}}\cdots\sum_{A_{j-1}\in P_{\mathbb{R}%
		}^{j-1}}\sum_{A_{j+1}\in P_{\mathbb{R}}^{j+1}}\ldots\sum_{A_{J}\in
		P_{\mathbb{R}}^{J}}\mathbb{P}\left(  Y_{d_1}(z_{k},z_{k^{\prime}})\in
	A_{1},\ldots,Y_{d_J}(z_{k},z_{k^{\prime}})\in A_{J}\right)  \\
	 \leq&\,\sum_{A_{1}\in P_{\mathbb{R}}^{1}}\cdots\sum_{A_{j-1}\in P_{\mathbb{R}%
		}^{j-1}}\sum_{A_{j+1}\in P_{\mathbb{R}}^{j+1}}\ldots\sum_{A_{J}\in
		P_{\mathbb{R}}^{J}}\min_{z\in\mathcal{Z}_{(k,k^{\prime})}}\sum_{\xi=1}%
	^{J}\mathbb{P}\left(  Y\in A_{\xi},D=d_{\xi}|Z=z\right)  ,
\end{align*}
which, together with \eqref{eq.density inequality simplified}, implies the third condition:
\begin{align}\label{eq.third condition ordered multi simplified}
	\max_{P_{\mathbb{R}}^{1},\ldots,P_{\mathbb{R}}^{J}\in\mathscr{P}}\max
	_{j\in\left\{  1,\ldots,J\right\}  }\sup_{A_{j}\in\mathcal{R}
	}\bigg\{  \max_{z\in\mathcal{Z}_{(k,k^{\prime})}}\mathbb{P}(  Y\in
	A_j,D=d_{j}|Z=z)  -\varphi_{j}\left(  A_{j},\mathcal{Z}_{(k,k^{\prime})},P_{\mathbb{R}}^{1},\ldots,P_{\mathbb{R}}^{J}\right)  \bigg\}
	\leq0,
\end{align}
where%
\begin{align*}
	&\varphi_{j}\left(  A_{j},\mathcal{W},P_{\mathbb{R}}^{1},\ldots,P_{\mathbb{R}%
	}^{J}\right) \\
	=&\,\sum_{A_{1}\in P_{\mathbb{R}}^{1}}\cdots\sum_{A_{j-1}\in
		P_{\mathbb{R}}^{j-1}}\sum_{A_{j+1}\in P_{\mathbb{R}}^{j+1}}\ldots\sum
	_{A_{J}\in P_{\mathbb{R}}^{J}}\min_{z\in\mathcal{W}}\sum
	_{\xi=1}^{J}\mathbb{P}\left(  Y\in A_{\xi},D=d_{\xi}|Z=z\right)
\end{align*}
for all $\mathcal{W}\subseteq\mathcal{Z}$.

Next, we reformulate the testable implications in
\eqref{eq.first condition ordered multi simplified}--\eqref{eq.third condition ordered multi simplified} to define
$\mathscr{Z}_{2}$ and $\widehat{\mathscr{Z}_{2}}$. Define the function spaces
\begin{align}\label{def.function spaces KM multi}
	&  \mathcal{G}_{Z}=\left\{  1_{\mathbb{R}\times\mathbb{R}\times\left\{
		z_{k}\right\}  }:1\leq k\leq K\right\}  ,\mathcal{H}_{D}=\{1_{\mathbb{R}%
		\times\{d\}\times\mathbb{R}},d\in\mathcal{D}\},\mathcal{H}_{B}=\left\{
	1_{B\times\mathbb{R}\times\mathbb{R}}:B\in\mathcal{R}\right\}
	,\nonumber\\
	&  \text{and }\bar{\mathcal{H}}_{B}=\left\{  1_{B\times\mathbb{R}%
		\times\mathbb{R}}:B\text{ is a closed, open, or half-closed interval in
	}\mathbb{R}\right\}  .
\end{align}
Let $P$ and $\widehat{P}$ be defined as in Section \ref{sec.estimation validity set}. Define a map $\psi:\bar{\mathcal{H}}_{B}\times\mathcal{H}%
_{D}\times\mathcal{G}_{Z}\rightarrow\mathbb{R}$ such that
\[
\psi(h,f,g)=\frac{P(h\cdot f\cdot g)}{P(g)}%
\]
for every $(h,f,g)\in\bar{\mathcal{H}}_{B}\times\mathcal{H}_{D}\times
\mathcal{G}_{Z}$. 
Moreover, define a map $\mathbb{H}$ such that if $P_{\mathbb{R}}\in\mathscr{P}$ with $P_{\mathbb{R}}=\{C_1,\ldots,C_N\}$ and $C_{k}\in\mathcal{R}$ for all $k\in\{1,\ldots,N\}$, then 
\begin{align}
    \mathbb{H}(P_{\mathbb{R}})=\{1_{C\times\mathbb{R}\times\mathbb{R}}:C\in P_{\mathbb{R}}\}.
\end{align}
Let $\mathrm{P}\left(  \mathcal{G}_{Z}\right)  $ be the collection of all nonempty subsets of $\mathcal{G}_Z$.
Then for every $\mathcal{G}_{S}\in\mathrm{P}\left(  \mathcal{G}_{Z}\right)  $,
define
\[
\psi_{1}\left(  \mathcal{G}_{S}\right)  =\max_{P_{\mathbb{R}}\in
	\mathscr{P}}\max_{f\in\mathcal{H}_{D}}\sum_{h\in\mathbb{H}\left(
	P_{\mathbb{R}}\right)  }\max_{g\in\mathcal{G}_{S}}\psi\left(  h,f,g\right)
-1,
\]%
\[
\psi_{2}\left(  \mathcal{G}_{S}\right)  =1-\min_{P_{\mathbb{R}}^{1}%
	,\ldots,P_{\mathbb{R}}^{J}\in\mathscr{P}}\sum_{h_{1}\in\mathbb{H}\left(
	P_{\mathbb{R}}^{1}\right)  }\cdots\sum_{h_{J}\in\mathbb{H}\left(
	P_{\mathbb{R}}^{J}\right)  }\min_{g\in\mathcal{G}_{S}}\sum_{j=1}^{J}%
\psi\left(  h_{j},f_{j},g\right)  ,
\]
and%
\[
\psi_{3}\left(  \mathcal{G}_{S}\right)  =\max_{P_{\mathbb{R}}^{1}%
	,\ldots,P_{\mathbb{R}}^{J}\in\mathscr{P}}\max_{j\in\left\{  1,\ldots
	,J\right\}  }\sup_{h_{j}\in\mathcal{H}_{B}}\left\{  \max_{g\in\mathcal{G}_{S}%
}\psi\left(  h_{j},f_{j},g\right)  -\tilde{\varphi}_{j}\left(  h_{j}%
,\mathcal{G}_{S},P_{\mathbb{R}}^{1},\ldots,P_{\mathbb{R}}^{J}\right)
\right\}  ,
\]
where $f_{j}=1_{\mathbb{R}\times\left\{  d_{j}\right\}  \times\mathbb{R}}$
and
\begin{align*}
&\tilde{\varphi}_{j}\left(  h_{j},\mathcal{G}_{S},P_{\mathbb{R}}^{1},\dots
,P_{\mathbb{R}}^{J}\right) \\
=&\,\sum_{h_{1}\in\mathbb{H}\left(  P_{\mathbb{R}%
	}^{1}\right)  }\cdots\sum_{h_{j-1}\in\mathbb{H}\left(  P_{\mathbb{R}}%
	^{j-1}\right)  }\sum_{h_{j+1}\in\mathbb{H}\left(  P_{\mathbb{R}}^{j+1}\right)
}\cdots\sum_{h_{J}\in\mathbb{H}\left(  P_{\mathbb{R}}^{J}\right)  }\min
_{g\in\mathcal{G}_{S}}\sum_{\xi=1}^{J}\psi\left(  h_{\xi},f_{\xi},g\right)  .
\end{align*}
For every $\mathcal{Z}_{(k,k^{\prime})}\in\mathscr{Z}$, let
$\mathcal{G}(\mathcal{Z}_{(k,k^{\prime})})=\{  1_{\mathbb{R}%
	\times\mathbb{R}\times\{z_{k}\}},1_{\mathbb{R}\times\mathbb{R}\times
	\{z_{k^{\prime}}\}}\}$.
The conditions in \eqref{eq.first condition ordered multi simplified}--\eqref{eq.third condition ordered multi simplified}
imply that for every $\mathcal{Z}_{(k,k^{\prime})}\in\mathscr{Z}_{\bar{M}}$, $\psi_{l}(\mathcal{G}(\mathcal{Z}_{(k,k^{\prime})}))\leq0$ for all
$l\in\{1,2,3\}$. Thus, we define $\mathscr{Z}_{2}$ by
\[
\mathscr{Z}_{2}=\left\{  \mathcal{Z}_{(k,k^{\prime})}\in\mathscr{Z}:\psi
_{l}(\mathcal{G}(\mathcal{Z}_{(k,k^{\prime})}))\leq0,l\in\{1,2,3\}\right\}  .
\]
Let $\widehat{\psi}:\bar{\mathcal{H}}_{B}\times\mathcal{H}
_{D}\times\mathcal{G}_{Z}\rightarrow\mathbb{R}$ be the sample analog of $\psi$ such that
\[
\widehat{\psi}(h,f,g)=\frac{\widehat{P}(h\cdot f\cdot g)}{\widehat{P}(g)}%
\]
for every $(h,f,g)\in\bar{\mathcal{H}}_{B}\times\mathcal{H}_{D}\times
\mathcal{G}_{Z}$. Let $\widehat{\psi}_l$ be the sample analog of $\psi_l$ for $l\in\{1,2,3\}$, which replaces $\psi$ in $\psi_l$ by $\widehat{\psi}$. 
We define the estimator $\widehat{\mathscr{Z}_{2}}$ for $\mathscr{Z}_{2}$ by
\[
\widehat{\mathscr{Z}_{2}}=\left\{  \mathcal{Z}_{(k,k^{\prime})}\in
\mathscr{Z}:\sqrt{T_{n}}\widehat{\psi}_{l}(\mathcal{G}(\mathcal{Z}%
_{(k,k^{\prime})}))\leq t_{n},l\in\{1,2,3\}\right\}  ,
\]
where $T_{n}=n\cdot \prod_{k=1}^K\widehat{P}(1_{\mathbb{R}\times \mathbb{R}
	\times \{z_k\}})$, $t_{n}\rightarrow\infty$, and $t_{n}/\sqrt{n}\rightarrow0$ as
$n\rightarrow\infty$.

To establish consistency of $\widehat{\mathscr{Z}_{2}}$, we state and prove some auxiliary lemmas.

\begin{lemma}\label{lemma.VC class}
	The function space $\mathcal{H}_B$ is a VC class with VC index $V(\mathcal{H}_B)=3$.
\end{lemma}

\begin{proof}[Proof of Lemma \ref{lemma.VC class}]
	The proof closely follows the strategy of the proof of Lemma C.2 of \citet{sun2021ivvalidity}.
\end{proof}

We define 
\begin{align}\label{eq.V}
	\mathcal{V}=\{h\cdot f\cdot g:h\in\bar{\mathcal{H}}_{B},f\in\mathcal{H}%
	_{D},g\in\mathcal{G}_{Z}\}	\text{ and }\mathcal{\tilde{V}}=\mathcal{V}%
	\cup\mathcal{G}_{Z}.
\end{align}

\begin{lemma}\label{lemma.Donsker and Glivenko-Cantelli}
	The function space $\tilde{\mathcal{V}}$ defined in \eqref{eq.V} is Donsker and pre-Gaussian uniformly in $Q\in\mathcal{P}$, and $\tilde{\mathcal{V}}$ is Glivenko--Cantelli uniformly in $Q\in\mathcal{P}$.
\end{lemma} 

\begin{proof}[Proof of Lemma \ref{lemma.Donsker and Glivenko-Cantelli}]
	The proof closely follows the strategies of the proofs of Lemmas C.5 and C.6 of \citet{sun2021ivvalidity}.
\end{proof}

The following proposition establishes consistency of $\widehat{\mathscr{Z}_{2}}$. 
\begin{proposition}
	\label{prop.consistent G hat pairwise Z2} Under Assumption \ref{ass.iid data}, $\mathbb{P}(\widehat{\mathscr{Z}_{2}}=\mathscr{Z}_{2}%
	)\rightarrow1$.
\end{proposition}

\begin{proof}[Proof of Proposition \ref{prop.consistent G hat pairwise Z2}]
	Let $\mathcal{C}_{2}$ be the set of all $\mathcal{G}(\mathcal{Z}
	_{(k,k^{\prime})})$ with $\mathcal{Z}_{(k,k^{\prime})}\in \mathscr{Z}_{2}$ and
	$\widehat{\mathcal{C}_{2}}$ be the set of all $\mathcal{G}(\mathcal{Z}
	_{(k,k^{\prime})})$ with $\mathcal{Z}_{(k,k^{\prime})}\in \widehat{\mathscr{Z}_{2}}$.
	First, we have that
	\begin{align*}
		\mathbb{P}\left(  \mathcal{C}_{2}\setminus\widehat{\mathcal{C}_{2}}%
		\neq\varnothing\right)  \leq &  \,\mathbb{P}\left(  \max_{\mathcal{G}%
			_{S}\in\mathcal{C}_{2}\setminus\widehat{\mathcal{C}_{2}}}\sqrt{T_{n}%
		}\left\{  \widehat{\psi}_{1}\left(  \mathcal{G}_{S}\right)  -\psi_{1}\left(
		\mathcal{G}_{S}\right)  \right\}  >t_{n}\right)  \\
		&  +\mathbb{P}\left(  \max_{\mathcal{G}_{S}\in\mathcal{C}_{2}%
			\setminus\widehat{\mathcal{C}_{2}}}\sqrt{T_{n}}\left\{  \widehat{\psi}%
		_{2}\left(  \mathcal{G}_{S}\right)  -\psi_{2}\left(  \mathcal{G}_{S}\right)
		\right\}  >t_{n}\right)  \\
		&  +\mathbb{P}\left(  \max_{\mathcal{G}_{S}\in\mathcal{C}_{2}%
			\setminus\widehat{\mathcal{C}_{2}}}\sqrt{T_{n}}\left\{  \widehat{\psi}%
		_{3}\left(  \mathcal{G}_{S}\right)  -\psi_{3}\left(  \mathcal{G}_{S}\right)
		\right\}  >t_{n}\right)  .
	\end{align*}
	By Theorem 1.3.6 (continuous mapping) of \citet{van1996weak},
	\begin{align*}
		&  \max_{\mathcal{G}_{S}\in\mathcal{C}_{2}}\sqrt{T_{n}}\left\vert
		\max_{P_{\mathbb{R}}\in\mathscr{P}}\max_{f\in\mathcal{H}_{D}}\sum
		_{h\in\mathbb{H}\left(  P_{\mathbb{R}}\right)  }\max_{g\in\mathcal{G}_{S}%
		}\widehat{\psi}\left(  h,f,g\right)  -\max_{P_{\mathbb{R}}\in\mathscr{P}}%
		\max_{f\in\mathcal{H}_{D}}\sum_{h\in\mathbb{H}\left(  P_{\mathbb{R}}\right)
		}\max_{g\in\mathcal{G}_{S}}\psi\left(  h,f,g\right)  \right\vert \\
		&  \leq\max_{\mathcal{G}_{S}\in\mathcal{C}_{2}}\max_{P_{\mathbb{R}}%
			\in\mathscr{P}}\max_{f\in\mathcal{H}_{D}}\sum_{h\in\mathbb{H}\left(
			P_{\mathbb{R}}\right)  }\max_{g\in\mathcal{G}_{S}}\sqrt{T_{n}}\left\vert
		\widehat{\psi}\left(  h,f,g\right)  -\psi\left(  h,f,g\right)  \right\vert
		\leadsto\mathbb{G}_{1}%
	\end{align*}
	for some random element $\mathbb{G}_{1}$. Then it follows that
	\begin{align*}
		\mathbb{P}\left(  \max_{\mathcal{G}_{S}\in\mathcal{C}_{2}%
			\setminus\widehat{\mathcal{C}_{2}}}\sqrt{T_{n}}\left\{  \widehat{\psi}%
		_{1}\left(  \mathcal{G}_{S}\right)  -\psi_{1}\left(  \mathcal{G}_{S}\right)
		\right\}  >t_{n}\right)  
		&\leq  \mathbb{P}\left(  \max_{\mathcal{G}_{S}\in\mathcal{C}_{2}}%
		\sqrt{T_{n}}\left\vert \widehat{\psi}_{1}\left(  \mathcal{G}_{S}\right)
		-\psi_{1}\left(  \mathcal{G}_{S}\right)  \right\vert >t_{n}\right)
		\\&\rightarrow0.
	\end{align*}
	Similarly, we have that
	\[
	\mathbb{P}\left(  \max_{\mathcal{G}_{S}\in\mathcal{C}_{2}\setminus
		\widehat{\mathcal{C}_{2}}}\sqrt{T_{n}}\left\{  \widehat{\psi}_{2}\left(
	\mathcal{G}_{S}\right)  -\psi_{2}\left(  \mathcal{G}_{S}\right)  \right\}
	>t_{n}\right)  \rightarrow0
	\]
	and
	\[
	\mathbb{P}\left(  \max_{\mathcal{G}_{S}\in\mathcal{C}_{2}\setminus
		\widehat{\mathcal{C}_{2}}}\sqrt{T_{n}}\left\{  \widehat{\psi}_{3}\left(
	\mathcal{G}_{S}\right)  -\psi_{3}\left(  \mathcal{G}_{S}\right)  \right\}
	>t_{n}\right)  \rightarrow0.
	\]
	Thus, $\mathbb{P(}\mathcal{C}_{2}\setminus\widehat{\mathcal{C}_{2}}%
	\neq\varnothing)\rightarrow0$.
	
	Next, let $\mathcal{C}$ be the set of all $\mathcal{G}(\mathcal{Z}_{(k,k')})$ with $\mathcal{Z}_{(k,k')}\in\mathscr{Z}$. Clearly, $\mathcal{C}$ is a finite set. If $\mathcal{C}\setminus\mathcal{C}_2\neq \varnothing$, there is some $\delta>0$ such that
	$ \min_{\mathcal{G}_{S}\in\mathcal{C}\setminus\mathcal{C}_{2}}\max_{l\in\{1,2,3\}}\psi_{l}\left(\mathcal{G}_{S}\right) >\delta$.
Then we have that
	\begin{align*}
		\mathbb{P}\left(  \widehat{\mathcal{C}_{2}}\setminus\mathcal{C}_{2}%
		\neq\varnothing\right)  \leq &  \,\mathbb{P}\left(  \max_{\mathcal{G}%
			_{S}\in\widehat{\mathcal{C}_{2}}\setminus\mathcal{C}_{2}}\psi_{1}\left(
		\mathcal{G}_{S}\right)  >\delta,\max_{\mathcal{G}_{S}\in\widehat
			{\mathcal{C}_{2}}\setminus\mathcal{C}_{2}}\sqrt{T_{n}}\widehat{\psi}%
		_{1}\left(  \mathcal{G}_{S}\right)  \leq t_{n}\right)  \\
		&  +\mathbb{P}\left(  \max_{\mathcal{G}_{S}\in\widehat{\mathcal{C}_{2}%
			}\setminus\mathcal{C}_{2}}\psi_{2}\left(  \mathcal{G}_{S}\right)  >\delta
		,\max_{\mathcal{G}_{S}\in\widehat{\mathcal{C}_{2}}\setminus\mathcal{C}%
			_{2}}\sqrt{T_{n}}\widehat{\psi}_{2}\left(  \mathcal{G}_{S}\right)  \leq
		t_{n}\right)  \\
		&  +\mathbb{P}\left(  \max_{\mathcal{G}_{S}\in\widehat{\mathcal{C}_{2}%
			}\setminus\mathcal{C}_{2}}\psi_{3}\left(  \mathcal{G}_{S}\right)  >\delta
		,\max_{\mathcal{G}_{S}\in\widehat{\mathcal{C}_{2}}\setminus\mathcal{C}%
			_{2}}\sqrt{T_{n}}\widehat{\psi}_{3}\left(  \mathcal{G}_{S}\right)  \leq
		t_{n}\right)  .
	\end{align*}
	By Lemma \ref{lemma.Donsker and Glivenko-Cantelli} and Lemma 1.9.3 of \citet{van1996weak},
	$
	\Vert \widehat{\psi}  -\psi	\Vert_{\infty} \rightarrow0
	$
	almost uniformly. Then we have that%
	\begin{align*}
		\,\max_{\mathcal{G}_{S}\in
			\mathcal{C}} &  \left\vert \widehat{\psi}_{1}\left(  \mathcal{G}%
		_{S}\right)  -\psi_{1}\left(  \mathcal{G}_{S}\right)  \right\vert \\
		= &  \,\max_{\mathcal{G}_{S}\in
			\mathcal{C}}\left\vert \max_{P_{\mathbb{R}}\in\mathscr{P}}\max
		_{f\in\mathcal{H}_{D}}\sum_{h\in\mathbb{H}\left(  P_{\mathbb{R}}\right)  }%
		\max_{g\in\mathcal{G}_{S}}\widehat{\psi}\left(  h,f,g\right)  -\max
		_{P_{\mathbb{R}}\in\mathscr{P}}\max_{f\in\mathcal{H}_{D}}\sum_{h\in
			\mathbb{H}\left(  P_{\mathbb{R}}\right)  }\max_{g\in\mathcal{G}_{S}}%
		\psi\left(  h,f,g\right)  \right\vert \\
		\leq &  \max_{\mathcal{G}_{S}\in
			\mathcal{C}}\max_{P_{\mathbb{R}}\in\mathscr{P}}\max_{f\in\mathcal{H}_{D}%
		}\sum_{h\in\mathbb{H}\left(  P_{\mathbb{R}}\right)  }\max_{g\in\mathcal{G}%
			_{S}}\left\vert \widehat{\psi}\left(  h,f,g\right)  -\psi\left(  h,f,g\right)
		\right\vert \rightarrow0
	\end{align*}
	almost uniformly. Similarly, it follows that%
	\[
	\max_{\mathcal{G}_{S}\in\mathcal{C}
	}\left\vert \widehat{\psi}_{2}\left(  \mathcal{G}_{S}\right)  -\psi_{2}\left(
	\mathcal{G}_{S}\right)  \right\vert \rightarrow0 \text{ and } \max_{\mathcal{G}_{S}\in\mathcal{C}
	}\left\vert \widehat{\psi}_{3}\left(  \mathcal{G}_{S}\right)  -\psi_{3}\left(
	\mathcal{G}_{S}\right)  \right\vert \rightarrow0
	\]
	almost uniformly. So for every $\varepsilon>0$, there is a measurable set $A\subseteq\Omega$ with $\mathbb{P}(A)\ge 1-\varepsilon$ such that for all large $n$,
	\[
	\max_{l\in\{1,2,3\}}\max_{\mathcal{G}_{S}\in\mathcal{C}
	}\left\vert \widehat{\psi}_{l}\left(  \mathcal{G}_{S}\right)  -\psi_{l}\left(
	\mathcal{G}_{S}\right)  \right\vert\le\frac{\delta}{2}
	\]
	uniformly on $A$.
	Thus, it follows that for every $l\in\{1,2,3\}$,
	\begin{align*}
		&\lim_{n\to\infty}\mathbb{P}\left(  \max_{\mathcal{G}_{S}\in\widehat{\mathcal{C}_{2}%
			}\setminus\mathcal{C}_{2}}\psi_{l}\left(  \mathcal{G}_{S}\right)  >\delta
		,\max_{\mathcal{G}_{S}\in\widehat{\mathcal{C}_{2}}\setminus\mathcal{C}%
			_{2}}\sqrt{T_{n}}\widehat{\psi}_{l}\left(  \mathcal{G}_{S}\right)  \leq
		t_{n}\right)\\
		\leq &\, \lim_{n\to\infty}\mathbb{P}\left( \left\{ \max_{\mathcal{G}_{S}\in\widehat{\mathcal{C}_{2}%
			}\setminus\mathcal{C}_{2}}\psi_{l}\left(  \mathcal{G}_{S}\right)  >\delta
		,\max_{\mathcal{G}_{S}\in\widehat{\mathcal{C}_{2}}\setminus\mathcal{C}%
			_{2}}\sqrt{T_{n}}\widehat{\psi}_{l}\left(  \mathcal{G}_{S}\right)  \leq
		t_{n}\right\}\cap A \right)+\mathbb{P}(A^c)\\
		 \leq &\, \lim_{n\to\infty} \mathbb{P}\left(  \frac{\delta}{2}\leq\max_{\mathcal{G}_{S}%
			\in\widehat{\mathcal{C}_{2}}\setminus\mathcal{C}_{2}}\widehat{\psi}
		_{l}\left(  \mathcal{G}_{S}\right)  \leq\frac{t_{n}}{\sqrt{T_{n}}}\right)+\varepsilon=\varepsilon.
	\end{align*}
	Since $\varepsilon$ can be arbitrarily small, we have that 
	\begin{align*}
	    \mathbb{P}\left(  \max_{\mathcal{G}_{S}\in\widehat{\mathcal{C}_{2}%
			}\setminus\mathcal{C}_{2}}\psi_{l}\left(  \mathcal{G}_{S}\right)  >\delta
		,\max_{\mathcal{G}_{S}\in\widehat{\mathcal{C}_{2}}\setminus\mathcal{C}%
			_{2}}\sqrt{T_{n}}\widehat{\psi}_{l}\left(  \mathcal{G}_{S}\right)  \leq
		t_{n}\right)\to0.
	\end{align*}
	This implies $\mathbb{P(}\widehat{\mathcal{C}_{2}}\setminus\mathcal{C}%
	_{2}\neq\varnothing)\rightarrow0$. Thus,%
	\[
	\mathbb{P}\left(  \widehat{\mathcal{C}_{2}}\neq\mathcal{C}_{2}\right)
	\leq\mathbb{P}\left(  \widehat{\mathcal{C}_{2}}\setminus\mathcal{C}_{2}%
	\neq\varnothing\right)  +\mathbb{P}\left(  \mathcal{C}_{2}\setminus
	\widehat{\mathcal{C}_{2}}\neq\varnothing\right)  \rightarrow0.
	\]
\end{proof}

\subsection{Partially Valid Instruments for Multivalued Ordered Treatments}\label{sec.ordered treatment partial} 

Here we extend the analysis in Section \ref{sec.partially valid instruments binary D} to multivalued ordered treatments. We follow the setup in Section \ref{sec.ordered treatment}. Consider the following generalized version of Definition \ref{def.partially validy instrument binary D}.
\begin{definition}
	\label{def.partial validity} Suppose the instrument $Z$ is pairwise valid for the (multivalued ordered) treatment $D$ with the largest validity pair set $\mathscr{Z}_{\bar{M}}$. If there is a validity pair set $$\mathscr{Z}_{{M}}=\{(z_{k_1},z_{k_2}),(z_{k_2},z_{k_3}),\ldots,(z_{k_{M-1}},z_{k_M})\}$$ for some $M>0$, then the instrument $Z$ is called a \textbf{partially valid instrument} for the treatment $D$. The set $\mathcal{Z}_M=\{z_{k_1},\ldots,z_{k_M}\}$ is called a \textbf{validity value set} of $Z$.
\end{definition}

\begin{assumption}\label{ass.first stage partial}
	The validity value set $\mathcal{Z}_M$ satisfies that 
	\begin{align}
	    E[g(Z_i)D_i|Z_i\in\mathcal{Z}_{M}]-E[D_i|Z_i\in\mathcal{Z}_{M}]\cdot E[g(Z_i)|Z_i\in\mathcal{Z}_{M}]\neq0.
	\end{align}
\end{assumption}

Suppose that we have access to a consistent estimator $\widehat{\mathcal{Z}_0}$ of the  validity value set $\mathcal{Z}_M$, that is, $\mathbb{P}(\widehat{\mathcal{Z}_{0}}=\mathcal{Z}_{M})\rightarrow 1$. Then we can use $\widehat{\mathcal{Z}_{0}}$ to construct a VSIV estimator, $\widehat{\theta}_{1}$,  for a weighted average of ACRs based on model \eqref{eq.VSIV estimation binary D}, where $D$ is now a multivalued ordered treatment. The following theorem presents the asymptotic properties of the VSIV estimator, generalizing Theorem \ref{thm.IV estimator asymptotics binary D}. Theorem \ref{thm.IV estimator asymptotics} is an extension of Theorem 2 of \citet{imbens1994identification} and Theorem 2 of \citet{angrist1995two} to the case where the instrument is partially but not fully valid.
\begin{theorem}\label{thm.IV estimator asymptotics}
	Suppose that the instrument $Z$ is partially valid for the treatment $D$ as defined in Definition \ref{def.partial validity} with a validity value set $\mathcal{Z}_M=\{z_{k_1},\dots,z_{k_M} \}$, and that the estimator $\widehat{\mathcal{Z}_0}$ for $\mathcal{Z}_M$ satisfies  $\mathbb{P}(\widehat{\mathcal{Z}_{0}}=\mathcal{Z}_{M})\rightarrow 1$. Under Assumptions \ref{ass.iid data} and \ref{ass.first stage partial}, it follows that $\widehat{\theta}_{1}\overset{p}\rightarrow\theta_1 $, where 
	\begin{align*}
		\theta_{1}=\frac{E\left[  g\left(  Z_{i}\right)
			Y_{i}|  Z_{i}\in\mathcal{Z}_{M}  \right]  -E\left[
			Y_{i}|  Z_{i}\in\mathcal{Z}_{M}  \right]  E\left[  g\left(
			Z_{i}\right)|  Z_{i}\in\mathcal{Z}_{M}  \right]  }{E\left[
			g\left(  Z_{i}\right)  D_{i}|  Z_{i}\in\mathcal{Z}_{M}
			\right]  -E\left[  D_{i}|  Z_{i}\in\mathcal{Z}_{M}  \right]
			E\left[  g\left(  Z_{i}\right)|  Z_{i}\in\mathcal{Z}_{M} \right]  }.
	\end{align*}
	Also, $\sqrt{n}( \widehat{\theta}_{1}-\theta_1 ) \overset{d}\rightarrow N\left( 0,\Sigma_1 \right) $, where 
	$\Sigma_1$ is provided in \eqref{eq.asymptotic beta1}. In addition, the quantity $\theta _{1}$ can be interpreted as
	the weighted average of $\{ \beta _{k_{2},k_{1}},\ldots ,\beta_{k_{M},k_{M-1}}\} $ defined in \eqref{eq.beta}. Specifically, $\theta _{1}=\sum_{m=1}^{M-1}\mu
	_{m}\beta _{k_{m+1},k_{m}}$ with 
	\begin{align*}
		\mu _{m}=
		\frac{\left[ p\left( z_{k_{m+1}}\right) -p\left( z_{k_{m}}\right) 
			\right] \sum_{l=m}^{M-1}\mathbb{P}\left( Z_{i}=z_{k_{l+1}}|Z_i\in\mathcal{Z}_M\right) \left\{
			g\left( z_{k_{l+1}}\right) -E\left[ g\left( Z_{i}\right) |Z_i\in\mathcal{Z}_M \right] \right\} }{\sum_{l=1}^{M}\mathbb{P}\left(
			Z_{i}=z_{k_{l}}|Z_i\in\mathcal{Z}_M\right) p\left( z_{k_{l}}\right) \left\{ g\left(
			z_{k_{l}}\right) -E\left[ g\left( Z_{i}\right)|Z_i\in\mathcal{Z}_M \right] \right\} },
	\end{align*}
	$p\left( z_{k}\right) =E\left[ D_{i}|Z_{i}=z_{k}\right] $, and $\sum_{m=1}^{M-1}\mu _{m}=1$.
\end{theorem}

\begin{proof}[Proof of Theorem \ref{thm.IV estimator asymptotics}]
	By the formula of the VSIV estimator in \eqref{eq.VSIV estimator binary D},%
	\begin{align*}
		\widehat{\theta}_{1}  =\frac{\frac{n_z}{n}\frac{1}{n}\sum_{i=1}^{n}g\left(  Z_{i}\right)  Y_{i}1\left\{
			Z_{i}\in\widehat{\mathcal{Z}_0}\right\}  -\bar{Y}_{\widehat{\mathcal{Z}_0%
			}}\frac{1}{n}\sum_{i=1}^{n}g\left(  Z_{i}\right)  1\left\{  Z_{i}\in
			\widehat{\mathcal{Z}_0}\right\}  }{\frac{n_z}{n}\frac{1}{n}\sum_{i=1}^{n}g\left(
			Z_{i}\right)  D_{i}1\left\{  Z_{i}\in\widehat{\mathcal{Z}_0}\right\}
			-\bar{D}_{\widehat{\mathcal{Z}_0}}\frac{1}{n}\sum_{i=1}^{n}g\left(
			Z_{i}\right)  1\left\{  Z_{i}\in\widehat{\mathcal{Z}_0}\right\}  },
	\end{align*}
	where%
	\[
	\bar{Y}_{\widehat{\mathcal{Z}_0}}=\frac{1}{n}\sum_{i=1}^{n}Y_{i}1\left\{
	Z_{i}\in\widehat{\mathcal{Z}_0}\right\}  \text{ and }\bar{D}_{\widehat
		{\mathcal{Z}_0}}=\frac{1}{n}\sum_{i=1}^{n}D_{i}1\left\{  Z_{i}\in
	\widehat{\mathcal{Z}_0}\right\}  .
	\]
	We first have%
	\begin{align*}
		&  \frac{1}{n}\sum_{i=1}^{n}g\left(  Z_{i}\right)  Y_{i}1\left\{  Z_{i}%
		\in\widehat{\mathcal{Z}_0}\right\}  \\
		= &\, \frac{1}{n}\sum_{i=1}^{n}g\left(  Z_{i}\right)  Y_{i}1\left\{  Z_{i}%
		\in\mathcal{Z}_M\right\}  +\left[  \frac{1}{n}\sum_{i=1}^{n}g\left(
		Z_{i}\right)  Y_{i}\left\{  1\left\{  Z_{i}\in\widehat{\mathcal{Z}_0%
		}\right\}  -1\left\{  Z_{i}\in\mathcal{Z}_M\right\}  \right\}  \right]
	\end{align*}
	with%
	\begin{align*}
		\left\vert \frac{1}{n}\sum_{i=1}^{n}g\left(  Z_{i}\right)  Y_{i}\left\{
		1\left\{  Z_{i}\in\widehat{\mathcal{Z}_0}\right\}  -1\left\{  Z_{i}%
		\in\mathcal{Z}_M\right\}  \right\}  \right\vert \leq\frac{1}{n}\sum_{i=1}^{n}\left\vert g\left(  Z_{i}\right)
		Y_{i}\right\vert 1\left\{  \widehat{\mathcal{Z}_0}\neq\mathcal{Z}
		_{M}\right\}  .
	\end{align*}
	Since $n^{-1}\sum_{i=1}^n\left\vert g\left(  Z_{i}\right)  Y_{i}\right\vert
	 \overset{p}\rightarrow E\left[  \left\vert g\left(  Z_{i}\right)  Y_{i}\right\vert
	\right]  $ and for every small $\varepsilon>0$,
	\[
	\mathbb{P}\left(  1\left\{  \widehat{\mathcal{Z}_0}\neq\mathcal{Z}%
	_{M}\right\}  >\varepsilon\right)  =\mathbb{P}\left(  \widehat{\mathcal{Z}%
		_{0}}\neq\mathcal{Z}_M\right)  \rightarrow0,
	\]
	we have that%
	\begin{align*}
		\frac{1}{n}\sum_{i=1}^{n}g\left(  Z_{i}\right)  Y_{i}1\left\{  Z_{i}%
		\in\widehat{\mathcal{Z}_0}\right\}   &  =\frac{1}{n}\sum_{i=1}^{n}g\left(
		Z_{i}\right)  Y_{i}1\left\{  Z_{i}\in\mathcal{Z}_M\right\}  +o_{p}\left(
		1\right)  \\
		&   \overset{p}\rightarrow E\left[  g\left(  Z_{i}\right)  Y_{i}1\left\{  Z_{i}%
		\in\mathcal{Z}_M\right\}  \right]  .
	\end{align*}
	Recall that $n_z=\sum_{i=1}^{n}1\{Z_i\in\widehat{\mathcal{Z}_0}\}$. Then we can show that $n_z/n\overset{p}\to \mathbb{P}(Z_i\in\mathcal{Z}_M)$ as $n\to\infty$. 
	Similarly, we have that $\bar{Y}_{\widehat{\mathcal{Z}_0}}\overset{p}\rightarrow
	E\left[  Y_{i}1\left\{  Z_{i}\in\mathcal{Z}_M\right\}  \right]  $,
	$\bar{D}_{\widehat{\mathcal{Z}_0}} \overset{p}\rightarrow E\left[  D_{i}1\left\{
	Z_{i}\in\mathcal{Z}_M\right\}  \right]  $, $n^{-1}\sum_{i=1}^n g\left(
	Z_{i}\right)  1\{Z_{i}\in\widehat{\mathcal{Z}_0}\} \overset{p}\rightarrow E\left[
	g\left(  Z_{i}\right)  1\left\{  Z_{i}\in\mathcal{Z}_M\right\}  \right]  $,
	and $n^{-1}\sum_{i=1}^n g\left(  Z_{i}\right)  D_{i}1\{Z_{i}\in\widehat{\mathcal{Z}%
		_{0}}\} \overset{p}\rightarrow E\left[  g\left(  Z_{i}\right)  D_{i}1\left\{  Z_{i}%
	\in\mathcal{Z}_M\right\}  \right]  $. Thus, it follows that%
	\[
	\widehat{\theta}_{1} \overset{p}\rightarrow \frac{ \frac{E\left[  g\left(  Z_{i}\right)
			Y_{i}1\left\{  Z_{i}\in\mathcal{Z}_M\right\}  \right]}{\mathbb{P}(Z_i\in\mathcal{Z}_M)}  -\frac{E\left[
			Y_{i}1\left\{  Z_{i}\in\mathcal{Z}_M\right\}  \right]}{\mathbb{P}(Z_i\in\mathcal{Z}_M)}   \frac{E\left[  g\left(
			Z_{i}\right)  1\left\{  Z_{i}\in\mathcal{Z}_M\right\}  \right]}{\mathbb{P}(Z_i\in\mathcal{Z}_M)}   }{\frac{E\left[
			g\left(  Z_{i}\right)  D_{i}1\left\{  Z_{i}\in\mathcal{Z}_M\right\}
			\right]}{\mathbb{P}(Z_i\in\mathcal{Z}_M)}   -\frac{E\left[  D_{i}1\left\{  Z_{i}\in\mathcal{Z}_M\right\}  \right]}{\mathbb{P}(Z_i\in\mathcal{Z}_M)} 
		\frac{E\left[  g\left(  Z_{i}\right)  1\left\{  Z_{i}\in\mathcal{Z}_M\right\}
			\right]}{\mathbb{P}(Z_i\in\mathcal{Z}_M)}   }=\theta_{1}.
	\]

	Next, we derive the asymptotic distribution of $\sqrt{n}(\widehat{\theta}_{1}%
	-\theta_{1})$. Define a function $f:\mathbb{R}^{6}\rightarrow\mathbb{R}$ by
	\[
	f\left(  x\right)  =\frac{x_{1}/x_{6}-x_{2}x_{3}/x_{6}^{2}}{x_{4}/x_{6}%
		-x_{5}x_{3}/x_{6}^{2}}
	\]
	for every $x\in\mathbb{R}^{6}$ with $x=\left(  x_{1},x_{2},x_{3},x_{4},x_{5},x_6\right)  ^{T}$ such that $f(x)$ is well defined. We can obtain the gradient
	of $f$, denoted $f^{\prime}$, by $f^{\prime}\left(  x\right)  =\left(
	f_{1}^{\prime}\left(  x\right)  ,f_{2}^{\prime}\left(  x\right)
	,f_{3}^{\prime}\left(  x\right)  ,f_{4}^{\prime}\left(  x\right)
	,f_{5}^{\prime}\left(  x\right)  ,f_{6}^{\prime}\left(  x\right)  \right)
	^{T}$, where%
	\begin{align*}
		f_{1}^{\prime}\left(  x\right)   &  =\frac{x_{6}}{x_{4}x_{6}-x_{5}
			x_{3}},f_{2}^{\prime}\left(  x\right)  =\frac{-x_{3}
		}{x_{4}x_{6}-x_{5}x_{3}},f_{3}^{\prime}\left(  x\right)
		=\frac{-x_{2}x_{4}x_{6}+x_{5}x_{1}x_{6}}{\left(  x_{4}x_{6}-x_{5}x_{3}\right)
			^{2}},\\
		f_{4}^{\prime}\left(  x\right)   &  =-\frac{\left(  x_{1}x_{6}-x_{2}%
			x_{3}\right)  x_{6}}{\left(  x_{4}x_{6}-x_{5}x_{3}\right)  ^{2}},f_{5}%
		^{\prime}\left(  x\right)  =\frac{x_{3}\left(  x_{1}x_{6}-x_{2}x_{3}\right)
		}{\left(  x_{4}x_{6}-x_{5}x_{3}\right)  ^{2}},\text{ and }f_{6}^{\prime
		}\left(  x\right)  =\frac{-x_{1}x_{5}x_{3}+x_{2}x_{3}x_{4}}{\left(  x_{4}%
			x_{6}-x_{5}x_{3}\right)  ^{2}}
	\end{align*}
	for every $x=(x_1,x_2,x_3,x_4,x_5,x_6)^T$ such that all the above derivatives are well defined.
	Then we can rewrite
	\[
	\sqrt{n}(\widehat{\theta}_{1}-\theta_{1})=\sqrt{n}\left\{  f\left(  \widehat{W}%
	_{n}\right)  -f\left(  W\right)  \right\}  ,
	\]
	where
	\[
	\widehat{W}_{n}=\left(
	\begin{array}
		[c]{c}%
		\frac{1}{n}\sum_{i=1}^{n}g\left(  Z_{i}\right)  Y_{i}1\left\{  Z_{i}%
		\in\widehat{\mathcal{Z}_0}\right\}  \\
		\bar{Y}_{\widehat{\mathcal{Z}_0}}\\
		\frac{1}{n}\sum_{i=1}^{n}g\left(  Z_{i}\right)  1\left\{  Z_{i}\in
		\widehat{\mathcal{Z}_0}\right\}  \\
		\frac{1}{n}\sum_{i=1}^{n}g\left(  Z_{i}\right)  D_{i}1\left\{  Z_{i}%
		\in\widehat{\mathcal{Z}_0}\right\}  \\
		\bar{D}_{\widehat{\mathcal{Z}_0}}\\
		\frac{1}{n}\sum_{i=1}^{n}1\left\{  Z_{i}\in
		\widehat{\mathcal{Z}_0}\right\}
	\end{array}
	\right)  \text{ and }W=\left(
	\begin{array}
		[c]{c}%
		E\left[  g\left(  Z_{i}\right)  Y_{i}1\left\{  Z_{i}\in\mathcal{Z}_{M}\right\}  \right]  \\
		E\left[  Y_{i}1\left\{  Z_{i}\in\mathcal{Z}_{M}\right\}  \right]  \\
		E\left[  g\left(  Z_{i}\right)  1\left\{  Z_{i}\in\mathcal{Z}_{M}\right\}
		\right]  \\
		E\left[  g\left(  Z_{i}\right)  D_{i}1\left\{  Z_{i}\in\mathcal{Z}%
		_{M}\right\}  \right]  \\
		E\left[  D_{i}1\left\{  Z_{i}\in\mathcal{Z}_{M}\right\}  \right]  \\
		E\left[  1\left\{  Z_{i}\in\mathcal{Z}_{M}\right\}  \right]
	\end{array}
	\right)  .
	\]
	For every small $\varepsilon>0$, we have $\mathbb{P}(\sqrt{n}%
	1\{\widehat{\mathcal{Z}_0}\neq\mathcal{Z}_{M}\}>\varepsilon
	)=\mathbb{P}(\widehat{\mathcal{Z}_0}\neq\mathcal{Z}_{M})\rightarrow0.$
	By assumption, $n^{-1}\sum_{i=1}^n\left\vert g\left(  Z_{i}\right)  Y_{i}\right\vert
	 \overset{p}\rightarrow E\left[  \left\vert g\left(  Z_{i}\right)  Y_{i}\right\vert
	\right]  $, and we have that
	\begin{align*}
		&  \sqrt{n}\left\vert \frac{1}{n}\sum_{i=1}^{n}g\left(  Z_{i}\right)
		Y_{i}1\left\{  Z_{i}\in\widehat{\mathcal{Z}_0}\right\}  -\frac{1}{n}%
		\sum_{i=1}^{n}g\left(  Z_{i}\right)  Y_{i}1\left\{  Z_{i}\in\mathcal{Z}%
		_{M}\right\}  \right\vert \\
		= &  \,\sqrt{n}\left\vert \frac{1}{n}\sum_{i=1}^{n}g\left(  Z_{i}\right)
		Y_{i}\left[  1\left\{  Z_{i}\in\widehat{\mathcal{Z}_0}\right\}
		-1\left\{  Z_{i}\in\mathcal{Z}_{M}\right\}  \right]  \right\vert \\
		\leq &  \,\frac{1}{n}\sum_{i=1}^{n}\left\vert g\left(  Z_{i}\right)
		Y_{i}\right\vert \left(  \sqrt{n}1\left\{  \widehat{\mathcal{Z}_0}\neq
		\mathcal{Z}_{M}\right\}  \right)  =o_{p}\left(  1\right)  .
	\end{align*}
	Similarly, we have that
	\begin{align*}
		&  \sqrt{n}\left(  \widehat{W}_{n}-W\right)  \\
		= &  \,\sqrt{n}\frac{1}{n}\sum_{i=1}^{n}\left(
		\begin{array}
			[c]{c}%
			g\left(  Z_{i}\right)  Y_{i}1\left\{  Z_{i}\in\mathcal{Z}_M\right\}
			-E\left[  g\left(  Z_{i}\right)  Y_{i}1\left\{  Z_{i}\in\mathcal{Z}%
			_{M}\right\}  \right]  \\
			Y_{i}1\left\{  Z_{i}\in\mathcal{Z}_M\right\}  -E\left[  Y_{i}1\left\{
			Z_{i}\in\mathcal{Z}_{M}\right\}  \right]  \\
			g\left(  Z_{i}\right)  1\left\{  Z_{i}\in\mathcal{Z}_M\right\}  -E\left[
			g\left(  Z_{i}\right)  1\left\{  Z_{i}\in\mathcal{Z}_{M}\right\}  \right]  \\
			g\left(  Z_{i}\right)  D_{i}1\left\{  Z_{i}\in\mathcal{Z}_M\right\}
			-E\left[  g\left(  Z_{i}\right)  D_{i}1\left\{  Z_{i}\in\mathcal{Z}%
			_{M}\right\}  \right]  \\
			D_{i}1\left\{  Z_{i}\in\mathcal{Z}_M\right\}  -E\left[  D_{i}1\left\{
			Z_{i}\in\mathcal{Z}_{M}\right\}  \right]  \\
			1\left\{  Z_{i}\in\mathcal{Z}_M\right\}  -E\left[  1\left\{  Z_{i}%
			\in\mathcal{Z}_{M}\right\}  \right]
		\end{array}
		\right)  +o_{p}\left(  1\right)  \overset{d}\rightarrow \,N\left(  0,\Sigma\right) ,
	\end{align*}
	where $\Sigma=E[  VV^{T}]  $ and
	\[
	V=\left(
	\begin{array}
		[c]{c}%
		g\left(  Z_{i}\right)  Y_{i}1\left\{  Z_{i}\in\mathcal{Z}_M\right\}
		-E\left[  g\left(  Z_{i}\right)  Y_{i}1\left\{  Z_{i}\in\mathcal{Z}%
		_{M}\right\}  \right]  \\
		Y_{i}1\left\{  Z_{i}\in\mathcal{Z}_M\right\}  -E\left[  Y_{i}1\left\{
		Z_{i}\in\mathcal{Z}_{M}\right\}  \right]  \\
		g\left(  Z_{i}\right)  1\left\{  Z_{i}\in\mathcal{Z}_M\right\}  -E\left[
		g\left(  Z_{i}\right)  1\left\{  Z_{i}\in\mathcal{Z}_{M}\right\}  \right]  \\
		g\left(  Z_{i}\right)  D_{i}1\left\{  Z_{i}\in\mathcal{Z}_M\right\}
		-E\left[  g\left(  Z_{i}\right)  D_{i}1\left\{  Z_{i}\in\mathcal{Z}%
		_{M}\right\}  \right]  \\
		D_{i}1\left\{  Z_{i}\in\mathcal{Z}_M\right\}  -E\left[  D_{i}1\left\{
		Z_{i}\in\mathcal{Z}_{M}\right\}  \right]  \\
		1\left\{  Z_{i}\in\mathcal{Z}_M\right\}  -E\left[  1\left\{  Z_{i}%
		\in\mathcal{Z}_{M}\right\}  \right]
	\end{array}
	\right)  .
	\]
	By the multivariate delta method, we have that
	\begin{align}\label{eq.asymptotic beta1}
		\sqrt{n}(\widehat{\theta}_{1}-\theta_{1})=\sqrt{n}\left\{  f\left(  \widehat{W}%
		_{n}\right)  -f\left(  W\right)  \right\}  \overset{d}\rightarrow f^{\prime}\left(
		W\right)  ^{T}\cdot N\left(  0,\Sigma\right)  .
	\end{align}

	Now we follow the strategy  of \citet{imbens1994identification} and have that
	\begin{align*}
		&  \frac{E\left[  g\left(  Z_{i}\right)  Y_{i}1\left\{  Z_{i}\in
			\mathcal{Z}_M\right\}  \right]  }{\mathbb{P}\left(  Z_{i}\in\mathcal{Z}%
			_M\right)  }-\frac{E\left[  Y_{i}1\left\{  Z_{i}\in\mathcal{Z}_M\right\}
			\right]  }{\mathbb{P}\left(  Z_{i}\in\mathcal{Z}_M\right)  }\frac{E\left[
			g\left(  Z_{i}\right)  1\left\{  Z_{i}\in\mathcal{Z}_M\right\}  \right]
		}{\mathbb{P}\left(  Z_{i}\in\mathcal{Z}_M\right)  }\\
		= &  \frac{\sum_{k=1}^{K}\mathbb{P}\left(  Z_{i}=z_{k}\right)  E\left[
			Y_{i}1\left\{  Z_{i}\in\mathcal{Z}_M\right\}  |Z_{i}=z_{k}\right]  \left\{
			g\left(  z_{k}\right)  1\left\{  z_{k}\in\mathcal{Z}_M\right\}
			-\frac{E\left[  g\left(  Z_{i}\right)  1\left\{  Z_{i}\in\mathcal{Z}%
				_M\right\}  \right]  }{\mathbb{P}\left(  Z_{i}\in\mathcal{Z}_M\right)
			}\right\}  }{\mathbb{P}\left(  Z_{i}\in\mathcal{Z}_M\right)  }\\
		= &  \sum_{m=1}^{M}\mathbb{P}\left(  Z_{i}=z_{k_{m}}|Z_{i}\in\mathcal{Z}%
		_M\right)  E\left[  Y_{i}|Z_{i}=z_{k_{m}}\right]  \left\{  g\left(
		z_{k_{m}}\right)  -E\left[  g\left(  Z_{i}\right)  |Z_{i}\in\mathcal{Z}%
		_M\right]  \right\}  .
	\end{align*}
	Then we write
	\begin{align}
		&  \sum_{m=1}^{M}\mathbb{P}\left(  Z_{i}=z_{k_{m}}|Z_{i}\in\mathcal{Z}%
		_M\right)  E\left[  Y_{i}|Z_{i}=z_{k_{m}}\right]  \left\{  g\left(
		z_{k_{m}}\right)  -E\left[  g\left(  Z_{i}\right)  |Z_{i}\in\mathcal{Z}%
		_M\right]  \right\}  \nonumber\label{eq.numerator 1}\\
		= &  \sum_{m=1}^{M-1}\mathbb{P}\left(  Z_{i}=z_{k_{m+1}}|Z_{i}\in
		\mathcal{Z}_M\right)  E\left[  Y_{i}|Z_{i}=z_{k_{m+1}}\right]  \left\{
		g\left(  z_{k_{m+1}}\right)  -E\left[  g\left(  Z_{i}\right)  |Z_{i}%
		\in\mathcal{Z}_M\right]  \right\}  \nonumber\\
		&  +\mathbb{P}\left(  Z_{i}=z_{k_{1}}|Z_{i}\in\mathcal{Z}_M\right)  E\left[
		Y_{i}|Z_{i}=z_{k_{1}}\right]  \left\{  g\left(  z_{k_{1}}\right)  -E\left[
		g\left(  Z_{i}\right)  |Z_{i}\in\mathcal{Z}_M\right]  \right\}  .
	\end{align}
	By \eqref{eq.beta}, we have
	\begin{align*}
		E\left[  Y_{i}|Z_{i}=z_{k_{m+1}}\right]   &  =\beta_{k_{m+1},k_{m}}\left(
		E\left[  D_{i}|Z_{i}=z_{k_{m+1}}\right]  -E\left[  D_{i}|Z_{i}=z_{k_{m}%
		}\right]  \right)  +E\left[  Y_{i}|Z_{i}=z_{k_{m}}\right]  \\
		&  =\sum_{l=1}^{m}\beta_{k_{l+1},k_{l}}\left(  E\left[  D_{i}|Z_{i}%
		=z_{k_{l+1}}\right]  -E\left[  D_{i}|Z_{i}=z_{k_{l}}\right]  \right)
		+E\left[  Y_{i}|Z_{i}=z_{k_{1}}\right]  ,
	\end{align*}
	and thus it follows that
	\begin{align*}
		&  \sum_{m=1}^{M-1}\mathbb{P}\left(  Z_{i}=z_{k_{m+1}}|Z_{i}\in\mathcal{Z}%
		_M\right)  E\left[  Y_{i}|Z_{i}=z_{k_{m+1}}\right]  \left\{  g\left(
		z_{k_{m+1}}\right)  -E\left[  g\left(  Z_{i}\right)  |Z_{i}\in\mathcal{Z}%
		_M\right]  \right\}  \\
		= &  \sum_{m=1}^{M-1}\bigg\{\mathbb{P}\left(  Z_{i}=z_{k_{m+1}}|Z_{i}%
		\in\mathcal{Z}_M\right)  \big\{\sum_{l=1}^{m}\beta_{k_{l+1},k_{l}}\left[
		p\left(  z_{k_{l+1}}\right)  -p\left(  z_{k_{l}}\right)  \right]  \big\}\\
		&  \cdot\left\{  g\left(  z_{k_{m+1}}\right)  -E\left[  g\left(  Z_{i}\right)|Z_{i}\in\mathcal{Z}_M\right]  \right\}  \bigg\}\\
		&  +\sum_{m=1}^{M-1}\mathbb{P}\left(  Z_{i}=z_{k_{m+1}}|Z_{i}\in
		\mathcal{Z}_M\right)  E[Y_{i}|Z_{i}=z_{k_{1}}]\left\{  g\left(  z_{k_{m+1}%
		}\right)  -E\left[  g\left(  Z_{i}\right)  |Z_{i}\in\mathcal{Z}_M\right]
		\right\}  .
	\end{align*}
	By \eqref{eq.numerator 1}, this implies that
	\begin{align*}
		&  \sum_{m=1}^{M}\mathbb{P}\left(  Z_{i}=z_{k_{m}}|Z_{i}\in\mathcal{Z}%
		_M\right)  E\left[  Y_{i}|Z_{i}=z_{k_{m}}\right]  \left\{  g\left(
		z_{k_{m}}\right)  -E\left[  g\left(  Z_{i}\right)  |Z_{i}\in\mathcal{Z}%
		_M\right]  \right\}  \nonumber\label{eq.numerator 2}\\
		= &  \sum_{m=1}^{M-1}\bigg\{\mathbb{P}\left(  Z_{i}=z_{k_{m+1}}|Z_{i}%
		\in\mathcal{Z}_M\right)  \big\{\sum_{l=1}^{m}\beta_{k_{l+1},k_{l}}\left[
		p\left(  z_{k_{l+1}}\right)  -p\left(  z_{k_{l}}\right)  \right]
		\big\}\nonumber\\
		&  \cdot\left\{  g\left(  z_{k_{m+1}}\right)  -E\left[  g\left(  Z_{i}\right)
		|Z_{i}\in\mathcal{Z}_M\right]  \right\}  \bigg\},
	\end{align*}
	where we use $\sum_{m=1}^{M}\mathbb{P}\left(  Z_{i}%
	=z_{k_{m}}|Z_{i}\in\mathcal{Z}_M\right)  \left\{  g\left(  z_{k_{m}%
	}\right)  -E\left[  g\left(  Z_{i}\right)  |Z_{i}\in\mathcal{Z}_M\right]
	\right\}  =0$. Furthermore, we obtain
	\begin{align*}
		&  \sum_{m=1}^{M-1}\mathbb{P}\left(  Z_{i}=z_{k_{m+1}}|Z_{i}\in\mathcal{Z}%
		_M\right)  \left\{  \sum_{l=1}^{m}\beta_{k_{l+1},k_{l}}\left[  p\left(
		z_{k_{l+1}}\right)  -p\left(  z_{k_{l}}\right)  \right]  \right\}  \tilde
		{g}\left(  z_{k_{m+1}}\right)  \\
		= &  \,\mathbb{P}\left(  Z_{i}=z_{k_{2}}|Z_{i}\in\mathcal{Z}_M\right)
		\left\{  \beta_{k_{2},k_{1}}\left[  p\left(  z_{k_{2}}\right)  -p\left(
		z_{k_{1}}\right)  \right]  \right\}  \tilde{g}\left(  z_{k_{2}}\right)
		+\cdots\\
		&  +\mathbb{P}\left(  Z_{i}=z_{k_{M}}|Z_{i}\in\mathcal{Z}_M\right)  \left\{
		\sum_{l=1}^{M-1}\beta_{k_{l+1},k_{l}}\left[  p\left(  z_{k_{l+1}}\right)
		-p\left(  z_{k_{l}}\right)  \right]  \right\}  \tilde{g}\left(  z_{k_{M}%
		}\right)  \\
		= &  \,\sum_{m=1}^{M-1}\left\{\beta_{k_{m+1},k_{m}}\left[  p\left(  z_{k_{m+1}%
		}\right)  -p\left(  z_{k_{m}}\right)  \right]  \sum_{l=m}^{M-1}\mathbb{P}%
		\left(  Z_{i}=z_{k_{l+1}}|Z_{i}\in\mathcal{Z}_M\right)  \tilde{g}\left(
		z_{k_{l+1}}\right)\right\}  ,
	\end{align*}
	where $\tilde{g}\left(  z\right)  =g\left(  z\right)  -E\left[  g\left(
	Z_{i}\right)  |Z_{i}\in\mathcal{Z}_M\right]  $ for all $z$. Similarly, we
	have
	\begin{align*}
		&  \frac{E\left[  g\left(  Z_{i}\right)  D_{i}1\left\{  Z_{i}\in
			\mathcal{Z}_M\right\}  \right]  }{\mathbb{P}\left(  Z_{i}\in\mathcal{Z}%
			_M\right)  }-\frac{E\left[  D_{i}1\left\{  Z_{i}\in\mathcal{Z}_M\right\}
			\right]  }{\mathbb{P}\left(  Z_{i}\in\mathcal{Z}_M\right)  }\frac{E\left[
			g\left(  Z_{i}\right)  1\left\{  Z_{i}\in\mathcal{Z}_M\right\}  \right]
		}{\mathbb{P}\left(  Z_{i}\in\mathcal{Z}_M\right)  }\\
		= &  \sum_{m=1}^{M}\mathbb{P}\left(  Z_{i}=z_{k_{m}}|Z_{i}\in\mathcal{Z}%
		_M\right)  p\left(  z_{k_{m}}\right)  \left\{  g\left(  z_{k_{m}}\right)
		-E\left[  g\left(  Z_{i}\right)  |Z_{i}\in\mathcal{Z}_M\right]  \right\},
	\end{align*}
	which is nonzero by Assumption \ref{ass.first stage partial}.
	Thus, we have $\theta_{1}=\sum_{m=1}^{M-1}\mu_{m}\beta_{k_{m+1},k_{m}}$ with%
	\begin{align*}
		\mu_{m}=
		\frac{\left[  p\left(  z_{k_{m+1}}\right)  -p\left(  z_{k_{m}}\right)
			\right]  \sum_{l=m}^{M-1}\mathbb{P}\left(  Z_{i}=z_{k_{l+1}}|Z_{i}%
			\in\mathcal{Z}_M\right)  \left\{  g\left(  z_{k_{l+1}}\right)  -E\left[
			g\left(  Z_{i}\right)  |Z_{i}\in\mathcal{Z}_M\right]  \right\}  }{\sum
			_{l=1}^{M}\mathbb{P}\left(  Z_{i}=z_{k_{l}}|Z_{i}\in\mathcal{Z}_M\right)
			p\left(  z_{k_{l}}\right)  \left\{  g\left(  z_{k_{l}}\right)  -E\left[
			g\left(  Z_{i}\right)  |Z_{i}\in\mathcal{Z}_M\right]  \right\}  }.
	\end{align*}
	Now we show that $\sum_{m=1}^{M-1}\mu_{m}=1$. First, we have that
	\begin{align*}
		&  \sum_{m=1}^{M-1}\left[  p\left(  z_{k_{m+1}}\right)  -p\left(  z_{k_{m}%
		}\right)  \right]  \sum_{l=m}^{M-1}\mathbb{P}\left(  Z_{i}=z_{k_{l+1}}%
		|Z_{i}\in\mathcal{Z}_M\right)  \left\{  g\left(  z_{k_{l+1}}\right)
		-E\left[  g\left(  Z_{i}\right)  |Z_{i}\in\mathcal{Z}_M\right]  \right\}  \\
		= &  \left[  p\left(  z_{k_{2}}\right)  -p\left(  z_{k_{1}}\right)  \right]
		\sum_{l=1}^{M-1}\mathbb{P}\left(  Z_{i}=z_{k_{l+1}}|Z_{i}\in\mathcal{Z}%
		_M\right)  \left\{  g\left(  z_{k_{l+1}}\right)  -E\left[  g\left(
		Z_{i}\right)  |Z_{i}\in\mathcal{Z}_M\right]  \right\}  +\cdots\\
		&  +\left[  p\left(  z_{k_{M}}\right)  -p\left(  z_{k_{M-1}}\right)  \right]
		\mathbb{P}\left(  Z_{i}=z_{k_{M}}|Z_{i}\in\mathcal{Z}_M\right)  \left\{
		g\left(  z_{k_{M}}\right)  -E\left[  g\left(  Z_{i}\right)  |Z_{i}%
		\in\mathcal{Z}_M\right]  \right\}  \\
		= &  \sum_{l=2}^{M}\mathbb{P}\left(  Z_{i}=z_{k_{l}}|Z_{i}\in\mathcal{Z}%
		_M\right)  p\left(  z_{k_{l}}\right)  \left\{  g\left(  z_{k_{l}}\right)
		-E\left[  g\left(  Z_{i}\right)  |Z_{i}\in\mathcal{Z}_M\right]  \right\}  \\
		&  -p\left(  z_{k_{1}}\right)  \sum_{l=2}^{M}\mathbb{P}\left(  Z_{i}=z_{k_{l}%
		}|Z_{i}\in\mathcal{Z}_M\right)  \left\{  g\left(  z_{k_{l}}\right)
		-E\left[  g\left(  Z_{i}\right)  |Z_{i}\in\mathcal{Z}_M\right]  \right\}  \\
		= &  \sum_{l=1}^{M}\mathbb{P}\left(  Z_{i}=z_{k_{l}}|Z_{i}\in\mathcal{Z}%
		_M\right)  p\left(  z_{k_{l}}\right)  \left\{  g\left(  z_{k_{l}}\right)
		-E\left[  g\left(  Z_{i}\right)  |Z_{i}\in\mathcal{Z}_M\right]  \right\}  ,
	\end{align*}
	where we use the equality that $\sum_{l=1}^{M}\mathbb{P}\left(  Z_{i}%
	=z_{k_{l}}|Z_{i}\in\mathcal{Z}_M\right)  \left\{  g\left(  z_{k_{l}}\right)
	-E\left[  g\left(  Z_{i}\right)  |Z_{i}\in\mathcal{Z}_M\right]  \right\}
	=0$. This implies that $\sum_{m=1}^{M-1}\mu_{m}=1$.	
\end{proof}

\subsection{Varying Underlying Distributions}
In the main text and Appendix \ref{sec.multivalued treatment}, we consider a fixed underlying distribution for the data. In this section, we extend the results to varying underlying distributions. 

\begin{assumption}\label{ass.iid data varying DGPs}
For each $n$, $\{(Y_i,D_i,Z_i)\}_{i=1}^{n}$ is an i.i.d.\ sample distributed according to some probability measure $P_n\in\mathcal{P}$ such that all relevant moments exist.
\end{assumption}

Assumption \ref{ass.iid data varying DGPs} allows the underlying distribution $P_n$ to change as $n$ increases. The following assumption provides a limit for the sequence of the probability measures $\{P_n\}$.

\begin{assumption}\label{ass.probability path}
	There is a probability measure $P\in\mathcal{P}$ such that 
	\begin{equation}\label{eq.probaility path}
	\lim_{n\rightarrow\infty}\int\left(  \sqrt{n}\left\{  \mathrm{d}P_{n}^{1/2}-\mathrm{d}P^{1/2}\right\}  -\frac{1}{2}v_0\mathrm{d}P^{1/2}\right)  ^{2}=0
	\end{equation}
	for some measurable function $v_0$, where $\mathrm{d}P_{n}^{1/2}$ and $\mathrm{d}P^{1/2}$ denote the square roots of the densities of $P_{n}$ and $P$, respectively.
\end{assumption}

Assumption \ref{ass.probability path} requires that the sequence of probability distributions $\{P_n\}$ converge to $P$, following the setup in \citet[p.~406]{van1996weak}. It corresponds to (29) in \citet{fang2014inference} and Assumption 3.2 in \citet{sun2021ivvalidity}. This assumption allows the data generating process to change with the sample size. It trivially holds if the data generating process does not change with $n$ so that $P_n=P$ for all $n$, as in the main text. 
For simplicity of notation, when there is no confusion, $E$ denotes the expectation under $P_n$ for every $n$.

For every $\mathcal{Z}_{(k,k^{\prime})}$, define a function space $$\mathcal{H}_{\mathcal{Z}_{(k,k^{\prime})}}%
=\{h_{1(k,k^{\prime})},h_{2(k,k^{\prime})},h_{3(k,k^{\prime})}%
,h_{4(k,k^{\prime})},h_{5(k,k^{\prime})},h_{6(k,k^{\prime})}\}$$ such that
\begin{align*}
h_{1(k,k^{\prime})}\left(  y,d,z\right)   &  =g\left(  z\right)  y1\left\{z\in\mathcal{Z}_{(k,k^{\prime})}\right\}  ,\\
h_{2(k,k^{\prime})}\left(  y,d,z\right)   &  =y1\left\{  z\in\mathcal{Z}%
_{(k,k^{\prime})}\right\}  ,\\
h_{3(k,k^{\prime})}\left(  y,d,z\right)   &  =g\left(  z\right)  1\left\{
z\in\mathcal{Z}_{(k,k^{\prime})}\right\}  ,\\
h_{4(k,k^{\prime})}\left(  y,d,z\right)   &  =g\left(  z\right)  d1\left\{
z\in\mathcal{Z}_{(k,k^{\prime})}\right\}  ,\\
h_{5(k,k^{\prime})}\left(  y,d,z\right)   &  =d1\left\{  z\in\mathcal{Z}%
_{(k,k^{\prime})}\right\}  ,\text{ and}\\
h_{6(k,k^{\prime})}\left(  y,d,z\right)   &  =1\left\{  z\in\mathcal{Z}%
_{(k,k^{\prime})}\right\}  .
\end{align*}

\begin{assumption}\label{ass.first stage varying DGPs}
	For every $\mathcal{Z}_{(k,k')}\in\mathscr{Z}_{\bar{M}}$, 
	\begin{align}\label{eq.first stage varing DGP}
     P(h_{4(k,k')})/P(h_{6(k,k')})-P(h_{5(k,k')})/P(h_{6(k,k')})\cdot P(h_{3(k,k')})/P(h_{6(k,k')})\neq0.
	\end{align}
\end{assumption}

Assumption \ref{ass.first stage varying DGPs} imposes a first-stage condition under $P$ for every $\mathcal{Z}_{(k,k')}\in\mathscr{Z}_{\bar{M}}$. The following theorem is an extension of Theorem \ref{thm.IV estimator asymptotics pairwise} with a convergent sequence of probability distributions $\{P_n\}$. Note that under Assumption \ref{ass.probability path}, for every $(z_k,z_{k'})$, $\beta_{(  k,k^{\prime})  }^{1}$ defined in \eqref{eq.VSIV true beta} and $\beta_{k^{\prime},k}$ defined in \eqref{eq.beta}
could be different for every $n$ since the expectations under $P_n$ in $\beta_{(  k,k^{\prime})  }^{1}$ and $\beta_{k^{\prime},k}$ could be different for every $n$. Let $\widehat\beta_1$ and $\beta_1$ be defined as in Appendix \ref{sec.ordered treatment}.

\begin{theorem}\label{thm.IV estimator asymptotics pairwise varying DGPs}
	Suppose that the instrument $Z$ is pairwise valid for the treatment $D$ as defined in Definition \ref{def.partial validity pairwise} with the largest validity pair set $\mathscr{Z}_{\bar{M}}=\{(z_{k_1},z_{k_1^{\prime}}),\ldots,(z_{k_{\bar{M}}},z_{k_{\bar{M}}^{\prime}})\}$ for every $n$, and that the estimator $\widehat{\mathscr{Z}_0}$ satisfies $\mathbb{P}(\widehat{\mathscr{Z}_0}=\mathscr{Z}_{\bar{M}})\to 1$. Under Assumptions \ref{ass.iid data varying DGPs}--\ref{ass.first stage varying DGPs}, $\sqrt{n}( \widehat{\beta}_{1}-\beta_1 ) \overset{d}\to N\left( 0,\Sigma \right) $, where 
	$\Sigma$ is defined in \eqref{eq.weak convergence pairwise2 varying DGPs}. In addition, when $n$ is sufficiently large, $\beta_{(  k,k^{\prime})  }^{1}=\beta_{k^{\prime},k}$ for every $(z_k,z_{k^{\prime}})\in\mathscr{Z}_{\bar{M}}$.
\end{theorem}

\begin{proof}[Proof of Theorem \ref{thm.IV estimator asymptotics pairwise varying DGPs}]
This proof modifies that of Theorem \ref{thm.IV estimator asymptotics pairwise} under the convergent sequence of probability distributions $\{P_n\}$.

Let $\mathcal{H}_{\mathcal{Z}}=\cup_{\mathcal{Z}_{(k,k^{\prime})}}\mathcal{H}_{\mathcal{Z}_{(k,k^{\prime})}}$. Clearly, since $\mathcal{H}_{\mathcal{Z}}$ is a finite set,
it is a Donsker class. Then we define a map $\eta_{(k,k^{\prime})}%
:\ell^{\infty}(  \mathcal{H}_{\mathcal{Z}_{(k,k^{\prime})}})
\rightarrow\mathbb{R}^{6}$ such that
\[
\eta_{(k,k^{\prime})}\left(  \psi\right)  =\left(  \psi\left(
h_{1(k,k^{\prime})}\right)  ,\psi\left(  h_{2(k,k^{\prime})}\right)
,\psi\left(  h_{3(k,k^{\prime})}\right)  ,\psi\left(  h_{4(k,k^{\prime}%
)}\right)  ,\psi\left(  h_{5(k,k^{\prime})}\right)  ,\psi\left(
h_{6(k,k^{\prime})}\right)  \right)  ^{T}.
\]
Then define another map $\eta:\ell^{\infty}\left(  \mathcal{H}_{\mathcal{Z}}\right)
\rightarrow\mathbb{R}^{6K\left(  K-1\right)  }$ by
\[
\eta\left(  \psi\right)  =\left(  \eta_{(1,2)}\left(  \psi\right)  ^{T}%
,\ldots,\eta_{(1,K)}\left(  \psi\right)  ^{T},\ldots,\eta_{(K,1)}\left(
\psi\right)  ^{T},\ldots,\eta_{(K,K-1)}\left(  \psi\right)  ^{T}\right)  ^{T}.
\]

 For every $\mathcal{Z}_{(  k,k^{\prime})  }\in\mathscr{Z}$, we define
	\[
	W_{i}\left(  \mathcal{Z}_{(  k,k^{\prime})  }\right)  =\left(
	\begin{array}
		[c]{c}%
		g\left(  Z_{i}\right)  Y_{i}1\left\{  Z_{i}\in\mathcal{Z}_{\left(
			k,k^{\prime}\right)  }\right\}  \\
		Y_{i}1\left\{  Z_{i}\in\mathcal{Z}_{\left(
			k,k^{\prime}\right)  }\right\}  \\
		g\left(  Z_{i}\right)  1\left\{  Z_{i}\in\mathcal{Z}_{\left(
			k,k^{\prime}\right)  }\right\}  \\
		g\left(  Z_{i}\right)  D_{i}1\left\{  Z_{i}\in\mathcal{Z}_{\left(
			k,k^{\prime}\right)  }\right\}  \\
		D_{i}1\left\{  Z_{i}\in\mathcal{Z}_{\left(
			k,k^{\prime}\right)  }\right\}  \\
		1\left\{  Z_{i}\in\mathcal{Z}_{\left(
			k,k^{\prime}\right)  }\right\}
	\end{array}
	\right)  ,
	\]%
	\[
	\widehat{W}_{n}\left(  \mathcal{Z}_{(  k,k^{\prime})  }\right)
	=\,\frac{1}{n}\sum_{i=1}^{n}W_{i}\left(  \mathcal{Z}_{(  k,k^{\prime})  }\right)
	, \text{ and } W\left(  \mathcal{Z}_{(  k,k^{\prime})
	}\right)  = E[W_{i}\left(  \mathcal{Z}_{(  k,k^{\prime})  }\right)].
	\]
	Also, we let
	\begin{align*}
		&\widehat{W}_{n}   =\left(  \widehat{W}_{n}\left(  \mathcal{Z}_{\left(
			1,2\right)  }\right)^{T}  ,\ldots,\widehat{W}_{n}\left(  \mathcal{Z}_{\left(
			1,K\right)  }\right)^{T}  ,\ldots,\widehat{W}_{n}\left(  \mathcal{Z}_{\left(
			K,1\right)  }\right)^{T}  ,\ldots,\widehat{W}_{n}\left(  \mathcal{Z}_{\left(
			K,K-1\right)  }\right)^{T}  \right)  ^{T}\\
		&\text{and }W  =\left(  W\left(  \mathcal{Z}_{\left(  1,2\right)
		}\right)^{T}  ,\ldots,W\left(  \mathcal{Z}_{\left(  1,K\right)  }\right)^{T}
		,\ldots,W\left(  \mathcal{Z}_{\left(  K,1\right)  }\right)^{T}  ,\ldots
		,W\left(  \mathcal{Z}_{\left(  K,K-1\right)  }\right)^{T}  \right)  ^{T}.
	\end{align*}
	
Let $\widehat{P}_n$ denote the empirical probability measure of $P_n$ for every $n$. By Theorem 3.10.12 of \citet{van1996weak}, we have that $\sqrt{n}(\widehat{P}_n-P_n)\leadsto \mathbb{G}_P$ under $P_n$, where $\mathbb{G}_P$ is a tight Brownian bridge. Theorem 3.10.12 of \citet{van1996weak} also implies $\widehat{P}_n\to P$ in probability.
Since $\eta$ is linear and continuous, by continuous mapping theorem,%
\begin{align}
\sqrt{n}\left(  \widehat{W}_{n}-W\right)  =\sqrt{n}\left(  \eta(
\widehat{P}_n)  -\eta\left(  P_n\right)  \right)  =\eta\left(  \sqrt{n}(
\widehat{P}_n-P_n)  \right)  \leadsto\eta\left(  \mathbb{G}_P\right) =N(0,\Sigma_P),
\end{align}
where
\[
\Sigma_{P}=E\left[  \eta\left(  \mathbb{G}_P\right)  \eta\left(  \mathbb{G}_P
\right)  ^{T}\right]  .
\] 	
Since $P_n\to P$ by Theorem 3.10.12 of \citet{van1996weak}, $\eta(P_n)\to \eta(P)$ and we denote $\eta(P)$ by $W_P$. 
 
	Define a function $f:\mathbb{R}^{6}\rightarrow\bar{\mathbb{R}}$ by
	\[
	f\left(  x\right)  =\frac{x_{1}/x_{6}-x_{2}x_{3}/x_{6}^{2}}{x_{4}/x_{6}%
		-x_{5}x_{3}/x_{6}^{2}}
	\]
	for every $x\in\mathbb{R}^{6}$ with $x=\left(  x_{1},x_{2},x_{3},x_{4}
	,x_{5},x_6\right)  ^{T}$ such that $f(x)$ is well defined.
	We can obtain the gradient
	of $f$, denoted $f^{\prime}$, by $f^{\prime}\left(  x\right)  =\left(
	f_{1}^{\prime}\left(  x\right)  ,f_{2}^{\prime}\left(  x\right)
	,f_{3}^{\prime}\left(  x\right)  ,f_{4}^{\prime}\left(  x\right)
	,f_{5}^{\prime}\left(  x\right)  ,f_{6}^{\prime}\left(  x\right)  \right)
	^{T}$ with
	\begin{align*}
		f_{1}^{\prime}\left(  x\right)   &  =\frac{x_{6}}{x_{4}x_{6}-x_{5}
			x_{3}},f_{2}^{\prime}\left(  x\right)  =\frac{-x_{3}
		}{x_{4}x_{6}-x_{5}x_{3}},f_{3}^{\prime}\left(  x\right)
		=\frac{-x_{2}x_{4}x_{6}+x_{5}x_{1}x_{6}}{\left(  x_{4}x_{6}-x_{5}x_{3}\right)
			^{2}},\\
		f_{4}^{\prime}\left(  x\right)   &  =-\frac{\left(  x_{1}x_{6}-x_{2}%
			x_{3}\right)  x_{6}}{\left(  x_{4}x_{6}-x_{5}x_{3}\right)  ^{2}},f_{5}%
		^{\prime}\left(  x\right)  =\frac{x_{3}\left(  x_{1}x_{6}-x_{2}x_{3}\right)
		}{\left(  x_{4}x_{6}-x_{5}x_{3}\right)  ^{2}},\text{ and }f_{6}^{\prime
		}\left(  x\right)  =\frac{-x_{1}x_{5}x_{3}+x_{2}x_{3}x_{4}}{\left(  x_{4}%
			x_{6}-x_{5}x_{3}\right)  ^{2}}
	\end{align*}
	for every  $x=\left(  x_{1},x_{2},x_{3},x_{4},x_{5},x_6\right)  ^{T}$ such that all the above derivatives are well defined.

For every $\varepsilon>0$ and every $\mathcal{Z}_{(k,k^{\prime})}$, by assumption we have that for every $\rho\geq0$,
\begin{align}\label{eq.set consistency varying DGPs}
\mathbb{P}\left(  n^{\rho}\left\vert 1\left\{  \mathcal{Z}_{(k,k^{\prime}%
)}\in\widehat{\mathscr{Z}_{0}}\right\}  -1\left\{  \mathcal{Z}_{(k,k^{\prime
})}\in\mathscr{Z}_{\bar{M}}\right\}  \right\vert >\varepsilon\right)    \leq\mathbb{P}\left(  \widehat{\mathscr{Z}_{0}}\neq\mathscr{Z}_{\bar{M}%
}\right)  \rightarrow0.
\end{align}
This implies that if $1\{\mathcal{Z}_{(k,k^{\prime})}\in\mathscr{Z}_{\bar{M}%
}\}=0$, then
\begin{align}\label{eq.I consistency varying DGPs}
n^{\rho}1\{  \mathcal{Z}_{(k,k^{\prime})}\in\widehat{\mathscr{Z}_{0}%
}\}  =o_{p}\left(  1\right).
\end{align}

Without loss of generality, we suppose $\mathscr{Z}_{\bar{M}}=\{
\mathcal{Z}_{(1,2)},\mathcal{Z}_{(1,3)},\ldots,\mathcal{Z}_{(K-1,K)}\}
$ and $\mathscr{Z}\setminus\mathscr{Z}_{\bar{M}}=\{
\mathcal{Z}_{(2,1)},\mathcal{Z}_{(3,1)},\ldots,\mathcal{Z}_{(K,K-1)}\}
$ for simplicity. For every 
$\mathcal{Z}_{\left(  k,k'\right)  }\notin\mathscr{Z}_{\bar{M}}$, by
Assumption \ref{ass.first stage varying DGPs}, it is possible that
\begin{align}\label{eq.first stage eq varying DGPs}
P(h_{4(k,k')})/P(h_{6(k,k')})-P(h_{5(k,k')})/P(h_{6(k,k')})\cdot P(h_{3(k,k')})/P(h_{6(k,k')})=0.
\end{align}
For every $w=(  w_{1}^T,\ldots,w_{\left(  K-1\right)K  }^T)^T  $ with
$w_{j}=(  w_{j1},\ldots,w_{j6})^T  $ for every $j $, define
\begin{align*}
&\mathcal{F}_{1}\left(  w\right)  =\left(  f\left(  w_{1}\right)
,\ldots,f\left(  w_{(K-1)K/2  }\right)  \right)  ^{T}\text{ and
}\\
&\mathcal{F}_{0}\left(  w\right)  =\left(  f\left(  w_{K(K-1)/2+1}\right)
,\ldots,f\left(  w_{(K-1)K  }\right)  \right)  ^{T}
.
\end{align*}
For every
$\mathscr{Z}_{s}\subseteq\mathscr{Z}$, define
\begin{align*}
\mathcal{I}_{1}\left(  \mathscr{Z}_{s}\right)   =  \left(
\begin{array}
[c]{cccc}%
1\left\{  \mathcal{Z}_{(1,2)}\in\mathscr{Z}_{s}\right\}   &  &  & \\
& 1\left\{  \mathcal{Z}_{(1,3)}\in\mathscr{Z}_{s}\right\}   &  & \\
&  & \ddots & \\
&  &  & 1\left\{  \mathcal{Z}_{(K-1,K)}\in\mathscr{Z}_{s}\right\}
\end{array}
\right) 
\end{align*}
and 
\begin{align*}
\mathcal{I}_{0}\left(  \mathscr{Z}_{s}\right)   =  \left(
\begin{array}
[c]{cccc}%
1\left\{  \mathcal{Z}_{(2,1)}\in\mathscr{Z}_{s}\right\}   &  &  & \\
& 1\left\{  \mathcal{Z}_{(3,1)}\in\mathscr{Z}_{s}\right\}   &  & \\
&  & \ddots & \\
&  &  & 1\left\{  \mathcal{Z}_{(K,K-1)}\in\mathscr{Z}_{s}\right\}
\end{array}
\right)  .
\end{align*}
Then we can write
\[
\sqrt{n}\left(  \widehat{\beta}_{1}-\beta_{1}\right)  =\sqrt{n}\left\{
\left(\begin{array}
[c]{c}%
\mathcal{I}_{1}\left(  \widehat{\mathscr{Z}_{0}}\right)  \mathcal{F}%
_{1}\left(  \widehat{W}_{n}\right)  \\
\mathcal{I}_{0}\left(  \widehat{\mathscr{Z}_{0}}\right)  \mathcal{F}%
_{0}\left(  \widehat{W}_{n}\right)
\end{array}\right)
-%
\left(\begin{array}
[c]{c}%
\mathcal{I}_{1}\left(  \mathscr{Z}_{\bar{M}}\right)  \mathcal{F}_{1}\left(
W\right)  \\
\mathcal{I}_{0}\left(  \mathscr{Z}_{\bar{M}}\right)  \mathcal{F}_{0}\left(
W\right)
\end{array}\right)
\right\}  .
\]

First, we have that
\begin{align*}
  \sqrt{n}\left\{  \mathcal{I}_{1}\left(  \widehat{\mathscr{Z}_{0}}\right)
\mathcal{F}_{1}\left(  \widehat{W}_{n}\right)  -\mathcal{I}_{1}\left(
\mathscr{Z}_{\bar{M}}\right)  \mathcal{F}_{1}\left(  W\right)  \right\}  
  =&\,\sqrt{n}\left\{  \mathcal{I}_{1}\left(  \widehat{\mathscr{Z}_{0}}\right)
\mathcal{F}_{1}\left(  \widehat{W}_{n}\right)  -\mathcal{I}_{1}\left(
\widehat{\mathscr{Z}_{0}}\right)  \mathcal{F}_{1}\left(  W\right)  \right\}
\\
&  +\sqrt{n}\left\{  \mathcal{I}_{1}\left(  \widehat{\mathscr{Z}_{0}}\right)
\mathcal{F}_{1}\left(  W\right)  -\mathcal{I}_{1}\left(  \mathscr{Z}_{\bar{M}%
}\right)  \mathcal{F}_{1}\left(  W\right)  \right\}  .
\end{align*}
The Jacobian matrix $\mathcal{F}_{1}^{\prime}\left(  W\right)  $ of
$\mathcal{F}_{1}$ at $W$ can be obtained with the derivatives of $f$. Then by \eqref{eq.set consistency varying DGPs} and delta method, it is easy to show that
\begin{align*}
  \sqrt{n}\left\{  \mathcal{I}_{1}\left(  \widehat{\mathscr{Z}_{0}}\right)
\mathcal{F}_{1}\left(  \widehat{W}_{n}\right)  -\mathcal{I}_{1}\left(
\mathscr{Z}_{\bar{M}}\right)  \mathcal{F}_{1}\left(  W\right)  \right\}   &=\mathcal{I}_{1}\left(  \widehat{\mathscr{Z}_{0}}\right)  \sqrt{n}\left\{
\mathcal{F}_{1}\left(  \widehat{W}_{n}\right)  -\mathcal{F}_{1}\left(
W\right)  \right\}  +o_{p}\left(  1\right)  \\
&  \overset{d}{\rightarrow}\mathcal{I}_{1}\left(  \mathscr{Z}_{\bar{M}%
}\right)  \mathcal{F}_{1}^{\prime}\left(  W_P\right)  N\left(  0,\Sigma
_{P}\right)  .
\end{align*}
Second, by assumption and \eqref{eq.0timesinfinity}, 
\begin{align*}
\sqrt{n}\left\{  \mathcal{I}_{0}\left(  \widehat{\mathscr{Z}_{0}}\right)
\mathcal{F}_{0}\left(  \widehat{W}_{n}\right)  -\mathcal{I}_{0}\left(
\mathscr{Z}_{\bar{M}}\right)  \mathcal{F}_{0}\left(  W\right)  \right\}  
=\sqrt{n}\mathcal{I}_{0}\left(  \widehat{\mathscr{Z}_{0}}\right)  \mathcal{F}_{0}\left(\widehat{W}_{n}  \right)  .
\end{align*}
For every $\mathcal{Z}_{(k,k')}\notin\mathscr{Z}_{\bar{M}}$ such that \eqref{eq.first stage eq varying DGPs} holds, 
\begin{align*}
\sqrt{n}1\{\mathcal{Z}_{(k,k')}\in\widehat{\mathscr{Z}_0}\}f(\widehat{W}_n(\mathcal{Z}_{(k,k')}))= {n}1\{\mathcal{Z}_{(k,k')}\in\widehat{\mathscr{Z}_0}\} \frac{A_{n}}{\sqrt{n}B_{n}},
\end{align*}
where%
\begin{align*}
A_{n} =&\,\frac{1}{n}\sum_{i=1}^{n}g\left(  Z_{i}\right)  Y_{i}1\left\{
Z_{i}\in\mathcal{Z}_{\left(  k,k'\right)  }\right\}  \frac{1}{n}\sum
_{i=1}^{n}1\left\{  Z_{i}\in\mathcal{Z}_{\left(  k,k'\right)  }\right\}  \\
&  -\frac{1}{n}\sum_{i=1}^{n}g\left(  Z_{i}\right)  1\left\{  Z_{i}%
\in\mathcal{Z}_{\left(  k,k'\right)  }\right\}  \frac{1}{n}\sum_{i=1}%
^{n}Y_{i}1\left\{  Z_{i}\in\mathcal{Z}_{\left(  k,k'\right)  }\right\} 
\end{align*}
and%
\begin{align*}
B_{n}  =&\,\frac{1}{n}\sum_{i=1}^{n}g\left(  Z_{i}\right)  D_{i}1\left\{
Z_{i}\in\mathcal{Z}_{\left(  k,k'\right)  }\right\}  \frac{1}{n}\sum
_{i=1}^{n}1\left\{  Z_{i}\in\mathcal{Z}_{\left(  k,k'\right)  }\right\}  \\
&  -\frac{1}{n}\sum_{i=1}^{n}g\left(  Z_{i}\right)  1\left\{  Z_{i}%
\in\mathcal{Z}_{\left(  k,k'\right)  }\right\}  \frac{1}{n}\sum_{i=1}%
^{n}D_{i}1\left\{  Z_{i}\in\mathcal{Z}_{\left(  k,k'\right)  }\right\}  .
\end{align*}
Define a map $h$ such that for every $x\in\mathbb{R}^6$ with $x=\left(  x_{1},\ldots,x_{6}\right)^T  $,
\[
h\left(  x\right)  =x_{4}x_{6}-x_{3}x_{5}.
\]
Let $W_P(  \mathcal{Z}_{(  k,k')  })$ denote $\eta_{(k,k')}(P)$ and $h^{\prime}(  W_P(  \mathcal{Z}_{(  k,k')  }))  $ be the Jacobian matrix of $h$ at $W_P(  \mathcal{Z}_{\left(k,k'\right)  })  $. Then by delta method,
\begin{align*}
\sqrt{n}B_{n} &  =\sqrt{n}\left(  h\left(  \widehat{W}_{n}\left(
\mathcal{Z}_{\left(  k,k'\right)  }\right)  \right)  -h\left(  W_P\left(
\mathcal{Z}_{\left(  k,k'\right)  }\right)  \right)  \right)\\
&\overset{d}{\rightarrow}h^{\prime}\left(  W_P\left(  \mathcal{Z}_{\left(k,k'\right)  }\right)  \right)  N\left(  0,\Sigma_{\left(  k,k'\right)
}\right)  ,
\end{align*}
where by Theorem 3.10.12 of \citet{van1996weak}, $$\Sigma_{\left(  k,k'\right)  }=E\left[  \eta_{(k,k')}\left(  \tilde{\mathbb{G}}_P\right)  \eta_{(k,k')}\left(  \tilde{\mathbb{G}}_P
\right)  ^{T}\right] $$
and $\tilde{\mathbb{G}}_P$ is some random element such that $\tilde{\mathbb{G}}_P(u)=\mathbb{G}_P(u)+P(uv_0)$ for every measurable $u$.
Also, it is easy to show that
\begin{align*}
A_{n}\overset{p}{\rightarrow}
P(h_{1(k,k')})P(h_{6(k,k')})-P(h_{2(k,k')})P(h_{3(k,k')}).
\end{align*}
Note that by \eqref{eq.I consistency varying DGPs}, $n\mathcal{I}_{0}(\widehat{\mathscr{Z}_{0}})=o_{p}\left(
1\right)  $. Thus, $\sqrt{n}1\{\mathcal{Z}_{(k,k')}\in\widehat{\mathscr{Z}_0}\}f(\widehat{W}_n(\mathcal{Z}_{(k,k')}))\overset{p}{\rightarrow} 0$.
Similarly, for every $\mathcal{Z}_{(k,k')}\notin\mathscr{Z}_{\bar{M}}$ such that \eqref{eq.first stage eq varying DGPs} does not hold, it is easy to show that
\begin{align*}
\sqrt{n}1\{\mathcal{Z}_{(k,k')}\in\widehat{\mathscr{Z}_0}\}f(\widehat{W}_n(\mathcal{Z}_{(k,k')}))= \sqrt{n}1\{\mathcal{Z}_{(k,k')}\in\widehat{\mathscr{Z}_0}\} \frac{A_{n}}{B_{n}}\overset{p}{\rightarrow} 0.
\end{align*}
This implies that
\[
\sqrt{n}\left\{  \mathcal{I}_{0}\left(  \widehat{\mathscr{Z}_{0}}\right)
\mathcal{F}_{0}\left(  \widehat{W}_{n}\right)  -\mathcal{I}_{0}\left(
\mathscr{Z}_{\bar{M}}\right)  \mathcal{F}_{0}\left(  W\right)  \right\}
\overset{p}{\rightarrow}0.
\]
By Lemma 1.10.2(iii) and Example 1.4.7 (Slutsky's lemma) of \citet{van1996weak},
\begin{align}\label{eq.weak convergence pairwise2 varying DGPs}
\sqrt{n}\left(  \widehat{\beta}_{1}-\beta_{1}\right) & =\sqrt{n}\left\{
\left(\begin{array}
[c]{c}%
\mathcal{I}_{1}\left(  \widehat{\mathscr{Z}_{0}}\right)  \mathcal{F}%
_{1}\left(  \widehat{W}_{n}\right)  \\
\mathcal{I}_{0}\left(  \widehat{\mathscr{Z}_{0}}\right)  \mathcal{F}%
_{0}\left(  \widehat{W}_{n}\right)
\end{array}\right)
-%
\left(\begin{array}
[c]{c}%
\mathcal{I}_{1}\left(  \mathscr{Z}_{\bar{M}}\right)  \mathcal{F}_{1}\left(
W\right)  \\
\mathcal{I}_{0}\left(  \mathscr{Z}_{\bar{M}}\right)  \mathcal{F}_{0}\left(
W\right)
\end{array}\right)
\right\}  \notag\\
&  \overset{d}{\rightarrow}
\left(\begin{array}
[c]{c}%
\mathcal{I}_{1}\left(  \mathscr{Z}_{\bar{M}}\right)  \mathcal{F}_{1}^{\prime
}\left(  W_P\right)  N\left(  0,\Sigma_{P}\right)  \\
0
\end{array}\right).
\end{align}

Since $P_n\to P$ by Theorem 3.10.12 of \citet{van1996weak}, Assumption \ref{ass.first stage varying DGPs} gives that for every $\mathcal{Z}_{(k,k')}\in\mathscr{Z}_{\bar{M}}$, for sufficiently large $n$,
\begin{align*}
	&  \frac{E\left[  g\left(  Z_{i}\right)  D_{i}1\left\{  Z_{i}\in
		\mathcal{Z}_{(k,k^{\prime})}\right\}  \right]  }{\mathbb{P}\left(  Z_{i}%
		\in\mathcal{Z}_{(k,k^{\prime})}\right)  }-\frac{E\left[  D_{i}1\left\{
		Z_{i}\in\mathcal{Z}_{(k,k^{\prime})}\right\}  \right]  }{\mathbb{P}\left(
		Z_{i}\in\mathcal{Z}_{(k,k^{\prime})}\right)  }\frac{E\left[  g\left(
		Z_{i}\right)  1\left\{  Z_{i}\in\mathcal{Z}_{(k,k^{\prime})}\right\}  \right]
	}{\mathbb{P}\left(  Z_{i}\in\mathcal{Z}_{(k,k^{\prime})}\right)  }\neq0.
	\end{align*}
Then we can follow the proof of Theorem \ref{thm.IV estimator asymptotics pairwise} to show that when $n$ is sufficiently large, $\beta_{(  k,k^{\prime})  }^{1}=\beta_{k^{\prime},k}$ for every $(z_k,z_{k^{\prime}})\in\mathscr{Z}_{\bar{M}}$.
\end{proof}

Appendix \ref{sec.estimation Z_0 ordered} provides the consistent estimation for $\mathscr{Z}_0$ under a fixed underlying distribution. It is straightforward to extend the results to varying underlying distributions $\{P_n\}$ under Assumption \ref{ass.probability path} by applying Theorem 3.10.12 of \citet{van1996weak}. Similarly, we can also extend the results to multivalued unordered treatments.
We omit the proofs of these results.

\subsection{VSIV Estimation vs. Pretest}
\label{app:pretest}
As an alternative to VSIV estimation, we could first test the IV validity assumptions for every pair of values of $Z$ using the methods of \citet{huber2015testing}, \citet{kitagawa2015test}, \citet{mourifie2016testing}, and \citet{sun2021ivvalidity}. Only if a pair passes the test, we then proceed and estimate the corresponding LATE.

We first consider this pretest procedure in a simple case where the instrument $Z$ is binary with $\mathcal{Z}=\{z_1,z_2\}$. 
Let $\Phi_n$ denote the test function of a test for IV validity, such that $\Phi_n=1$ indicates rejection and $\Phi_n=0$ indicates non-rejection. 
Let $\widehat{\gamma}$ be the estimator for LATE $\gamma$ proposed by \citet{imbens1994identification}. Then we have that for all $a\in\mathbb{R}$,
\begin{align*}
    \mathbb{P}(\sqrt{n}(\widehat{\gamma}-\gamma)\le a|\Phi_n=0)=\frac{\mathbb{P}(\Phi_n=0, \sqrt{n}(\widehat{\gamma}-\gamma)\le a)}{\mathbb{P}(\Phi_n=0)}.
\end{align*}
Suppose $Z$ is valid. If $\mathbb{P}(\Phi_n=0)\not\to 1$ as in tests with a fixed significance level $\alpha>0$ \citep{huber2015testing,kitagawa2015test,mourifie2016testing,sun2021ivvalidity}, then the limit of the conditional probability $\mathbb{P}(\sqrt{n}(\widehat{\gamma}-\gamma)\le a|\Phi_n=0)$ may be different from that of $\mathbb{P}(\sqrt{n}(\widehat{\gamma}-\gamma)\le a)$. This implies that after the pretest we would need to consider the distribution of $\sqrt{n}(\widehat{\gamma}-\gamma)$ conditional on $\Phi_n=0$, which is not tractable given the non-standard nature of the tests for IV validity. If we use the unconditional limiting distribution of $\sqrt{n}(\widehat{\gamma}-\gamma)$ instead to do the inference, the result could be misleading. 

The cases where $Z$ is multivalued are analogous but more complicated. Suppose now $\mathcal{Z}=\{z_1,z_2,z_3\}$. 
Let $\Phi_n^{(z,z')}$ denote the test function for IV validity for pair $(z,z')$, and $\widehat{\gamma}_{(z,z')}$ be the estimator for LATE ${\gamma}_{(z,z')}$ for pair $(z,z')$. Suppose in this case $\mathscr{Z}_{\bar{M}}=\{(z_1,z_2), (z_1,z_3)\}$. Let $A_n=\{\Phi_n^{(z_1,z_2)}=0,\Phi_n^{(z_1,z_3)}=0,\Phi_n^{(z_2,z_3)}=1\}$. Then the limit of the conditional probability
\begin{align}\label{eq.gamma conditional distribution}
    &\mathbb{P}(\sqrt{n}(\widehat{\gamma}_{(z_1,z_2)}-{\gamma}_{(z_1,z_2)})\le a_1,\sqrt{n}(\widehat{\gamma}_{(z_1,z_3)}-{\gamma}_{(z_1,z_3)})\le a_2|A_n)\notag\\
    =&\,{\mathbb{P}(\sqrt{n}(\widehat{\gamma}_{(z_1,z_2)}-{\gamma}_{(z_1,z_2)})\le a_1,\sqrt{n}(\widehat{\gamma}_{(z_1,z_3)}-{\gamma}_{(z_1,z_3)})\le a_2,A_n)}/{\mathbb{P}(A_n)}
\end{align}
could be different from that of $\mathbb{P}(\sqrt{n}(\widehat{\gamma}_{(z_1,z_2)}-{\gamma}_{(z_1,z_2)})\le a_1,\sqrt{n}(\widehat{\gamma}_{(z_1,z_3)}-{\gamma}_{(z_1,z_3)})\le a_2)$ for some $a_1,a_2\in\mathbb{R}$ if $\mathbb{P}(\Phi_n^{(z_1,z_2)}=0,\Phi_n^{(z_1,z_3)}=0,\Phi_n^{(z_2,z_3)}=1)\not\to1$.

In our framework, the VSIV estimator for every pair $(z_k,z_{k'})$ could also be constructed with
$$1\{(z_k,z_{k'})\in\widehat{\mathscr{Z}_0}\}=1-\Phi_n^{(z_k,z_{k'})},$$
which means if the pair $(z_k,z_{k'})$ is not rejected ($\Phi_n^{(z_k,z_{k'})}=0$), then we include $(z_k,z_{k'})$ in $\widehat{\mathscr{Z}_0}$. The VSIV estimator can then be written as 
$\widehat{\beta}^1_{(k,k')}=(1-\Phi_n^{(z_k,z_{k'})})\cdot\widehat{\gamma}_{(z_k,z_{k'})}$,
where $\widehat{\gamma}_{(z_k,z_{k'})}$ denotes the traditional LATE estimator.
This estimator is similar to that of \citet{leeb2005model} in the model selection context. That is, we select the valid pairs of values of $Z$ and construct the estimator based on pretest selection.  
Under this setting, if $(z_k,z_{k'})$ is valid, we may have $\mathbb{P}((z_k,z_{k'})\in\widehat{\mathscr{Z}_0})=\mathbb{P}(\Phi_n^{(z_k,z_{k'})}=0)\not\to1$, which corresponds to the conservative model selection discussed in \citet{leeb2005model}. It follows that the weak convergence result in Theorem \ref{thm.IV estimator asymptotics pairwise} fails.

The above discussion assumes that $\alpha$ is fixed. If we allow $\alpha$ to depend on $n$ and let $\alpha\to0$ at some particular rate, we may have that in \eqref{eq.gamma conditional distribution},
\begin{align*}
    \lim_{n\to\infty}\mathbb{P}(\Phi_n^{(z_1,z_2)}=0,\Phi_n^{(z_1,z_3)}=0,\Phi_n^{(z_2,z_3)}=1)= 1
\end{align*}
and therefore
\begin{align*}
    &\lim_{n\to\infty}\mathbb{P}(\sqrt{n}(\widehat{\gamma}_{(z_1,z_2)}-{\gamma}_{(z_1,z_2)})\le a_1,\sqrt{n}(\widehat{\gamma}_{(z_1,z_3)}-{\gamma}_{(z_1,z_3)})\le a_2|A_n)\notag\\
    =&\lim_{n\to\infty}{\mathbb{P}(\sqrt{n}(\widehat{\gamma}_{(z_1,z_2)}-{\gamma}_{(z_1,z_2)})\le a_1,\sqrt{n}(\widehat{\gamma}_{(z_1,z_3)}-{\gamma}_{(z_1,z_3)})\le a_2)}.
\end{align*}
In this case, it might be possible to derive a result as in Theorem \ref{thm.IV estimator asymptotics pairwise} with $\widehat{\beta}^1_{(k,k')}=(1-\Phi_n^{(z_k,z_{k'})})\cdot\widehat{\gamma}_{(z_k,z_{k'})}$. The significance level $\alpha$ would be the tuning parameter that needs to be determined in practice. We leave the derivation of such a result for future research.

\section{Proofs and Supplementary Results for Appendix \ref{sec.unordered treatment}}\label{sec.proofs unordered treatment}

\subsection{Proofs for Appendix \ref{sec.unordered treatment}}

\begin{proof}[Proof of Lemma \ref{lemma.equivalent characterizations unordered pairwise}]
	(i) $\Leftrightarrow$ (ii). We closely follow the proof for ``(i)
	$\Leftrightarrow$ (ii)'' in Theorem T-3 of \citet{heckman2018unordered}. By Lemma L-5
	of \citet{heckman2018unordered}, if ${B}_{d(k,k')} $ is lonesum, then
	no $2\times2$ sub-matrix of ${B}_{d(k,k')}$ takes the form
	\begin{align}\label{eq.binary form for lonesum}
		\left(
		\begin{array}
			[c]{cc}%
			1 & 0\\
			0 & 1
		\end{array}
		\right)  \text{ or }\left(
		\begin{array}
			[c]{cc}%
			0 & 1\\
			1 & 0
		\end{array}
		\right)  .
	\end{align}
	Since ${B}_{d(k,k')}=1\{  \mathcal{K}_{(k,k')}R=d\}  $, (i) $\Rightarrow$
	(ii). Suppose (ii) holds. Then no $2\times2$ sub-matrix of ${B}_{d(k,k')}$
	takes the form in \eqref{eq.binary form for lonesum} by the definition of ${B}_{d(k,k')}$. By Lemmas L-6 and L-8 of \citet{heckman2018unordered}, (i) holds.
	
	(i) $\Rightarrow$ (iii) $\Rightarrow$ (ii). If for every $d\in\mathcal{D}$,
	${B}_{d(k,k')}$ is lonesum, by Lemma L-9 of \citet{heckman2018unordered}, 
	\[
	{B}_{d(k,k')}\left(  1,l\right)  \leq {B}_{d(k,k')}\left(  2,l\right)
	\text{ for all }l,\text{ or } {B}_{d(k,k')}\left(  1,l\right)  \geq
	{B}_{d(k,k')}\left(  2,l\right)  \text{ for all }l\text{.}
	\]
	Because the value of
	$(  D_{z_k}  ,D_{z_{k'}}  )  $
	must be equal to $(  \mathcal{K}_{(k,k')}R\left(  1,l\right)  ,\mathcal{K}_{(k,k')}R\left(
	2,l\right)  )  $ for some $l$, it follows that
	\begin{align*}
		1\left\{  D_{z_k}  =d\right\}     \leq1\left\{
		D_{z_{k'}}  =d\right\}  
		\text{ or }1\left\{  D_{z_k}  =d\right\}   
		\geq1\left\{  D_{z_{k^{\prime}}}  =d\right\} .
	\end{align*}
	Thus the following sub-matrices will not appear in $\mathcal{K}_{(k,k')}R$:
	\[
	\left(
	\begin{array}
		[c]{cc}%
		d & d^{\prime}\\
		d^{\prime\prime} & d
	\end{array}
	\right)  \text{ or }\left(
	\begin{array}
		[c]{cc}%
		d^{\prime} & d\\
		d & d^{\prime\prime}%
	\end{array}
	\right),
	\]
	where $d'\neq d$ and $d''\neq d$.
\end{proof}

\begin{proof}[Proof of Theorem \ref{theorem.counterfactuals identified pairwise}]
	The proof follows a strategy similar to that of the proof of Theorem T-6 in
	\citet{heckman2018unordered}. We first write
	\begin{align}\label{eq.PKS}
	\mathbb{P}\left(  \mathcal{M}_{(k,k^{\prime})}S\in\Sigma_{d(k,k^{\prime}%
		)}\left(  t\right)  \right)  =b_{d(k,k^{\prime})}\left(  t\right)
	P_{S(k,k^{\prime})}.
	\end{align}
	Also, since%
	\begin{align*}
		&  E\left[  \kappa\left(  Y_d(z_k,z_{k'})  \right)  1\left\{  \mathcal{M}%
		_{(k,k^{\prime})}S\in\Sigma_{d(k,k^{\prime})}\left(  t\right)  \right\}
		\right]  \\
		=&\,E\left[  E\left[  \kappa\left(  Y_d(z_k,z_{k'})  \right)  1\left\{
		\mathcal{M}_{(k,k^{\prime})}S\in\Sigma_{d(k,k^{\prime})}\left(  t\right)
		\right\}  |1\left\{  \mathcal{M}_{(k,k^{\prime})}S\in\Sigma_{d(k,k^{\prime}%
			)}\left(  t\right)  \right\}  \right]  \right]  \\
		=&\,E\left[  \kappa\left(  Y_d(z_k,z_{k'})  \right)  |\mathcal{M}%
		_{(k,k^{\prime})}S\in\Sigma_{d(k,k^{\prime})}\left(  t\right)  \right]
		\cdot\mathbb{P}\left(  \mathcal{M}_{(k,k^{\prime})}S\in\Sigma_{d(k,k^{\prime
			})}\left(  t\right)  \right)
	\end{align*}
	and%
	\begin{align*}
		&  E\left[  \kappa\left(  Y_d(z_k,z_{k'})  \right)  1\left\{  \mathcal{M}%
		_{(k,k^{\prime})}S\in\Sigma_{d(k,k^{\prime})}\left(  t\right)  \right\}
		\right]  \\
		=&\,E\left[  \kappa\left(  Y_d(z_k,z_{k'})  \right)  \sum_{l=1}^{L_{(k,k')}}1\left\{
		\mathcal{M}_{(k,k^{\prime})}S=s_{l}\right\}  1\left\{  s_{l}\in\Sigma
		_{d(k,k^{\prime})}\left(  t\right)  \right\}  \right]  =b_{d(k,k^{\prime})}\left(  t\right)  Q_{S(k,k^{\prime})}\left(  d\right),
	\end{align*}
	we have that
	\begin{align}\label{eq.EY}
	E\left[  \kappa\left(  Y_d(z_k,z_{k'})  \right)  |\mathcal{M}_{(k,k^{\prime
		})}S\in\Sigma_{d(k,k^{\prime})}\left(  t\right)  \right]  =\frac
	{b_{d(k,k^{\prime})}\left(  t\right)  Q_{S(k,k^{\prime})}\left(  d\right)
	}{b_{d(k,k^{\prime})}\left(  t\right)  P_{S(k,k^{\prime})}}.
	\end{align}
	Now we suppose $(z_k,z_{k'})\in\mathscr{Z}_{\bar{M}}$. By definition,  $P_{Z(k,k^{\prime})}\left(  d\right)  =B_{d(k,k^{\prime
		})}P_{S(k,k^{\prime})}$ and $Q_{Z(k,k^{\prime})}\left(  d\right)
	=B_{d(k,k^{\prime})}Q_{S(k,k^{\prime})}\left(  d\right)  $, so by Lemma L-2 of
	\citet{heckman2018unordered},
	\begin{align*}
		&  b_{d(k,k^{\prime})}\left(  t\right)  P_{S(k,k^{\prime})}=b_{d(k,k^{\prime
			})}\left(  t\right)  \left[  B_{d(k,k^{\prime})}^{+}P_{Z(k,k^{\prime})}\left(
		d\right)  +\left(  I-B_{d(k,k^{\prime})}^{+}B_{d(k,k^{\prime})}\right)
		\lambda_{P}\right]  \text{ and }\\
		&  b_{d(k,k^{\prime})}\left(  t\right)  Q_{S(k,k^{\prime})}\left(  d\right)
		=b_{d(k,k^{\prime})}\left(  t\right)  \left[  B_{d(k,k^{\prime})}%
		^{+}Q_{Z(k,k^{\prime})}\left(  d\right)  +\left(  I-B_{d(k,k^{\prime})}%
		^{+}B_{d(k,k^{\prime})}\right)  \lambda_{Q}\right]  ,
	\end{align*}
	where $\lambda_{P}$ and $\lambda_{Q}$ are some real-valued vectors.
	
	We next show that $b_{d(k,k^{\prime})}\left(  t\right)  [I-B_{d(k,k^{\prime}%
		)}^{+}B_{d(k,k^{\prime})}]=0$. First, by the proof of Lemma L-16 of
	\citet{heckman2018unordered} and Lemma
	\ref{lemma.equivalent characterizations unordered pairwise} in this paper, if
	$B_{d(k,k^{\prime})}\left(  \cdot,l\right)  $ and $B_{d(k,k^{\prime})}\left(
	\cdot,l^{\prime}\right)  $ have the same sum, then these two vectors are
	identical. Thus, by the definition of the set $\Sigma_{d(k,k^{\prime})}\left(t\right)  $, for all $s_{l},s_{l^{\prime}}\in\Sigma_{d(k,k^{\prime}	)}\left(  t\right)  $, $B_{d(k,k^{\prime})}\left(  \cdot,l\right)
	=B_{d(k,k^{\prime})}\left(  \cdot,l^{\prime}\right)  $. Let $C_{d(k,k^{\prime
		})}\left(  t\right)  =B_{d(k,k^{\prime})}\left(  \cdot,l\right)  $ with $l$
	satisfying that $s_{l}\in\Sigma_{d(k,k^{\prime})}\left(  t\right)  $, where
	$s_{l}$ is the $l$th column of $\mathcal{K}_{(k,k^{\prime})}R$. Define by
	$C_{d(k,k^{\prime})}=(C_{d(k,k^{\prime})}(1),C_{d(k,k^{\prime})}(2))$ the
	matrix consisting of all unique nonzero vectors in $B_{d(k,k^{\prime})}$.\footnote{Without loss of generality, we assume that both $C_{d(k,k^{\prime})}(1)$ and $C_{d(k,k^{\prime})}(2)$ exist.}
	Then clearly $C_{d(k,k^{\prime})}$ has full column rank and
	$C_{d(k,k^{\prime})}^{T}C_{d(k,k^{\prime})}$ has full rank. Thus,
	$(C_{d(k,k^{\prime})}^{T}C_{d(k,k^{\prime})})^{-1}$ exists. Let
	$D_{d(k,k^{\prime})}=(b_{d(k,k^{\prime})}\left(  1\right)  ^{T}%
	,b_{d(k,k^{\prime})}(2)^{T})^{T}$. Since by the definition of
	$b_{d(k,k^{\prime})}\left(  t\right)  $, $b_{d(k,k^{\prime})}\left(  t\right)
	\cdot b_{d(k,k^{\prime})}\left(  t^{\prime}\right)  ^{T}=0$ for $t\neq
	t^{\prime}$, $D_{d(k,k^{\prime})}$ has full row rank and $(D_{d(k,k^{\prime}%
		)}D_{d(k,k^{\prime})}^{T})^{-1}$ exists. We then decompose $B_{d(k,k^{\prime
		})}=C_{d(k,k^{\prime})}\cdot D_{d(k,k^{\prime})}$.\footnote{See Remark A.3 of \citet{heckman2018unordered}.}
	
	Now by similar arguments as in the proof of Lemma L-17 of \citet{heckman2018unordered}, we can
	show that the Moore--Penrose pseudo inverse of ${B}_{d(k,k^{\prime})}$ is
	\[
	B_{d(k,k^{\prime})}^{+}=D_{d(k,k^{\prime})}^{T}(D_{d(k,k^{\prime}%
		)}D_{d(k,k^{\prime})}^{T})^{-1}(C_{d(k,k^{\prime})}^{T}C_{d(k,k^{\prime}%
		)})^{-1}C_{d(k,k^{\prime})}^{T}.
	\]
	For $t\in\left\{  1,2\right\}  $, we can write $b_{d(k,k^{\prime})}\left(
	t\right)  =e_{t}D_{d(k,k^{\prime})}$, where $e_{t}$ is a row vector in which
	the $t$th element is $1$ and the other element is $0$. 
 
 Then we have
	that
	\begin{align*}
		&b_{d(k,k^{\prime})}\left(  t\right)  [I-B_{d(k,k^{\prime})}^{+}%
		B_{d(k,k^{\prime})}] =b_{d(k,k^{\prime})}\left(  t\right)
		-b_{d(k,k^{\prime})}\left(  t\right)  B_{d(k,k^{\prime})}^{+}B_{d(k,k^{\prime
			})}\\
		=&\,b_{d(k,k^{\prime})}\left(  t\right)  -e_{t}D_{d(k,k^{\prime}%
			)}D_{d(k,k^{\prime})}^{T}(D_{d(k,k^{\prime})}D_{d(k,k^{\prime})}^{T}%
		)^{-1}(C_{d(k,k^{\prime})}^{T}C_{d(k,k^{\prime})})^{-1}C_{d(k,k^{\prime})}%
		^{T}C_{d(k,k^{\prime})}\cdot D_{d(k,k^{\prime})}\\
		=&\,0.
	\end{align*}
	This implies that $b_{d(k,k^{\prime})}\left(  t\right)  P_{S(k,k^{\prime})}$
	and $b_{d(k,k^{\prime})}\left(  t\right)  Q_{S(k,k^{\prime})}\left(  d\right)
	$ can be identified as
	\begin{align*}
		&b_{d(k,k^{\prime})}\left(  t\right)  P_{S(k,k^{\prime})}=b_{d(k,k^{\prime}%
			)}\left(  t\right)  B_{d(k,k^{\prime})}^{+}P_{Z(k,k^{\prime})}\left(
		d\right) \\ &\text{and }b_{d(k,k^{\prime})}\left(  t\right)  Q_{S(k,k^{\prime}%
			)}\left(  d\right)  =b_{d(k,k^{\prime})}\left(  t\right)  B_{d(k,k^{\prime}%
			)}^{+}Q_{Z(k,k^{\prime})}\left(  d\right)  .
	\end{align*}
	Thus, \eqref{eq.PKS} and \eqref{eq.EY} show that
	\begin{align*}
		&\mathbb{P}\left(  \mathcal{M}_{(k,k^{\prime})}S\in\Sigma_{d(k,k^{\prime}%
			)}\left(  t\right)  \right)  =b_{d(k,k^{\prime})}\left(  t\right)
		B_{d(k,k^{\prime})}^{+}P_{Z(k,k^{\prime})}\left(  d\right) \\ 
		&\text{and	}E\left[  \kappa\left(  Y_d(z_k,z_{k'})  \right)  |\mathcal{M}%
		_{(k,k^{\prime})}S\in\Sigma_{d(k,k^{\prime})}\left(  t\right)  \right]
		=\frac{b_{d(k,k^{\prime})}\left(  t\right)  B_{d(k,k^{\prime})}^{+}%
			Q_{Z(k,k')}\left(  d\right)  }{b_{d(k,k^{\prime})}\left(  t\right)  B_{d(k,k^{\prime
				})}^{+}P_{Z(k,k')}\left(  d\right)  }
	\end{align*}
	are identified. Define
	\[
Z_{Pi}=(1\left\{  Z_{i}=z_{1}\right\}  ,\ldots,1\left\{  Z_{i}=z_{K}\right\}
),
\]%
\[
P_{DZi}\left(  d\right)  =\left(  1\left\{  D_{i}=d,Z_{i}=z_{1}\right\}
,\ldots,1\left\{  D_{i}=d,Z_{i}=z_{K}\right\}  \right)  ^{T}\text{ for all }d,
\]%
\[
Q_{YDZi}\left(  d\right)  =\left(  \kappa\left(  Y_{i}\right)  1\left\{
D_{i}=d,Z_{i}=z_{1}\right\}  ,\ldots,\kappa\left(  Y_{i}\right)  1\left\{
D_{i}=d,Z_{i}=z_{K}\right\}  \right)  ^{T}\text{ for all }d,
\]
and
\[
W_{i}=\left(  Z_{Pi},P_{DZi}\left(  d_{1}\right)  ^{T},\ldots,P_{DZi}\left(
d_{J}\right)  ^{T},Q_{YDZi}\left(  d_{1}\right)  ^{T},\ldots,Q_{YDZi}\left(
d_{J}\right)  ^{T}\right)  ^{T}.
\]
	By multivariate central limit theorem, $\sqrt{n}(\widehat{W}-W)\overset{d}\rightarrow
	N\left(  0,\Sigma_{W}\right)  $, where
	\begin{align}\label{eq.sigmaW}
	\Sigma_W=E[(W_i-W)(W_i-W)^T],    
	\end{align}
	and therefore $\widehat{W} \overset{p}\rightarrow W$.
	Also, for every $\varepsilon>0$,  $\mathbb{P}(\sqrt{n}\Vert\mathds{1}(\widehat{\mathscr{Z}_0})-\mathds{1}({\mathscr{Z}_{\bar{M}}})\Vert_2>\varepsilon)\le \mathbb{P}(\widehat{\mathscr{Z}_0}\neq{\mathscr{Z}_{\bar{M}}})\to0$ by assumption.
	Then, by Lemma 1.10.2(iii) and Example 1.4.7 (Slutsky's lemma) of
	\citet{van1996weak},
	\begin{align*}
		\sqrt{n}\left\{  \left(  \widehat{W}^T,\mathds{1}(\widehat{\mathscr{Z}_0})^T\right)^T  -\left(  W^T,\mathds{1}({\mathscr{Z}_{\bar{M}}})^T\right)^T
		\right\}  
		 \overset{d}\rightarrow\left(  N\left(  0,\Sigma
		_{W}\right)^T  ,0^T\right)^T.
	\end{align*}
\end{proof}

\begin{proof}[Proof of Lemma \ref{lemma.beta unordered}]
If $(z_k,z_{k^{\prime}})\in\mathscr{Z}_{\bar{M}}$ and  $\Sigma_{d(k,k')}(t)=\Sigma_{d'(k,k')}(t')$, then $Y_{dz_k}=Y_d(z_k,z_{k'})$ a.s. and $Y_{d'z_{k'}}=Y_{d'}(z_k,z_{k'})$ a.s.
By \eqref{eq.counterfactuals pairwise}, it follows that
\begin{align}
		\beta_{(k,k')}(d,d',t,t')=\left\{\frac{b_{d(  k,k^{\prime})}\left(  t\right)  B_{d(  k,k^{\prime})}^{+}  Q_{Z(  k,k^{\prime})}\left(  d\right) }{b_{d(  k,k^{\prime})}\left(  t\right)  B_{d(  k,k^{\prime})}^{+}  P_{Z(  k,k^{\prime})}\left(  d\right) }
		-\frac{b_{d'(  k,k^{\prime})}\left(  t'\right)  B_{d'(  k,k^{\prime})}^{+}  Q_{Z(  k,k^{\prime})}\left(  d'\right) }{b_{d'(  k,k^{\prime})}\left(  t'\right)  B_{d'(  k,k^{\prime})}^{+}  P_{Z(  k,k^{\prime})}\left(  d'\right) }\right\}.
	\end{align}

If $(z_k,z_{k^{\prime}})\notin\mathscr{Z}_{\bar{M}}$ or  $\Sigma_{d(k,k')}(t)\neq\Sigma_{d'(k,k')}(t')$, clearly the lemma holds. 
\end{proof}

\begin{proof}[Proof of Theorem \ref{thm.bias reduction unordered}]
	The proof is similar to that of Theorem \ref{thm.bias reduction multi ordered}.
\end{proof}

\subsection{Definition and Estimation of $\mathscr{Z}_{0}$}\label{sec.estimation Z_0 unordered}

\subsubsection{Definition and Estimation of $\mathscr{Z}_1$}

Following \citet[main text and Appendix D]{sun2021ivvalidity}, we provide the definitions of $\mathscr{Z}_1$ and its estimator. Suppose the instrument $Z$ is pairwise valid with $\mathscr{Z}_{\bar
	{M}}=\{(z_{k_{1}},z_{k_{1}^{\prime}}),\ldots,(z_{k_{\bar{M}}},z_{k_{\bar{M}%
	}^{\prime}})\}$. 
Fix $(z,z^{\prime})\in\mathscr{Z}_{\bar{M}}$. For every $d\in\mathcal{D}$, if $1\{  D_{z'}
=d\}  \leq1\{  D_z  =d\}  $ a.s., we have that
\begin{align}\label{eq.inequality unordered}
	\mathbb{P}(  Y\in B,D=d|Z=z^{\prime})   &=E\left[  1\left\{
	Y_d(z,z')  \in B\right\}  \times1\left\{  D_{z'}
	=d\right\}  \right]\notag  \\
	&  \leq E\left[  1\left\{  Y_d(z,z')  \in B\right\}  \times1\left\{
	D_z  =d\right\}  \right]  =\mathbb{P}\left(  Y\in
	B,D=d|Z=z\right)
\end{align}
for all Borel sets $B$. 
Denote the $2^{J}$ $J$-dimensional different binary vectors by $v_{1}%
,\ldots,v_{2^{J}}$, where
\[
v_{1}=\left(
\begin{array}
	[c]{c}%
	0\\
	0\\
	\vdots\\
	0
\end{array}
\right)  ,v_{2}=\left(
\begin{array}
	[c]{c}%
	1\\
	0\\
	\vdots\\
	0
\end{array}
\right)  ,\ldots,v_{2^{J}}=\left(
\begin{array}
	[c]{c}%
	1\\
	1\\
	\vdots\\
	1
\end{array}
\right)  .
\]
Let $\mathcal{L}:\mathcal{D}\to \{1,\ldots,J\}$ map $d\in\mathcal{D}$ to $d$'s index in $\mathcal{D}$ so that if $d=d_j$, it follows that $\mathcal{L}(d)=j$. For every $q\in\{  1,\ldots,2^{J}\}  $, define $f_{q}:\left\{
d_{1},\ldots,d_{J}\right\}  \rightarrow\left\{  1,-1\right\}  $ by
$f_{q}\left(  d\right)  =\left(  -1\right)  ^{v_{q}\left(  \mathcal{L}(d)\right)  }$.
For every fixed $(z,z')\in\mathscr{Z}_{\bar{M}}$, there is $q\in\{1,\ldots,2^J\}$ such that
\begin{align*}
	f_q(d)\cdot\{\mathbb{P}\left(  Y\in B,D=d|Z=z'\right)  -\mathbb{P}\left(  Y\in B,D=d|Z=z\right)\}\le0
\end{align*} 
for all $d\in\mathcal{D}$ and all closed intervals $B$.
Then for all $q\in\{  1,\ldots,2^{J}\}  $, define%
\begin{align*}
	\mathrm{H}_{q}  &  =\left\{  f_{q}\left(  d\right)  \cdot1_{B\times\left\{
		d\right\}  \times\mathbb{R}}:B\text{ is a closed interval in }\mathbb{R}%
	\text{, }d\in\mathcal{D}\right\}  \text{ and}\\
	\mathrm{\bar{H}}_{q}  &  =\left\{  f_{q}\left(  d\right)  \cdot1_{B\times
		\left\{  d\right\}  \times\mathbb{R}}:B\text{ is a closed, open, or
		half-closed interval in }\mathbb{R}\text{, }d\in\mathcal{D}\right\}  .
\end{align*}
Furthermore, define the following function spaces
\begin{align}\label{def.function spaces unordered}
	{\mathrm{G}}=\left\{ \left( 1_{\mathbb{R}\times \mathbb{R}\times \left\{
		z_j\right\} },1_{\mathbb{R}\times \mathbb{R}\times \left\{ z_k\right\} }\right) :j,k\in\{1,\ldots,K\},j<k\right\} , \mathrm{H}=\cup_{q=1}^{2^{J}}\mathrm{H}_{q},\text{ and }\mathrm{\bar{H}}
	=\cup_{q=1}^{2^{J}}\mathrm{\bar{H}}_{q}.
\end{align}
Let $P$ and $\widehat{P}$ be defined as in Section \ref{sec.estimation validity set}. Let $\phi$, $\sigma^2$, $\widehat{\phi}$, and $\widehat{\sigma}^2$ be defined in a way similar to that in Section \ref{sec.estimation validity set} but for all $\left( h,g\right) \in {\bar{	\mathrm{H}}}\times\mathrm{G}$. Also, we let $\Lambda(P)=\prod_{k=1}^{K}P(1_{\mathbb{R}\times\mathbb{R}\times\{z_k\}})$  and $T_{n}=n\cdot \prod_{k=1}^K\widehat{P}(1_{\mathbb{R}\times \mathbb{R}
	\times \{z_k\}})$.
By similar arguments as in the proof of Lemma 3.1 in \citet{sun2021ivvalidity}, ${\sigma}^{2}$ and $\widehat{\sigma}^{2}$ are uniformly bounded in $(h,g)\in {\bar{	\mathrm{H}}}\times\mathrm{G}$.

The following lemma reformulates the testable restrictions in terms of $\phi$.
\begin{lemma}\label{lemma.superset of Z unordered pairwise}
	Suppose that the instrument $Z$ is pairwise valid for the treatment $D$ with the largest validity pair set $\mathscr{Z}_{\bar{M}}=\{(z_{k_1},z_{k_1^{\prime}}),\ldots,(z_{k_{\bar{M}}},z_{k_{\bar{M}}^{\prime}})\}$. For every $m\in\{1,\ldots,\bar{M}\}$, it follows that $\min_{q\in\{  1,\ldots,2^{J}\}  }\sup_{h\in\mathrm{H}_{q}}\phi ( h,g) =0$ with $g=( 1_{\mathbb{R}\times \mathbb{R}\times \{	z_{k_m}\} },1_{\mathbb{R}\times \mathbb{R}\times \{ z_{k_{m}^{\prime}}\} })$.
\end{lemma}

\begin{proof}[Proof of Lemma \ref{lemma.superset of Z unordered pairwise}]
	Since we can find $a\in\mathbb{R}$ and $d\in\mathcal{D}$ such that
	$P\left(  1_{\left\{  a\right\}  \times\left\{  d\right\}  \times\mathbb{R}%
	}\right)  =0$, then we have $\sup_{h\in\mathrm{H}_{q}}\phi\left(
	h,g\right)  \geq0$ for every $q$ and every $g\in\mathrm{G}$. So for every $g\in\mathrm{G}$, it follows that $\min_{q\in\left\{  1,\ldots,2^{J}\right\}  }\sup_{h\in\mathrm{H}_{q}}	\phi\left(  h,g\right)  \ge0$.
	Let $h_{Bd}=1_{B\times\left\{  d\right\}  \times\mathbb{R}}$ for every closed interval $B$ and every $d\in\mathcal{D}$. Fix $m\in\{1,\ldots,\bar{M}\}$. By assumption, for every $d\in\mathcal{D}$, we have
	\begin{align*}
		&\phi\left(  h_{Bd},g\right)     =\frac{P\left(  h_{Bd}\cdot g_{2}\right)
		}{P\left(  g_{2}\right)  }-\frac{P\left(  h_{Bd}\cdot g_{1}\right)  }{P\left(
			g_{1}\right)  }\leq0\text{ for every closed interval }B,\text{ }\\
		&\text{or }\phi\left(  -h_{Bd},g\right)    =\frac{-P\left(  h_{Bd}\cdot
			g_{2}\right)  }{P\left(  g_{2}\right)  }-\frac{-P\left(  h_{Bd}\cdot
			g_{1}\right)  }{P\left(  g_{1}\right)  }\leq0\text{ for every closed interval
		}B,
	\end{align*}
	where $g_1=1_{\mathbb{R}\times\mathbb{R}\times\{z_{k_m}\}}$, $g_2=1_{\mathbb{R}\times\mathbb{R}\times\{z_{k_{m}^{\prime}}\}}$, and $g=(g_1,g_2)$.
	This implies
	that there is $\mathrm{H}_{q}$ such that $\sup_{h\in\mathrm{H}_{q}}\phi\left(
	h,g\right)  \leq0$. Thus, it follows that $\min_{q\in\left\{  1,\ldots,2^{J}\right\}  }\sup_{h\in\mathrm{H}_{q}}	\phi\left(  h,g\right)  =0$.
\end{proof}

By Lemma \ref{lemma.superset of Z unordered pairwise}, we define  
\begin{align}\label{eq.G0 pair unordered}
	&\mathrm{G}_{1}=\left\{ g\in \mathrm{G}:\min_{q\in\left\{  1,\ldots,2^{J}\right\}  }\sup_{h\in\mathrm{H}_{q}}\phi\left(
	h,g\right) =0\right\} \text{ and }\notag\\
	&\widehat{\mathrm{G}_{1}}=\left\{ g\in 
	\mathrm{G}:\sqrt{T_n}\left\vert \min_{q\in\left\{  1,\ldots,2^{J}\right\}  }\sup_{h\in\mathrm{H}_{q}}\frac{\widehat{\phi}
		\left( h,g\right)}{\xi_{0}\vee \widehat{\sigma}(h,g)} \right\vert \leq \tau
	_{n}\right\}
\end{align}
with $\tau_{n}\to\infty$ and $\tau_{n}/\sqrt{n}\to0$ as $n\to\infty$, where $\xi_{0}$ is a small positive number. 
We define ${\mathscr{Z}_1}$ as the collection of all $(z,z^{\prime})$ that are associated with some $g\in{\mathrm{G}_1}$:
\begin{align}\label{eq.Z1 pair unordered}
	\mathscr{Z}_1=\left\{ (z_{k},z_{k^{\prime}})\in\mathscr{Z}: g=( 1_{\mathbb{R}\times \mathbb{R}\times \{	z_{k}\} },1_{\mathbb{R}\times \mathbb{R}\times \{ z_{k^{\prime}}\} })\in\mathrm{G}_1\right\}.
\end{align}
We use $\widehat{\mathrm{G}_1}$ to construct the estimator of $\mathscr{Z}_1$, denoted by $\widehat{\mathscr{Z}_1}$, which is defined as the set of all $(z,z^{\prime})$ that are associated with some $g\in\widehat{\mathrm{G}_1}$ in the same way $\mathscr{Z}_1$ is defined based on $\mathrm{G}_1$:
\begin{align}\label{eq.Z1_hat pair unordered}
	\widehat{\mathscr{Z}_1}=\left\{ (z_{k},z_{k^{\prime}})\in\mathscr{Z}: g=( 1_{\mathbb{R}\times \mathbb{R}\times \{	z_{k}\} },1_{\mathbb{R}\times \mathbb{R}\times \{ z_{k^{\prime}}\} })\in\widehat{\mathrm{G}_1}\right\}.
\end{align}

To derive the desired consistency result, we state and prove an additional auxiliary lemma.

\begin{lemma}
	\label{lemma.uniform and weak convergence unordered} Under Assumption \ref{ass.iid data unordered},  $\widehat{\phi}\to \phi$, $T_n/n\to \Lambda(P)$, and $\widehat{\sigma}\to \sigma$ almost uniformly. In addition, $\sqrt{T_n}(  \widehat{\phi}-\phi)  \leadsto\mathbb{G}$ for some random element
	$\mathbb{G}$, and for all $\left(  h,g\right)  \in\bar{\mathrm{H}}\times\mathrm{G}$ with $g=(g_1,g_2)$, the variance $Var\left(  \mathbb{G}\left(  h,g\right)  \right)=\sigma^2(h,g)$.
\end{lemma}

\begin{proof}[Proof of Lemma \ref{lemma.uniform and weak convergence unordered}]
	Note that the spaces $\bar{\mathrm{H}}$ and $\mathrm{ G}$ defined in \eqref{def.function spaces unordered} are similar to the spaces $\bar{\mathcal{H}}$ and $\mathcal{G}_P$ defined in \eqref{def.function spaces}. The lemma can be proved following a strategy similar to that of the proof of Lemma \ref{lemma.uniform and weak convergence}.
\end{proof}	

\begin{proposition}
	\label{prop.consistent G hat unordered pairwise}  Suppose the instrument $Z$ is pairwise valid for the treatment $D$ as defined in Definition \ref{def.partial validity unordered pairwise}. Under Assumption \ref{ass.iid data unordered},  $\mathbb{P}(\widehat{\mathrm{G}_{1}}=\mathrm{G}_{1})\rightarrow 1$, and thus $\mathbb{P}(\widehat{\mathscr{Z}_{1}}=\mathscr{Z}_{1})\rightarrow 1$.
\end{proposition}

\begin{proof}[Proof of Proposition \ref{prop.consistent G hat unordered pairwise}]
	First, suppose $\mathrm{G}_{1}\neq\varnothing$. Then we have that 
	\begin{align*}
		\min
		_{q\in\left\{  1,\ldots,2^{J}\right\}  }\sup_{h\in\mathrm{H}_{q}}\{{\phi\left(
			h,g\right)  }/({\xi_{0}\vee\widehat{\sigma}\left(  h,g\right)  })\}=0 
	\end{align*}
	for all
	$g\in\mathrm{G}_{1}$.	
	Under the constructions, we have that
	\begin{align*}
		&  \lim_{n\rightarrow\infty}\mathbb{P}\left(  \mathrm{G}_{1}\setminus
		\widehat{{\mathrm{G}}_{1}}\neq\varnothing\right)  \\
		\leq &  \lim_{n\rightarrow\infty}\mathbb{P}\left(  \max_{g\in\mathrm{G}_{1}%
		}\sqrt{T_{n}}\left\vert
		\begin{array}
			[c]{c}%
			\min_{q\in\left\{  1,\ldots,2^{J}\right\}  }\sup_{h\in\mathrm{H}_{q}}%
			\frac{\widehat{\phi}\left(  h,g\right)  }{\xi_{0}\vee\widehat{\sigma}\left(
				h,g\right)  }\\
			-\min_{q\in\left\{  1,\ldots,2^{J}\right\}  }\sup_{h\in\mathrm{H}_{q}}%
			\frac{\phi\left(  h,g\right)  }{\xi_{0}\vee\widehat{\sigma}\left(  h,g\right)
			}%
		\end{array}
		\right\vert >\tau_{n}\right)  \\
		= &  \lim_{n\rightarrow\infty}\mathbb{P}\left(  \max_{g\in\mathrm{G}_{1}}%
		\sqrt{T_{n}}\left\vert
		\begin{array}
			[c]{c}%
			-\max_{q\in\left\{  1,\ldots,2^{J}\right\}  }\left(  -\sup_{h\in\mathrm{H}%
				_{q}}\frac{\widehat{\phi}\left(  h,g\right)  }{\xi_{0}\vee\widehat{\sigma
				}\left(  h,g\right)  }\right)  \\
			+\max_{q\in\left\{  1,\ldots,2^{J}\right\}  }\left(  -\sup_{h\in\mathrm{H}%
				_{q}}\frac{\phi\left(  h,g\right)  }{\xi_{0}\vee\widehat{\sigma}\left(
				h,g\right)  }\right)
		\end{array}
		\right\vert >\tau_{n}\right)  \\
		\leq &  \lim_{n\rightarrow\infty}\mathbb{P}\left(  \max_{g\in\mathrm{G}_{1}%
		}\sup_{h\in\mathrm{H}}\sqrt{T_{n}}\left\vert \frac{\widehat{\phi}\left(
			h,g\right)  -\phi\left(  h,g\right)  }{\xi_{0}\vee\widehat{\sigma}\left(
			h,g\right)  }\right\vert >\tau_{n}\right)  .
	\end{align*}

	By Lemma \ref{lemma.uniform and weak convergence unordered}, $\sqrt{T_{n}%
	}(\widehat{\phi}-\phi)\leadsto\mathbb{G}$ and $\widehat{\sigma}\rightarrow
	\sigma$ almost uniformly, which implies that $\widehat{\sigma}\leadsto\sigma$
	by Lemmas 1.9.3(ii) and 1.10.2(iii) of \citet{van1996weak}. Consequently, by Example
	1.4.7 (Slutsky's lemma) and Theorem 1.3.6 (continuous mapping) of
	\citet{van1996weak}, we have that
	\[
	\max_{g\in\mathrm{G}_{1}}\sup_{h\in\mathrm{H}}\sqrt{T_{n}}\left\vert
	\frac{\widehat{\phi}\left(  h,g\right)  -\phi\left(  h,g\right)  }{\xi_{0}%
		\vee\widehat{\sigma}\left(  h,g\right)  }\right\vert \leadsto\max
	_{g\in\mathrm{G}_{1}}\sup_{h\in\mathrm{H}}\left\vert \frac{\mathbb{G}\left(
		h,g\right)  }{\xi_{0}\vee\sigma\left(  h,g\right)  }\right\vert .
	\]
	Since $\tau_{n}\rightarrow\infty$, we have that $\lim_{n\rightarrow\infty
	}\mathbb{P}(\mathrm{G}_{1}\setminus\widehat{{\mathrm{G}}_{1}}\neq
	\varnothing)=0$.
	
	If $\mathrm{G}_{1}={\mathrm{G}}$, then clearly $\lim_{n\rightarrow\infty
	}\mathbb{P(}\widehat{{\mathrm{G}}_{1}}\setminus\mathrm{G}_{1}\neq
	\varnothing)=0$. Suppose now $\mathrm{G}_{1}\neq{\mathrm{G}}$. Since
	${\mathrm{G}}$ is a finite set and $\widehat{\sigma}$ is uniformly bounded,
	then there is a $\delta>0$ such that
	\[
	\min_{g\in{\mathrm{G}}\setminus\mathrm{G}_{1}}\left\vert \min_{q\in\left\{
		1,\ldots,2^{J}\right\}  }\sup_{h\in\mathrm{H}_{q}}\frac{\phi\left(
		h,g\right)  }{\xi_{0}\vee\widehat{\sigma}\left(  h,g\right)  }\right\vert
	>\delta.
	\]
	By Lemma \ref{lemma.uniform and weak convergence unordered}, $\widehat{\phi
	}\rightarrow\phi$ almost uniformly. Thus, for every $\varepsilon>0$, there is a measurable set $A$ with
	$\mathbb{P}(A)\geq1-\varepsilon$ such that for sufficiently large $n$,
	\begin{align}\label{eq.consequence of almost uniform convergence phi}
	\max_{g\in\mathrm{G}}\left\vert\left\vert
	\min_{q\in\left\{  1,\ldots,2^{J}\right\}  }\sup_{h\in\mathrm{H}_{q}}%
	\frac{\widehat{\phi}\left(  h,g\right)  }{\xi_{0}\vee\widehat{\sigma}\left(
		h,g\right)  }\right\vert -\left\vert \min_{q\in\left\{  1,\ldots,2^{J}\right\}
	}\sup_{h\in\mathrm{H}_{q}}\frac{\phi\left(  h,g\right)  }{\xi_{0}\vee
		\widehat{\sigma}\left(  h,g\right)  }\right\vert \right\vert\le \frac{\delta}{2}%
	\end{align}
	uniformly on $A$. We now have that
	\begin{align*}
		&  \lim_{n\rightarrow\infty}\mathbb{P}\left(  \widehat{{\mathrm{G}}_{1}%
		}\setminus\mathrm{G}_{1}\neq\varnothing\right)  \\
		\leq &  \lim_{n\rightarrow\infty}\mathbb{P}\left(
		\begin{array}
			[c]{c}%
			\left\{  \max_{g\in\widehat{{\mathrm{G}}_{1}}\setminus\mathrm{G}_{1}%
			}\left\vert \min_{q\in\left\{  1,\ldots,2^{J}\right\}  }\sup_{h\in
				\mathrm{H}_{q}}\frac{\phi\left(  h,g\right)  }{\xi_{0}\vee\widehat{\sigma
				}\left(  h,g\right)  }\right\vert >\delta\right\}  \\
			\cap\left\{  \max_{g\in\widehat{{\mathrm{G}}_{1}}\setminus\mathrm{G}_{1}}%
			\sqrt{T_{n}}\left\vert \min_{q\in\left\{  1,\ldots,2^{J}\right\}  }\sup
			_{h\in\mathrm{H}_{q}}\frac{\widehat{\phi}\left(  h,g\right)  }{\xi_{0}%
				\vee\widehat{\sigma}\left(  h,g\right)  }\right\vert \leq\tau_{n}\right\}
			\cap A
		\end{array}
		\right)  +\mathbb{P}(A^{c})\\
		\leq &  \lim_{n\rightarrow\infty}\mathbb{P}\left(  \sqrt{\frac{T_{n}}{n}}%
		\frac{\delta}{2}<\max_{g\in\widehat{{\mathrm{G}}_{1}}\setminus\mathrm{G}_{1}%
		}\sqrt{\frac{{T_{n}}}{{n}}}\left\vert \min_{q\in\left\{  1,\ldots
			,2^{J}\right\}  }\sup_{h\in\mathrm{H}_{q}}\frac{\widehat{\phi}\left(
			h,g\right)  }{\xi_{0}\vee\widehat{\sigma}\left(  h,g\right)  }\right\vert
		\leq\frac{\tau_{n}}{\sqrt{n}}\right)  +\varepsilon=\varepsilon,
	\end{align*}
	because $\tau_{n}/\sqrt{n}\rightarrow0$ as $n\rightarrow\infty$. Here,
	$\varepsilon$ can be arbitrarily small. Thus we have that $\mathbb{P(}%
	\widehat{{\mathrm{G}}_{1}}=\mathrm{G}_{1})\rightarrow1$, because
	$\mathbb{P(}\mathrm{G}_{1}\setminus\widehat{{\mathrm{G}}_{1}}\neq
	\varnothing)\rightarrow0$ and $\mathbb{P(}\widehat{{\mathrm{G}}_{1}}%
	\setminus\mathrm{G}_{1}\neq\varnothing)\rightarrow0$.
	
	Second, suppose $\mathrm{G}_{1}=\varnothing$. This implies that
	\[
	\min_{g\in{\mathrm{G}}}\left\vert \min_{q\in\left\{  1,\ldots,2^{J}\right\}
	}\sup_{h\in\mathrm{H}_{q}}\frac{\phi\left(  h,g\right)  }{\xi_{0}\vee
		\widehat{\sigma}\left(  h,g\right)  }\right\vert >\delta
	\]
	for some $\delta>0$. Thus, with \eqref{eq.consequence of almost uniform convergence phi} we now have that
	\begin{align*}
		&  \lim_{n\rightarrow\infty}\mathbb{P}\left(  \widehat{{\mathrm{G}}_{1}}%
		\neq\varnothing\right)  \\
		\leq &  \lim_{n\rightarrow\infty}\mathbb{P}\left(
		\begin{array}
			[c]{c}%
			\left\{  \max_{g\in\widehat{{\mathrm{G}}_{1}}}\left\vert \min_{q\in\left\{
				1,\ldots,2^{J}\right\}  }\sup_{h\in\mathrm{H}_{q}}\frac{\phi\left(
				h,g\right)  }{\xi_{0}\vee\widehat{\sigma}\left(  h,g\right)  }\right\vert
			>\delta\right\}  \\
			\cap\left\{  \max_{g\in\widehat{{\mathrm{G}}_{1}}}\sqrt{T_{n}}\left\vert
			\min_{q\in\left\{  1,\ldots,2^{J}\right\}  }\sup_{h\in\mathrm{H}_{q}}%
			\frac{\widehat{\phi}\left(  h,g\right)  }{\xi_{0}\vee\widehat{\sigma}\left(
				h,g\right)  }\right\vert \leq\tau_{n}\right\}  \cap A
		\end{array}
		\right)  +\mathbb{P}(A^{c})\\
		\leq &  \lim_{n\rightarrow\infty}\mathbb{P}\left(  \sqrt{\frac{T_{n}}{n}}%
		\frac{\delta}{2}<\max_{g\in\widehat{{\mathrm{G}}_{1}}}\sqrt{\frac{{T_{n}}}%
			{{n}}}\left\vert \min_{q\in\left\{  1,\ldots,2^{J}\right\}  }\sup
		_{h\in\mathrm{H}_{q}}\frac{\widehat{\phi}\left(  h,g\right)  }{\xi_{0}%
			\vee\widehat{\sigma}\left(  h,g\right)  }\right\vert \leq\frac{\tau_{n}}%
		{\sqrt{n}}\right)  +\varepsilon=\varepsilon,
	\end{align*}
	because $\tau_{n}/\sqrt{n}\rightarrow0$ as $n\rightarrow\infty$. Here,
	$\varepsilon$ can be arbitrarily small. Thus, $\mathbb{P}(\widehat
	{{\mathrm{G}}_{1}}={\mathrm{G}}_{1})=1-\mathbb{P}(\widehat{{\mathrm{G}}_{1}%
	}\neq\varnothing)\rightarrow1$.
\end{proof}

Proposition \ref{prop.consistent G hat unordered pairwise} is also related to the contact set estimation in \citet{sun2021ivvalidity}. Since  $\mathrm{G}$ is a finite set, we can obtain the stronger result in Proposition \ref{prop.consistent G hat unordered pairwise}, that is, $\mathbb{P}(\widehat{\mathrm{G}_1}=\mathrm{G}_1)\rightarrow 1$.

\subsubsection{Definition and Estimation of $\mathscr{Z}_2$}

The definition of $\mathscr{Z}_2$ is the same as that in Appendix \ref{sec.Z2 multi ordered} because the necessary conditions provided by \citet{kedagni2020generalized} are for the exclusion and statistical independence conditions. Therefore, the estimator of $\mathscr{Z}_2$ can be constructed as in Section \ref{sec.Z2 multi ordered}.

\section{Additional Simulation Evidence and Application Results}\label{sec.additional simulations}

\subsection{Simulations with Balanced DGPs}\label{sec.simulation balanced}

The simulations in Section \ref{sec.simulation} are constructed based on the empirical example in \citet{kedagni2020discordant}. 
In this section, we modify DGP (0) in Section \ref{sec.simulation} so that the generated data are more balanced. For modified DGP (0), we specify $U\sim\mathrm{Unif}(0,1)$, let $Z=1\{U \le 0.25\}+2\times\{0.25<U \le 0.5\}+3\times\{0.5<U\le 0.75\}+4\times 1\{U>0.75\}$, and keep everything else the same as in Section \ref{sec.simulation}. Tables \ref{tab:DGP0Balanced} and \ref{tab:DGP0InferenceBalanced} present the results for DPG (0) with balanced data. For $c=0.6$, the selection rates are high and are increasing as the sample size increases. The coverage rates are converging to $95\%$.

\begin{table}[h!]
	
	\centering
	\caption{Validity Pair Set Estimation for DGP (0) with Balanced Data}
	\scalebox{0.9}{
		\begin{tabular}{  c  c  c  c  c  c  c  c  }
			\hline
			\hline
$n$  &	$c $ & (1, 2) & {(1, 3)} & {(1, 4)} 
			 & {{(2, 3)}} & {(2, 4)}
			 & {{(3, 4)}} \\
			\hline
\multirow{10}{*}{1230}   &  0.1 & 0.000  & 0.000  & 0.000  & 0.000  & 0.000  & 0.000  \\ 
    &    0.2 & 0.000  & 0.000  & 0.000  & 0.000  & 0.000  & 0.000  \\ 
    &    0.3 & 0.000  & 0.000  & 0.000  & 0.000  & 0.000  & 0.000  \\ 
    &    0.4 & 0.000  & 0.001  & 0.035  & 0.000  & 0.001  & 0.000  \\ 
    &    0.5 & 0.296  & 0.180  & 0.562  & 0.335  & 0.532  & 0.293  \\ 
    &    0.6 & 0.826  & 0.827  & 0.975  & 0.883  & 0.951  & 0.852  \\ 
    &    0.7 & 0.981  & 0.984  & 0.998  & 0.995  & 0.996  & 0.987  \\ 
    &    0.8 & 0.997  & 0.997  & 1.000  & 1.000  & 0.999  & 1.000  \\ 
    &    0.9 & 1.000  & 0.999  & 1.000  & 1.000  & 1.000  & 1.000  \\ 
    &    1 & 1.000  & 1.000  & 1.000  & 1.000  & 1.000  & 1.000  \\ 
        \hline
\multirow{10}{*}{2460}        
    &    0.1 & 0.000  & 0.000  & 0.000  & 0.000  & 0.000  & 0.000  \\ 
    &    0.2 & 0.000  & 0.000  & 0.000  & 0.000  & 0.000  & 0.000  \\ 
    &    0.3 & 0.000  & 0.000  & 0.000  & 0.000  & 0.000  & 0.000  \\ 
    &    0.4 & 0.005  & 0.002  & 0.280  & 0.010  & 0.013  & 0.003  \\ 
    &    0.5 & 0.585  & 0.552  & 0.947  & 0.650  & 0.741  & 0.602  \\ 
    &    0.6 & 0.968  & 0.958  & 0.997  & 0.981  & 0.993  & 0.972  \\ 
    &    0.7 & 1.000  & 0.995  & 1.000  & 1.000  & 0.999  & 0.998  \\ 
    &    0.8 & 1.000  & 1.000  & 1.000  & 1.000  & 1.000  & 1.000  \\ 
    &    0.9 & 1.000  & 1.000  & 1.000  & 1.000  & 1.000  & 1.000  \\ 
    &    1 & 1.000  & 1.000  & 1.000  & 1.000  & 1.000  & 1.000  \\ 
 \hline
			\hline
		\end{tabular}
	}
	
	\label{tab:DGP0Balanced}
\end{table}
\begin{table}[h!]
	
	\centering
	\caption{Coverage Rates of the Confidence Intervals for DGP (0) with Balanced Data}
	\scalebox{0.9}{
		\begin{tabular}{  c  c  c  c  c  c  c  c  }
			\hline
			\hline
$n$  &	$c $ & (1, 2) & {(1, 3)} & {(1, 4)} 
			 & {{(2, 3)}} & {(2, 4)}
			 & {{(3, 4)}} \\
			\hline
\multirow{10}{*}{1230}   
    &    0.1 & 1.000  & 1.000  & 1.000  & 1.000  & 1.000  & 1.000  \\ 
    &    0.2 & 1.000  & 1.000  & 1.000  & 1.000  & 1.000  & 1.000  \\ 
    &    0.3 & 1.000  & 1.000  & 1.000  & 1.000  & 1.000  & 1.000  \\ 
    &    0.4 & 1.000  & 0.999  & 1.000  & 1.000  & 1.000  & 1.000  \\ 
    &    0.5 & 0.990  & 0.995  & 0.980  & 0.990  & 0.970  & 0.984  \\ 
    &    0.6 & 0.969  & 0.970  & 0.960  & 0.962  & 0.945  & 0.964  \\ 
    &    0.7 & 0.957  & 0.965  & 0.958  & 0.957  & 0.945  & 0.957  \\ 
    &    0.8 & 0.956  & 0.964  & 0.958  & 0.957  & 0.945  & 0.957  \\ 
    &    0.9 & 0.956  & 0.964  & 0.958  & 0.957  & 0.945  & 0.957  \\ 
    &    1 & 0.956  & 0.964  & 0.958  & 0.957  & 0.945  & 0.957  \\ 
        
        \hline
\multirow{10}{*}{2460}        
    &    0.1 & 1.000  & 1.000  & 1.000  & 1.000  & 1.000  & 1.000  \\ 
    &    0.2 & 1.000  & 1.000  & 1.000  & 1.000  & 1.000  & 1.000  \\ 
    &    0.3 & 1.000  & 1.000  & 1.000  & 1.000  & 1.000  & 1.000  \\ 
    &    0.4 & 0.999  & 1.000  & 0.983  & 0.999  & 0.999  & 1.000  \\ 
    &    0.5 & 0.977  & 0.973  & 0.955  & 0.980  & 0.966  & 0.979  \\ 
    &    0.6 & 0.955  & 0.947  & 0.949  & 0.967  & 0.951  & 0.962  \\ 
    &    0.7 & 0.952  & 0.945  & 0.949  & 0.967  & 0.951  & 0.962  \\ 
    &    0.8 & 0.952  & 0.945  & 0.949  & 0.967  & 0.951  & 0.962  \\ 
    &    0.9 & 0.952  & 0.945  & 0.949  & 0.967  & 0.951  & 0.962  \\ 
    &    1 & 0.952  & 0.945  & 0.949  & 0.967  & 0.951  & 0.962  \\ 
 \hline
			\hline
		\end{tabular}
	}
	
	\label{tab:DGP0InferenceBalanced}
\end{table}

\subsection{Local Violations}
In this section, we present results for local violations of the testable implications. The DGPs (5) and (6) are designed based on DGP (1) in Section \ref{sec.simulation}. The degree to which the testable implications are violated decreases from DGP (1), to DGP (5), and to DGP (6).

\begin{enumerate}[start=5,label=(\arabic*):]
    \item  For $(z_1,z_2)\in\mathscr{Z}_P$, $N_{1z_1}\sim N(\mu_{(z_1,z_2)},1)$, $N_{1z_2}\sim N(0,1)$, $N_{0z_1}\sim N(0,1)$, $N_{0z_2}\sim N(\mu_{(z_1,z_2)},1)$ with $\mu_{(1,2)}=-0.6$, $\mu_{(1,3)}=-0.8$, $\mu_{(1,4)}=-1$, $\mu_{(2,3)}=-0.6$, $\mu_{(2,4)}=-0.8$, $\mu_{(3,4)}=-0.6$, $N_{dz}\sim N(0,1)$ for $d\in\{0,1\}$ and $z\in\{1,2,3,4\}\setminus \{z_1,z_2\}$, 
    $Y=\sum_{z=1}^{4}1\{Z=z\}\times(\sum_{d=0}^{1} 1\{D=d\}\times N_{dz})$

    \item  For $(z_1,z_2)\in\mathscr{Z}_P$, $N_{1z_1}\sim N(\mu_{(z_1,z_2)},1)$, $N_{1z_2}\sim N(0,1)$, $N_{0z_1}\sim N(0,1)$, $N_{0z_2}\sim N(\mu_{(z_1,z_2)},1)$ with $\mu_{(1,2)}=-0.3$, $\mu_{(1,3)}=-0.5$, $\mu_{(1,4)}=-0.7$, $\mu_{(2,3)}=-0.3$, $\mu_{(2,4)}=-0.5$, $\mu_{(3,4)}=-0.3$, $N_{dz}\sim N(0,1)$ for $d\in\{0,1\}$ and $z\in\{1,2,3,4\}\setminus \{z_1,z_2\}$, 
    $Y=\sum_{z=1}^{4}1\{Z=z\}\times(\sum_{d=0}^{1} 1\{D=d\}\times N_{dz})$
    
\end{enumerate}

Tables \ref{tab:DGP1(1)} and \ref{tab:DGP1(2)} present the RMSEs of $\sqrt{n}(\widehat{\beta}_{(k,k')}^1-{\beta}_{(k,k')}^1)$ for DGPs (5) and (6) and our preferred choice of tuning parameter with $c=0.6$. As the violations become smaller and thus harder to detect, the RMSEs are getting higher. (Note that the RMSEs of $\sqrt{n}(\widehat{\beta}_{(k,k')}^1-{\beta}_{(k,k')}^1)$ may not be decreasing in the sample size due to the rescaling by $\sqrt{n}$.) 

\begin{table}[h!]
	
	\centering
	\caption{RMSEs for DGP (5)}
	\scalebox{0.9}{
		\begin{tabular}{  c  c  c  c  c  c  c  c  c  c  c  c  c  c }
			\hline
			\hline
$n$  &	$c $ & (1, 2) & {(1, 3)} & {(1, 4)} 
			 & {{(2, 3)}} & {(2, 4)} 
			 & {{(3, 4)}} \\
			\hline
{1230}   
    &    0.6 & 0.000  & 4.147  & 1.645  & 5.005  & 1.294  & 11.080  \\ 
        
    \hline
{2460}    
    &    0.6 & 0.000  & 4.235  & 1.832  & 5.859  & 2.221  & 14.249  \\ 

 \hline
			\hline
		\end{tabular}
	}
	
	\label{tab:DGP1(1)}
\end{table}
\begin{table}[h!]
	
	\centering
	\caption{RMSEs for DGP (6)}
	\scalebox{0.9}{
		\begin{tabular}{  c  c  c  c  c  c  c  c  c  c  c  c  c  c }
			\hline
			\hline
$n$  &	$c $ & (1, 2) & {(1, 3)} & {(1, 4)} 
			 & {{(2, 3)}} & {(2, 4)} 
			 & {{(3, 4)}} \\
			\hline
{1230}  
    &    0.6 & 1.380  & 8.007  & 2.798  & 9.132  & 3.317  & 14.608  \\ 
        
    \hline
{2460}   
    &    0.6 & 4.074  & 12.944  & 4.123  & 14.471  & 4.807  & 18.312  \\ 

 \hline
			\hline
		\end{tabular}
	}
	
	\label{tab:DGP1(2)}
\end{table}

\subsection{Results from Test of IV Validity} \label{sec.validity set estimation using tests}

In Section \ref{sec.application}, we estimate the validity pair set using the proposed method in the paper. In this section, we apply the approaches of \citet{kitagawa2015test} and \citet{sun2021ivvalidity} to test the validity for each pair in $\mathscr{Z}_P$ in the empirical example of \citet{heckman2001four} and \citet{kedagni2020discordant}. We may follow \citet{kitagawa2015test} and \citet{sun2021ivvalidity} to obtain the $p$-values for each pair of the values of $Z$ based on the testable implications for every pair. 

Table \ref{tab:validity set estimation using test} shows the $p$-values from the test of \citet{sun2021ivvalidity} for each pair in $\mathscr{Z}_P$ in the empirical application with different values of trimming parameter $\xi$ ($\bar{v}_{\xi}$ is the probability
measure that assigns equal probabilities (weights) to the values of $\xi$). 
We only reject the validity of the first pair at the 10\% level. That is, the test for IV validity is less effective at removing invalid pairs than VSIV estimation in this application. One possible reason for this result is that we specifically choose $c$ based on simulation results so that VSIV estimation has a high power for excluding invalid pairs in small samples. 

\begin{table}[H]
	
	\centering
	\caption{$p$-values from Test of IV Validity}
	\scalebox{1}{
		\begin{tabular}{  c  c  c  c  c  c  c  c  c  c  c  c  c  c }
			\hline
			\hline
	$\xi $ & (1, 2) & {(1, 3)} & {(1, 4)} 
			 & {{(2, 3)}} & {(2, 4)} 
			 & {{(3, 4)}} \\
			\hline
0.07 & 0.070  & 0.730  & 1.000  & 0.996  & 1.000  & 0.998  \\ 
        0.1 & 0.081  & 0.743  & 1.000  & 0.996  & 1.000  & 0.998  \\ 
        0.13 & 0.081  & 0.743  & 1.000  & 0.996  & 1.000  & 0.998  \\ 
        0.16 & 0.081  & 0.743  & 1.000  & 0.996  & 1.000  & 0.998  \\ 
        0.19 & 0.081  & 0.743  & 1.000  & 0.996  & 1.000  & 0.998  \\ 
        0.22 & 0.081  & 0.743  & 1.000  & 0.996  & 1.000  & 0.998  \\ 
        0.25 & 0.081  & 0.743  & 1.000  & 0.996  & 1.000  & 0.998  \\ 
        0.28 & 0.081  & 0.743  & 1.000  & 0.996  & 1.000  & 0.998  \\ 
        0.3 & 0.081  & 0.743  & 1.000  & 0.996  & 1.000  & 0.998  \\ 
        1 & 0.081  & 0.743  & 1.000  & 0.996  & 1.000  & 0.998  \\ 
        $\bar{v}_{\xi}$ & 0.078  & 0.742  & 1.000  & 0.996  & 1.000  & 0.998  \\

 \hline
			\hline
		\end{tabular}
	}
	
	\label{tab:validity set estimation using test}
\end{table}

\putbib    
\end{bibunit}

\end{document}